\documentclass[a4paper, fontsize=10pt, twoside=true, numbers=noenddot]{kaobook}

\ifxetexorluatex
	\usepackage{polyglossia}
	\setmainlanguage{english}
\else
	\usepackage[english]{babel} 
\fi
\usepackage[english=british]{csquotes}	
\usepackage{amsmath,nicefrac,bbm,kaobiblio,xcolor,kaorefs,physics,xfrac}
\usepackage{caption}
\usepackage{enumitem}
\usepackage{adjustbox}
\usepackage{wrapfig}
\usepackage{graphicx}
\usepackage{pdfpages}

\DeclareUnicodeCharacter{0301}{'}
\DeclareUnicodeCharacter{0302}{\^{}}
\newcommand{\ot}{\otimes}

\definecolor{definitioncolor}{RGB}{240,240,255}
\definecolor{theoremcolor}{RGB}{200, 218, 191}
\definecolor{lemmacolor}{RGB}{233, 190, 190}
\definecolor{observationcolor}{RGB}{255, 232, 205}
\definecolor{propositioncolor}{RGB}{217, 213, 209}
\definecolor{corollarycolor}{RGB}{194, 185, 211}
\definecolor{kaoboxcolor}{RGB}{238,237,238}
\definecolor{equationcolor}{RGB}{222,94,100}
\definecolor{citationcolor}{RGB}{94,102,228}

\hypersetup{
    linkcolor=equationcolor,
    citecolor=equationcolor,
    filecolor=black,
    urlcolor=equationcolor
}

\newcommand{\horrule}[1]{\noindent\rule{\linewidth}{#1}}

\newcommand{\ms}[1]{\textsf{#1}}
\DeclareMathAlphabet{\mycal}{OMS}{cmsy}{m}{n}
\newcommand{\TT}{\mycal{T}}
\newcommand{\FF}{\mycal{F}}
\newcommand{\OO}{\mycal{O}}
\newcommand{\D}{\mathcal{D}}
\newcommand{\T}{\mathcal{T}}
\newcommand{\N}{\mathcal{N}}
\newcommand{\E}{\mathcal{E}}
\renewcommand{\S}{\mathcal{S}}
\def\P{ {\cal P} } 
\newcommand{\iden}{\mathbbm{1}}
\renewcommand{\v}[1]{\ensuremath{\boldsymbol #1}}
\DeclareFontFamily{U}{mathb}{\hyphenchar\font45}
\DeclareFontShape{U}{mathb}{m}{n}{
	<-6> mathb5 <6-7> mathb6 <7-8> mathb7
	<8-9> mathb8 <9-10> mathb9
	<10-12> mathb10 <12-> mathb12
}{}
\DeclareSymbolFont{mathb}{U}{mathb}{m}{n}
\DeclareMathSymbol{\llcurly}{\mathrel}{mathb}{"CE}
\DeclareMathSymbol{\ggcurly}{\mathrel}{mathb}{"CF}

\usepackage[framed=true]{kaotheorems}

\begin{document}

\frontmatter 


\includepdf{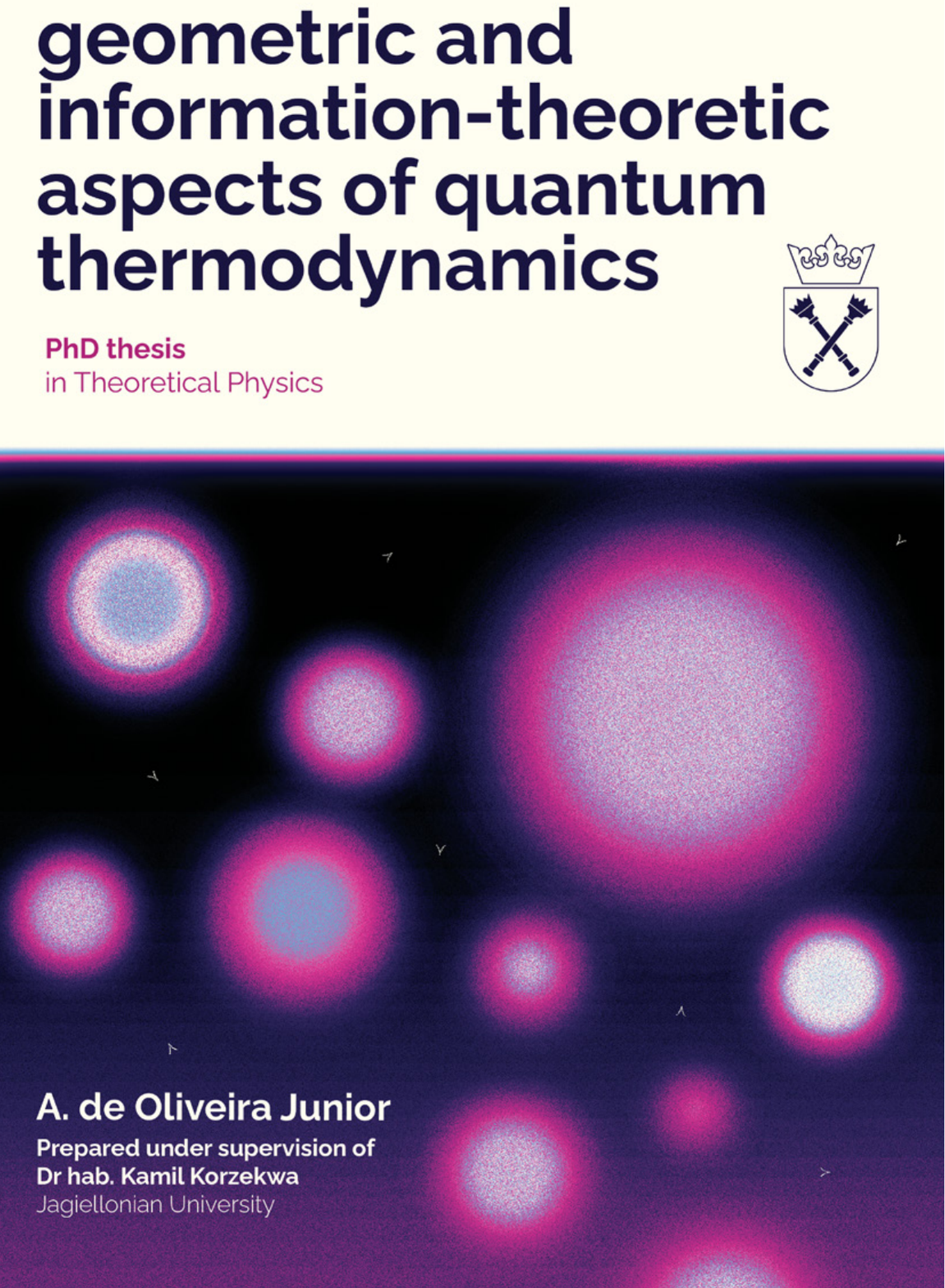}

\begin{titlepage}
    \begin{center}
        \normalfont \normalsize 
        \textsc{Jagiellonian University in Cracow \\ Faculty of Physics, Astronomy and Applied Computer Science} \\ [5pt]
        \begin{figure}[H]
            \centering
            \includegraphics[width=78.475mm]{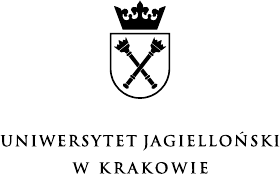}
        \end{figure}
        \Large Alexssandre de Oliveira Junior \\
        \horrule{0.5pt} \\[0.4cm]
        {\fontsize{24}{39}\selectfont Geometric and information-theoretic aspects \\  of quantum thermodynamics} \\[5pt]
        \normalsize PhD thesis \\ in Theoretical Physics
        \horrule{2pt} \\[0.5cm] 

    \end{center}
    \vfill
    \flushright{Prepared under supervision of \\ Dr hab. Kamil Korzekwa \\ Jagiellonian University \\}
    \begin{center}
        Kraków, October 11, 2023
    \end{center}
\end{titlepage}

\newpage
\index{Cover}

\begin{center}
\Large Oświadczenie    
\end{center}

Ja niżej podpisany Alexssandre de Oliveira Junior, doktorant Wydziału Fizyki, Astronomii i Informatyki Stosowanej Uniwersytetu Jagiellońskiego oświadczam, że przedłożona przeze mnie rozprawa doktorska pt. “Geometric and information-theoretic aspects of quantum thermodynamics” jest oryginalna i przedstawia wyniki badań wykonanych przeze mnie osobiście, pod kierunkiem dr. hab. Kamila Korzekwy. Pracę napisałem samodzielnie. Oświadczam, że moja rozprawa doktorska została opracowana zgodnie z Ustawą o prawie autorskim i prawach pokrewnych z dnia 4 lutego 1994 r. (Dziennik Ustaw 1994 nr 24 poz. 83 wraz z późniejszymi zmianami). Jestem świadom, że niezgodność niniejszego oświadczenia z prawdą ujawniona w dowolnym czasie, niezależnie od skutków prawnych wynikających z ww. ustawy, może spowodować unieważnienie stopnia nabytego na podstawie tej rozprawy.

\vspace{1.5cm}
\begin{minipage}{2in}
Kraków, 10 października 2023
\end{minipage}
\hfill
\newpage
\index{Oświadczenie}

\topskip0pt
\vspace*{\fill}

\begin{center}
\Large Streszczenie    
\end{center}

Termodynamika i mechanika kwantowa reprezentują dwa fundamentalne paradygmaty, które ewoluowały w potężne i skuteczne teorie naukowe, zdolne do opisywania ogromnej gamy zjawisk fizycznych z niezwykłą precyzją. Termodynamika ma swoje korzenie w XIX wieku, pochodzące z prób zrozumienia silników parowych podczas rewolucji przemysłowej. Natomiast mechanika kwantowa wyłoniła się z argumentów dotyczących natury promieniowania ciała doskonale czarnego. Przez wiele dekad te dwie teorie rozwijały się niezależnie. Termodynamika skupiała się na opisie makroskopowych właściwości układów, nie zagłębiając się w szczegóły mikroskopowe, podczas gdy mechanika kwantowa koncentrowała się na studiowaniu atomów i cząstek subatomowych. W miarę jak termodynamika stopniowo poszerzała swoje granice w celu zbadania układów mikroskopowych, mechanika kwantowa torowała drogę dla nowego obszaru badań skoncentrowanego na koncepcji, że układy kwantowe mogą być wykorzystywane do zadań obliczeniowych i informacyjnych bardziej efektywnie niż klasyczne. W pewnym momencie odkryto także, że termodynamika narzuca fizyczne ograniczenia na przetwarzanie informacji. Konwergencja tych dwóch paradygmatów nastąpiła, gdy naukowcy zaczęli badać rzeczywistość kwantową, starając się zrozumieć prawa termodynamiczne rządzące układami kwantowymi.

W tej pracy badam różne aspekty jednego z fundamentalnych pytań termodynamiki: jakie przemiany stanu mogą przechodzić układy kwantowe podczas oddziaływania z łaźnią termiczną przy określonych ograniczeniach? Ograniczenia te mogą dotyczyć zachowania całkowitej energii, efektów pamięciowych lub uwzględnienia skończonego rozmiaru układu. Korzystając z minimalnych założeń na temat wspólnej dynamiki układu i kąpieli, tj. że układ złożony jest zamknięty i ewoluuje unitarnie zachowując energię, wywodzę jawną konstrukcję zbioru wszystkich stanów, do których układ w danym stanie początkowym może ewoluować termodynamicznie lub z których może ewoluować. Pozwala to na scharakteryzowanie i zrozumienie struktury termodynamicznej strzałki czasu. Taka konstrukcja jest ogólna i opiera się jedynie na założeniu o zachowaniu całkowitej energii i termiczności kąpieli. W rezultacie, kolejne pytanie jaka dynamika może być obserwowana, gdy efekty pamięci pomiędzy układem a kąpielą są niezaniedbywalne, a także jak procesy termodynamiczne zależą od efektów pamięci, skłoniło mnie do opracowania formalizmu łączącego procesy termodynamiczne bez pamięci z dowolnie niemarkowskimi. To z kolei rzuca światło na sposób kwantyfikacji efektów pamięci w protokołach termodynamicznych. Wreszcie, badam także problem charakteryzacji optymalnych przemian termodynamicznych, gdy fluktuacje wokół wielkości termodynamicznych są porównywalne ze średnimi. Prowadzi mnie to do zbadania wystarczających i koniecznych warunków leżących u podstaw przemian termodynamicznych układów o skończonym rozmiarze. Jako pierwszy charakteryzuję optymalne przemiany dla ważnej klasy procesów termodynamicznych układów o skończonym rozmiarze przygotowanych w superpozycji różnych stanów energetycznych. Co więcej, dowodzę również zależność między fluktuacjami energii swobodnej stanu początkowego układu a minimalną ilością energii swobodnej rozpraszanej podczas procesu. Formułuję więc słynną relację fluktuacji-dyssypacji w języku informacji kwantowej.

Ostatnia część tej pracy skupia się na badaniu zjawiska wszechobecnego w nauce, w szczególności w termodynamice, a mianowicie zjawiska katalizy. Polega ona na wykorzystaniu dodatkowego układu (katalizatora) do umożliwienia procesów, które w przeciwnym razie byłyby niemożliwe. Przez ostatnie dwie dekady ta koncepcja rozprzestrzeniła się w dziedzinie fizyki kwantowej, efekt ten jednak zwykle opisywany jest w ramach wysoce abstrakcyjnego formalizmu. Pomimo swoich sukcesów, podejście to ma trudności z pełnym uchwyceniem zachowania fizycznie realizowalnych układów, co ogranicza praktyczną stosowalność katalizy kwantowej. Nasuwa się więc pytanie czy kataliza kwantowa może wyjść poza teorię i wejść w kontekst praktyczny? Innymi słowy, jak można przekształcić koncepcję katalizy kwantowej z czysto teoretycznej idei w narzędzie, które można zastosować w praktyce? Pokażę w jaki sposób można to osiągnąć w paradygmatycznym układzie optyki kwantowej, a mianowicie w modelu Jaynesa-Cummingsa, w którym atom oddziałuje z wnęką optyczną. Atom odgrywa tu rolę katalizatora i pozwala na deterministyczne generowanie nieklasycznych stanów światła we wnęce, co potwierdzone jest statystyką sub-Poissonowską lub negatywnością Wignera.

\newpage
\begin{center}
\Large Abstract    
\end{center}

Thermodynamics and quantum mechanics represent two fundamental paradigms that have evolved into powerful and successful scientific theories, capable of describing a vast array of physical phenomena with remarkable precision. Thermodynamics has its roots in the 19th century, originating from efforts to understand steam engines during the Industrial Revolution. Quantum mechanics, on the other hand, emerged from consistency arguments concerning the nature of black-body radiation. For many decades, these two theories developed independently. Thermodynamics concentrated on describing the macroscopic properties of systems without delving into microscopic details, while quantum mechanics focused on the study of atoms and subatomic particles. As thermodynamics gradually pushed its boundaries to explore microscopic systems, quantum mechanics paved the way for a new research field centred on the notion that quantum systems could be harnessed for computational and informational tasks. In due course, it was discovered that thermodynamics imposed physical constraints on information processing. The convergence of these two paradigms occurred when scientists began probing the quantum realm, seeking to understand the thermodynamic laws governing quantum systems.

In this thesis, I investigate various aspects of one of the most fundamental questions in thermodynamics: \emph{what state transformations can quantum systems undergo while interacting with a thermal bath under specific constraints?} These constraints may involve total energy conservation, memory effects, or finite-size considerations. Using minimal assumptions on the joint system-bath dynamics, namely that the composite
system is closed and evolves unitarily via an energy-preserving unitary, I will derive an explicit construction of the set of states to which a given state can thermodynamically evolve to or evolve from. This allows one to characterise and understand the structure of the thermodynamic arrow of time. Such a construction is general and relies only on the assumption of total energy conservation and the thermality of the bath. As a result, a follow-up question \emph{what dynamics are observed when memory effects between the system and bath are non-negligible}, as well as \emph{how thermodynamic processes are affected by memory effects}, will lead to development of a framework bridging the gap between memoryless and arbitrarily non-Markovian thermodynamic processes. This, in turn, sheds light on how to quantify the role played by memory effects in thermodynamic protocols. Next, to understand how optimal thermodynamic processing is affected when one goes beyond the thermodynamic limit -- where fluctuations around thermodynamic quantities are comparable to averages -- I ask \emph{what are the necessary and sufficient conditions for the existence of thermodynamic transformations between different non-equilibrium states of few-particle systems}. The answer to such a question will lead to the necessary and sufficient conditions underlying thermodynamic transformations of finite-size systems and the characterisation of thermodynamic processes from finite-size systems that may be in superposition of different energy states. What is more, I will also prove a relation between the free energy fluctuations of the initial state of the system and the minimal amount of free energy dissipated during the process. This allows for the formulation of the famous fluctuation-dissipation relations within a quantum information framework.

Finally, the last part of this thesis focuses on studying a ubiquitous phenomenon in science and, in particular, in thermodynamics, the so-called \emph{catalysis}. It consists of using an auxiliary system (a catalyst) to enable processes that would otherwise be impossible. Over the last two decades, this notion has spread to the field of quantum physics. However, this effect is typically described within a highly abstract framework. Despite its successes, this approach struggles to fully capture the behavior of physically realisable systems, thereby limiting the applicability of quantum catalysis in practical scenarios. This gives rise to the question: \emph{what if quantum catalysis could go beyond theory and step into practical context?} In other words, \emph{how can one translate the concept of quantum catalysis from being a purely theoretical notion to a tool that can be practically implemented}. Strikingly, I will show this effect in a paradigmatic quantum optics setup, namely the Jaynes-Cummings model, where an atom interacts with an optical cavity. The atom plays the role of the catalyst, and allows for the deterministic generation of non-classical light in the cavity as witnessed by sub-Poissonian statistics or Wigner negativity.

\vspace*{\fill}
\newpage
\index{abstract}

\begin{center}
\Large Acknowledgements 
\end{center}

I would like to express my heartfelt gratitude to my supervisor, Kamil Korzekwa, for the unwavering support, trust, intellectual freedom, and exceptional guidance. Above all, our discussions and your friendship have been among the most significant aspects of the past three years. I was very fortunate to be guided and mentored by Prof. Karol $\dot{Z}$yczkowski and Prof. Michał Horodecki. My sincere appreciation goes to both, for their physical insights, countless interesting questions, and mentorship. 

This thesis is indebted to the remarkable collaborations I have had in recent years. I am immensely thankful to one of my closest collaborators, Jakub Czartowski. It is through our endless discussions that numerous projects and ideas were born, laying the foundation for this thesis. A special thanks also goes to Tanmoy Biswas, with whom I had the pleasure of collaborating. We both began our PhDs together and immediately tackled a very challenging problem. Our countless Skype calls and meetings about quantum thermo would surely set a Guinness World Record. 

Throughout my PhD, I met incredible people who stood by me through my PhD meltdowns and moments of joy. Without Aritra Sinha, my first friend in Poland, these past three years would have been much sadder and boring. Our daily morning coffees became a cherished ritual. I am thankful for all the discussions I had with my incredible office mates: Gerard Angles Munné, for being so supportive and the funny one (say tak tak tak), together we discovered Krakow and collected stories that will be always with me; Albert Rico for your incredible patience and for pushing me beyond my comfort zone (I always felt that we were so similar, yet so different); Moises Bermejo Morán, for your amazing friendship and for being the type of physicist I aspire to be (you went beyond the probability simplex); and Steffen Kessler, the chillest person I ever met. Thank you for all the laughs and for arguing on my behalf with the other three about turning on the AC (we never found stk on geocaching).

I was very lucky to have met Jake Xuereb in my last year of my PhD. His inspiring way of looking at physics (and problems) and constantly buzzing me with quantum thermo puzzles are a source of inspiration.

Furthermore, I extend my gratitude to all members of the Jagiellonian Quantum Information Team. A special mention goes to Felix Huber; thank you for the stimulating discussions, memorable times over beers and climbs, and most importantly, for your friendship. Oliver Reardon-Smith, you helped me with so many things that I don't even know how to start -- from enlightening chats about physics to helping me with English grammar (and for fixing the bibliography style of this thesis). Your patience has been invaluable. Not to mention, your cakes are simply the best! Fereshte Shahbeigi, you are my go-to whenever a query about quantum channels arises. Roberto Salazar, our collaborations and discussions have been greatly rewarding. Grzegorz Rajchel-Mieldzioć, or Grzeg, for all the nice talks, discussions, beers, and for buying me a Guarana during my first 2 months in Poland. Lastly, to Korand Szymanski, for introducing and helping me with this template and for the very pleasant conversations.

I would also like to take the opportunity to deeply thank Martí Perarnau Llobet, for inviting me to spend time with his group in Geneva. I am thankful for the engaging discussions, insightful questions, and for granting me the opportunity to explore new horizons and go beyond the boundaries of my research (I hope to join you in the next l'Escalade). Many thanks to Nicolas Brunner for all the very interesting discussions and for supporting my stay in Geneva. I am grateful and indebted to Patryk Lipka-Bartosik. Thanks for guiding me throughout our project, and for your incredible patience. I have learned a lot from you three.

Obrigado aos meus amigos do Brasil. Arthur Faria, que começou a estudar termodinâmica comigo, por todas as discussões, conversas e orientações (se tivéssemos feito doutorado no mesmo lugar, com certeza teríamos um dinossauro colado na parede do nosso escritório); ao grande Carlos Eduardo 'Dudu', meu amigo de infância que acompanhou a minha jornada, me escutou e me deu grandes conselhos (de guitarristas de uma banda punk a um físico e um filósofo, hein?). Vocês são parte disso e essenciais.

Finalmente, eu gostaria de agradecer à minha família por todo o amor, carinho e apoio. 
Para você, Gabriel, eu poderia agradecer por me ajudar a escolher a paleta de cores desta tese e por quase 70\% das imagens aqui, mas eu decido agradecer pela sua amizade (e não, não foi pelo the big bang theory!). Por fim, muito obrigado Kiara, pelo incondicional suporte, por ser minha melhor amiga e minha companheira. 

\newpage
\index{acknowledgments}

\section*{List of publications}

\subsection*{Publications and preprints completed during PhD studies}

The following works cover a wide range of topics, with a primary focus on quantum thermodynamics and quantum information theory. They are presented below in reverse chronological order:

\begin{enumerate}[label={[\arabic*]}]
    
    \item \textbf{Quantum catalysis in cavity QED}, A. de Oliveira Junior, M. Perarnau-Llobet, N. Brunner, P. Lipka-Bartosik, 
    {\href{https://arxiv.org/abs/2305.19324}{\color{teal}{\textit{\textit{arXiv:2305.19324}}} (2023).}}
    \item \textbf{Thermal recall: Memory-assisted Markovian thermal processes\footnote{\label{statement} JCz and AOJ contributed equally to this work.}}, J. Czartowski, A. de Oliveira Junior, K. Korzekwa, {\href{https://journals.aps.org/prxquantum/abstract/10.1103/PRXQuantum.4.040304}{\color{teal}\textit{PRX Quantum \textbf{4}, 040304} (2023).}}
    \item \textbf{Geometric structure of thermal cones}, A. de Oliveira Junior, J. Czartowski, K. Życzkowski, K. Korzekwa, {\href{https://journals.aps.org/pre/abstract/10.1103/PhysRevE.106.064109}{\color{teal}\textit{Phys. Rev. E} \textbf{106}, 064109 (2022).}}
    \item \textbf{Unravelling the non-classicality role in Gaussian heat engines}, A. de Oliveira Junior, MC de Oliveira, {\href{https://www.nature.com/articles/s41598-022-13811-z}{\color{teal} \textit{Sci. Rep.} \textbf{12}, 10412 (2022).}}
    \item \textbf{Resource theory of Absolute Negativity}, R. Salazar, J. Czartowski, A. de Oliveira Junior, {\href{https://arxiv.org/abs/2205.13480}{\color{teal} \textit{\textit{arXiv:2205.13480}} (2022).}}
    \item \textbf{Fluctuation-dissipation relations for thermodynamic distillation processes}, T. Biswas, A. de Oliveira Junior, M. Horodecki, K. Korzekwa, {\href{https://journals.aps.org/pre/abstract/10.1103/PhysRevE.105.054127}{\color{teal}\textit{Phys. Rev. E} \textbf{105}, 054127 (2022).}}
    \item \textbf{Machine classification for probe-based quantum thermometry}, F.S. Luiz, A. de Oliveira Junior, F.F. Fanchini, G.T. Landi, {\href{https://journals.aps.org/pra/abstract/10.1103/PhysRevA.105.022413}{\color{teal}\textit{Phys. Rev. A} \textbf{105}, 022413 (2022).}}
    \item \textbf{Spin-orbit implementation of Solovay-Kitaev decomposition of single-qubit channel}, M.H.M. Passos, A. de Oliveira Junior, M.C. de Oliveira, A.Z. Khoury, J.A.O. Huguenin, {\href{https://journals.aps.org/pra/abstract/10.1103/PhysRevA.102.062601}{\color{teal}\textit{Phys. Rev. A} \textbf{102}, 062601, (2020).}}
    \item \textbf{Ultra-fast Kinematic Vortices in Mesoscopic Superconductors: The Effect of the Self-Field}, L.R. Cadorim, A. de Oliveira Junior, E. Sardella,  {\href{https://www.nature.com/articles/s41598-020-75748-5}{\color{teal}\textit{Sci. Rep.} \textbf{10}, 18662 (2020).}}
\end{enumerate}

To maintain thematic coherence and a clear narrative line, this thesis is primarily composed of four published (accepted) papers:
\begin{itemize}
    \item \textbf{Geometric structure of thermal cones}, A. de Oliveira Junior, J. Czartowski, K. Życzkowski, K. Korzekwa, {\href{https://journals.aps.org/pre/abstract/10.1103/PhysRevE.106.064109}{\color{teal}\textit{Phys. Rev. E} \textbf{106}, 064109 (2022).}}
    \item \textbf{Thermal recall: Memory-assisted Markovian thermal processes\footnotemark[\value{footnote}]}, J. Czartowski, A. de Oliveira Junior, K. Korzekwa, {\href{https://journals.aps.org/prxquantum/abstract/10.1103/PRXQuantum.4.040304}{\color{teal}\textit{PRX Quantum \textbf{4}, 040304} (2023).}}
    \item \textbf{Fluctuation-dissipation relations for thermodynamic distillation processes}, T. Biswas, A. de Oliveira Junior, M. Horodecki, K. Korzekwa, {\href{https://journals.aps.org/pre/abstract/10.1103/PhysRevE.105.054127}{\color{teal}\textit{Phys. Rev. E} \textbf{105}, 054127 (2022).}}
    \item \textbf{Quantum catalysis in cavity QED}, A. de Oliveira Junior, M. Perarnau-Llobet, N. Brunner, P. Lipka-Bartosik, 
    {\href{https://arxiv.org/abs/2305.19324}{\color{teal}{\textit{\textit{Phys. Rev. Res \v{X}. XXXX}}} (2023).}}
\end{itemize}

\index{list-of-publication}

\let\cleardoublepage\clearpage


\begingroup 
\pagenumbering{gobble}
\setlength{\textheight}{230\hscale} 

\etocstandarddisplaystyle 
\etocstandardlines 

{
  \hypersetup{linkcolor=black}
\tableofcontents* 
}


\newpage ~~
\thispagestyle{empty}
\endgroup

{
  \let\oldclearpage\clearpage 
  \let\clearpage\bigskip 


 
\setchapterstyle{kao} 
\mainmatter 

\pagelayout{wide} 
\setcounter{page}{0}
\addpart{What is this thesis about?}
\pagelayout{margin} 
\chapter{Modern thermodynamics in a nutshell}\label{C:Introduction}

\begin{center}
    \emph{From classical to quantum thermodynamics}
\end{center}

\raisebox{-0.1\height}{\includegraphics[width=0.1\textwidth]{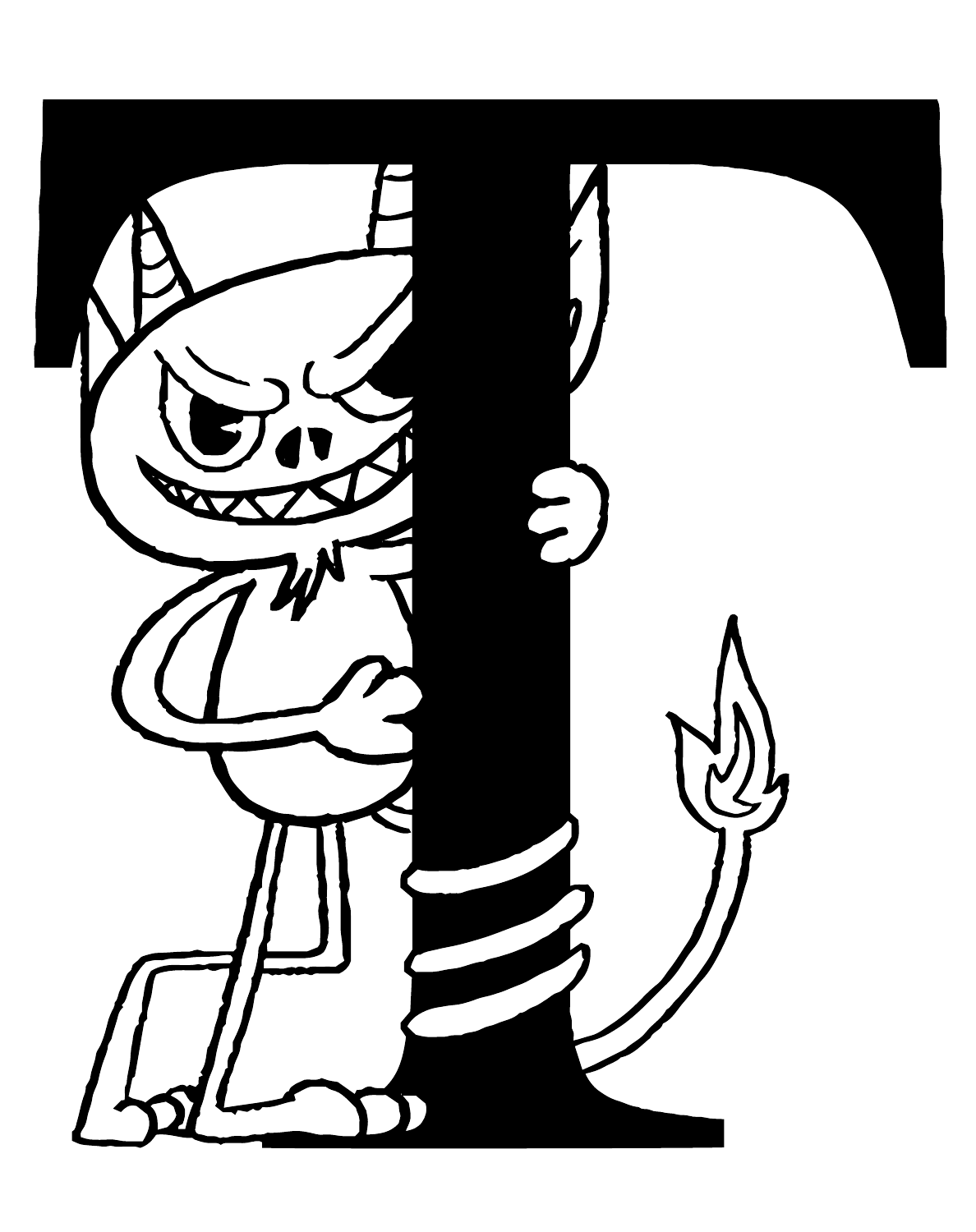}}\hspace{-0.2cm}
hermodynamics is a fundamental theory for describing the behaviour of complex systems. It offers clear and simple guidelines for characterising macroscopic systems in equilibrium, allowing us to predict and control state transformations using a small number of macroscopic variables~\cite{fermi, callen}. Unlike other theories in physics, which begin with a microscopic approach and proceed to a macroscopic description, thermodynamics is unconcerned with microscopic details and instead takes only bulk properties into account\sidenote{This is a historical strategy as the theory was conceived at the end of the eigh- teenth century when the Industrial Revolution began to use heat to generate movement. At the time, no physical theory existed that could adequately explain the nature of heat at the microscopic level.}. Over three centuries, the theory has withstood major scientific revolutions, such as the advent of general relativity and quantum mechanics, establishing itself as one of the pillars of physics. Its importance was unequivocally stated by Einstein, who declared~\cite{einstein}:
\begin{quote}
    \emph{``The only physical theory of universal content, which I am convinced, that within the framework of applicability of its basic concepts will never be overthrown.''}
\end{quote}

Its successes and universality are rooted in four laws, which naturally boil down empirical observations of nature:

\begin{enumerate}[start=0]
    \item \textbf{Zeroth law.} \emph{If two systems are each in thermal equilibrium with a third system, they are also in thermal equilibrium with one another.} \\
    \item \textbf{First law.} \emph{The energy exchanged between a system and its environment in a physical process can be split into two contributions, namely work $W$ and heat} $Q$: \mbox{$dU = \delta Q + \delta W$}\footnote{
    The notation $\delta$ is used to indicate that $W$ and $Q$ are inexact differentials, i.e., path-dependent.}.\\
    \item \textbf{Second law.} \emph{The entropy of a closed system undergoing a spontaneous physical process either increases or remains the same: $\Delta S \geq 0$.}\\
    \item \textbf{Third law.} \emph{As the temperature approaches absolute zero, the change in entropy also tends toward zero: \mbox{$\lim_{T\to 0} \Delta S = 0$}.}
\end{enumerate}

There are many ways to state and rephrase the aforementioned laws, particularly the second law of thermodynamics, as reflected in the formulations by Kelvin\sidenote{\emph{``A transformation whose only final result is to transform into work heat extracted from a source which is at the same temperature throughout is impossible.''}~\cite{fermi}}, Clausius\sidenote{\emph{``A transformation whose only final result is to transfer heat from a body at a given temperature to a body at a higher temperature is impossible''}~\cite{fermi}}, and Carnot\sidenote{\emph{``No engine operating between two heat reservoirs can be more efficient than a Carnot engine operating between those same reservoirs.''}~\cite{Carnot}}. 
The key takeaway is that thermodynamics imposes fundamental constraints on possible physical processes. These constraints are expressed as a set of relations among macroscopic quantities, which are defined in what is known as \emph{thermodynamic limit}. This idealised scenario involves a macroscopic system composed of $n \to \infty$ particles undergoing changes in such a way that the system remains approximately in thermal equilibrium at all times.

In the late 19th century, as the atomic theory gained popularity, scientists began conceptualising a gas as a vast collection of bouncing balls confined within a chamber of finite volume. This sparked the interest of Boltzmann, who ended up establishing a connection between entropy—a quantity previously defined phenomenologically—and the volume of a specific region in phase space, an entity defined in classical mechanics~\cite{boltzmann2022lectures}. From this point onwards, statistical mechanics emerged as a complementary framework to thermodynamics, aiming to bridge the gap between macroscopic and microscopic descriptions. By combining techniques from statistics and mechanics, it provided a new route for explaining the physical properties of matter in bulk through the lens of the dynamical behaviour of its microscopic constituents~\cite{pathria2016statistical}. Importantly, it elucidated the notion of equilibrium and showed that thermodynamic variables can be interpreted as averages of microscopic quantities. For instance, thermal energy may be associated with the statistical mean of the kinetic energy of system's particles. Consequently, the laws of thermodynamics possess a probabilistic nature: they are expected to hold on average, but there is nothing to preclude their temporary violation when we go beyond the thermodynamic limit. As thermodynamics typically involves exceptionally large systems, the law of large numbers tells us that the likelihood of observing significant deviations from the average value essentially becomes negligible.

\begin{marginfigure}[0.5cm]
	\includegraphics[width=4.718cm]{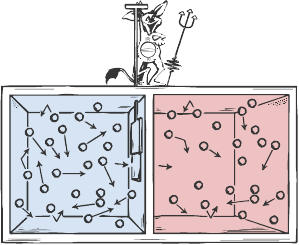}
    \caption{\emph{Maxwell's Gedanken experiment}. An intelligent being, a demon, controls a door between two gas chambers. As gas molecules approach the door, the demon selectively opens and closes it, allowing only fast-moving molecules to pass through in one direction and slow-moving molecules to pass through in the other direction. Since the kinetic temperature of a gas is directly linked to the velocities of its constituent molecules, the demon's strategy results in one chamber warming up while the other chamber cools down. This reduces the overall entropy of the system without requiring any work, thereby violating the second law of thermodynamics.}
	\label{Fig:maxwellsdemon}
\end{marginfigure}

The probabilistic nature of the laws of thermodynamics originates from the interplay between information and thermodynamics. The famous gedanken experiment, known as Maxwell's demon~\cite{leff1990maxwell, Vedral2009}, suggests that, given information about the positions and momenta of particles, one could reduce the entropy of a gas of particles without expending any work, seemingly violating the second law of thermodynamics (see Fig.~\ref{Fig:maxwellsdemon} for an extended discussion). Essentially, having information allows one to violate the second law. However, the recognition of the thermodynamic significance of information is perhaps best captured by Szilárd’s engine~\cite{szilard1929entropieverminderung}, a simple setup that exploits one bit of information (the outcome of an unbiased yes/no measurement) to implement a cyclic process that extracts $1/\beta \log 2$ of energy as work from a thermal reservoir at inverse temperature $\beta$ (see Fig.~\ref{Fig:szilard-engine} for a pictorial representation). As with Maxwell's demon, the Szilárd engine can overcome the second law of thermodynamics whenever some information about the state of the system is available. The proposal put forth by Szilárd has undergone thorough analysis for decades, revealing possible limitations of the scheme and shedding light on the origin of the entropy decrease. While Szilárd did acknowledge the need for an entropy increase during the measurement process to compensate for the entropy reduction, he did not explicitly address the significance of the demon's memory. Efforts to resolve this conundrum have been many and varied~\cite{leff1990maxwell, Vedral2009}. The puzzle was finally solved by Charles Bennett~\cite{bennett1973logical}, who elucidated the necessity of resetting (or erasing) the demon's memory to complete the cycle. Through the resolution of this puzzle, it became apparent that thermodynamics imposes physical constraints on information processing. In particular, the second law can be reformulated as a statement that \emph{no thermodynamic process can result solely in the erasure of information}. Every time information is erased, the erasure process is accompanied by a fundamental heat cost, i.e., an entropy increase in the environment. Alternatively, \emph{Landauer's Principle}~\cite{landauer1991information} states that the erasure process has an unavoidable energetic cost, with the minimum possible amount of energy required to erase a completely unknown bit of information given by $1/\beta\log 2$.

The advancements in statistical mechanics, along with the integration of thermodynamics and information theory, have led scientists to go beyond the standard thermodynamics scenario. Whilst classical thermodynamics primarily focuses on thermodynamic equilibrium, recent developments have paved the way to investigate non-equilibrium states, revealing thermodynamic properties in such scenarios.
\begin{marginfigure}[-2.2cm]
	\includegraphics[width=4.718cm]{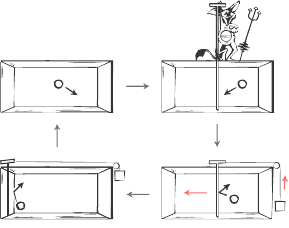}
    \caption{\emph{Szilárd engine}. A chamber containing a single atom is initially in an unknown position. A demon measures the atom's position and inserts a movable wall with an attached mass. As the gas in the chamber expands, the attached mass is raised, resulting in the performance of work. This mechanism appears to enable the conversion of information into useful energy, seemingly contradicting the second law of thermodynamics.}
	\label{Fig:szilard-engine}
\end{marginfigure}
For small deviations from equilibrium, linear response theory~\cite{kubo1957statistical} provides insight into how thermodynamic quantities relate to fluctuations defined at equilibrium. Fluctuation-dissipation theorems are crucial examples as they enable us to understand the thermodynamic response of systems under the effect of perturbations~\cite{kubo1966fluctuation,weber1956fluctuation}. For example, when classical systems evolve under a Hamiltonian and are coupled to a thermal bath, the free energy difference satisfies the following relation, provided the switching is sufficiently slow~\cite{hermans1991simple}: $\Delta F = \langle W \rangle - \beta\sigma_W/2$, where $\beta$ is the inverse temperature, $\langle W \rangle$ is the average work and $\sigma_W = \langle W^2 \rangle - \langle W \rangle^2$ is the work variance. This simple relation tells us that the cost of driving a system out of equilibrium is accounted for by some work dissipation, which is related to thermal fluctuations characterised at equilibrium. Emerging approaches like stochastic thermodynamics~\cite{Seifert2008,Seifert_2012} pick up where the standard methods start to fail and provide a framework for extending classical thermodynamics to finite-size and non-equilibrium systems. Its range of applicability spans from single colloidal particles trapped by time-dependent laser fields~\cite{sekimoto1998langevin,PhysRevLett.89.050601,PhysRevLett.92.140601} to enzymes~\cite{seifert2005fluctuation, schmiedl2007entropy,PhysRevE.79.011917} and thermoelectric devices involving single electron transport~\cite{PhysRevLett.89.116801,PhysRevLett.94.096601,esposito2009thermoelectric,PhysRevE.85.031117}. This field has played a crucial role in uncovering \emph{fluctuation relations}~\cite{campisi2011colloquium}, which establish bounds for processes occurring when systems are driven away from equilibrium. A celebrated example is the Jarzynski equality~\cite{jarzynski1997nonequilibrium}\sidenote{The Jarzynski equality establishes a connection between the free energy differences of two states and the irreversible work performed along an ensemble of trajectories that link those states:
\begin{equation*}
   \langle e^{-\beta w} \rangle = e^{-\beta \Delta F}.
\end{equation*}}, a remarkable relation that allows one to express the free energy difference between two equilibrium states by a non-linear average over the work required to drive the system in a non-equilibrium process from one state to the other. By comparing probability distributions for the work spent in the forward and backward (time-reversed) processes, Crooks found a "refinement" of the Jarzynski relation, now called the Crooks fluctuation theorem~\cite{crooks1999entropy}\sidenote{The Crooks relation offers a refined fluctuation theorem in comparison to the Jarzynski equality, given by the equation:
\begin{equation}
    \frac{P_F(w)}{P_B(w)} = e^{-\beta(\Delta F - w)},
\end{equation}
Here, $P_F(w)$ and $P_B(w)$ represent the forward and backward probability densities of dissipating a specific amount of work~$w$.}. Other results linking fluctuations and thermodynamic quantities have been derived in the past years~\cite{sevick2008fluctuation}, allowing one to discuss thermodynamics beyond its limits.

When dealing with even smaller systems, quantum effects come into play, and fluctuations are no longer solely of thermal origin but also quantum in nature. New effects, such as entanglement and coherence, naturally raise the question of how classical formulations of thermodynamics are affected at the level of a few quanta. Addressing these questions entails a preliminary inquiry into how concepts like heat, work, and temperature can be extended to the quantum realm. Apart from the drive to clarify fundamental aspects, the onset of a new quantum revolution involving the miniaturisation of devices to the nanoscale has also sparked questions about implementing quantum effects to devise optimal thermodynamic protocols and explore the possibilities for achieving real quantum advantages. As a result, the emerging field known as \emph{quantum thermodynamics}~\cite{gemmer2009quantum,kosloff2013quantum, Vinjanampathy_2016,Goold2016, binder2018thermodynamics,Deffner2019,RevModPhys.93.035008,strasberg2022quantum}, seeks to extend standard thermodynamics to systems of sizes well below the thermodynamic limit, in non-equilibrium situations, and with the full inclusion of quantum effects. 

Whilst the field has flourished over the last two decades, its roots stretch back to the middle of the last century, when physicists realised that qubits could achieve temperatures below absolute zero~\cite{PhysRev.103.20}, and that three-level masers could be regarded as heat engines~\cite{PhysRevLett.2.262,PhysRev.156.343}. Shortly thereafter, in the 70s, mathematical approaches leaning on the theory of open quantum systems formalised the study of such systems as models of heat engines~\cite{kosloff1984quantum,alicki1979quantum}. 
Finally, the field has consolidated in recent years, as our increasing ability to manipulate and control systems at smaller scales has enabled experimental physicists to miniaturise heat engines, pumps, and refrigerators into the quantum realm~\cite{blickle2012realization,martinez2016brownian, koski2014experimental, thierschmann2015three, rossnagel2016single, PhysRevLett.123.240601,maslennikov2019quantum}.\sidenote{See Book~\cite{halpern2018quantumbook} (Chapter 6) for a detailed historical account of the evolution of quantum thermodynamics.}

Due to the broad nature of the field, there have been several different approaches in recent years that have either been unified or combined into a single approach. In the context of capturing thermodynamics at the level of a few quanta, two prominent candidates emerge. The first candidate is based on the theory of open quantum systems~\cite{breuer2002theory,rivas2012open}, which describes the continuous evolution of a quantum system over time using a type of differential equation known as a master equation. This approach is typically model-dependent and describes the thermalisation of a system that is weakly coupled to a large environment.

The second approach, known as the \emph{resource theory of thermodynamics}~\cite{Janzing2000,horodecki2013fundamental,brandao2015second,Lostaglio2019,ng2019resource}, is general, model-independent, and aims to characterise the possible physical evolution of systems under certain constraints, such as energy conservation, locality, etc. Unlike the first approach, it does not rely on a master equation description. Aiming to capture thermodynamic properties out of certain constraints, it does not focus on a specific model or explicitly solve for a particular dynamics.

Although at first glance, both frameworks may seem like completely distinct approaches, they can be mapped to each other and are simply different sides of the same coin~\cite{lostaglio2021continuous}. Where both approaches differ lies in the toolkits used and the questions asked. 

\begin{center}
    \emph{What am I gonna find here?}
\end{center}

This thesis focuses on the resource-theoretic approach~\cite{chitambar2019quantum}. Its objectives include characterising thermodynamic transformations across different regimes and identifying optimal methods for exploiting them, while taking into account realistic constraints. These constraints arise from fundamental principles, such as the second law of thermodynamics, as well as more practical considerations like limited access to memory. The same approach also enables us to revisit longstanding questions, such as that of the thermodynamic arrow of time, which imposes a fundamental asymmetry in the flow of events. Finally, the last chapter of this thesis is dedicated to a novel connection between quantum optics and resource theories, opening up new directions from both theoretical and experimental perspectives.

\chapter{Motivation and contributions}
\section{Geometric structure of thermal cones}

\begin{center}
    \emph{State of the art and motivation}
\end{center}

The second law of thermodynamics imposes a fundamental asymmetry in the flow of events. The so-called thermodynamic arrow of time introduces an ordering that divides the system's state space into past, future and incomparable regions. This \emph{thermodynamic ordering} can be explored by focusing on the most general type of energy-conserving interaction between the system and a thermal bath. These transformations encode the structure of the thermodynamic arrow of time by telling us which states can be reached from a given state. While substantial insights have been gained from studying the future, explicit characterisation of both the incomparable region and the past has not been addressed.

\begin{center}
    \emph{Main results}
\end{center}

I\sidenote{In the first chapter, although the efforts described are a product of both my work and that of my collaborators, I have chosen to use the pronoun 'I' for simplicity and to emphasise the presentation of the results. From Chapter~\ref{C:mathematical_preliminaries} onwards, the pronoun 'we' is employed in line with conventional academic writing and to reflect collective contributions.} have made a significant step towards geometrically understanding the thermodynamic arrow of time by investigating the concept of \emph{thermal cones}, i.e., sets of states that a given state can thermodynamically evolve to (the future thermal cone) or evolve from (the past thermal cone). I provided an explicit construction of the past thermal cone and the incomparable region for a given $d$-dimensional energy-incoherent state, and delved into a detailed analysis of their volumes. Together with previously known results on the future thermal cone, this construction offers a comprehensive view of the thermodynamically allowed evolution.  

In deriving my results, I introduced a novel tool—an embedding lattice for thermomajorisation—which might be of independent interest to scientists working on information-theoretic approaches to thermodynamics. Lastly, I showed how these findings can be used to study possible state transformations in theories of entanglement and coherence via the common majorisation framework.

\section{Thermal recall: Memory-assisted Markovian thermal processes}

\begin{center}
    \emph{State of the art and motivation}
\end{center}

One of the major endeavours currently undertaken by the quantum community is the challenge of extending thermodynamics to the level of just a few particles, while still incorporating all relevant properties into a single framework. To this end, two primary approaches have emerged. The first approach is based on Markovian master equations, while the second focuses on thermal operations. Despite the success of both frameworks, each comes with its own specific limitations. The former models memoryless dynamics, whereas the latter describes arbitrarily non-Markovian processes. However, neither approach allows for easy incorporation of finite-memory effects, which necessitates further attention. Although previous research in non-Markovian dynamics has developed model-dependent approaches, there is a need for a general, model-independent and comprehensive framework.

\begin{center}
    \emph{Main results}
\end{center}

I bridged the gap between thermal operations and Markovian thermal processes by introducing the concept of memory-assisted Markovian thermal processes (MeMTP) -- memoryless thermodynamic processes that are promoted to non-Markovianity by allowing partial control over the bath’s degrees of freedom. My main contribution was to put forward a family of algorithms that enabled interpolation between the regime of memoryless dynamics and that with full control over both the system and the bath. Furthermore, I proved that, in the infinite temperature limit, all thermodynamic transformations induced by arbitrarily non-Markovian dynamics are recovered using MeMTP when memory is large enough. For the finite temperature regime, I have demonstrated the convergence to a subset of transitions and, based on strong numerical evidence, posed a conjecture that the convergence extends to arbitrary transitions.

These results opened the door to studying the role memory plays in the performance of thermodynamic protocols. First, I investigated work extraction in the intermediate regime of limited memory, comparing it with the two traditional scenarios of no memory and complete control. Secondly, I addressed the problem of cooling a two-level system using a two-dimensional memory characterised by a non-trivial Hamiltonian. This minimal model demonstrated that, when memory is brought into the picture, one can further cool a system below the ambient temperature.

\section{Fluctuation-dissipation relations for thermodynamic distillation processes}

\begin{center}
    \emph{State of the art and motivation}
\end{center}

The development of novel quantum technologies is closely linked to our level of understanding of the laws of thermodynamics at the scale of a few quanta. In this regime, thermal and quantum fluctuations around thermodynamic quantities become significant, and dissipation becomes unavoidable. This implies that the system's free energy is inevitably lost during a thermodynamic process that transforms the system's state, which limits its potential as a quantum resource. From both fundamental and applied perspectives, it is crucial to comprehend the ultimate limits of such dissipation. This brings the need to develop a theoretical framework that accommodates quantum effects such as superpositions, while simultaneously handling finite-size systems where fluctuations around thermodynamic averages are dominant. This regime inherently poses a fundamental question: \emph{what are the necessary and sufficient conditions underlying thermodynamic transformations of finite-size systems?} Variants of this question, such as single-shot and asymptotic state interconversion, have been addressed, but the interplay between fluctuations and dissipation has remained unresolved.

\begin{center}
    \emph{Main results}
\end{center}

I developed a genuinely quantum framework allowing one to push the thermodynamic description beyond macroscopic systems. Specifically, I was able to find necessary and sufficient conditions for the existence of a thermodynamic transformation between different non-equilibrium states of a few-particle systems. This allowed me to establish a precise relationship between the free energy fluctuations of the system's initial out-of-equilibrium state and the minimal free energy dissipated during a thermodynamic process. As a result, I was able to recast the well-established fluctuation-dissipation relations within the the context of quantum information. This, in turn, enabled me to pinpoint the optimal performance of thermodynamic protocols, including work extraction, information erasure, and thermodynamically-free communication, up to second-order asymptotics based on the number of processed systems. These findings represent a pioneering analysis of such thermodynamic protocols for quantum states exhibiting coherence between distinct energy eigenstates, especially in the intermediate regime of large yet finite $N$.

\section{Catalysis in cavity QED}

\begin{center}
    \emph{State of the art and motivation}
\end{center}

Catalysis is a ubiquitous phenomenon in science. It consists of using an auxiliary system (a catalyst) to enable processes that would otherwise be impossible. Over the last two decades, this notion has spread to the field of quantum physics. It has become instrumental in revealing fundamental constraints on both entanglement manipulation and thermodynamic processes. Additionally, it has found many applications within quantum information theory. However, this effect is typically described within a highly abstract framework known as \emph{resource theories}. Despite its successes, this approach struggles to fully capture the behaviour of physically realisable systems, thereby limiting the applicability of quantum catalysis in practical scenarios. This poses a challenge for translating the concept of quantum catalysis from a purely theoretical construct to a practically implementable tool. 

\begin{center}
    \emph{Main results}
\end{center}

I delved into the practical aspects of quantum catalysis to determine its relevance and potential applications, especially in experimental settings. During my investigation, I uncovered the effect of quantum catalysis within a paradigmatic quantum optics setup -- namely, the Jaynes-Cummings model in which an atom interacts with an optical cavity. By using the atom as a catalyst, I demonstrated that this effect  can the generate non-classical states of light in the cavity. This insight prompted me to translate the known framework of catalytic transformations to the realm of quantum optics, thereby proving its usefulness in practical scenarios. In doing so, I identified which atomic states can act as catalysts and assessed the degree of non-classicality these states could induce in the cavity. Futheremore, I also elucidated the role of quantum correlations and quantum coherence in the creation of non-classical states of light during a catalytic process. These findings significantly broaden the horizons of quantum catalysis, pushing it beyond the theoretical boundaries of quantum resource theories and marking a crucial step forward in its practical applications within quantum science.

\section{Outline of the thesis} 

This thesis is structured as follows. Chapter~\ref{C:mathematical_preliminaries} begins with mathematical preliminaries, setting out the notation and introducing essential concepts required for subsequent chapters. These discussions include an examination of probability distributions and their transformations, as well as various types of majorisation such as thermomajorisation, continuous thermomajorisation, and approximate thermomajorisation. 

Next, Chapter~\ref{C:resource_theory_of_thermodynamics} provides an introduction to the resource theory of thermodynamics, predominantly encompassing previously known results. Following the formal introduction of the set of free operations, termed thermal operations, and a discussion of their characteristics, this framework is applied to the analysis of thermodynamic protocols. I then briefly introduce a paradigmatic approach to quantum thermodynamics based on Markovian master equations, which will be used and contrasted with the resource-theoretic approach in subsequent sections of this thesis.

Chapter~\ref{C:thermal_cones} is devoted to the study of the geometric structure of thermal cones. This chapter builds on the following original result~\cite{deoliveirajunior2022}. The discussion starts with the main results concerning the construction of majorisation cones and their interpretation in the thermodynamic context, as well as in other majorisation-based theories. Subsequently, I investigate the structure of thermal cones, which emerge from the thermomajorisation relation, using a novel tool known as the embedding lattice. An explicit characterisation of the incomparable and past thermal regions is then presented, followed by an in-depth analysis of their properties. New thermodynamic monotones, defined by the volumes of the past and future thermal cones, are introduced. I further delve into their operational interpretation and detail their properties. Finally, the chapter ends by extending the notion of thermal cones beyond diagonal states to the simplest case of a coherent qubit.

Chapter~\ref{C:memory-MTP} is dedicated to the study of memory-assisted Markovian thermal processes. This chapter is based on the following original result~\cite{czartowski2023thermal}. The chapter starts by highlighting the distinctions between the frameworks of thermal operations and Markovian thermal processes. Subsequently, the central concept of this chapter, namely, memory-assisted Markovian thermal processes is introduced. A protocol is then outlined that uses thermal memory states to approximate non-Markovian thermodynamic state transitions with Markovian thermal processes. Furthermore, it is demonstrated how this approximation converges to the complete set of transitions achievable via thermal operations as the memory size increases. Lastly, I explain how this framework can be used to quantify the role played by memory effects in thermodynamic protocols such as work extraction and cooling.

Chapter~\ref{C:finite-size} contains my original findings on fluctuation-dissipation relations presented through the language of information theory. This chapter is based on the following original result~\cite{PhysRevE.105.054127}. The chapter begins with a high-level description that provides a glimpse into my investigations and conveys the underlying physical intuition to a broad audience, without delving into the technical details of the developed framework. I then present the key results concerning optimal transformation error and the fluctuation-dissipation relation for incoherent and general pure states. Then, I discuss their thermodynamic interpretation and apply my findings to three specific thermodynamic protocols: work extraction, information erasure, and thermodynamically-free communication.

Chapter~\ref{C:CQED} is dedicated to the study of quantum catalysis in the paradigmatic quantum optics model of Jaynes-Cummings. This chapter is based on the following manuscript~\cite{junior2023quantum}. I start by discussing the general aspects of catalytic transformations and bridging the gap with quantum optics. This is done by uncovering a catalytic process that enables the generation of non-classical states of light within the Jaynes-Cummings model. Next, I investigate the mechanism of this catalytic process and identify two crucial ingredients: \emph{correlations} and \emph{quantum coherence}. I then treat the problem analytically by explicitly solving this model. This allows me to identify which states serve as catalysts and explore the degree of non-classicality induced by these atomic states. Finally, I conclude the chapter by discussing the generality of catalysis and potential future research directions.

\section{A few more details}

From Chapter~\ref{C:mathematical_preliminaries} onwards, defintions and notable results are presented in titled boxes with the following color code:
\begin{definition}
A precise and organised meaning to a new term.
\end{definition}

\begin{proposition}
A statement that is derived from what has been discussed.
\end{proposition}

\begin{lemma}
A minor, proven proposition that is used as a stepping stone to a larger result.
\end{lemma}
\begin{theorem}
Larger result.
\end{theorem}

\begin{corollary}
A proposition that follows from the larger result.
\end{corollary}

\begin{observation}
A statement or a remark that is made based on informal analysis. 
\end{observation}

Not all proofs in the text are immediately presented after lemmas, theorems, or corollaries. Only short proofs that directly follow from those results are provided. This choice is intended to assist readers who are not interested in delving into detailed derivations and would prefer to skip them entirely. Each chapter always ends with a section entitled \emph{``Derivation of the Results,"} in which auxiliary tools and results are derived and discussed in detail. 

The thesis includes several comments (or solved problems) framed like this

\begin{kaobox}[frametitle= Title of the comment or problem]
Description of the comment or problem.
\end{kaobox}

Typically, comments and problems are interspersed throughout the text rather than being crucial for the understanding of the thesis. These comments and problems can be skipped without hindering the overall reading. However, they do serve to complement the description of the topics being discussed and provide additional insights. 
Some relevant comments, extensions of arguments, or even curious observations are displayed as margin notes.

\pagelayout{wide} 
\addpart{Background}
\pagelayout{margin} 
\chapter{The geometry of quantum states}\label{C:mathematical_preliminaries}

\section{Very general and brief remarks about quantum mechanics}

Any quantum system\marginnote{This section is intended to provide a concise introduction to the fundamental concepts of quantum theory discussed throughout this thesis. It is not meant to be a comprehensive overview of quantum mechanics; for that purpose, the recommended books of Peres~\cite{peres1995quantum}, Cohen-Tannoudji~\cite{cohen1977quantum}, and Nielsen and Chuang~\cite{nielsen2010quantum} are highly suggested.} can be described by a complex Hilbert space, denoted as $\mathcal{H}$. This is a vector space over the complex field that is equipped with an inner product. A specific vector in $\mathcal{H}$ can be represented using the convenient Dirac notation as $\ket{\psi}$, which is read as \emph{ket psi}. The inner product of two vectors, $\ket{\psi}$ and $\ket{\phi}$, is expressed as $\braket{\psi}{\phi}$. Through the use of the inner product, we can define a linear functional, denoted as $\bra{\phi}$ and read as \emph{bra phi}, for each vector $\ket{\phi}$. As a result, the inner product forms what is referred to as a \emph{bracket} in this context. This vector space can either be finite or infinite dimensional. In this thesis, except in Chapter~\ref{C:CQED}, we will focus solely on finite Hilbert spaces.

The mathematical object used to describe the state of a physical system at a given instant in time is called a state. In quantum mechanics, this state is represented by a density operator, denoted as $\rho$, that acts on the Hilbert space $\mathcal{H}$ associated with the physical system it describes. Formally, a density operator $\rho$ is defined by the following conditions:
\begin{enumerate}[label=(C\arabic*)]
  \item \textbf{Positive Semi-definite.} For all $\ket{\psi}$ in $\mathcal{H}$ it follows that $\braket{\psi|\rho}{\psi} \geq 0$, or simply $\rho \geq 0.$
  \item \textbf{Normalisation.} The trace of the density operator is equal to one, i.e., $\tr(\rho) = 1$. 
  \item \textbf{Hermiticity.} The density operator is Hermitian $\rho = \rho^\dagger$. 
\end{enumerate}
These conditions guarantee that the density operator provides a valid quantum mechanical description of a system's state, encapsulating both pure states and statistical mixtures of states. The condition of positive semi-definiteness ensures that the probabilities of outcomes obtained from measurements are real and non-negative. The normalisation condition confirms that the total probability of all possible states of the quantum system sums to one, reflecting the statistical nature of quantum mechanics. Lastly, the requirement of Hermiticity ensures that all eigenvalues of the density operator are real numbers, a necessary prerequisite for them to represent valid probabilities. As a result, every density operator can be written as the convex combination of
unidimensional projectors:
\begin{equation}
    \rho = \sum_k p_k \ketbra{\psi_k}{\psi_k} \:\: \operatorname{with} \: \: \sum_k p_k = 1 \:\: \operatorname{and} \: \: \forall k: p_k \geq 0.
\end{equation}
This decomposition is generally not unique. 

We will be interested in the Hilbert space of a composiste quantum system comprising a sytem with Hilbert space $\mathcal{H}_A$ and a system with Hilbert space $\mathcal{H}_B$, such that the joint Hilbert space is given by the tensorial product $\mathcal{H}_{AB} = \mathcal{H}_A \otimes \mathcal{H}_B$. For a given quantum state $\rho_{AB}$ of the composite system, the reduced state of subsystem $A$ can be obtained by taking the partial trace over the degrees of freedom of subsystem $B$. This operation is denoted as $\rho_A = \tr_B(\rho_{AB})$, where $\tr_B$ represents the partial trace over subsystem $B$\sidenote{The reduced state on $\mathcal{H}_A$, obtained as the partial trace over $\mathcal{H}_B$, is defined as
\begin{equation*}
\rho_A:= \sum_{j}(\mathbbm{1}_A\otimes \bra{j}_B) \rho_{AB} (\mathbbm{1}_A \otimes \ket{j}_B),
\end{equation*}
where $\mathbbm{1}$ denotes the identity and ${ \ket{j}_B}$ is an orthonormal basis of $\mathcal{H}_B$.
}. In this thesis, subsystem $B$ will represent the environment, or heat bath, with which the main system, subsystem $A$, interacts. The state of the environment is typically represented by a thermal Gibbs state, which is a statistical ensemble at equilibrium.

The state of an isolated quantum system (time-independent case) follows the Lioville-von Neumann equation
\begin{equation}
    \frac{\partial \rho(t)}{\partial t} = -i[H,\rho(t)]
\end{equation}
where [.,.] denotes the commutator, $[A,B]=AB-BA$, and $H$ is the Hamiltonian of the system. Throughout this thesis, we will set $\hbar =1$. 

Isolated, or closed, quantum systems satisfies two important properties:
\begin{enumerate}[label=(P\arabic*)]
  \item \textbf{Unitary evolution.} The dynamics of a closed system undergoes a unitary time evolution given by the operator $U(t) = \exp(-iHt)$. Consequently, if $\rho(0)$ denotes the initial state of the system, the evolved state at time $\tau$ is given by $\rho(\tau) = U(\tau) \rho(0) U^{\dagger}(\tau)$.\sidenote{If the Hamiltonian is time-dependent, the evolution of $\rho(t)$ is still described by a unitary operator $U(t)$, but its computation is more cumbersome as it follows the so-called Dyson series, i.e., 
\begin{align*}
    U(t) &= \sum_{n=0}^{\infty}(-i)^n \int_{0}^{t}dt_1 H(t_1) \times ... \\ &\hspace{2cm}\times ... \int_{0}^{t}dt_n H(t_n) \\
    &\equiv \exp_+\left[-i \int_{0}^{t}ds H(t_s)\right],
\end{align*}
where the subscript $+$ denotes time ordering.}
  \item \textbf{Purity.} The eigenvalues $p_k$ of $\rho(t)$ do not change in time, i.e.,
  \begin{equation}
      \rho(t) = \sum_k p_k \ketbra{\psi_k(t)}{\psi_k(t)},
  \end{equation}
  where the $p_k$'s are time-independent probabilities and $\ket{\psi_k(t)}$ belongs to an orthonormal basis of wave functions, i.e., $\braket{\psi_i(t)}{\psi_j(t)} = \delta_{ij}$, where $\delta_{ij}$ is the Kronecker delta. The purity of $\rho(t)$, defined as $\mu := \tr[\rho(t)^2]$, is conserved in time. If the state is pure, then $\mu = 1$.
\end{enumerate}

A general evolution of an open quantum system is described by a quantum channel, which is defined as a completely positive, trace-preserving map (CPTP) acting on the quantum state $\rho$. Typical microscopic derivations lead to a master equation of the following general form~\cite{kossakowski1972quantum,lindblad1976generators}
\begin{equation}
    \frac{\partial \rho(t)}{\partial t} = -\frac{i}{\hbar}[H,\rho(t)] + \mathcal{L}_t[\rho(t)].
\end{equation}
The first term on the right-hand side  represents the Hamiltonian evolution of the system, governing the closed (reversible) quantum dynamics. The second term, known as the Lindbladian or dissipator, governs the open (irreversible) quantum dynamics. This captures the effects of decoherence and dissipation due to the system's interaction with its environment. The specific form of the Lindbladian depends on the type of interaction, and can generally be written as~\cite{breuer2002theory}:
\begin{equation}
	\label{eq:lindbladian}
\mathcal{L}_t(\rho) = \sum_{i} r_i(t) \left[ L_i(t) \rho L_i(t)^\dag - \frac{1}{2}\Bigl\{L_i(t)^\dag L_i(t), \rho\Bigl\}\, \right],
\end{equation} 
where $\{\cdot,\cdot\}$ denotes the anticommutator, $L_i(t)$ represents time-dependent jump operators, and $r_i(t)\geq 0$ are time-dependent, non-negative rates associated with the different jump processes. These jump rates $r_i(t)$ and operators $L_i(t)$ embody the specifics of the system's interaction with its environment.

\newpage

\section{Dynamics of quasi-classical states}\label{Sec:quasi-classical_states}

Consider a finite system whose state is described by a $d$-dimensional vector \mbox{$\v{p} = (p_1, \dots, p_d)$} with $\sum_{i=1}^d p_i = 1$ and $p_i \geq 0$ for all $i \in \{1, \dots, d\}$. These states are referred to as \emph{quasi-classical} and they belong to the probability simplex, denoted by $\Delta_d$, which represents the space of normalised vectors of dimension $d$ with real entries [see Fig.~\ref{Fig:space_of_states}]:

\begin{marginfigure}[-2.25300cm]
	\includegraphics[width=41.1378mm]{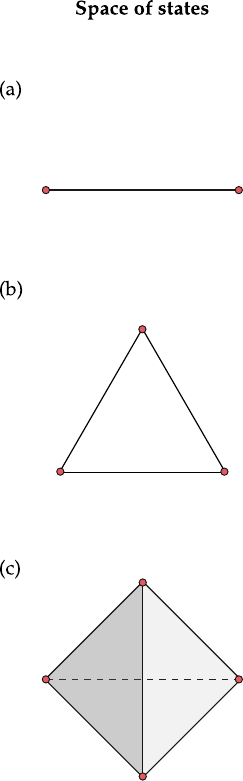}
	\caption{\emph{Space of states.} The state space of a quasi-classical system is represented by a simplex. A two-level system (a) corresponds to a 1-dimensional simplex, a line segment; a three-level system (b) corresponds to a 2-dimensional simplex, a triangle and its interior and a four-level system (c) corresponds to a 3-dimensional simplex, a tetrahedron.}
	\label{Fig:space_of_states}
\end{marginfigure}

\begin{equation}\label{Eq:probability_simplex}
    \Delta_d= \left\{(p_1, ..., p_d) \in \mathbb{R}^d : \sum_i p_i = 1\right\}.   
\end{equation}
Probability vectors will be denoted by bold lowercase letters $\v p \in \Delta_d$ and their corresponding cumulative counterparts by bold uppercase vectors $\v P: P_j = \sum_{i=1}^j p_i$ for $j \in \{1, \dots, d\}$.

A stochastic matrix $\Lambda$ represents the most general evolution between elements of $\Delta_d$. These matrices are linear maps that transform probability vectors into new probability vectors. Thus, its components satisfy the following constraints:
\begin{equation}\label{Eq:stochastic_matrix_conditions}
    \Lambda_{ij} \geq 0 \quad \text{and} \quad \sum_i \Lambda_{ij} = 1 .
\end{equation}
These conditions ensure that each entry of $\Lambda$ is non-negative and that its rows sum to unity, which corresponds to the conservation of probability. The matrix $\Lambda$ describes the dynamics of a system, in a state $\v p$, by the matrix-vector product:
\begin{equation}\label{Eq:stochastic_matrix_evolution}
    \v q = \Lambda \v{p}  \quad \text{or, equivalently,} \quad \v q = \sum_{j=1}^d \Lambda_{ij} p_j .
\end{equation}

A stochastic matrix that connects two states $\v p$ and $\v q$ is referred to as a \emph{process}. If $\Lambda$ can be generated by a continuous Markov process, it is said to be \emph{embeddable}~\cite{davies2010embeddable} and the process is \emph{memoryless}. More precisely, we introduce a rate matrix or generator $L$ as a matrix with finite entries that satisfies
\begin{equation}
    L_{ij} \geq 0 \quad \text{and} \quad \sum_{i} L_{ij} = 0.
\end{equation}
Then, a continuous one-parameter family $L(t)$ of rate matrices generates a family of stochastic processes $\Lambda(t)$ satisfying
\begin{equation}\label{Eq:master-equation}
    \frac{d}{dt}\Lambda (t) = L(t) \Lambda (t) \quad \text{with} \quad \Lambda(0) = \iden.
\end{equation}
The purpose of the control $L(t)$ is to achieve a target stochastic process $\Lambda$ at a given final time $t_f$, i.e., $\Lambda = \Lambda(t_f)$\sidenote{Equation~\ref{Eq:master-equation} is also known as a \emph{master equation}.}. If it is feasible to achieve such a target process for some choice of $L(t)$, then $\Lambda$ is said to be embeddable\footnote{Determining which stochastic matrices $\Lambda$ are embeddable remains a challenging open problem that has been extensively studied for decades. The complete characterisation is currently limited to $2 \times 2$~\cite{Kingman1962}, $3 \times 3$~\cite{cuthbert1973logarithm,johansen1974some,carette1995characterizations} and $4\times 4$~\cite{casanellas2020embedding} stochastic matrices, although various necessary conditions are known~\cite{goodman1970intrinsic,shahbeigi2023quantum}.}.

A stochastic matrix $\Lambda$ that additionally satisfies the condition \mbox{$\sum_j \Lambda_{ij} =1$} is called \emph{bistochastic matrix}. The term ``bistochastic'' refers to the fact that both rows and columns of the matrix sum to unity. Due to its significance in subsequent discussions, this matrix will be denoted by $\Lambda^{0}$. This set of matrices is completely characterised via the following theorem~\cite{bhatia1996matrix}:
\begin{theorem}[Birkhoff theorem]\label{Thm:Birkhoff}
    The set of bistochastic matrices is a convex set whose extreme points are permutation matrices.
\end{theorem}
\begin{marginfigure}[-1.451cm]
        \centering
	\includegraphics[width=3.771cm]{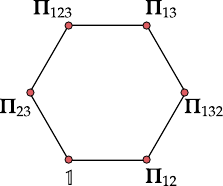}
	\caption{\emph{Birkhoff polytope.} The set $\mathcal{B}_3$ is a convex hull whose extreme points are permutation matrices (denoted by $\v \Pi_{ij}$).}
	\label{Fig:birkhoff_polytope}
\end{marginfigure}
A proof of this theorem can be found in Chapter 2 of Ref.~\cite{bhatia1996matrix}. 

The above result is a fundamental theorem stating that every bistochastic matrix can be represented as a probabilistic mixture of permutation matrices. The set of bistochastic matrices of dimension $N\times N$, denoted by $\mathcal{B}_N$, is called the Birkhoff polytope. This set forms a polytope in $\mathbb{R}^{(d-1)^2}$ with $d!$ vertices and with centre occupied by the uniform matrix with all entries equal to $1/d$~(see Fig.~\ref{Fig:birkhoff_polytope} for a pictorial representation of $\mathcal{B}_3$).

A notable property of bistochastic matrices is that they preserve the identity element of the space of probability distributions, i.e., the \emph{uniform distribution}
\begin{equation}
\v \eta = \frac{1}{d}\Bigl(1, ..., 1\Bigl).   
\end{equation}
\begin{marginfigure}[-2.4cm]
	\includegraphics[width=4.718cm]{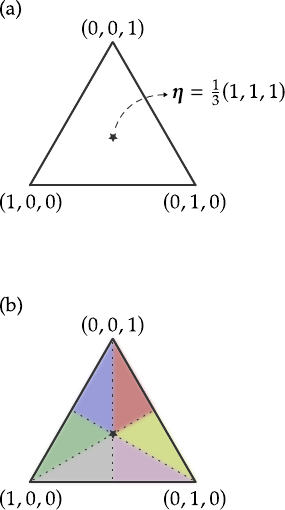}
	\caption{\emph{Properties of $\Delta_3$.} The state space of a three-level system consists of sharp states situated at the corners and a uniform state, represented by a star, in the centre. Its symmetry allows one for partitioning $\Delta_3$ into $3! = 6$ Weyl chambers, shown in (b) with different colors, which are asymmetrical simplices characterised by a particular ordering of the components in the probability vector.}
	\label{Fig:space-of-states-d-3}
\end{marginfigure}
Note that $\v \eta$ is maximally uncertain as all possible outcomes are equally likely and no information is available to predict which outcome will occur. Conversely, the ``opposite'' of a uniform state is the sharp state $\v s_k$, with $(\v s_k)_j = \delta_{jk}$, where one outcome is certain to occur, and all others have zero probability. Thus, it is a state in which all the information necessary to predict the outcome is available. In the space of states, sharp states are located at the vertices of the probability simplex, while the uniform state is at the centre~[see Fig.~\hyperref[Fig:space-of-states-d-3]{\ref{Fig:space-of-states-d-3}a}]. 

It is important to note that the space of states is symmetric under the action of the symmetric group $\mathcal{S}_d$. This symmetry arises from the constraints that define the probability simplex, which remain unchanged under permutations. As a result, the probability simplex can be divided into $d!$ equal parts, known as \emph{Weyl chambers}~\cite{hall2013lie}. Each chamber is composed of probability vectors that are ordered in non-decreasing order by a specific permutation. The chamber corresponding to the identity permutation is referred to as the canonical Weyl chamber (gray triangle in Fig.~\hyperref[Fig:space-of-states-d-3]{\ref{Fig:space-of-states-d-3}b}), and all other chambers are images of the canonical chamber under the action of the Weyl group.

In future considerations, states will always be ordered relative to a reference state, and their original order will not play any role in its description. The reordering will be based on the particular context and will possess the characteristics of a \emph{partial ordering}. This generalises the notion of total ordering by allowing incomparability between elements. Formally, a partial order $L$ is a binary relation $\succ$ over a set $S$ that satisfies three conditions:
\begin{enumerate}[label=\roman*]
\item Reflexivity: $s_1 \succ s_1$
\item Transitivity: if $s_1 \succ s_2$ and $s_2 \succ s_3$ then $s_1 \succ s_3$
\item Antisymmetry: if $s_1 \succ s_2$ and $s_2 \succ s_1$ then $s_1 = s_2$
\end{enumerate}
Binary relations satisfying only the first two properties are known as preorders. We will focus on a special kind of partial order known as a lattice and, later in Chapter~\ref{C:resource_theory_of_thermodynamics}, interpret it from a thermodynamic perspective. More precisely, the notion of lattice is defined as follows:
\begin{marginfigure}[-1cm]
	\includegraphics[width=4.7839 cm]{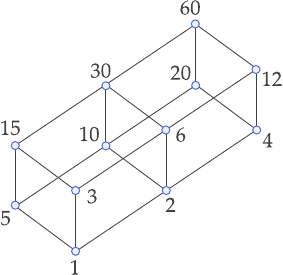}
	\caption{\emph{Partially ordered set $D_{60}$.} All the divisors, when ordered by divisibility, form a partially ordered set, which constitutes a lattice. For example, $6 \vee 10 = 30$ and $6 \wedge 10 = 2$}
	\label{Fig:lattice-example}
\end{marginfigure}

\begin{definition}[Lattice]\label{def_Lattice} A partially ordered set $(L, \leq)$ forms a lattice if for every pair of elements $\v p ,\v q \in L$, there exists a least upper bound, called \textbf{join} and denoted by $\v p \vee \v q$, such that $\v p \vee \v q \geq \v p$ and $\v p \vee \v q \geq \v q$; and a greatest lower bound, called \textbf{meet} and denoted by $\v p \wedge \v q$, such that $\v p \wedge\v q \leq \v p$ and $\v p \wedge \v q \leq \v q$.
\end{definition}

\section{Majorisation zoo}

In the previous section, we reviewed the fundamental concepts of probability distributions and stochastic matrices, along with their underlying geometrical properties. Building on this foundation, we will now explore an essential tool in the theory of statistical comparisons: \emph{majorisation} and its \emph{variants}. The concept of majorisation was first introduced by Muirhead~\cite{muirhead1902some} and later popularised by Hardy, Littlewood, and Pólya~\cite{hardy1952inequalities}. It is a powerful and easy-to-use tool that is widely applied to compare two probability distributions and assess their disorder. This concept has broad applications in various fields, including, economics~\cite{lorenz1905methods,dalton1920measurement}, computer science~\cite{parker1980conditions,polak1986majorization} and quantum physics~\cite{nielsen1999conditions,nielsen2001majorization,jonathan1999entanglement}. Notably, majorisation has been particularly useful in quantum mechanics, where it originated in entanglement manipulation~\cite{nielsen1999conditions,} before spreading to quantum thermodynamics~\cite{brandao2013resource}, coherence theory~\cite{Du2015}, and other subfields~\cite{van2023continuous}. The core of majorisation and its variants involves ordering a given probability distribution based on specific criteria and then applying a set of conditions to compare this initial distribution with a target one. Intriguingly, a direct link also exists between majorisation and stochastic matrices, which is further connected to the existence of certain processes.

\subsection{Majorisation}

We begin by defining majorisation\sidenote{Majorisation can be extended to density matrices, in which case it is viewed as a preorder of their spectra. More specifically, we say that $\rho \succ \sigma$ if $\v \lambda(\rho) \succ \v \lambda(\sigma)$, where $\v \lambda(\chi)$ denotes the vector of eigenvalues of a matrix $\chi$~\cite{hardy1952inequalities,horn1954doubly}} as follows:
\begin{definition}[Majorisation]\label{def_Majorisation} Given two $d$-dimensional probability distributions $\v p$, $\v q \in \Delta_d$, we say that $\v{p}$ \emph{majorises} $\v{q}$, and denote it by $\v p \succ \v q$, if and only if
\begin{equation} \label{eq_majorisation}
    \sum_{i=1}^k p_i^{\downarrow}\geq\sum_{i=1}^k q_i^{\downarrow} \quad \text{for all} \quad  k\in\{1\dots d\},
\end{equation}
where $\v{p}^{\downarrow}$ denotes the vector $\v{p}$ rearranged in a non-increasing order. 
\end{definition}
\begin{figure*}
\includegraphics[width=15.843cm]{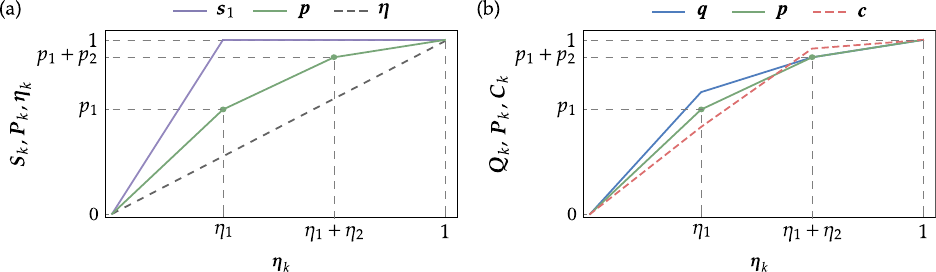}
	\caption{\emph{A geometric view of majorisation.} (a) Majorisation curves for a three-dimensional sharp state $\v s_1$, $\v p$ and uniform state $\v \eta$ and (b) for three different states $\v q$, $\v p$ and $\v c$. Note that while $\v p$ majorises $\v q$ [since $f_{\v p}(x)$ is never below $f_{\v q}(x)$], both states are incomparable with $\v c$, as their majorisation curves cross with $f_{\v c}(x)$.}
	\label{Fig:majorisation-curves}
\end{figure*}

Equivalently, the majorisation relation can be expressed in a more geometric way by defining a \emph{majorisation curve}\sidenote{Also known as the \emph{Lorenz curve}, it was introduced by American economist Max O. Lorenz as a way to visualize wealth distribution and inequality among the population of the United States~\cite{lorenz1905methods}.}.

\begin{definition}[Majorisation curve]\label{def_Majorisation_curve} Let $\v p$ and $\v \eta$ be a $d$-dimensional probability vector and a uniform state, respectively. The majorisation curve is a piece-wise linear curve $f_{\v p}(x)$ in $\mathbb{R}^2$ obtained by joining the origin $(0,0)$ and the points
\begin{equation}
\left(\sum_{i=1}^{k} \eta_i, \sum_{i=1}^{k} p^{\downarrow}_i \right) \quad \text{for} \quad k \in \{1, ..., d\}. 
\end{equation}
\end{definition}

A distribution $\v{p}$ majorises $\v{q}$ if and only if the majorisation curve $f_{\v p}(x)$ of $\v{p}$ is always above that of $\v{q}$,
\begin{equation}
\v p \succ \v q \iff \forall x\in \left[0,1\right]:~f_{\v p}(x) \geq  f_{\v q}(x) \, .
\end{equation}
Majorisation does not introduce a total order. A pair of states $\v p$ and $\v q$ may be incomparable with each other, in the sense that neither $\v{p}$ majorises~$\v{q}$, nor $\v{q}$ majorises $\v{p}$. In terms of majorisation curves, this implies that both curves intersect each other~ (see Fig.~\hyperref[Fig:majorisation-curves]{\ref{Fig:majorisation-curves}b}).

The partial ordering of probability vectors induced by majorisation can be interpreted as a formalisation of the concept of disorder with respect to the uniform distribution $\v \eta$. First, note that sharp distributions majorise all other distributions, and all distributions in turn majorise the uniform distribution $\v \eta$. Second, one can link majorisation to the concept of entropy via a specific class of functions -- those that preserve the majorisation-induced partial order structure. This can be more precisely illustrated through the following definition:
\begin{definition}[Schur-convex functions]\label{def_schur-convex} A function $f: \mathbbm{R} \to \mathbbm{R}$ is called Schur-convex if and only if
\begin{equation}
\v p \succ \v q \Rightarrow f(\v p) \geq f(\v q) 
\end{equation}
and Schur-concave if and only if $\v p \succ \v q \Rightarrow f(\v p) \leq f(\v q)$.
\end{definition}
The function $f$ is a homomorphism from the partially ordered set $(\mathbbm{R}^n, \succ)$ to the totally ordered set of real numbers. Examples of such functions encompass all R\'{e}nyi entropies, which, for a $d$-dimensional probability distribution $\v p$, are defined as follows
\begin{equation}
    H_{\alpha}(\v p) := \frac{\operatorname{sgn(\alpha)}}{1-\alpha}\log\left(\sum_{i=1}^d p^{\alpha}_i\right),
\end{equation}
where $\alpha \in \mathbbm{R}$. The cases $\alpha \to \pm \infty$ and $\alpha \to 1$ are defined by suitable limits $ H_1(\v p) = -\sum_{i=1}^d p_1 \log p_i \:\: ,\:\: H_{\infty}(\v p) = - \log \underset{i}{\operatorname{max}}\:p_i \:\: ,\:\:  H_{-\infty} = \log \underset{i}{\operatorname{min}} \: p_i$. The case $\alpha = 0$ is known as the Burg entropy $H_0(\v p) = \frac{1}{d}\sum_{i=1}^d \log p_i$. An example of the Rényi entropy for different values of $\alpha$ is shown in Fig.~\ref{Fig:renyeentropy}.

\begin{marginfigure}[-2.458cm]
	\includegraphics[width=4.756cm]{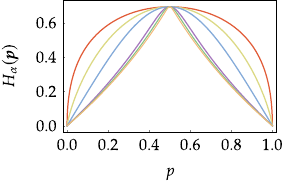}
	\caption{\emph{Rényi entropy}. Entropy of a random variable $\v p = (p, 1-p)$ as a function of $p$ for the following values of $\alpha = 0.5$ (red), $1$ (yellow), $2$ (blue), $10$ (purple), $20$ (green) and $50$ (orange).}
	\label{Fig:renyeentropy}
\end{marginfigure}

\begin{kaobox}[frametitle=Majorisation hierarchy]
For $d$-dimensional vectors, we have
\begin{equation}\label{Eq:majorisation-example}
(1,0..,0) \succ \Bigl(\frac{1}{2},\frac{1}{2},..,0\Bigl) \succ ... \succ \Bigl(\frac{1}{d-1}, ..., \frac{1}{d-1},0\Bigl) \succ \Bigl(\frac{1}{d}, ..., \frac{1}{d}\Bigl).\nonumber
\end{equation}
\end{kaobox}

We are now ready to establish the connection between majorisation relations and bistochastic state transformations, as captured by the renowned Hardy-Littlewood-P\'{o}lya theorem~\cite{hardy1952inequalities}.

\begin{theorem}[Majorisation $\&$ bistochastic matrices]\label{Thm:HLP}
There exists a bistochastic matrix $\Lambda^{0}$, $\Lambda^{0} \v \eta=\v \eta$, mapping $\v{p}$ to $\v{q}$ if and only if $\v{p} \succ \v{q}$.
\end{theorem}

See Chapter 2 of Ref.~\cite{marshall1979inequalities}, for a proof of the Hardy-Littlewood-P\'{o}lya  theorem. Now, let us briefly discuss the significance of this result. Firstly, it is important to note that the existence of processes connecting two probability distributions is directly related to a majorisation relation. This constitutes the bedrock of this thesis, as most of the presented results focus on generalisations of this theorem within a given context. For instance, the conditions for the existence of entanglement transformations between pure bipartite states under Local Operations and Classical Communication (LOCC) can be reframed in terms of majorisation, allowing us to determine whether a given LOCC exists~\cite{nielsen1999conditions}. As we shall see, at infinite temperature [$\beta = 0$ in Eq.\eqref{Eq:thermal-state}] or when energy levels are degenerate, thermodynamic transformations are described by bistochastic matrices. In such a scenario, the Hardy-Littlewood-P\'{o}lya  theorem plays a vital role in identifying the set of allowed transformations. Thus, Theorem.~\ref{Thm:HLP} serves as a cornerstone for subsequent results.

\subsection{Thermomajorisation}\label{Subsec-thermomajorisation}

The partial ordering of probability vectors, which is induced by majorisation, can be conceptualised as formalising the measure of disorder relative to the uniform distribution $\v \eta$. One might then pose the question of whether a majorisation relative to a general probability distribution can be defined, so that disorder is measured relative to a generic non-uniform distribution. This question was addressed in Refs.~\cite{veinott1971least,ruch1978mixing}, where the concept of $\v d$-majorisation was formalised and introduced\sidenote{Very recently, Ref.~\cite{VomEndedMajorisation} further investigated the geometric and topological properties of $\v d$-majorisation. Moreover, the concept of strict positivity proves to be a useful tool in analysing $\v d$-majorisation on matrices and their order properties. See~\cite{VomEndeStrict} for a detailed discussion.}. Mathematically, for a given vector $\v d$, we say that $\v p$ \emph{$\v d$-majorises} $\v q$, and denote by $\v p \succ_{\v d}\v q$, if there exists a stochastic matrix $\Lambda^{\v d}$, which leaves the vector $\v d$ invariant and maps $\v p$ into $\v q$ -- such a matrix then is called $\v d$-stochastic.

In this section, we introduce the notion of thermomajorisation~\cite{horodecki2013fundamental} -- the thermodynamic analogue of majorisation. It defines a partial order relation relative to the thermal distribution
\begin{equation}\label{Eq:thermal-state}
    \v \gamma = \frac{1}{\sum_{i=1}^{d}e^{-\beta E_i}} \Bigl(e^{-\beta E_1}, ..., e^{-\beta E_d} \Bigl),
\end{equation}
rather than the uniform distribution $\v \eta$. Generally, we will say that a given distribution $\v p$ \emph{thermomajorises} a target distribution $\v q$ if there exist a stochastic matrix that preserves the Gibbs state $\v \gamma$ and maps $\v p$ into $\v q$. Due to the importance of this class of matrices, we define it as follows.

\begin{definition}[Gibbs-preserving matrix]
A stochastic matrix is called Gibbs-preserving (GP), and denoted by $\Lambda^{\beta}$, if it leaves $\v \gamma$ invariant:
\begin{equation}
    \Lambda^{\beta} \v \gamma = \v \gamma.
\end{equation}
\end{definition}

To formally define the concept of thermomajorisation, we start by defining the concept of an embedding map. 

\begin{definition}[Embedding map]
Given a thermal distribution $\v{\gamma}$ with rational entries, \mbox{$\gamma_i=D_i/D$} and \mbox{$D_i,D\in\mathbb{N}$}, the embedding map $\Gamma$ sends a $d$-dimensional probability distribution $\v{p}$ to a $D$-dimensional probability distribution $\hat{\v{p}}:=\Gamma(\v{p})$ as follow:
	\begin{equation}
	\hat{\v{p}}=\left[\phantom{\frac{i}{i}}\!\!\!\!\right.\underbrace{\frac{p_1}{D_1},\dots,\frac{p_1}{D_1}}_{D_1\mathrm{~times}},
	\,\dots\,,\underbrace{\frac{p_d}{D_d},\dots,\frac{p_d}{D_d}}_{D_d\mathrm{~times}}\left.\phantom{\frac{i}{i}}\!\!\!\!\right]. \label{eq:embedding}
	\end{equation}
\end{definition}
\begin{marginfigure}[-4.3cm]
	\includegraphics[width=4.718cm]{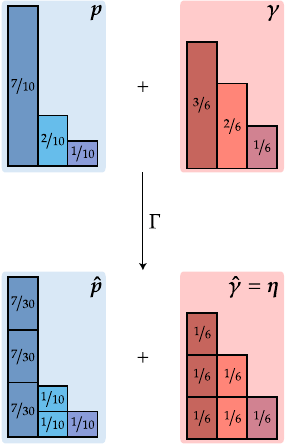}
	\caption{\emph{Embedding scheme.} Loosely speaking one can understand the embedding map as a transformation that allows for translating between different descriptions of a physical system. If $\v p$ represents a statistical description of a system in the canonical ensemble, then $\hat{\v p}$ describes the same state in the microcanonical ensemble~\cite{Egloff_2015}. This can be seen from the representation of the embedded version of $\v p$ -- the sum of smaller blocks yields the larger one. Note that the embedded version of a thermal start results in a flat state.
 }
	\label{Fig:embedding}
\end{marginfigure}
Irrational values of $\gamma_i$ can be approached with arbitrarily high accuracy by choosing a sufficiently large value of $D$. Although it is not necessary, for the purpose of using the embedding map as a technique, we will express the thermal distribution $\v \gamma$ as a probability vector with rational entries,
\begin{equation}\label{eq:definitiongibbsvec}
    \v \gamma = \qty(\frac{D_1}{D}, ..., \frac{D_d}{D}).
\end{equation}
Moreover, we will refer to the sets of repeated elements above as embedding boxes~[see Fig.~\ref{Fig:embedding} for an illustrative example]. 

\begin{kaobox}[frametitle=Embedding map example]
Given a thermal distribution $\v \gamma = (\nicefrac{3}{6},\nicefrac{2}{6},\nicefrac{1}{6})$ and a three-dimensional probability distribution $\v p =(\nicefrac{7}{10},\nicefrac{2}{10},\nicefrac{1}{10})$, the embedding map $\Gamma$ sends $\v p$ into
\begin{equation*}
\hat{\v p} = \Biggl(\underbrace{\frac{7}{30},\frac{7}{30},\frac{7}{30}}_{\frac{7}{10}},\underbrace{\frac{2}{20},\frac{2}{20}}_{\frac{2}{10}},\frac{1}{10}\Biggl). 
\end{equation*}
See Fig.~\ref{Fig:embedding} for an illustrative example.
\end{kaobox}

We also observe the following effect by embedding both the maximally mixed states and the sharp states:

\begin{observation}[Embedded maximally mixed \& sharp states]
The embedded version of a thermal state $\v \gamma$ is a maximally mixed state over $D$ states
\begin{equation}
    \hat{\v\gamma} = \v \eta = \Bigl(\frac{1}{D}, ..., \frac{1}{D}\Bigl)
\end{equation}
and the embedded version of a sharp state $\v{s}_k$ is a \emph{flat state} that is maximally mixed over a subset of $D_k$ entries, with zeros otherwise:
\begin{equation}
\hat{\v{s}}_{k}=\v{f}_{k}:=\frac{1}{D_k}\left[\phantom{\frac{i}{i}}\!\!\!\!\right.\underbrace{0,\dots,0}_{\sum_{j=1}^{k-1}D_j},
\,\underbrace{1\dots 1}_{D_k}\,,\underbrace{0,\dots,0}_{\sum_{j=k+1}^d D_j}\left.\phantom{\frac{i}{i}}\!\!\!\!\right]\!. \label{eq:flat}
\end{equation}
\end{observation}

We now prove that the existence of a Gibbs-preserving matrix between $\v p$ and $\v q$ is equivalent to the existence of a bistochastic matrix between them:
\begin{equation}\label{Gibbs-preserving-and-bistochastic}
    \Lambda^{\beta} \v p = \v q \Longleftrightarrow \hat{\Lambda}^{\beta} \hat{\v p} = \hat{\v q},
\end{equation}
where $\hat{\Lambda}^{\beta} = \Gamma \Lambda^{\beta} \Gamma^{-1}$ is the embedded version of $\Lambda^{\beta}$. This in turn, allow us to state the following result:
\begin{lemma}[Embedded GP matrix]
The embedded version of a Gibbs-preserving matrix $\Lambda^{\beta}$ is a bistochastic matrix $\Lambda^0$.
\end{lemma}
\begin{proof}
The matrix elements of $\hat{\Lambda}^{\beta}$ are given by
\begin{equation}
    \hat{\Lambda}^{\beta}_{ij} = \sum_{k,l=1}^d\Gamma_{ik} \Lambda^{\beta}_{kl} \Gamma^{-1}_{lj},
\end{equation}
so the conditions for bistochasticity of $\hat{\Lambda}^{\beta}_{ij}$ yields
\begin{equation}\label{Eq:conditions-bistochasticity}
    \forall i: \sum_{j=1}^D \sum_{k,l=1}^d\Gamma_{ik} \Lambda^{\beta}_{kl} \Gamma^{-1}_{lj} = 1 \quad , \quad \forall j: \sum_{i=1}^D \sum_{k,l=1} \Gamma_{ik} \Lambda^{\beta}_{kl} \Gamma^{-1}_{lj} =1.
\end{equation}
Using the explicit forms\sidenote{By writing out the embedding matrix $\Gamma$, we obtain
\begin{equation*}
    \Gamma = \begin{pmatrix}
\frac{1}{D_1} & 0  & \cdots & 0 \\
\vdots & \vdots & \cdots & 0 \\
\frac{1}{D_1} & 0  & \cdots  & 0 \\
0 & \frac{1}{D_2} & \cdots & 0 \\
\vdots & \vdots & \cdots & 0 \\
0 & \frac{1}{D_2}  & \cdots  & 0 \\
\vdots & \vdots & \cdots & \vdots \\
0 & 0 & \cdots & \frac{1}{D_d} \\
\vdots & \cdots & \cdots & \vdots \\
0& 0& \cdots & \frac{1}{D_d} \\
\end{pmatrix} .
\end{equation*}
The left inverse $\Gamma^{-1}$ can be readily computed as
\[
\setlength{\arraycolsep}{2pt}
\Gamma^{-1} = \begin{pmatrix}
1 & \cdots &1 & 0 &\cdots &0 &\cdots &\cdots &\cdots &0 \\
0 &\cdots &0 &1 &\cdots &1 &\cdots &\cdots &\cdots &0 \\
\vdots & \vdots & \vdots &\vdots &\vdots &\vdots &\vdots &\vdots &\vdots &\vdots \\
0&\cdots &0 &0 &\cdots &0 &\cdots &1 &\cdots & 1 \\
\end{pmatrix}.
\]
\setlength{\arraycolsep}{5pt}
} of $\Gamma$ and $\Gamma^{-1}$, we note the following. First, for all $l$, we have $\sum_j \Gamma^{-1}_{lj} = D_l$. Second, for all $k$, it follows that $\sum_i \Gamma_{ik} =1$. Thus, we can simplify the
conditions specified by Eq.~\eqref{Eq:conditions-bistochasticity} to obtain the following
\begin{equation}
\forall i: \sum^{d}_{k,l=1}\Gamma_{ik} \Lambda^{\beta}_{kl} \Gamma^{-1}_{lj} =1 \quad , \quad  \forall j: \sum_{k,l=1} \Gamma_{ik} \Lambda^{\beta}_{kl} \Gamma^{-1}_{lj} =1.  
\end{equation}
Taking into account that for a fixed $i$ and $j$ there is just one non-zero element of $\Gamma$, $\Gamma_{ik} = 1/D_k$, and one non-zero element of $\Gamma^{-1}$, $\Gamma^{-1}_{lj} =1$, we get
\begin{equation}
    \sum_{l=1}^d \Gamma_{kl}D_l = D_k \quad , \quad \sum_{k=1}^d \Lambda^{\beta}_{kl} =1.
\end{equation}
The first condition is fulfilled because $\Lambda^{\beta}$ is Gibbs-preserving (recall that $\gamma_l = D_l/D$) and second condition is fulfilled because $\Lambda^{\beta}$ is a stochastic matrix.
\end{proof} 

The fact that the embedded version of a GP matrix is a bistochastic matrix enables us to make a connection to majorisation, as demonstrated by the following lemma.
\begin{lemma}[GP matrix \& majorisation]\label{Lemma:GP-majorisation}
There exists a Gibbs-stochastic matrix $\Lambda^{\beta}$ such that $\Lambda^{\beta}\v p = \v q$ if and only if $\Gamma(\v p) \succ \Gamma(\v q)$.
\end{lemma}
\begin{proof}
As shown in Equation~\eqref{Gibbs-preserving-and-bistochastic}, the existence of a Gibbs-preserving matrix between vectors $\v p$ and $\v q$ is equivalent to the existence of a bistochastic matrix between the pair of embbeded vectors $\hat{\v p}$ and $\hat{\v q}$. By leveraging this observation and combining it with Theorem~\ref{Thm:HLP}, we prove Lemma~\ref{Lemma:GP-majorisation}.
\end{proof}

Finally, one can establish a relationship that generalises the concept of majorisation. Specifically, we will show how the requirement that the embedded distribution $\hat{\v{p}}$ majorises $\hat{\v{q}}$ can be restated as a thermomajorisation condition using only $\v{p}$ and $\v{q}$. Let us start by first defining the thermodynamic-ordering known as $\beta$-ordering:
\begin{restatable}{definition}{betaordering}
    \emph{($\beta$-ordering)}. Let $\v p$ and $\v \gamma$ be a probability vector and its corresponding thermal Gibbs distribution. The $\beta$-ordering of $\v p$ is defined as the permutation $\pi_{\v p}$ that arranges the vector \mbox{$(p_1/\gamma_1, ..., p_d/\gamma_d)$} in a non-increasing order, i.e., 
\begin{equation}\label{Eq:beta-ordering}
\v p^{\, \beta} = \Bigl(p_{\v \pi^{-1}_{\v p}(1)}, ...,p_{\v \pi^{-1}_{\v p}(d)}\Bigl) .    
\end{equation}
Each permutation belonging to the symmetric group, $\v \pi \in \S_{d}$, defines a different $\beta$-ordering on the energy levels of the Hamiltonian $H$~[see Fig.~\ref{Fig:thermalchambers}].
\end{restatable}
\begin{marginfigure}[-4.6cm]
        \centering
	\includegraphics[width=3.818cm]{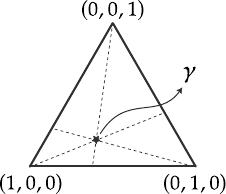}
	\caption{\emph{Chambers and $\beta$-ordering}. Each permutation defines a different \mbox{$\beta$-ordering}, which, in the space of states, is represented by a chamber. In the case of a three-level system, there are six distinct $\beta$-orderings, corresponding to six chambers. The asymmetry within each chamber arises due to the fact that $\beta > 0$. The thermal state $\v \gamma$ is depicted by a black star $\star$. }
	\label{Fig:thermalchambers}
\end{marginfigure}

The main motivation for introducing the $\beta$-ordering is due to the fact that sorting the embedded distribution $\hat{\v{p}}$ in non-increasing order corresponds to the $\beta$-ordering of the $d$-dimensional probability distribution $\v p$. To clarify, the vector $\hat{\v{p}}^{\downarrow}$ comprises groups of embedding boxes arranged in a non-increasing order, with each group consisting of $D_i$ elements. If each group is replaced with its sum, which is equal to $p_i$, one would obtain a $\beta$-ordered version of $\v p$.

\begin{kaobox}[frametitle=$\beta$-ordering example]
Given a thermal distribution $\v \gamma = (\nicefrac{3}{6},\nicefrac{2}{6},\nicefrac{1}{6})$ and the probability distributions $\v p =(\nicefrac{7}{10},\nicefrac{2}{10},\nicefrac{1}{10})$ and $\v q = (\nicefrac{1}{5},\nicefrac{16}{25},\nicefrac{4}{25})$, the $\beta$-ordered vectors of $\v p$ and $\v q$ are
\begin{equation*}
\v p^{\beta} = \qty(\frac{7}{10},\frac{2}{10},\frac{1}{10})\:\: , \:\: \v q^{\beta} = \qty(\frac{16}{25},\frac{4}{25},\frac{1}{5}).
\end{equation*}
\end{kaobox}
Next, we introduce the generalisation of the concept of majorisation curves, i.e., \emph{thermomajorisation curves}:
\begin{definition}[Thermomajorisation curve]\label{def_thermomajorisation_curve} Let $\v p$ and $\v \gamma$ be a probability vector and its thermal Gibbs distribution, respectively. The thermomajorisation curve is a piece-wise linear curve $f^{\beta}_{\v p}(x)$ in $\mathbb{R}^2$ obtained by joining the origin $(0,0)$ and the points
\begin{equation}
    \left(\sum_{i=1}^k\gamma^{\, \beta}_i,~\sum_{i=1}^k p^{\, \beta}_i\right):=\left(\sum_{i=1}^k\gamma_{\v \pi^{-1}_{\v{p}}(i)},~\sum_{i=1}^k p_{\v \pi^{-1}_{\v{p}}(i)}\right),
\end{equation}
for $k\in\{1,\dots,d\}$. 
\end{definition}
Given two $d$-dimensional probability distributions $\v p$ and $\v q$, and a fixed inverse temperature $\beta$, we say that $\v p$ \emph{thermomajorises} $\v q$ and denote it as $\v p \succ_{\beta} \v q$\sidenote{Historically, thermomajorisation should be denoted by $\succ_{\v{\gamma}}$ since the vector being preserved is $\v{\gamma}$. However, the community has adopted $\succ_{\beta}$, given that $\v{\gamma}$ is associated with the inverse temperature.}, if the thermomajorisation curve $f^{\, \beta}_{\v{p}}$ is above $f^{\, \beta}_{\v{q}}$ everywhere, i.e.,
\begin{equation}
    \v p \succ_{\beta} \v q \iff \forall x\in[0,1]:~ f^{\, \beta}_{\v{p}}(x) \geq f^{\, \beta}_{\v{q}}(x) \, .
\end{equation}
\begin{figure*}
\includegraphics[width=15.940cm]{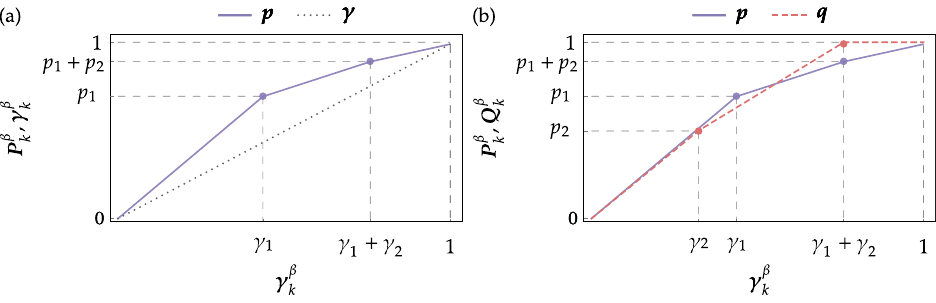}
	\caption{\emph{A geometric view of thermomajorisation}. (a) Thermomajorisation curves for a three-dimensional probability distribution $\v p$ with $\beta$-order (1,2,3) are presented alongside the curve associated with the thermal state $\v \gamma$. Panel (b) displays the thermomajorisation curves of $\v p$ and $\v q$, the latter having $\beta$-ordering (2,1,3). Both vectors are incomparable since their thermomajorisation curves intersect.}
	\label{Fig:thermo-majorisation-curves}
\end{figure*}

As before, it may happen that given two vectors, $\v p$ and $\v q$, are incomparable, meaning that $f^{\, \beta}_{\v p}$ and $f^{\, \beta}_{\v q}$ cross at a some point (see Fig.~\hyperref[Fig:thermo-majorisation-curves]{\ref{Fig:thermo-majorisation-curves}b}). Importantly, in the case of uniform equilibrium distributions, $\v \gamma$ = $\v \eta$, thermomajorisation is replaced by majorisation. For a fixed dimension, the sharp distribution with largest energy, \mbox{$\v s_d = (0, .., 1)$}, thermomajorises every other distribution, and every distribution thermomajorises $\v \gamma$.

As a result, we can show the following:
\begin{lemma}[Majorisation between embedded vectors]
The initial state $\v p$ thermomajorises the target state $\v q$ if and only if $\hat{\v p} \succ \hat{\v q}$.
\end{lemma}
\begin{proof}
The non-increasing order of $D$-dimensional probability distributions $\hat{\v p}$ and $\hat{\v q}$ corresponds to $\beta-$ordering the $d$-dimensional probability distributions $\v p$ and $\v q$. Moreover, by summing up each embedding box, we obtain the $\beta$-ordered versions of $\v p$ and $\v q$. Now, note that the majorisation curve of $\hat{\v p}$ and $\hat{\v q}$ will be composed of linear segments connecting the points $(\sum_{i=1}^k D^{\beta}_i, \sum_i x^{\beta}_i)$, where $x_i = p_i, q_i$ and $\gamma_i = D_i/D$. As a result, the curves $f_{ \hat{\v{p}}}$ and $f_{\hat{\v q}}$ correspond exactly to the thermomajorisation curves of $\v p$ and $\v q$, and therefore, $\v p \succ_{\beta} \v q$ is equivalent to $\hat{\v p} \succ \hat{\v q}$.
\end{proof}

Combining all these results, Theorem~\ref{Thm:HLP} can be generalised from bistochastic matrices
to Gibbs-preserving matrices:

\begin{restatable}{theorem}{GPthermo}\emph{(Thermomajorisation $\&$ GP matrices).} There exists a Gibbs-preserving matrix $\Lambda^{\beta}$ mapping $\v p$ to $\v q$ if and only if $\v p \succ_{\beta} \v q$. \label{Thm:HLPbeta}
\end{restatable}

\subsection{Continuous thermomajorisation}
Majorisation and thermomajorisation are relationships linking the initial and target distributions, while nothing is said about the path in $\Delta_d$ connecting these two distributions. One can then ask if there exists a continuous path within the probability simplex connecting these two distributions, such that the preceding distribution is thermomajorised by the succeeding one at any two points along this path. This question raises the concept of \emph{continuous thermomajorisation}~\cite{Lostaglio2019}, which can be mathematically defined as follows~(see also Fig.~\ref{Fig:continuous}).
\begin{marginfigure}[0cm]
	\includegraphics[width=4.718cm]{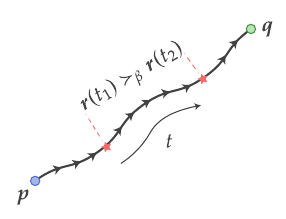}
	\caption{\emph{Continuous thermomajorisation.} A probability vector $\v{p}$ continuously thermomajorises $\v{q}$ if there exists a continuous probability path $\v{r}(t)$ connecting both vectors, ensuring that $\v{r}(t_1) \succ_{\beta} \v{r}(t_2)$ whenever $t_1 \leq t_2$}
	\label{Fig:continuous}
\end{marginfigure}
\begin{definition}[Continuous thermomajorisation]
	\label{def:markov_majo}
 	A distribution $\v{p}$ \emph{continuously thermomajorises}, denoted $\v{p} \ggcurly_\beta \v{q}$, if there exists a continuous path of probability distributions $\v{r}(t)$ for $t\in[0,t_f)$ such that
 	\begin{enumerate}
 		\item $\v{r}(0)=\v{p}$,
 		\item $\forall~ t_1,t_2\in[0,t_f):\quad t_1 \leq t_2 \Rightarrow\v{r}(t_1)\succ_\beta \v{r}(t_2)$,
 		\item $\v{r}(t_f)=\v{q}$.
 	\end{enumerate}
The path $\v{r}(t)$ is called \emph{thermomajorising trajectory} from $\v{p}$ to $\v{q}$.
\end{definition}

Let us make a few comments about the aforementioned definition. Firstly, it is worth noting that when there is a uniform fixed point, $\v{\gamma}=\v{\eta}$, the above definition corresponds to a continuous version of standard majorisation, denoted by $\ggcurly$~\cite{zylka1985note}. Secondly, determining whether a given initial state continuously thermomajorises a target state is a difficult problem, and unlike the other variants of majorisation presented so far, there is no continuous thermomajorisation curve for this type of majorisation that would facilitate a quick check. Nevertheless, the necessary and sufficient conditions are known~\cite{lostaglio2021continuous}. These comprise a complete set of entropy production inequalities that can be reduced to a finitely verifiable set of constraints. Moreover, Ref.~\cite{lostaglio2021continuous} presents an explicit algorithm for verifying the continuous thermomajorisation relation between any two vectors. An implementation of this algorithm in Mathematica is provided in Ref.~\cite{korzekwaalgorithm}.

To understand how the concept of continuous thermomajorisation relates to the existence of given processes between two distributions, we introduce the following definition~\cite{PhysRevX.11.021019}.

\begin{definition}[Markovian classical dynamics]
	\label{def:reach_class_mem}
	A distribution $\v{p}$ can be mapped to $\v{q}$ by a classical master equation with a fixed point $\v{\gamma}$ if a continuous one-parameter family of rate matrices, denoted by $L(t)$, exists. This family generates a set of stochastic matrices represented by $\Lambda^{\beta}(t)$, such that
\begin{equation}
\Lambda^{\beta}(t_f)\v{p} = \v{q}, \quad L(t)\v{\gamma} = \mathbf{0} \quad \text{for all } t \in [0,t_f).
\end{equation}
\end{definition}

This definition naturally generalise to quantum dynamics in the following way

\begin{definition}[Markovian quantum dynamics]
	\label{def:reach_quant_mem}
A distribution $\v{p}$ can be mapped to $\v{q}$ via a quantum master equation with a fixed point $\v{\gamma}$ if there exists a continuous one-parameter family of Lindbladians, denoted by $\mathcal{L}(t)$. This family generates quantum channels $\mathcal{E}^{\beta}(t)$ such that
\begin{equation}
\mathcal{E}^{\beta}(t_f)[\rho_{\v{p}} ]= \rho_{\v{q}}, \quad \mathcal{L}(t)[\rho_{\v{\gamma}}] = 0 \quad \textrm{for all } t \in [0,t_f).
\end{equation}
where $\rho_{\v t} = \sum_k t_k \ketbra{k}{k}$. 
\end{definition}

The previous definitions encapsulate the set of transformations realised by Markovian master equations with a fixed point $\v \gamma$. As will be properly defined in Section~\ref{C:resource_theory_of_thermodynamics}, we will refer to such quantum dynamics as \emph{Markovian thermal processes}. Crucially, the concept of continuous thermomajorisation encompasses all constraints of memoryless processes on population dynamics. This is expressed in the following theorem:

\begin{theorem}[Continuous thermomajorisation \& Markovian thermal processes] There exists a Markovian thermal processes mapping $\v p(0)$ to $\v p(t_f)$ if and only if $\v p \ggcurly_{\beta} \v p(t_f)$.
\end{theorem}
The proof of the above theorem can be found in Appendix A of Ref.~\cite{Lostaglio2019}. As a result, continuous thermomajorisation provides a full set of constraints for the population's evolution.

\begin{proposition}[Equivalence between thermomajorisations] If a pair of states, $\v p$ and $\v q$, have the same $\beta$-ordering, and $\v p \succ_{\beta} \v q$, then $\v p \ggcurly_{\beta} \v q$. 
\end{proposition}
The proof of the above proposition can be found in Appendix A of Ref.~\cite{Lostaglio2019}. 

Given that continuous thermomajorisation characterises Markovian processes, while thermomajorisation characterises general processes, including non-Markovian ones, the aforementioned result highlights that all the complications associated with Markovianity (or advantages arising from non-Markovianity) stem from crossing the boundary between Weyl chambers. This observation will play a crucial role in Chapter~\ref{C:memory-MTP} as we explore the significance of memory effects in stochastic processes with a fixed point.

\subsection{Approximate thermomajorisation}

\begin{marginfigure}[-1.5cm]
	\includegraphics[width=4.718cm]{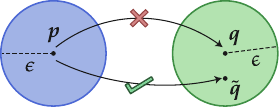}
    \caption{\emph{Approximate thermomajorisation}. While an input distribution $\v p$ may not thermomajorise a target distribution $\v q$, it can still thermomajorise a distribution $\tilde{\v q}$ that closely approximates the target. The $\epsilon$-circle represents sets of probability distributions that are each at a distance less than $\epsilon$ from both $\v p$ and $\v q$.}
	\label{Fig:approximate-majorisation}
\end{marginfigure}

Although an initial distribution $\v p$ may not thermomajorise a target distribution $\v q$, it may thermomajorise a final distribution $\tilde{\v q}$ that is close to the target\footnote{Note that the same definition applies to majorisation, as this represents a particular case of thermomajorisation when $\beta=0$.}. This idea gives rise to the concept of approximate thermomajorisation, meaning that for a given distribution and a specified distance $\delta$, we allow the final state to differ from the target state, as long as it is close enough~(see Fig.~\ref{Fig:approximate-majorisation}).

We will measure the distance between states using the \emph{infidelity},
\begin{equation}\label{eq:DefFidelity}
	\delta(\v{p},\v{q}):=1-F(\v{p},\v{q}) \,,
\end{equation}
where the \emph{fidelity} (or Bhattacharyya coefficient) is
\begin{equation}
F(\v{p},\v{q}):= \left( \sum_{i=1}^d \sqrt{p_i q_i} \right)^2 \,.
\end{equation}

The two important properties of the infidelity that we will utilise throughout this thesis are as follows. First, since fidelity is non-decreasing under stochastic maps, we have
\begin{equation}
\delta(\Lambda\v{p},\Lambda\v{q})\leq\delta(\v{p},\v{q}).
\end{equation}
Second, the distance $\delta$ between two probability vectors is identical to that between their embedded versions, i.e.,
\begin{equation}
\delta(\v{p},\v{q})=\delta(\hat{\v{p}},\hat{\v{q}}),
\end{equation}
which can be verified by direct calculation.

Two commonly referred to definitions of approximate majorisation and thermomajorisation are known as 'pre' and 'post', which are defined as follows~\cite{Chubb2018beyondthermodynamic}:
\begin{definition}[Pre-and-post-thermomajorisation]
    A distribution $\v{q}$ is said to be "$\epsilon$-pre-thermomajorised" by a distribution $\v{p}$, denoted as $\v{p} _{\epsilon}{\succ_\beta} \v{q}$, if there exists a distribution $\tilde{\v{p}}$ such that:
\begin{align}
\tilde{\v{p}}\succ_\beta \v{q} \quad \text{and} \quad \delta(\v{p},\tilde{\v{p}})\leq \epsilon,
\end{align}
where $\delta(\cdot,\cdot)$ represents the infidelity.

On the other hand, a distribution $\v{q}$ is said to be "$\epsilon$-post-thermomajorised" by a distribution $\v{p}$, denoted as $\v{p} \succ_\beta^\epsilon \v{q}$, if there exists a distribution $\tilde{\v{q}}$ such that:
\begin{align}
\v{p}\succ_\beta \tilde{\v{q}} \quad \text{and} \quad \delta(\v{q},\tilde{\v{q}})\leq \epsilon.
\end{align}
In the special case where $\beta=0$ and thermomajorisation corresponds to majorisation, we refer to it as pre- and post-majorisation, denoted as $_{\epsilon}{\succ}$ and $\succ^\epsilon$, respectively.
\end{definition}

Let us provide some comments regarding the definition mentioned above. Firstly, it is worth noting that according to Theorem~\ref{Thm:HLPbeta}, $\v{p} _{\epsilon}{\succ_\beta}~ \v{q}$ implies the existence of a state $\tilde{\v{p}}$ in the vicinity of $\v{p}$ and a mapping $\Lambda^\beta$ that transforms it into $\v{q}$. Similarly, $\v{p} \succ_\epsilon^\beta \v{q}$ implies the existence of a mapping $\Lambda^\beta$ that transforms $\v{p}$ into $\tilde{\v{q}}$, which lies in the vicinity of $\v{q}$. 

\begin{center}
    \textbf{\emph{Vidal, Jonathan and Nielsen's algorithm for approximate majorisation} }
\end{center}

We now describe the algorithm introduced in Ref.~\cite{vidal2000approximate} for constructing the optimal distribution satisfying the following constraint:
\begin{equation}
    \tilde{\v p}^{\star} = \underset{\tilde{\v p}: \tilde{\v p} \succ \v q}{\operatorname{arg max}} \: F(\v p, \tilde{\v p}).
\end{equation}
In other words, for a given pair of states $\v p$ and $\v q$, such that $\v p \nsucc \v q$, we seek the optimal $\tilde{\v p}^{\star}$, which maximises the fidelity for a given $\epsilon$.

We begin by assuming that all the distributions are ordered in a non-increasing manner. Then, for any distribution $\v t$, we define the following quantity:
\begin{equation}
    E^{\v t}_k:= \sum^{d}_{i=k} t_i \quad , \quad \Delta^{k_2}_{k_1}(\v t):= \sum^{k_2-1}_{i=k_1}t_i = E^{\v t}_{k_1}-E^{\v t}_{k_2}.
\end{equation}
It is straightforward to verify that if $\v p \succ \v q$, then $E^{\v p}_k \leq E^{\v q}_k$ for all k. Now, for a given $\v p$ and $\v q$, the construction of $\tilde{\v p}^{\star}$ is obtained by the following iterative procedure. Set $l_0 = d+1$ and define"
\begin{equation}
    l_j :=  \underset{k < l_{j-1}}{\operatorname{arg max}}\frac{\Delta^{l_{j-1}}_k(\v q)}{\Delta^{l_{j-1}}_k(\v p)} \quad , \quad r_j := \frac{\Delta^{l_{j-1}}_{l_j}(\v q)}{\Delta^{l_{j-1}}_{l_j}(\v p)}.
\end{equation}

If the minimisation defining $l_j$ does not yield a unique solution, then $l_j$ is chosen to be the smallest possible value. We will also denote $N$ as an index for which $l_N=1$. The $i$-th entry of the optimal vector for $i\in{l_{j},\dots,l_{j-1}-1}$ is then given by
\begin{equation}
\tilde{p}^{\star}_i=r_jp_i.
\end{equation}

It is direct to verify that $\tilde{\v{p}}^\star$ is normalised, and $\tilde{\v{p}}^\star\succ\v{q}$, as the construction guarantees that \mbox{$E^{\v{p}}_k\leq E^{\v{q}}_k$} for all $k$. Furthermore, the optimal fidelity between $\v{p}$ and a distribution that majorises $\v{q}$ is given by
\begin{eqnarray}
\sqrt{F(\v{p},\tilde{\v{p}}^{\star})}&=&\sum_{i=1}^d\sqrt{p_i\tilde{p}_i^{\star}}=\sum_{j=1}^{N}\sum_{i=l_j}^{l_{j-1}-1} \sqrt{p_i\tilde{p}^\star_i} \nonumber\\
&=&\sum_{j=1}^{N} \left(\frac{\Delta_{l_{j}}^{l_{j-1}}(\v{q})}{\Delta_{l_j}^{l_{j-1}}(\v{p})}\right)^{\frac{1}{2}}\sum_{i=l_j}^{l_{j-1}-1} p_i
\nonumber\\
&=&\sum_{j=1}^{N} \left(\Delta_{l_{j}}^{l_{j-1}}(\v{q})\Delta_{l_j}^{l_{j-1}}(\v{p})\right)^{\frac{1}{2}}.
\end{eqnarray}
The critical observation, which is essential for proving the optimality of the above is that for all $j$, we have
\begin{equation}
r_j< r_{j+1}.
\end{equation}
This is deduced from the definitions of $l_j$, $r_j$, and the fact that for $a,b,c,d> 0$, one observes (see to Ref.~\cite{vidal2000approximate} for details)
\begin{equation}
\frac{a}{b}\leq\frac{a+c}{b+d}\quad\Longleftrightarrow\quad \frac{a}{b}<\frac{c}{d}.
\end{equation}

\begin{kaobox}[frametitle=Approximate majorisation]
Let $\v{p} = (0.7, 0.2, 0.1)$ and $\v{q} = (0.75, 0.13, 0.12)$. Since $p_1 < q_1$ and $p_1 + p_2 > q_1 + q_2$, the two vectors are incomparable. By using the algorithm mentioned above (available in Mathematica code form here~\cite{adeoliveirajunior2023code}), we obtain the approximate state $\tilde{\v{p}} = (0.75, 0.17, 0.08)$, with a fidelity of $F(\v{p}, \tilde{\v{p}}) = 0.774597$.
\end{kaobox}

\subsection{Catalytic majorisation}

As previously discussed, majorisation introduces a partial ordering in the space of states, resulting in not all states being comparable to each other. There generally exists a pair of states, $\v p$ and $\v q$, such that neither $\v p \succ \v q$ nor $\v q \succ \v p$. In this scenario, we refer to $\v p$ and $\v q$ as incomparable. Nevertheless, there exists a possibility of 'borrowing' an ancillary system $\v c$ that aids in performing a transformation, but its resources are not used up. In the language we have introduced, this implies that the composite system satisfies the following equation
\begin{equation}\label{Eq:catalysis}
    \v p \otimes \v c \succ \v q \otimes \v c .
\end{equation}
Therefore, the auxiliary system $\v c$ allows one to perform an otherwise impossible transformation without being disturbed, as it is returned unchanged at the end of the process. Such an auxiliary system is referred to as a \emph{catalyst}. Different ways of relaxing Eq.\eqref{Eq:catalysis} lead to different types of catalysis (see Section III of Ref.~\cite{lipkabartosik2023catalysis} for an extensive discussion on the topic). However, here we will primarily discuss the case where Eq.~\eqref{Eq:catalysis} is strictly satisfied, i.e., the state of the catalyst is precisely the same as at the beginning and the joint system does not become correlated. 

\begin{kaobox}[frametitle=Catalytic majorisation~\cite{jonathan1999entanglement}]
The pair of states $\v p = (0.5,0.25,0.25,0)$ and $\v q = (0.4,0.4,0.1,0.1)$ are incomparable as neither $\v p \nsucc \v q$ nor $\v q \nsucc \v p$. It can easily be checked that $p_1 > q_1$ but $p_1+p_2 < q_1 + q_2$. However, note that the two-level catalyst $\v c = (0.6,0.4)$, leads to the following joint states
\begin{align*}
    \v p \otimes \v c &= (0.30, 0.20, 0.15, 0.15, 0.10, 0.10, 0.00, 0.00), \\
     \v q \otimes \v c &= (0.24, 0.24, 0.16, 0.16, 0.06, 0.06, 0.04, 0.04), 
\end{align*}
so that now $\v p\otimes \v c \succ \v q\otimes \v c$.
\end{kaobox}

There are three important properties of catalytic transformations that deserves to be highlighted~\cite{jonathan1999entanglement, PhysRevA.99.052348}:

\begin{enumerate}[label=(P\arabic*)]
	\item \textbf{Largest and smallest element.} Given two incomparable states $\v p$ and $\v q$, along with a catalyst $\v c$, such that $\v{p} \otimes \v{c} \succ \v{q} \otimes \v{c}$, it follows that the largest element of $\v{q}$ is always smaller than the largest element of $\v{p}$, and conversely, the smallest element of $\v{q}$ is always larger than the smallest element of $\v{p}$
        \item \textbf{Uniform state.} No transformation can be catalysed by a uniform state $\v \eta$.
        \item  \textbf{Minimal dimension} When $d = 2$, either $\v{p} \succ \v{q}$ or $\v{q} \succ \v{p}$, rendering the use of a catalyst irrelevant. For two 3-dimensional incomparable states $\v{p}$ and $\v{q}$, there exists no catalyst vector $\v{c}$ such that $\v{p} \otimes \v{c} \succ \v{q} \otimes \v{c}$. Therefore, a catalytic transformation is only possible when $d \geq 4$
\end{enumerate}
Properties (P1), (P2), and (P3) also apply to thermomajorisation. However, in this context, the uniform state is replaced by the thermal Gibbs distribution. While thermomajorisation allows incomparability even when $d=2$, property (P3) continues to hold true in this scenario~\cite{brandao2015second}.

The conditions for a catalytic state-transformation were first derived in Ref.~\cite{Klimesh2007,turgut2007necessary}, and are known as trumping conditions. More precisely, Given two $d$-dimensional probability distributions $\v p$ and $\v q$ and a catalyst $\v c$, we say that $\v p$ trumps $\v q$ when
\begin{equation}
    \v p \otimes \v c \succ \v q \otimes \v c    \Longleftrightarrow H_{\alpha}(\v p) \leq H_{\alpha}(\v q).
\end{equation}

\section{Concluding remarks}

This chapter aimed to provide a concise introduction to the key concepts that will be explored throughout this thesis. We focused on presenting well-established results without delving into their proofs or intricate mathematical details. Since much of the mathematics in this thesis will take place in the probability simplex, special emphasis was placed on the properties of probability distributions, majorisation and its variants. Many notions outlined in this chapter will be further clarified in their specific contexts later on.

\chapter{Resource theory of thermodynamics}\label{C:resource_theory_of_thermodynamics}

Resource theories emerged from the desire to manipulate entanglement~\cite{RevModPhys.81.865,plenio2007introduction} and quickly spread to other areas of quantum information~\cite{RevModPhys.89.041003,theurer2017resource,Albarelli2018,Lostaglio2019,RevModPhys.79.555,RevModPhys.86.419}. As a theoretical toolkit capable of uncovering the limitations of a given system undergoing a certain process, it became a widely adopted framework in recent decades. At its core, it involves identifying and characterising which operations are easy and hard to implement when subject to constraints, such as locality, experimental difficulties in preparing specific superpositions, or fundamental restrictions induced by physical laws like energy conservation. This is where the term \emph{resource} comes into place. The natural question then is \emph{what distinguishes these resources as resources?} If one thinks about entanglement as a resource, Bell states would be the most valuable states among all others, while separable states would fall into the opposite category, namely \emph{free}. As a result, from a resource point of view, a given context results in a partial ordering of the set of quantum states, with the hardest to prepare at the top, and easiest at the bottom. A generic resource theory is determined by three ingredients~\cite{chitambar2019quantum}: (i) free operations $\FF$, (ii) free states $\OO$ and (iii) a set of monotones $\mathcal{M}$.

A class of free operations is a set of transformations that can be implemented at no cost. Similarly, free states are those that can be generated and used without any expenditure of resources. On the other hand, monotones are real-valued functions that quantify the amount of a resource in a given state that is not free. Given the aforementioned operations, free and resource states, one can ask about the fundamental limitations of manipulating states. More precisely, for a pair of states $\rho$ and $\sigma$, our interest lies in determining the possibility of transforming these states using the allowed operations. Naturally, this process involves identifying the set of necessary and sufficient conditions for the given transformation:
\begin{equation}\label{Eq:general-transformation}
    \rho  \xrightarrow[?]{\text{allowed operations}} \sigma.
\end{equation}
The conditions underlying Eq.~\eqref{Eq:general-transformation} can be either necessary or sufficient, and in some cases, both. These are typically expressed in terms of monotones, indicating that for a transformation to be possible one needs $\mathcal{M}(\rho) \geq \mathcal{M}(\sigma)$. However, it is not always the case that one can find a set of monotones that completely characterises a given transformation.

A celebrated example of this approach is the entanglement theory, which can be cast as a resource theory~\cite{RevModPhys.81.865}. The starting point is to assume that Alice and Bob are two distant parties who are capable of creating any local quantum states in their own labs and manipulating them via arbitrary local operations. Additionally, we also assume that both can easily communicate with each other classically. This defines a set of operation known as local operations and classical communication (LOCC) that one takes as a free set. Such operations allow Alice and Bob to create a joint, correlated bipartite quantum states, $\rho_{AB}$. However, by imposing the constraint that Alice and Bob can only employ LOCC, they are unable to create any entanglement. As a result, the joint quantum state $\rho_{AB}$ is separable and, hence, free. Thus, any initial entangled state that Alice and Bob share becomes a valuable resource. This pre-existing entanglement can be harnessed for practical purposes, such as teleportation~\cite{PhysRevLett.70.1895,bouwmeester1997experimental}. In the light of Eq.~\eqref{Eq:general-transformation}, suppose Alice and Bob jointly possess a pure state $\ket{\psi}$. Using local operations and classical communication, one can ask whether is possible for Alice and Bob to transform $\ket{\psi}$ into another state $\ket{\phi}$:
\begin{equation}\label{Eq:general-transformation-LOCC}
    \ket{\psi}  \xrightarrow[]{\text{LOCC}} \ket{\phi} 
\end{equation}
The answer to this question will be presented in Chapter~\ref{C:thermal_cones}. The key takeaway is that is for transformations involving bipartite pure states [as in Eq.~\eqref{Eq:general-transformation-LOCC}], the necessary and sufficient conditions are known~\cite{nielsen1999conditions}.

When discussing thermodynamics, one of the most fundamental questions to ask is \emph{what state transformations systems can undergo while interacting with a thermal bath}. Since the laws of thermodynamics impose fundamental constraints on how states can be manipulated, one can adopt the resource-theoretic perspective and cast this theory as resource theory. Several different approaches have been explored~\cite{gour2015resource, yunger2016microcanonical, halpern2016beyond, sparaciari2017resource,halpern2018beyond, singh2019quantum}, but almost all can be described within a common framework. The central idea is to define free operations as those that conserve certain extensive properties (such as energy, particle number, etc.) when a given system interacts with a bath. Our focus is on the resource theory of athermality~\cite{brandao2013resource}, which involves just energy conservation. The concept of free operations, initially introduced by Janzing~\cite{Janzing2000} and subsequently expanded upon~\cite{horodecki2013fundamental,brandao2015second,Lostaglio2019,ng2019resource}, encompasses all physical dynamics that conserve the total energy when the system interacts with a thermal bath. These processes, referred to as \emph{thermal operations}~\cite{horodecki2013quantumness}, represent the most general approach to describing the joint system-bath dynamics. 

Such an approach allows one to go beyond the thermodynamic limit and the assumption of equilibrium. Its generality allows us to describe the structure of non-equilibrium states and quantify the thermodynamic properties of systems without explicitly solving for the specific dynamics. This opens the door to exploring how quantum properties can be harnessed and sheds light on optimal strategies for certain thermodynamic tasks. 

In this chapter, we introduce the resource-theoretic approach to thermodynamics and primarily cover well-known results. The concepts discussed here will serve as fundamental building blocks for the rest of the thesis. The chapter is organised as follows: firstly, we formally introduce the set of free operations known as \emph{thermal operations}. Next, we provide a brief interlude on information-theoretic quantities and their thermodynamic interpretations. The subsequent section focuses on addressing Eq.~\eqref{Eq:general-transformation} in the context of thermal operations and energy-incoherent states. Subsequently, we explore the applications of the aforementioned framework. We conclude the chapter by revisiting the open quantum system approach to quantum thermodynamics.

\newpage

\begin{marginfigure}[3cm]
	\includegraphics[width=44.046mm]{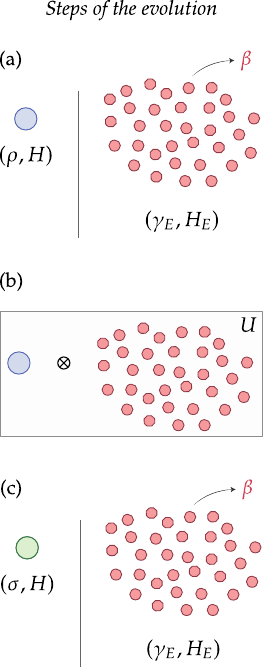}
	\caption{\emph{Thermal operation snapshots.} (a) A quantum system (represented by a blue circle), previously isolated, is coupled to a heat bath (shown as a collection of small red circles). (b)~The joint system is assumed to be closed and thus evolves unitarily, conserving the total energy. (c)~After some time, the system (now in a different state and represented by a green circle) is decoupled from the bath.}
	\label{Fig:system-bath}
\end{marginfigure} 

\section{Resource theory of states out of equilibrium}

Consider a finite-dimensional quantum system in the presence of a heat bath at temperature $T$\footnote{Throughout this thesis, the terms "temperature $T$" and "inverse temperature $\beta$" will be used interchangeably}. The system under investigation is described by a Hamiltonian $H$, and it is prepared in a state $\rho$; while the heat bath, with a Hamiltonian $H_E$, is in a thermal equilibrium state,
\begin{equation}
\label{Eq:thermal_state}
    \gamma_E=\frac{e^{-\beta H_E}}{\textrm{tr}({e^{-\beta H_E}})},
\end{equation}
where $\beta=1/k_B T$ is the inverse temperature with $k_B$ denoting the Boltzmann constant. 

\begin{kaobox}[frametitle= $N$ particles system]
This framework can also be employed when the initial system is composed of $N$ non-interacting subsystems with the total Hamiltonian $H^N$ and a state $\rho^N$ given by
\begin{equation}\label{eq:N-noninteracting}
    H^N=\sum_{n=1}^N H_{n},\qquad \rho^N=\bigotimes_{n=1}^N \rho_{n}.
\end{equation}
A typical example of this setting is when initial and target systems are given by copies of independent and identical subsystems. More precisely, in this case, the family of initial systems is given by $H^N$ with $H^N_n=H$ and $\rho^N=\rho^{\otimes N}$.
\end{kaobox}

The evolution of the joint system is assumed to be closed and described by a unitary operator $U$ that is further required to conserve the total energy~(see Fig.~\ref{Fig:system-bath} for a pictorial representation)
\begin{equation}
\label{eq_energy_conservation}
    [U, H\otimes \iden_E+ \iden\otimes H_E] = 0.
\end{equation} 
This assumption is typically made to ensure that the transformation is consistent with the laws of thermodynamics. Moreover, note that in general $U$ does not commute with $H$ and $H_E$
individually, only with their sum. As a consequence, the set of allowed thermodynamic transformations is modelled by a set of operations called \emph{thermal operations} (TOs), which are defined as follows:
\begin{definition}[Thermal operations]\marginnote{One can also define \emph{catalytic thermal operations}~\cite{brandao2015second}. In addition to thermal operations, an ancilla, termed the catalyst, is borrowed under the condition that it is returned in the same state, in product form with the other systems.}\label{def:thermal-operations}
The set of thermal operations (TOs) consists of completely positive, trace-preserving (CPTP) maps that act on a system as
\begin{equation}
    \label{eq:thermal_ops}
    \mathcal{E}^{\beta}(\rho)=\textrm{Tr}_E\left[U\left(\rho\otimes\gamma_E\right)U^{\dagger}\right],
\end{equation}
with $U$ satisfying Eq.~\eqref{eq_energy_conservation} and the state $\gamma_E$ given by Eq.~\eqref{Eq:thermal_state} with an arbitrary Hamiltonian $H_E$.
\end{definition}

Thermal operations are designed to be as general as possible and capture a wide range of transformations, including both reversible and irreversible, as well as Markovian and non-Markovian ones. 

Note that if the system was initially out-of-equilibrium with respect to the thermal reservoir when they were brought into contact, then it is natural to expect that a thermal state can be prepared without the expenditure of any resource by letting the system equilibrate with the thermal bath. In this case, the only free state is the thermal Gibbs state and therefore, any out-of-equilibrium state is a resource. From the perspective of the channel defined in Eq.~\eqref{eq:thermal_ops}, one can uncover two essential properties:
\begin{enumerate}[label=(P\arabic*)]
	\item \begin{marginfigure}[-0.567cm]
	\includegraphics[width=4.753cm]{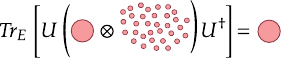}
	\caption{\emph{The second law}. Preserving the Gibbs state prevents the possibility of constructing a perpetual machine. Without this property, it would be possible to extract work from a system and re-equilibrate it, thereby contradicting the second law.}
	\label{TO_property_1}
\end{marginfigure}
        \textbf{Preserves the thermal state.} The Gibbs state of the system, 
	\begin{equation}
		\label{Eq:thermalstate}
		\gamma = \frac{e^{-\beta H}}{\textrm{tr}({e^{-\beta H}})},
	\end{equation}
	is a fixed point of the dynamics, i.e.,
	\begin{equation}
	       \mathcal{E}^{\beta}(\gamma) = \gamma. 
	\end{equation} 
	\item \label{p2} \begin{marginfigure}[-0.353cm]
	\includegraphics[width=4.756cm]{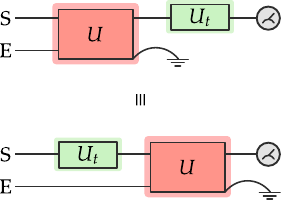}
	\caption{\emph{The first law}. This symmetry reflects the absence of an external time reference. If the thermal environment affected the system differently depending on when a thermal operation is applied, it would be possible to extract ``time-like'' information from the final state.}
	\label{TO_property_2}
\end{marginfigure}
        \textbf{Covariance.} Thermal operations are time-translation covariant,
	\begin{equation}
		\mathcal{E}^{\beta}(e^{-i Ht}\rho e^{i Ht}) = e^{-i Ht}\mathcal{E}^{\beta}(\rho) e^{i Ht}.
	\end{equation}
\end{enumerate}
The first property, combined with a contractive distance measure $\delta$ that satisfies $\delta(\rho,\gamma)\geq\delta[\mathcal{E}^{\beta}(\rho),\mathcal{E}^{\beta}(\gamma)]=\delta[\mathcal{E}^{\beta}(\rho),\gamma]$, reflects the fundamental principle of the second law of thermodynamics. This principle states that a system will evolve towards thermal equilibrium, meaning it cannot be transformed into any other out-of-equilibrium state without cost -- both properties are illustrated in Figs.~\ref{TO_property_1}~and~\ref{TO_property_2}. 

\begin{kaobox}[frametitle=A physical model of TO~\cite{lostaglio2018elementary}]
The Jaynes-Cummings model reproduces the behaviour of the set of thermal operations for a range of physically relevant parameters. This model describes the interaction between an atom with a single electromagnetic field mode in an optical cavity. 

In this case, the cavity is modelled by bosonic creation and annihilation operators $a^{\dagger}$ and $a$, with Hamiltonian \mbox{$H_{\ms C} = \omega_{\ms C} a^{\dagger} a$}, where $\omega_C$ is the angular frequency of the mode. The atom, a two-level system with energy levels $\ket{g}$ and $\ket{e}$, and raising and lowering operators \mbox{$\sigma_+ = \ketbra{e}{g}$} and \mbox{$\sigma_- = \ketbra{g}{e}$}, is described by the Hamiltonian \mbox{$H_{\ms A} = \omega_{\ms A} \sigma_z$}, where $\omega_{\ms A}$ is the transition frequency of the two-level system and $\sigma_z$ is the Pauli-z matrix. These two systems interact through the Jaynes-Cummings Hamiltonian, which in the rotating wave approximation, is given by:
\begin{align}
    H_{\text{tot}} = \omega_{\ms C} a^{\dagger} a + \omega_{\ms A} \sigma_z + g \left(\sigma_+ a + \sigma_- a^{\dagger} \right),
\end{align}
When the two-level system is driven on resonance, meaning \mbox{$\omega_{\ms C} = \omega_{\ms A}$}, the unitary evolution generated by the Jaynes-Cummings Hamiltonian, \mbox{$U(t) = e^{-i H_{\text{tot}} t}$}, conserves the total energy of both systems. This can be seen by noting that $[U(t), H_{\ms A}+H_{\ms C}] = 0$ for all times $t$.
\end{kaobox}

The evolution of systems towards thermal equilibrium reveals a fundamental aspect of the theory: any initial state cannot be transformed into any final state. In other words, not all quantum states are equally valuable, leading to a hierarchical ordering of the set of quantum states, with the most resourceful at the top and least resourceful at the bottom~(see Fig.~\ref{Partial_ordering_set_TO}). The natural question is how to quantify and evaluate the resources present in a given quantum state. This leads to the concept of thermodynamic monotones, which are mathematical functions that do not increase under free operations. They are formally defined as follows.
\begin{marginfigure}[0.0cm]
	\includegraphics[width=4.756cm]{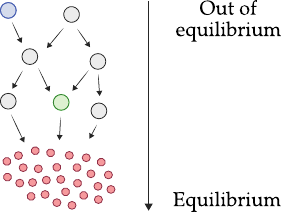}
	\caption{\emph{Thermodynamic ordering}. Thermodynamic transformations between states follow a partial ordering relation. The existence of a thermal operation mapping an initial state to a target state is depicted by an arrow. Note that, beyond the thermodynamic limit, not all states are comparable.}
	\label{Partial_ordering_set_TO}
\end{marginfigure}

\begin{definition}[Thermodynamic monotone]\label{def:thermodynamic-monotones}
A function $\mathcal{M}$ that maps the set of quantum states to non-negative real numbers $\mathbbm{R}_+ \cup \{0\}$ is a thermodynamic monotone if and only if it satisfies the following two conditions:
\begin{enumerate}
    \item $\mathcal{M}[\mathcal{E}^{\beta}(\rho)] \leq \mathcal{M}(\rho)$.
    \item $\mathcal{M}(\gamma) = 0$.
\end{enumerate}
For $\mathcal{E}^{\beta}$ being a thermal operation.
\end{definition}
In a resource-theoretic language, monotones are functions that quantify the amount of resources present in a given state and provide insight into how these resources can be transformed or manipulated using the allowed operations. Since the set of allowed operations cannot create resources for free, this requirement is encoded in the first property, while the second follows from the fact that free states do not have any resources. A non-trivial example of a thermodynamic monotone is the relative entropy or Kullback-Leibler divergence between a given state $\rho$ and the thermal Gibbs state $\gamma$. A detailed discussion of this example will be presented in the next section.

\subsection{Information-theoretic notions and their thermodynamic interpretation}\label{subsec:information-theoretic-notions}

In classical thermodynamics, an equilibrium system in the presence of a heat bath at inverse temperature $\beta$ is completely characterised by its free energy. Specifically, for an equilibrium state $\gamma$, with internal energy denoted by $U(\gamma)$ and Gibbs entropy $S(\gamma)$, the free energy is defined as
\begin{equation}
    F(\gamma) := U(\gamma) -\frac{1}{\beta}S(\gamma).
\end{equation}
This quantity is related to the partition function $Z$ via~\cite{callen}
\begin{equation}
    F(\gamma) = -\frac{1}{\beta}\log Z \quad \text{where} \quad  Z = \sum_i e^{-\beta E_i}.
\end{equation} 
One of the manifestations of the second law of thermodynamics is that for a system interacting with a bath in thermal equilibrium, the maximum amount of work that it can perform (that can be extracted from the system) is bounded by the difference $\Delta F$ between its initial and final free energy~\cite{fermi,callen}. Traditionally, the free energy has been defined only for states at thermal equilibrium. However, taking into account its operational meaning, one can extend its definition to investigate also the case of non-equilibrium states. More precisely, for any $d$-dimensional quantum state~$\rho$, the relative entropy~$D$ between~$\rho$ and a thermal Gibbs state~$\gamma$
\begin{equation}
    D(\rho\|\gamma):=\tr[\rho \left(\log \rho-\log\gamma\right)],\label{eq:D}
\end{equation}
can be interpreted as a non-equilibrium generalisation of the free energy difference between a state $\rho$ and a thermal state $\gamma$. This can be easily seen by manipulating Eq.~\eqref{eq:D} to arrive at
\begin{align}
    \frac{1}{\beta}D(\rho\|\gamma) &= \frac{1}{\beta}\tr[\rho(\log \rho - \log \gamma)] = \frac{1}{\beta}\tr(\rho \log \rho) - \frac{1}{\beta}\tr[\rho \log(\frac{e^{-\beta H}}{Z})] \nonumber\\
    &=-\frac{1}{\beta}S(\rho) +\tr(\rho H) +\frac{1}{\beta}\log Z, \label{eq:rel_ent_gibbs}
\end{align} 
where $S(\rho):=-\tr(\rho \log \rho)$ is the von Neumann entropy. Defining the non-equilibrium free energy as $\mathcal{F}(\rho) := \langle H\rangle_{\rho} -S(\rho)/\beta$, we observe that Eq.~\eqref{eq:rel_ent_gibbs} expresses a difference between free energies. One term originates from the non-equilibrium contribution, while the other stems from the equilibrium part:
\begin{equation}
    \frac{1}{\beta}D(\rho \| \gamma) = \mathcal{F}(\rho) - F(\gamma). 
\end{equation}
We can further define the corresponding relative entropy variance~$V$ and the function~$Y$ related to relative entropy skewness~\cite{Tomamichel2013,Chubb2018beyondthermodynamic,boes2020variance}:
\begin{subequations}
	\begin{align}
	V(\rho\|\gamma):=&\tr{\rho \left(\log \rho-\log\gamma-D(\rho\|\gamma)\right)^2},\label{eq:V}\\
	Y(\rho\|\gamma):=&\tr{\rho \left|\log \rho-\log\gamma - D(\rho\|\gamma)\right|^3}.\label{eq:W}
	\end{align}
\end{subequations}
Denoting the average of an observable $O$ in a state $\rho$ by \mbox{$\langle O\rangle_{\rho}=\tr(\rho O)$}, we can introduce a random variable \mbox{$\log \rho - \log \gamma$}, so that the divergence can be interpreted as the expectation value of the difference-likelihood. Similarly, the relative entropy variance and skewness correspond to the variance and the third moment of the random variable. 

The higher moments can then be understood as fluctuations of the non-equilibrium free energy content of the system. This is most apparent for pure states $\rho=\ketbra{\psi}{\psi}$, as $V$ simply describes energy fluctuations of the system:
\begin{align}
    \label{eq:energy_fluct}
    \frac{1}{\beta^2}V(\ketbra{\psi}{\psi}\|\gamma)&=\bra{\psi} H^2\ket{\psi}-\bra{\psi} H\ket{\psi}^2.
\end{align}
In particular, as noted in Ref.~\cite{Chubb2018beyondthermodynamic}, when $\rho=\gamma'$ is a thermal distribution at some different temperature \mbox{$T'\neq T$}, the expression for $V$ becomes
\begin{equation}
\label{eq:var_capacity}
    V(\gamma'\|\gamma)=\left(1-\frac{T'}{T}\right)^2 \cdot\frac{c_{T'}}{k_B},
\end{equation}
where
\begin{equation}
    c_{T'}=\frac{\partial}{\partial T'} \tr(\gamma'H)
\end{equation}
is the specific heat capacity of the system in a thermal state at temperature $T'$, and $k_B$ is the Boltzmann constant. 

\newpage

\section{Thermodynamic evolution of energy-incoherent states}

The formalism introduced so far is general and applies to any quantum state. However, we now focus on a specific class known as \emph{energy-incoherent states}. These states are diagonal in the energy-eigenbasis and are invariant under time translations, making the covariance restriction irrelevant in their description. While this restriction may suggest that the "quantum component" is lost, it is important to note that this is fully justified in the "quasi-classical regime", where systems are quantised and have a small number of energy levels, but the decoherence is so strong that interference between different energy levels is negligible. Formally, this set of state is defined as follows.
\begin{definition}[Energy-incoherent states]\label{def:energy-incoherent}
Let $\rho$ be a $d$-dimensional system described by a Hamiltonian $H = \sum_i E_i \ketbra{E_i}{E_i}$. If $\rho$ commutes with~$H$ it follows that $\rho$ can be written as:
\begin{equation}
\rho = \sum^{d}_{i=1}p_i \ketbra{E_i}{E_i},
\end{equation}
where $p_i$ are the eigenvalues of $\rho$, which coincide with the populations in the energy-eigenbasis. Furthermore, $\rho$ can be equivalently represented by a $d$-dimensional probability vector $\v p$ of its eigenvalues:
\begin{equation}
    \v p = \text{eig}(\rho) = (p_1, ..., p_d)
\end{equation}
\end{definition}
As a result, the state space of energy-incoherent states is the $(d-1)$-dimensional probability simplex $\Delta_d$. If we recall from Chapter~\ref{C:mathematical_preliminaries}, stochastic matrices represent the most general evolution between probability distributions. Hence, this leads to the question of whether studying thermodynamic transformations of energy-incoherent states allows one to replace a thermal operation, i.e., a CPTP map, with a stochastic matrix subject to certain constraints. This question is answered in the following theorem, which links the existence of a thermal operation between incoherent states to the existence of a Gibbs-preserving stochastic matrix between the probability distributions representing these states.
\begin{theorem}\label{thm_Thermaloperations}
    Let $\rho$ and $\sigma$ be quantum states commuting with the system Hamiltonian $H$, and $\gamma$ be its thermal Gibbs state with respect to the inverse temperature $\beta$. Denote their eigenvalues by $\v{p}$, $\v{q}$ and $\v{\gamma}$, respectively. Then, there exists a thermal operation $\mathcal{E}^{\beta}$, such that $\mathcal{E}^{\beta}(\rho)=\sigma$, if and only if there exists a stochastic map $\Lambda^{\beta}$ such that
		\begin{equation}
		    \Lambda^{\beta}\v{\gamma}=\v{\gamma} \quad \text{and} \quad \Lambda \v p=\v q . 
		\end{equation}
\end{theorem}

The proof of this theorem can be found in the supplementary material of Ref.~\cite{horodecki2013fundamental}.

As discussed in the previous section, in classical thermodynamics, state transformations are governed by the free energy. Specifically, one can transform an equilibrium state $\gamma$ into another equilibrium state $\gamma'$, if and only if \mbox{$F(\gamma) > F(\gamma')$}. However, when considering state transformations between different non-equilibrium states of a few-particle systems, the condition of decreasing free energy remains necessary but becomes insufficient. 

\begin{kaobox}[frametitle= Beyond the thermodynamic limit~\cite{brandao2015second}]
By going from the macroscopic (equilibrium) to the microscopic (nonequilibrium) realm, the traditional free energy function is replaced by a family of quantum free energies and transitions are now governed by a family of second laws. 

Define the free energies:
\begin{equation}\label{Eq-second-laws}
    F_{\alpha}(\rho, \gamma) := \beta D_{\alpha}(\rho \| \gamma) - \beta \log Z,
\end{equation}
with the Rényi divergences $D_{\alpha}(\rho \| \gamma)$ given by\sidenote{The cases of $\alpha \in \{ -\infty, 0, 1, \infty\}$ are defined via suitable limits:
\begin{align*}
    &D_0(\rho \| \gamma) = -\log \left(\sum_{i | p_i \neq 0} \gamma_i \right), \\ &D_1(\rho \| \gamma) = \sum_i p_i \log \frac{p_i}{\gamma_i}, \\
    &D_{\infty} = \log \text{max}_i \frac{p_i}{\gamma_i}, \\ &D_{-\infty}(\rho \| \gamma) =D_{\infty}(\gamma \| \rho).
\end{align*}
}
\begin{equation}
    D_{\alpha}(\rho \| \gamma) := \frac{\operatorname{sgn}(\alpha)}{\alpha-1}\log \sum_{i} p^{\alpha}_i \gamma^{1-\alpha}_i,
\end{equation}
where $p_i$ and $\gamma_i$ are the eigenvalues of $\rho$ and $\gamma$, respectively. 

If $F_{\alpha}(\rho, \gamma) \geq F_{\alpha}(\rho',\gamma)$ holds $\forall \alpha \geq 0$, then there exists a catalytic thermal operation that transforms $\rho$ into $\rho'$.
\end{kaobox}

In principle, determining the existence of a given thermal operation would require checking Eq.~\eqref{Eq-second-laws} for all $\alpha \in \mathbbm{R}$. However, this is only necessary; it becomes sufficient only when catalysts are allowed, i.e., when one can introduce states that aid the transformation without being degraded in the process. Nevertheless, this problem can be tackled from the point of view of partial order relations. This allows one to characterise the set of possible transitions in an advantageous way as it reduces the problem to a finite list of conditions, namely $n$ simple 1-norm inequalities. In Section~\ref{Subsec-thermomajorisation}, we revisited this problem in a mathematical way via Theorem~\ref{Thm:HLPbeta}. This theorem establishes a link between the existence of a Gibbs-preserving matrix and the concept of thermomajorisation. Consequently, a direct corollary arises when we combine Theorem~\ref{Thm:HLPbeta} with Theorem~\ref{thm_Thermaloperations}:
\begin{corollary}[Thermal operations $\&$ Thermomajorisation]
There exists a thermal operation mapping an incoherent state $\v p$ into an incoherent state $\v q$ if and only if $\v p \succ_{\beta} \v q.$
\end{corollary}

\begin{kaobox}[frametitle= $\beta$-swap - the thermodynamic permutation]

If a Gibbs-preserving matrix exists that connects $\v p$ to $\v q$, we say that there is a \emph{thermal process} between these two distributions.

Consider a qubit described by a Hamiltonian $H = E\ketbra{E}{E}$, and a thermal bath at inverse temperature $\beta$. Assuming that the qubit is in an energy incoherent state $\v p = (p,1-p)$, one can ask what is the set of achievable states from $\v p$ via thermal operations. 

This question can be easily answered using thermomajorisation relations. The key point is that the set of achievable state is given by a convex combination between the initial state and the state $\v p^{\pi^{\beta}_{12}} = (1-pe^{-\beta E}, pe^{-\beta E})$. Importantly, note that the state $\v p^{\pi_{12}}$ becomes a swap when $\beta = 0$. A direct swap of population between two energy levels cannot be achieved by thermal operations since the associated transition matrix is not Gibbs-stochastic. However, such a transformation that takes $\v p$ into $\v p^{\pi_{12}}$ is called $\beta$-swap and represents the thermodynamic analogue of a swap~\cite{lostaglio2018elementary}. In this example, such a transformation is accomplished by the following matrix:
\begin{equation}
 \Pi^{\beta}_{12} = \begin{pmatrix}
1-e^{-\beta E} & 1 \\
e^{-\beta E} & 0 \\
\end{pmatrix}  . 
\end{equation}
\end{kaobox}

The question raised in the above text box will be the central topic of the next chapter. Specifically, we will complement such a question by also asking the converse: \emph{What is the set of states that can be achieved from a given initial state via thermal operations, and what is the set of states that are incomparable with the initial state?} 

\begin{kaobox}[frametitle= Infinite temperature and the majorisation order]
In the infinite temperature case, $\beta = 0$, the thermal Gibbs state $\v \gamma$ is described by a uniform distribution
\begin{equation}
    \label{eq_uniform-state}
    \v \eta = \frac{1}{d}(1, ..., 1).
\end{equation}
Theorem~\ref{thm_Thermaloperations} implies that a state $\v{p}$ can be mapped to $\v{q}$ if and only if there exists a bistochastic matrix $\Lambda^0$, such that $\Lambda^0 \v{\eta}=\v{\eta}$, which transforms $\v{p}$ into $\v{q}$. In Chapter~\ref{C:mathematical_preliminaries}, we saw that there exists a bistochastic matrix $\Lambda$, $\Lambda^0 \v \eta=\v \eta$, mapping $\v{p}$ to $\v{q}$ if and only if $\v{p} \succ \v{q}$. 
\end{kaobox}

\subsection{Single-shot, intermediate regime and asymptotic case}

The scenario discussed so far is known as the \emph{single-shot} regime, where our sole objective is to ascertain when it is possible to transform one state into another using thermal operations. Therefore, our interest lies in individual instances of the protocol, and nothing has been mentioned about many runs, averages, or the asymptotic case. Thus, the natural question consists of having access to arbitrarily many copies of the initial state and asking whether it is possible to transform instances of one state to another with asymptotically vanishing error $\epsilon$:
\begin{equation}\label{Eq:optimal-rate}
    \rho^{\otimes N}  \xrightarrow[N\to \infty \:,\: \: M\to \infty\: \text{and}\: \epsilon \to 0]{\mathcal{E}^{\beta}} \sigma^{\otimes M} .
\end{equation}
The figure of merit when processing many subsystems is the conversion rate $R$, defined as the ratio between the number of output and input states 
\begin{equation}
    R = \frac{M}{N}.
\end{equation}
This conversion rate quantifies the efficiency of generating output states while processing multiple copies of the input state. Thus, we are interested in determining the optimal rate satisfying~\eqref{Eq:optimal-rate}. It was found that, for initial and target energy-incoherent states, the rate is given by the ratio of non-equilibrium free energies~\cite{brandao2013resource}
\begin{equation}\label{Eq:rate-asymptotic}
    R = \frac{D(\v p \| \v \gamma)}{D(\v q \| \v \gamma)},
\end{equation}
where $\v p$ and $\v q$ are the initial and target distributions. Hence, in this regime, all transformations become fully reversible and, therefore, there is no dissipation~[see Fig.~\hyperref[Fig:finite-size-rate]{\ref{Fig:finite-size-rate}a}]. This result bears a striking resemblance to that obtained in pure state entanglement theory, where the optimal interconversion rate is determined by the ratio of the entanglement entropies of the initial and target states~\cite{PhysRevA.53.2046}.

However, for finite $N$, this results no longer hold, and transformations become irreversible. The \emph{intermediate regime} has been thoroughly explored in the works~\cite{Chubb2018beyondthermodynamic, Chubb2019_2, PhysRevE.105.054127}, where there authors established a connection between the extreme case of single-shot thermodynamics with $N = 1$ and the asymptotic limit of $N \to \infty$. In particular, Eq.~\eqref{Eq:rate-asymptotic} gets a correction, which, up to second order asymptotics, is given by \begin{equation}
R(N, \epsilon) \simeq \frac{D(\v{p}\|\v{\gamma})}{D(\v{q}\|\v{\gamma})} \left( 1 + \sqrt{\frac{V(\v{p}\|\v{\gamma})}{N\, D(\v{p}\|\v{\gamma})^2}}\, Z_{1/\nu}^{-1}(\epsilon) \right),
\end{equation}
where $Z^{-1}_{\nu}$ is the inverse of the cumulative function of Rayleigh-normal distribution $Z_{\nu}$ introduced in Ref.~\cite{kumagai2016second} with $\nu$ given by
\begin{align}
\nu = \frac{V(\rho\|\gamma)/D(\rho\|\gamma)}{V(\sigma\|\gamma)/D(\sigma\|\gamma)}  \,,
\end{align}
and $\simeq$ denotes equality up to terms of order \mbox{$o(1/\sqrt{N})$}. See Fig.~\ref{Fig:finite-size-rate} for a schematic representation of the  asymptotic and intermediate regimes.
\begin{marginfigure}[-9cm]
	\includegraphics[width=4.756cm]{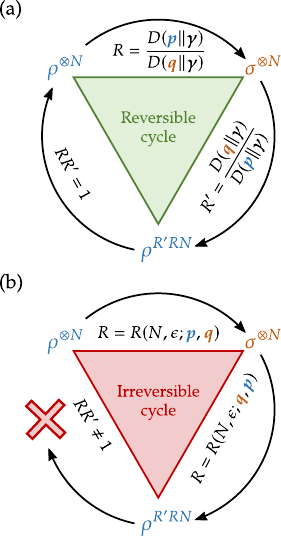}
	\caption{\emph{Intermediate regime}. (a) In the asymptotic limit, $N \to \infty$, the optimal conversion rate $R$ from $\rho^{\otimes N}$ (with eigenvalues $\v p$) to $\sigma^{\otimes RN}$ (with eigenvalues $\v q$) is equal to the inverse of the conversion rate from $\sigma$ to $\rho$. 
    (b) Finite $N$ corrections to conversion rates lead to irreversibility. Adapted with permission from Reference~\cite{Chubb2018beyondthermodynamic}}
	\label{Fig:finite-size-rate}
\end{marginfigure}
Such results enable the study of how thermodynamic cycles are influenced by irreversibility resulting from finite-size effects. For instance, the performance of a heat engine is affected when one of the heat baths it operates between is finite in size~\cite{Chubb2018beyondthermodynamic}. In Chapter~\ref{C:finite-size}, we will explore this regime and prove a relation between the amount of free energy dissipated in such processes and the free energy fluctuations of the initial state of the system.

\section{Applications}

The resource-theoretic approach allows us for a rigorous study of important thermodynamic protocols, such as \emph{work extraction}, \emph{Landauer erasure}, \emph{thermodynamically-free communication}, among others. In this section, we review the three mentioned ones. 

\newpage
\subsection{Work extraction}~\label{Subsec:work-extraction}

Upon moving from the macroscopic to the microscopic regime, the concept of deterministically extracting work from a single system becomes blurred due to large fluctuations. While it has been demonstrated that the second law of thermodynamics, as a statement regarding average work, remains valid at the scale of a few quanta~\cite{horodecki2013fundamental}, the precise definition of work in this context remains an open question. Different definitions suit different contexts, and a universally applicable ruling for all scenarios is still lacking.

Generally, definitions of work rely on controlling and changing external parameters that determine the system's Hamiltonian. The standard dynamical approach is based on the assumption that the average energy of a system varies over time due to energy exchanges occurring in two distinct qualitative ways~\cite{Alicki_1979,kosloff2013quantum}\sidenote{Let $\rho(t)$ be a quantum system with Hamiltonian $H(t)$ and average energy $U(t) = \tr[\rho(t) H(t)]$. One can define work as
\begin{equation*}
    W = -\int_{0}^{\tau} \textrm{dt} \, \tr[\rho(t)\qty(\frac{d H(t)}{dt})],
\end{equation*}
while the remaining energy is associated to heat,
\begin{equation*}
    Q = -\int_{0}^{\tau} \textrm{dt} \, \tr[\qty(\frac{d \rho(t)}{dt}) H(t)].
\end{equation*}
Here we adopt the convention that work is extracted from the system, while heat is absorbed, which explains the presence of the minus sign.}: (i) work-like energy associated with the variation of external parameters that can be controlled, and (ii) heat-like energy associated with non-controlled energy exchanges between the system and its surroundings. This approach is particularly powerful for analysing the behaviour of quantum heat engines. A second, and complementary, approach is based on the assumption that  work is a random variable, and a controlled Hamiltonian evolution is employed to determine the statistics of work~\cite{PhysRevE.90.032137,PhysRevE.92.042150,PhysRevE.93.022131}. The statistics of work play a central role in the study of the nonequilibrium thermodynamics of small systems both classical and quantum. Indeed, the diversity of approaches for constructing the work distribution gives rise to different proposals~\cite{ Miller2017,PhysRevE.93.022131, PhysRevLett.118.070601,PhysRevLett.123.230603}. 

Here, we will adopt a different approach that distinguishes itself from the standard methods in two ways. First, we avoid using any external agents, which means Hamiltonian changes will not be employed. Second, instead of examining the average and higher moments of the work distribution, we concentrate on the so-called \emph{single-shot regime}. This means our interest lies in individual instances of the protocol rather than repeated trials to characterise the statistics.

More precisely, work is defined in terms of a composite system undergoing a thermal operation. The setup involves the main system and an auxiliary system, a battery $B$, where energy can be reliably deposited and extracted in a controlled manner. Without loss of generality, we consider the battery to be a two-level system described by a Hamiltonian $H_B$ with eigenstates $\ket{0}_B$ and $\ket{1}_B$ corresponding to energies $0$ and $W_{\text{ext}}$, respectively\sidenote{The battery system can be described by either a continuous Hamiltonian or a Hamiltonian with a discrete spectrum, as long as its energy differences align with the amount of work we intend to extract.}. Assuming that the battery is initially prepared in the ground state, the work extraction protocol is equivalent to the existence of a thermal operation, whose effect consists of thermalising the main system, while exciting the battery (see Fig.~\ref{Fig:work-extraction})
\begin{equation}~\label{Eq-work-extraction}
    \E^{\beta}(\rho \otimes \ketbra{0}{0}_B) = \gamma \otimes \ketbra{1}{1}_B \,. 
\end{equation}

\begin{figure}[t]
\includegraphics[width=10.553cm]{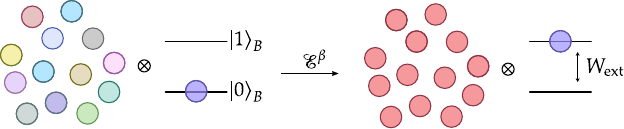}
	\caption{\emph{Work extraction processes.} Out-of-equilibrium system depicted by a collection of circles, each representing a distinct state and Hamiltonian, along with a battery system initially in the ground state. The composite system undergoes a thermal operation, which thermalises the main system while leaving the battery in an excited state. The extracted work is determined by the difference in energy of the battery. }
	\label{Fig:work-extraction}
\end{figure}
The maximum amount of work that can be extracted from $\rho$ is determined by the maximum value of $W_{\text{ext}}$\sidenote{The optimal $W$ is also called \emph{work of distillation}.} for which the aforementioned thermal operation exists. If, Eq.~\eqref{Eq-work-extraction} is strictly satisfied, then we say that we have achieved a \emph{deterministic work extraction}. 

\newpage
\begin{kaobox}[frametitle= No deterministic work extraction from mixed energy eigenstates]
Consider a composite system composed of a two-level system prepared in an energy-incoherent state $\v{p} = (p, 1-p)$ and a battery in the ground state. Assuming that the two-level system and the battery are described by the Hamiltonians $H = E\ketbra{1}{1}$ and $H_B = W\ketbra{1}{1}$, respectively; we ask for the maximum value of $W$, such that the following transformation is satisfied
\begin{equation}
    \v p \otimes (1,0) \xrightarrow[]{\mathcal{E}^{\beta}} \v \gamma \otimes (0,1).
\end{equation}
To answer this question, we first need to plot the thermomajorisation curves of the initial and target states. Next, note that both initial and target distributions reach a height of 1 at the $x$-axis positions $x = Z$ (initial state) and $x = Ze^{-\beta W}$ (target state), respectively. We can then observe that the optimal value of W occurs when $Z = Z e^{-\beta W}$. This equation reveals that this condition is satisfied when $W = 0$. Therefore, in a single-shot regime, deterministic work extraction from a convex combination of energy-eigenstates is not possible. However, notice that when either $p = 1$ or $p = 0$, deterministally one can extract $W = \frac{1}{\beta}\log Z$ and $W = E+\frac{1}{\beta}\log Z$, respectively. 
\end{kaobox}

If we relax the constraint that work must be extracted deterministically and instead allow for a probability of failure $\epsilon$, we enter the scenario known as $\epsilon$-\emph{deterministic work extraction}. The failure is modeled by assuming that the battery ends up in a state $\epsilon$ close to the excited state, i.e., $\v p_{B}=(\epsilon,1-\epsilon)$. 

While the main focus of the discussion is on a single subsystem, in Chapter~\ref{C:finite-size}, we will explore the scenario of extracting work from an initial system composed of asymptotically many non-interacting subsystems. These subsystems can either be energy-incoherent and non-identical, existing in different states with different Hamiltonians, or they can be pure and identical. In particular, this scenario allows one to treat the battery system as a part of the main system by selecting one specific subsystem to be in a pure ground state~(see Fig.~\ref{Fig:work-extraction}). 

\subsection{Information erasure}\label{subsec:information-erasure}

The connection between information and thermodynamics is as old as the thermodynamic theory itself, stretching back to the thought experiment known as Maxwell’s Demon to his exorcism. At the core of this connection lies the principle that processing information must adhere to the laws of thermodynamics. More precisely, erasing information has a fundamental energetic cost of $1/\beta\log 2$. 

\begin{figure}[t]
\includegraphics[width=10.553cm]{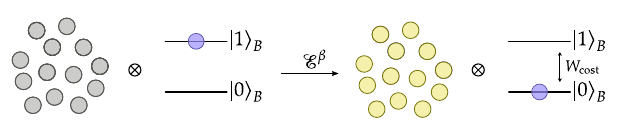}
	\caption{\emph{Information erasure.} The $N$ bits of information to be erased are represented by $N$ subsystems in a state $\rho^N$ with a trivial Hamiltonian $H^N$. The process is performed by attaching a battery system $B$ in an excited state $\ket{1}_B$ with energy $W^N_{\text {cost}}$, which measures the energetic cost of erasure. The erasure process resets the state $\rho^N$ to a fixed state $\ket{0}^{\otimes N}$, and de-excites the battery system. }
	\label{Fig:erasure}
\end{figure}

Similarly to the case of work extraction, the erasure process can also be formulated as a particular type of thermodynamic process. The setting comprises a bit of information, represented by a two-level system in a state $\rho$ with trivial Hamiltonian $H=0$, and a two-level battery system initially in an excited state $\ket{1}_B$ of energy $W_{\text{cost}}$, which measures the energetic cost of erasure. The processes involves resetting the state $\rho$ to a fixed state $\ket{0}$, while investing an amount of work $W_{\text{cost}}$. This is equivalent to the existence of the following thermal processes:
\begin{equation}
    \E^{\beta}(\rho \otimes \ketbra{1}{1}_B) = \ketbra{0}{0}\otimes\ketbra{0}{0}_B \,,    
\end{equation}
with the initial and target Hamiltonians being identical. If one wants to erase $N$ bits of information, then the initial state is replaced by $N$ two-level systems. The problem becomes more interesting when we allow a probability of failure $\epsilon$ and ask about the contribution arising from finite-size effects in the erasure processes. This problem will be tackle in Chapter~\ref{C:finite-size}.

\begin{kaobox}[frametitle= Landauer's principle]
Consider a composite system composed of a two-level system prepared in a mixed state $\v{p} = (1/2, 1/2)$ and a battery in the excited state. Assuming that the two-level system and the battery are described by the Hamiltonians $H = 0$ and $H_B = W_{\text{cost}}\ketbra{1}{1}$, respectively; we ask for the maximum value of $W$, such that the following transformation is satisfied
\begin{equation}
    \v p \otimes (0,1) \xrightarrow[]{\mathcal{E}^{\beta}} (1,0) \otimes (1,0).
\end{equation}
As before, this question is answered by drawing the thermomajorisation curves of the initial and target state. Then, by comparing the $x$-axis positions for when both curves reach a height of 1, we conclude that the optimal value of $W$, or the work cost, is given by $W_{\text{cost}} = 1/\beta \log 2$, which is precisely the fundamental amount required for erasing one bit of information. 
\end{kaobox}

 \subsection{Thermodynamically-free communication}
\label{sec:encoding-def}

Since thermodynamics is closely linked with information processing, one can also study thermodynamic aspects of communication. A traditional communication scenario in which Alice wants to encode and transmit classical information to Bob over a quantum channel consists of the following three steps~\cite{wilde2013quantum}. First, she encodes a message \mbox{$m \in \{1, ..., M \}$} by preparing a quantum system in a state $\rho_m$. Then, she sends it to Bob via a noisy quantum channel~$\mathcal{N}$. Finally, Bob decodes the original message by performing an optimal measurement on $\mathcal{N}(\rho_m)$. Crucially, in this standard scenario, both Alice and Bob are completely unconstrained, meaning that they can employ all encodings and decodings for free, and the only thing beyond their control is the noisy channel~$\N$.

Recently, a modification of this scenario was introduced that allows one to quantify the thermodynamic cost of communication~\cite{narasimhachar2019quantifying,korzekwa2019encoding}. More precisely, it is assumed for simplicity that Alice and Bob are connected via a noiseless channel, and Bob's decoding is still unconstrained. However, Alice is constrained to \emph{thermodynamically-free encodings}, meaning that encoded states $\rho_m$ can only arise from thermal operations acting on a given initial state $\rho$, interpreted as an \emph{information carrier}. Physically, this means that Alice obeys the second law of thermodynamics, in the sense that the encoding channel is constrained to use no thermodynamic resources other than the ones initially present in the information carrier~$\rho$. We illustrate this process in Fig.~\ref{fig:encoding}.

\begin{marginfigure}[-4.4cm]
	\includegraphics[width=4.756cm]{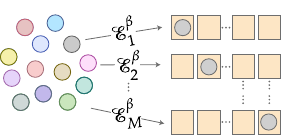}
	\caption{\emph{Thermodynamically-free encoding.} The thermal encoding of information can be captured by a thermodynamic process by considering $N$ independent subsystems in a state $\rho^N$ and with a Hamiltonian $H^N$ as an information carrier. The sender encodes a message \mbox{$m \in \{1, ..., M \}$} into it by applying a thermal operation $\E^{\beta}_{m}$ that transforms $\rho^N$ into mutually (almost) distinguishable states, and the receiver decodes the original message by performing a measurement on $\E^{\beta}_m(\rho^N)$.}
	\label{fig:encoding}
\end{marginfigure}

Now, the central question is: \emph{what is the optimal number of messages that can be encoded into $\rho$ in a thermodynamically-free way, so that the average decoding error is smaller than $\epsilon_{\mathrm{d}}$?} In Chapter~\ref{C:finite-size}, we will investigate the case when the information carrier is given by $N$ independent systems in a state $\rho^N$ and with a Hamiltonian $H^N$, as specified in Eq.~\eqref{eq:N-noninteracting}. Then, instead of asking for the optimal number of messages $M(\rho^N,\epsilon_{\mathrm{d}})$, we can equivalently ask for the optimal encoding rate~\cite{korzekwa2019encoding}:
\begin{equation}
\label{eq:optimal-encoding-rate}
    R(\rho^N, \epsilon_{\mathrm{d}}) := \frac{\log [M(\rho^N, \epsilon_{\mathrm{d}})]}{N} \,.   
\end{equation}

\section{The open quantum system approach to quantum thermodynamics}

Let us now contrast the resource-theoretic framework with the standard dynamical approach to quantum thermodynamics~\cite{kosloff2013quantum}. Specifically, this approach describes the continuous evolution of a quantum system over time using a differential equation known as a \emph{Markovian master equation}. Markovianity implies that no memory effects of the bath are taken into account. While memory effects can always be present in principle, several assumptions are made when modelling the joint system-bath dynamics. These assumptions, such as weak coupling, large bath size, and quickly decaying correlations, ensure the system's behaviour is described by a Lindblad master equation~\cite{kossakowski1972quantum,gorini1976completely,lindblad1976generators}
\begin{equation}
	\label{Eq:master_equation}
	\frac{d \rho(t)}{dt}  = -i \qty[H, \rho(t)] + \mathcal{L}_t\qty[\rho(t)].
\end{equation}
In the above, $[\cdot,\cdot]$ denotes the commutator and $\mathcal{L}_t$ is the Lindbladian with the following general form~\cite{breuer2002theory}:
\begin{equation}
	\label{eq:lindbladian2}
\mathcal{L}_t(\rho) = \sum_{i} r_i(t) \left[ L_i(t) \rho L_i(t)^\dag - \frac{1}{2}\Bigl\{L_i(t)^\dag L_i(t), \rho\Bigl\}\, \!\right]\!,\!
\end{equation} 
with $\{\cdot,\cdot\}$ denoting the anticommutator, $L_i(t)$ being time-dependent jump operators, and $r_i(t)\geq 0$ being time-dependent non-negative jump rates. While a general Lindbladian only requires the rates $r_i$ to be non-negative, Lindbladians arising from the interaction of a quantum system with a large heat bath have two standard properties~\cite{breuer2002theory,kosloff2013quantum}:
\begin{enumerate}[label=(P\arabic*)]
	\item\label{Lp1}\textbf{Steady state.} The Gibbs thermal state of the system, 
	\begin{equation}
		\gamma = \frac{e^{-\beta H}}{\tr(e^{-\beta H})},
	\end{equation}
	is a stationary solution of the dynamics, i.e.,
	\begin{equation}
		\forall t: \quad \mathcal{L}_t\gamma = 0. 
	\end{equation} 
	\item\label{Lp2}\textbf{Covariance.} The Lindbladian $\mathcal{L}_t$ commutes with the generator of the Hamiltonian dynamics $\mathcal{H}$ at all times $t$, i.e.,
	\begin{equation}
		\forall \rho:\quad\mathcal{L}_t[\mathcal{H}(\rho)] = \mathcal{H}[\mathcal{L}_t(\rho)].
	\end{equation}
\end{enumerate}

Quantum dynamics, as governed by master equations of the form in Eq.~\eqref{Eq:master_equation} and satisfying to properties~\ref{Lp1}~and~\ref{Lp2}, are termed Markovian thermal processes~\cite{lostaglio2021continuous}. They can be formally defined as:

\begin{definition}[Markovian thermal processes]
A channel is a Markovian thermal process (MTP) if it results from integrating a Markovian master equation, Eq.~\eqref{Eq:master_equation}, between time $0$ and $t_f \in [0, +\infty]$, where the Lindbladian $\mathcal{L}_t$ satisfies properties~\ref{Lp1}-\ref{Lp2}.
\end{definition}

Typically, such an approach is model-dependent. For instance, one might assume the dynamics of a quantum dot coupled to a fermionic bath~\cite{schaller2014open}, or a qubit interacting with an electromagnetic field~\cite{carmichael1999statistical}. The primary goal consists of deriving a specific master equation and subsequently solving for the associated dynamics. However, as before, we can ask whether it is possible to characterise the necessary and sufficient conditions for the existence of a Markovian thermal process between specific initial and final energy distributions of the system in a model-independent way. This problem was recently tackled in Ref.~\cite{lostaglio2021continuous,Korzekwa2022}. The authors introduced a hybrid framework that leverages the information theory tools discussed in this chapter, complementing the toolkit offered by the master equation formalism. Remarkably, it was found that the existence of a Markovian thermal process is linked to the notion of continuous thermomajorisation via the following theorem~\cite{lostaglio2021continuous}:
\begin{theorem}[Thermodynamic interconversion]
	\label{thm:main}
 	There exists a Markovian thermal process mapping $\v p(0)$ into $\v p(t_f)$ if and only if $\v p$ continously thermomajorises $\v p(t_f)$:
 \begin{equation}
 		\label{eq:continuousmajorization}
 		\v{p} \stackrel{ \textrm{MTP}}{\longmapsto} 	\v{p}(t_f) \Longleftrightarrow \v{p} \ggcurly_{\v{\beta}} \v{p}(t_f).
 	\end{equation}
\end{theorem}
The proof of this theorem can found in the Appendix A of Ref.~\cite{lostaglio2021continuous}. As a consequence, continuous thermomajorisation provides a complete set of constraints for the evolution of populations within the standard Markovian master equations approach. What is more, if $\v p \ggcurly_{\beta} \v p(t_f)$, there exists an universal set of controls that allows one to devise a thermalisation process that drives the system to
a final target state $\v p(t_f)$. These are a set of thermalisations acting only on two energy levels $(i, j)$ and is represented by the reset Markovian master equation
\begin{equation}
\frac{dp_i}{dt}=\frac{1}{\tau}\left[  \frac{\gamma_i}{\gamma_i+\gamma_j}(p_i+p_j)-p_i \right] \quad 	\text{with} \quad \frac{dp_j}{dt} = - 	\frac{dp_i}{dt},  
\end{equation}
whose solution describes an exponential relaxation to equilibrium:
\begin{equation}\label{Eq:thermalisation-reset}
	\v{p}^{(ij)}(t) = e^{-t/\tau} \v{p}^{(ij)}(0) + [p_i(0)+p_j(0)] (1-e^{-t/\tau}) \v{\gamma}^{(ij)}.
\end{equation} 
where $\v{p}^{(ij)}(t):=(p_i(t),p_j(t))$. The equation Eq.~\eqref{Eq:thermalisation-reset} can be represented in the form of a matrix equation as:
\begin{equation}
	\v{p}^{(ij)}(t) = T^{(ij)}(\lambda_t) \v{p}^{(ij)}(0),
\end{equation}
with $\lambda_t = 1- e^{-t/\tau}$ and
\begin{equation}
	\label{eq:elementarythermalization}
		T^{(ij)}(\lambda) =
		\begin{bmatrix}
			(1-\lambda) + \frac{  \lambda \gamma_i}{\gamma_i+\gamma_j} &	\lambda \frac{\gamma_i}{\gamma_i+\gamma_j} \\
			\lambda \frac{\gamma_j}{\gamma_i+\gamma_j} & 	(1-\lambda) +  \frac{\lambda \gamma_i}{\gamma_i+\gamma_j}
		\end{bmatrix}.
\end{equation}
Finally, the so-called \emph{elementary thermalisations} are a universal set of thermalisation controls:

\begin{theorem}[Universality of elementary thermalizations]\label{thm:universality}
	There exists a Markovian thermal process mapping $\v p(0)$ into $\v p(t_f)$ if only if there exists a finite sequence of elementary thermalisations such that 
	\begin{equation}
		\label{eq:plttrajectory}
		\v{p}(t_f)=T^{(i_f j_f)}(\lambda_f) \dots T^{(i_1 j_1)}(\lambda_1) \v{p}(0).
	\end{equation}
\end{theorem}
For the proof, see Appendix~A of Ref.~\cite{lostaglio2021continuous}. This result significantly simplifies the set of controls needed to generate the transformations achievable by the most general Markovian thermal process.

\section{Concluding remarks}
In this chapter, we introduced the resource-theoretic framework for studying thermodynamic transformations of quantum systems. We began by defining the set of thermal operations, which encapsulate all possible transformations that can occur without requiring resources beyond a single heat bath. Focusing specifically on energy-incoherent states, we identified the necessary and sufficient conditions that determine the existence of thermal operations between a pair of states. Intriguingly, the rules describing which state transformations are possible under thermal operations can be expressed as a partial-order relation between probability vectors corresponding to the initial and final states. In the infinite-temperature limit, these rules are captured by the majorisation relation, while in the finite temperature scenario, they are represented by thermomajorisation. Additionally, we demonstrated how this framework can be employed to analyse thermodynamic protocols such as work extraction, information erasure, and thermodynamically-free communication. We wrapped up the chapter by delving into a hybrid framework linking the resource-theoretic approach with the standard dynamical approach to quantum thermodynamics. Notably, the conditions needed for the existence of a Markovian thermal process that between a given initial and final energy distributions of the system can also be framed using the notion of partial-order. In this case, the concept continuous thermomajorisation between these states. Even more striking is the fact that this framework is constructive, returning explicit protocols for realising any possible Markovian transformation. These protocols are built upon elementary thermalisations, each acting on only two energy levels of the system, which have been proven to be universal controls.

Finally, one question we have not addressed is Which bath Hamiltonians matter for thermal operations? This question was answered in Ref.~\cite{vonEndebath}, where the author introduced the concept of Hamiltonians with a resonant spectrum relative to a reference one. It was then demonstrated that, in defining thermal operations, it suffices to consider only those bath Hamiltonians that satisfy this resonance property.

\pagelayout{wide} 
\addpart{Original results}
\pagelayout{margin} 
\chapter{Geometric structure of thermal cones}\label{C:thermal_cones}

Thermodynamic evolution of physical systems obeys a fundamental asymmetry imposed by nature. Known as the thermodynamic arrow of time~\cite{eddington1928nature}, it is a direct manifestation of the second law of thermodynamics, which states that the entropy of an isolated system cannot decrease~\cite{fermi, callen1985thermodynamics}. In other words, the thermodynamic evolution inherently distinguishes the past from the future: systems spontaneously evolve to future equilibrium states, but do not spontaneously evolve away from them. Even though recognition of the thermodynamic arrow of time is an old discussion~\cite{boltzmann1895certain,zermelo1896satz}, it still raises deep questions relevant both to philosophy and the foundations of physics~\cite{prigogine2000arrow, price2004origins}. Despite many attempts, the full understanding of the time asymmetry in thermodynamics seems to be still beyond our reach. 

The theoretical toolkit discussed in Chapter~\ref{C:resource_theory_of_thermodynamics} offers a comprehensive approach that allows us to reexamine previous inquiries, such as the nature of the thermodynamic arrow of time. An investigation from the viewpoint of order theory was initially performed in Ref.~\cite{Korzekwa2017}, where this issue was explored. However, the analysis mainly focused on structural differences between classical and quantum theories in contrast to the geometric aspects of thermal cones that we investigate here.

The aim of this chapter is to characterise the thermodynamic arrow of time by investigating the allowed transformations between energy-incoherent states that arise from the most general energy-conserving interaction between the system and a thermal bath. These transformations encode the structure of the thermodynamic arrow of time by telling us which states can be reached from a given state, which we refer to as the \emph{present state}, in accordance with the laws of thermodynamics. Under these constraints, the state space can then be naturally decomposed into three parts: the set of states to which the present state can evolve is called the \emph{future thermal cone}; the set of states that can evolve to the present state is called the \emph{past thermal cone}; while states that are neither in the past nor the future thermal cone form the \emph{incomparable region}.

While studying the future thermal cone has yielded substantial insights~\cite{Lostaglio2018elementarythermal,mazurek2018decomposability,mazurek2019thermal,e26020106,EndeExp}, the explicit characterisation of the incomparable region and the past thermal cone has not been performed. The core of this chapter relies on two main theorems that address this gap. The first theorem provides a geometric characterisation of the past majorisation cone, which is the set of probability distributions that majorises a given distribution. This, together with the incomparable region and future thermal cone, fully specifies the time-like ordering in the probability simplex in the limit of infinite temperature. The second result generalises the first one to the case of finite temperatures. 

Earlier works have established that the future (thermal) cone is convex~\cite{bhatia1996matrix,Lostaglio2018elementarythermal}. The results presented here extend this knowledge by demonstrating that the past (thermal) cone of a $d$-dimensional system can always be decomposed into $d!$ convex parts. Furthermore, in the zero-temperature limit, only one of these convex parts retains a non-zero volume, resulting in the entire past thermal cone becoming convex. Additionally, new thermodynamic monotones, based on the volume of the thermal cones, are introduced. 

The results discussed here can also be seen as an extension of the famous Hardy-Littlewood-P{\'o}lya theorem~\cite{hardy1952inequalities}, as they specify the past cone and the incomparable region in addition to the previously studied future cone. Therefore, they can also be employed to study other majorisation-based resource theories, such as the theory of entanglement~\cite{nielsen1999conditions} or coherence~\cite{Plenio2014, Du2015}. Concerning local operations, an analogy between special relativity and the set of pure states of bipartite systems was previously made in Ref.~\cite{zyczkowski2002}, where the authors correspondingly divided the state space into three parts. Here, we consider a more general partial-order structure, the thermomajorisation order, which generalises the previous and recovers it in the limit of infinite temperature.

This chapter starts by stating the main results concerning the construction of majorisation cones and discussing their interpretation within the thermodynamic setting, and in other majorisation-based theories. This construction is also generalised for probabilistic transformations. Next, the majorisation cones results are generalised to thermal cones, generated by thermomajorisation relation, where we also introduce the tool of embedding lattice, instrumental for the proof of the second main result. Then, we introduce thermodynamic monotones given by the volumes of the past and future thermal cone, discuss their intuitive operational interpretation and describe their properties. We also comment on the different natures of future and past cones for entanglement transformations. The technical derivation of the main results can be found in Section~\ref{sec_appendix}. 

\section{Majorisation cones}~\label{Sec:majorisation-cones}

The reachability of states under bistochastic matrices can be studied by introducing the notion of \emph{majorisation cones}, defined as follows:

\begin{definition}[Majorisation cones]\label{Def:majorisation_cones}
\begin{marginfigure}[-0.952cm]
	\includegraphics[width=4.651cm]{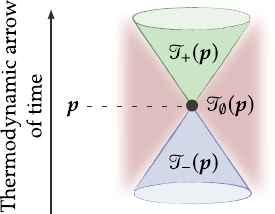}
	\caption[Fig:thermal-cones]{\emph{Thermal cones.} The thermodynamic arrow of time induces a time-like ordering that can be decomposed into past, incomparable, and future regions}
	\label{Fig:thermal-cones}
\end{marginfigure}
The set of states that a probability vector $\v p$ can be mapped to by bistochastic matrices is called the \textbf{future cone} $\T_+(\v{p})$. The set of states that can be mapped to $\v p$ by bistochastic matrices is called the \textbf{past cone} $\T_-(\v{p})$. The set of states that are neither in the past nor in the future cone of $\v p$ is called the \textbf{incomparable region} $\T_{\emptyset}(\v{p})$.
\end{definition}

Definition~\ref{Def:majorisation_cones} provides a framework for analysing the thermodynamics of energy-incoherent states in the infinite-temperature limit. Additionally, it allows us to study a broader class of state transformations that are governed by majorisation relations, such as those that arise in the resource theories of entanglement~\cite{nielsen1999conditions} or coherence~\cite{Plenio2014,Du2015, Streltsov2016,Streltsov2017}. We discuss these more general settings in detail in Sec.~\ref{Subsec:other-resource-theories}; however, for now, we will focus on thermodynamic transformations.

The future cone can be characterised using Birkhoff's Theorem~\ref{Thm:Birkhoff}, which states that any bistochastic matrix can be expressed as a convex combination of permutation matrices. Since there are $d!$ permutations in $\mathcal{S}_d$, the set of $d \times d$ bistochastic matrices forms a convex polytope with $d!$ vertices. Combining this result with Theorem~\ref{Thm:HLP}, we obtain the future cone of $\v p$:
\begin{corollary}[Future cone]~\label{Thm:Future_majoristion_cone} For a $d$-dimensional probability vector $\v{p}$, its future cone is given by
\begin{equation}
    \label{eq-futurecone}
    \T_{+}(\v p) = \operatorname{conv}\left[\left\{\Pi \v p \, , \mathcal{S}_d \ni \v \pi \mapsto \Pi  \right\}\right] ,
\end{equation}
where $\Pi$ denotes a permutation matrix corresponding to the permutation $\v \pi$ with $d$ elements, and $\operatorname{conv[S]}$ the convex hull of the set $S$.
\begin{marginfigure}[-3.364cm]
	\includegraphics[width=4.859cm]{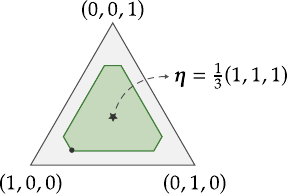}
	\caption[Fig:future-thermal-cone]{\emph{Achievability}. Future cone for a state with population $\v p = (0.7,0.2,0.1)$. Extreme points corresponds to all permutations of $\v p$.}
	\label{Fig:future-thermal-cone}
\end{marginfigure}
\end{corollary}

The above corollary implies that the future cone of $\v p$ is a convex set, with all distributions lying in $\T_{+}(\v p)$ being majorised by $\v p$ (see Fig.~\ref{Fig:future-thermal-cone} for an example considering a three-level system). Since the $d$-dimensional sharp distribution, $(0,..., 1,.., 0)$, majorises all probability distributions, its future cone is the entire probability simplex, which we will denote by $\v \Delta_d$.

If there is no transformation mapping $\v p$ into $\v q$ nor $\v q$ into~$\v p$, we say that these two states are incomparable. The incomparable region can be characterised by incorporating into the analysis the concept of quasi-probability distributions, which are defined by relaxing the non-negativity condition on the entries of a normalised probability distribution. The following result, the proof of which is employing the lattice structure of majorisation order and can be found in Section~\ref{sec_appendix}, specifies the incomparable region of $\v p$.

\begin{lemma}[Incomparable region]~\label{Lem:Incomparable_region0}
For a $d$-dimensional probability distribution \mbox{$\v{p} = (p_1,..., p_d)$}, consider the quasi-probability distributions $\v{t}^{(n)}$ constructed for each \mbox{$n \in \{1, ..., d\}$},
\begin{marginfigure}[-1.71cm]
	\includegraphics[width=4.754cm]{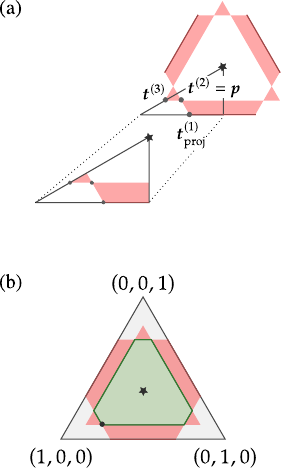}
	\caption[Fig:incomparable-thermal-cone]{\emph{Incomparability}. (a) Incomparable region for a state with population $\v p=(0.7,0.2,0.1)$, where the tangent vectors are highlighted in the first chamber. In (b), we illustrate both the future and incomparable regions.}
	\label{Fig:incomparable-thermal-cone}
\end{marginfigure}
\begin{equation}
\label{eq_tangentvectors}
    \v{t}^{(n)}=\left(t^{(n)}_1, p_n^{\downarrow}, ..., p_n^{\downarrow}, t^{(n)}_d\right) , 
\end{equation}
with 
\begin{align}
\label{eq:extremalpointspastan}
    t_1^{(n)} & = \sum_{i=1}^{n-1} p_i^{\downarrow} - (n-2) p_n^{\downarrow}, &
    t_d^{(n)} & = 1 - t^{(n)}_1 - (d-2)p_n^{\downarrow},
\end{align}
and define the following set
\begin{equation}
\label{eq:setP}
    \mathbb{T}:=\bigcup\limits_{j=1}^{d-1} \operatorname{conv}\left[\T_{+}\left(\v{t}^{(j)}\right) \cup \T_{+}\left(\v{t}^{(j+1)}\right)\right] .    
\end{equation}
Then, the incomparable region of $\v{p}$ is given by 
\begin{align}
\label{eq:incomparableconeIT}
    \T_{\emptyset}(\v{p}) = \left[\operatorname{int}(\mathbb{T})\backslash\T_+\left(\v p\right)\right]\cap \v{\Delta}_d  ,
\end{align}
where the backslash $\backslash$ denotes the set difference and $\operatorname{int}( \mathbb{T})$ represents the interior of  $\mathbb{T}$.
\end{lemma}

See Figs.~\hyperref[Fig:incomparable-thermal-cone]{\ref{Fig:incomparable-thermal-cone}a}~and~~\hyperref[Fig:incomparable-thermal-cone]{\ref{Fig:incomparable-thermal-cone}b} for an example illustrating the incomparable region for a three-level system.

We will refer to the quasi-probability distributions $\v{t}^{(n)}$ as \emph{tangent vectors}. The intuition behind this name and the importance of $\v{t}^{(n)}$ can be explained by noticing that any convex function $g(x)$ lies fully under its tangent at any point $y$, denoted as $t_y(x) \geq g(x)$, with equality guaranteed only for \mbox{$t_y(y) = g(y)$}. It follows from the definition of $\v t^{(n)}$ that its majorisation curve $f_{\v t^{(n)}}(x)$ is parallel to the $n$-th linear piece of $f_{\v p}(x)$  for $x\in\left[(n-1)/d,\,n/d\right]$, and the first and last elements of $\v t^{(n)}$ guarantee tangency and normalisation. Finally, since the adjacent linear fragments of $f_{\v p}(x)$ share the elbows of the function, the consequent tangent vectors $\v t^{(n)},\,\v t^{(n+1)}$ are both tangent at a selected elbow, $f_{\v p}(n/d) = f_{\v t^{(n)}}(n/d) = f_{\v t^{(n+1)}}(n/d)$. Therefore, any convex combination of the form $a \v t^{(n)} + (1-a)\v t^{(n+1)}$ will be ``tangent'' at the $n$-th elbow of the $\v p$ majorisation curve. The fact that $\v{t}^{(n)}$ may be a quasi-probability distribution does not pose a problem, since this vector can be projected back onto the probability simplex. The projected vector $\v{t}^{(n)}$ will be denoted by $\v{t}^{(n)}_{\text{proj}}$, and can be obtained by successively applying the map 
\begin{equation}
\left\{t^{(n)}_{m-1},\,t^{(n)}_m\right\} \longmapsto \left\{\min\left(t^{(n)}_{m-1}+t^{(n)}_m,\,t^{(n)}_{m-1}\right),\,\max\left(t^{(n)}_m,\,0\right)\right\}    
\end{equation}
to pairs of entries of $\v{t}^{(n)}$ going from $m = d$ to $m = 2$. In each step, the map either zeros the second component by shifting its value to the first one or, if the second component is non-negative, it leaves them both unperturbed. Geometrically, the state is shifted along the edges of the future cone of $\v{t}^{(n)}$ and every time it hits a plane defining one of the faces of the probability simplex $\Delta_d$, a new direction is selected, until the state is composed exclusively of non-negative entries.

Using Lemma~\ref{Lem:Incomparable_region0}, we can now prove the following theorem that specifies the past cone.
\begin{marginfigure}[-0.493cm]
	\includegraphics[width=4.859cm]{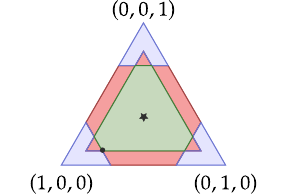}
	\caption[Fig:past-thermal-cone]{\emph{Majorisation cones.} Past region for a state with population \mbox{$\v p = (0.7,0.2,0.1)$}.}
	\label{Fig:past-thermal-cone}
\end{marginfigure}
\begin{theorem}[Past cone]
\label{Thm:past_cone}
The past cone of $\v{p}$ is given by
\begin{equation}
        \T_{-}(\v{p}) = \v{\Delta}_d \backslash \operatorname{int} (\mathbb{T}) \, .
\end{equation}
\end{theorem}
\begin{proof}
    One only needs to use the fact that
    \begin{equation}
        \label{eq:pastconeIT}
        \T_{-}(\v{p}) = \v{\Delta}_d \texttt{\textbackslash} \left(\T_{\emptyset}(\v{p})\cup \T_{+}(\v{p})\right)  \, ,
\end{equation}
and employ Lemma~\ref{Lem:Incomparable_region0} to replace $\T_{\emptyset}(\v{p})$ in the above with Eq.~\eqref{eq:incomparableconeIT}.
\end{proof}

Let us make a few comments on the above results. First, note that the incomparable region arises only for $d \geq 3$. This can be easily deduced from Lemma~\ref{Lem:Incomparable_region0}, as for $d=2$ the two extreme points, $\v t^{(1)}$ and $\v t^{(2)}$, are precisely the initial state~$\v{p}$. Second, the future thermal cone is symmetric with respect to the maximally mixed distribution $\v \eta$, and consequently, the incomparable and past cones also exhibit a particular symmetry around this point. As we shall see, this symmetry is lost when we go beyond the limit of infinite temperature. Third, although the past cone is not convex as a whole, it is convex when restricted to any single Weyl chamber.
Therefore, we may note that the tangent vectors $\v t^{(n)}$ provide the extreme points of the past not only from the viewpoint of a single-chamber but also to the entire probability simplex. This can be understood by noting that $\v t^{(n)}$ are located at the boundary between the incomparable and the past cone, and by symmetry, it applies to all their permuted versions. As a consequence, the past is constructed from $d!$ copies of the past in the canonical Weyl chamber, each copy transformed according to the corresponding permutation $\v{\pi}$~(see Fig.~\hyperref[fig-past-chamber]{\ref{fig-past-chamber}a, b}). 
Finally, one can make an analogy to special relativity with bistochastic matrices imposing a causal structure in the probability simplex $\Delta_d$. There exists a ``light cone'' for each point in $\Delta_d$, which divides the space into past, incomparable, and future regions~(see Fig.~\hyperref[fig-past-chamber]{\ref{fig-past-chamber}c}).

\begin{figure}[t]
\includegraphics[width=10.774cm]{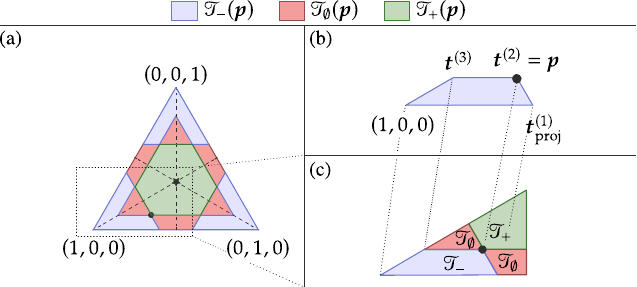}
	\caption[majorisation-curve]{\emph{Majorisation cones and Weyl chambers}. (a) Probability simplex $\Delta_3$ and a state $\v p = (0.6,0.3,0.1)$ represented by a black dot~$\bullet$ together with its majorisation cones. The division of $\Delta_3$ into different Weyl chambers is indicated by dashed lines with the central state \mbox{$\v \eta =(1/3,1/3,1/3)$} denoted by a black star $\bigstar$. (b) The past cone of a state~$\v{p}$ restricted to a given Weyl chamber is convex with the extreme points given by $\v{t}^{(n)}$ from Eq.~\eqref{eq_tangentvectors} and the sharp state. (c) The causal structure induced by bistochastic matrices (i.e., thermal operations in the infinite temperature limit) in a given Weyl chamber.}
	\label{fig-past-chamber}
\end{figure}

\begin{kaobox}[frametitle= Constructing the incomparable region]
To gain geometric intuition for interpreting Lemma~\ref{Lem:Incomparable_region0}, we consider a three-level system and construct its incomparable region.

Let $\v{p}=(0.7,0.2,0.1)$ be the population of the initial state. Then, Lemma~\eqref{Lem:Incomparable_region0} implies the existence of three different quasi-probability distributions $\v{t}^{(n)}$, namely $\v{t}^{(1)}=(0.7,0.7,-0.4)$, $\v{t}^{(2)}=(0.7,0.2,0.1)$, and $\v{t}^{(3)}=(0.8,0.1,0.1)$. The next step is to evaluate Eq.~\eqref{eq:setP}, i.e., calculate the future thermal cones for each quasi-probability distribution individually, and then take the convex union of all the cones. This process is illustrated below.
\begin{figure}[H]
\includegraphics[width=9.594cm]{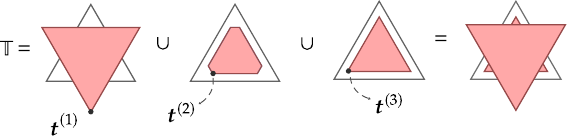}
\end{figure}
To obtain the incomparable region, we first take the set difference between the interior of $\mathbbm{T}$ and the future thermal cone of $\v{p}$
\begin{figure}[H]
\includegraphics[width=6.732cm]{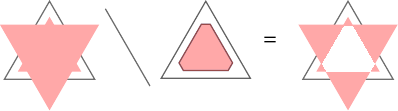}
\end{figure}
The next step to compute Eq.~\eqref{eq:incomparableconeIT} is to intersect with $\Delta_d$. Note that this procedure eliminates improper probability distributions. By doing so, we obtain the incomparable region, which is illustrated in Fig.~\hyperref[Fig:incomparable-thermal-cone]{\ref{Fig:incomparable-thermal-cone}a.}
\end{kaobox}

The central idea behind Lemma~\ref{Lem:Incomparable_region0} and Theorem~\ref{Thm:past_cone} can be better understood through a visualisation using partial-order diagrams. To illustrate the principles of such diagrams, we will first focus on the special case of a three-level system. In this case, the past cone has three non-trivial extreme points: $\v t^{(1)}$, $\v t^{(2)}$ and $\v t^{(3)}$. Furthermore, as shown in Fig.~\hyperref[Fig:latticediagram]{\ref{Fig:latticediagram}a}, these extreme points satisfy the following partial-order relation: the sharp state $\v s_1$ majorises both $\v t^{(1)}$ and $\v t^{(3)}$, and these two vectors majorise the initial state $\v p=\v{t}^{(2)}$. As it was proved in Lemma~\ref{Lem:Incomparable_region0}, the union of the future cones of these extreme points provides us, after subtracting the future of the vector $\v{p}$, with the incomparable region of $\v p$. In the particular case of $d=3$ since $\v{t}^{(1)},\,\v{t}^{(3)}\succ\v{t}^{(2)}$, we find that $\operatorname{conv}[\mathcal{T}_+(\v{t}^{(1)}),\mathcal{T}_+(\v{t}^{(2)})] = \mathcal{T}_+(\v{t}^{(1)})$ and similarly for $\v{t}^{(3)}$. However, it is important to note the fact that $\v t^{(1)}$ and $\v t^{(3)}$ are incomparable, and in turn, their respective future cones after subtracting future of $\v{p}$ characterise disjoint parts of the incomparable region. Finally, the tangent vector $\v t^{(2)}$ reduces to the original probability vector, $\v t^{(2)} = \v p$, only for $d = 3$, and this fact is fully understood from the construction of the $\v t^{(n)}$-vectors (see Section~\ref{sec_appendix}). 
\begin{marginfigure}[-5cm]
	\includegraphics[width=3.700cm]{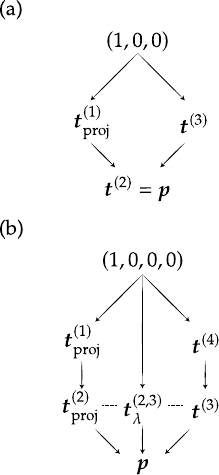}
	\caption[Partial-order diagrams for majorisation]{\emph{Hasse diagrams}. Graphical representation of the partially ordered set formed by the extreme points of the past cone in the probability simplex $\Delta_d$. Each arrow indicates that one of the elements precedes the other in the majorisation ordering. (a) Diagram for $d=3$. (b) Diagram for $d=4$. Note that any convex combination between $\v{t}^{(2)}_{\text{proj}}$ and $\v t^{(3)}$, here denoted by $\v{t}^{(2,3)}_{\lambda}$, results in an incomparable vector with respect to $\v{t}^{(1)}_{\text{proj}}$ and $\v{t}^{(4)}$.}
	\label{Fig:latticediagram}
\end{marginfigure}
It is evident from Lemma~\ref{Lem:Incomparable_region0} that each pair of tangent vectors $(\v t^{(n)},\,\v t^{(n+1)})$ characterises a given part of the incomparable region. However, notice that the futures of the extreme points considered one by one do not give the full description of the incomparable region -- one needs to consider their convex hulls to fill in the gaps. This particular feature of the construction can be demonstrated by considering the case of $d = 4$. Observe that in this case, we have a set of four tangent vectors $\v t^{(n)}$ with $n\in\{1,2,3,4\}$. Straightforward calculation shows that $\v t^{(1)}$ majorises $\v t^{(2)}$ and $\v t^{(4)}$ majorises $\v t^{(3)}$, therefore we find certain simplification, namely $\operatorname{conv}[\mathcal{T}_+(\v{t}^{(1)}),\mathcal{T}_+(\v{t}^{(2)})] = \mathcal{T}_+(\v{t}^{(1)})$ and similarly for $\v{t}^{(3)}$ and $\v{t}^{(4)}$. Nevertheless, $\v t^{(1)}$ is incomparable to $\v t^{(4)}$; similarly $\v t^{(2)}$ belongs to the incomparable region of $\v t^{(3)}$ [see Fig.~\hyperref[Fig:latticediagram]{\ref{Fig:latticediagram}b}].
From this we find the non-inclusions $\mathcal{T}_{+}(\v t^{(1)}) \not\subset \mathcal{T}_{+}(\v t^{(4)})$ and $\mathcal{T}_{+}(\v t^{(4)}) \not\subset \mathcal{T}_{+}(\v t^{(1)})$, similarly $\mathcal{T}_{+}(\v t^{(2)}) \not\subset \mathcal{T}_{+}(\v t^{(3)})$ and $\mathcal{T}_{+}(\v t^{(3)}) \not\subset \mathcal{T}_{+}(\v t^{(2)})$. Naively, one may be led to a conclusion that the incomparable region can be characterised by the future cones of $\v t^{(1)}$ and $\v t^{(4)}$ alone. However, any convex combination $\lambda \v{t}^{(2)} + (1-\lambda)\v {t}^{(3)} \equiv \v t^{(2,3)}_\lambda$ results in an incomparable vector $\v{t}^{(2,3)}_\lambda \in \mathcal{T}_\emptyset(\v{t}^{(i)})$ for $i = 1,2,3,4$ and $0<\lambda<1$
[see Fig.~\hyperref[Fig:latticediagram]{\ref{Fig:latticediagram}c}], and hence, in a ``new'' fragment of the incomparable region. In order to account for the entire incomparable region one must take the union of all future cones of $\v t^{(2,3)}_\lambda$ for $\lambda \in [0,1]$. This corresponds to the convex hull of the future cones\footnote{Geometrically, the mixture of $\v t^{(2)}$ and $\v t^{(3)}$ corresponds to the edge that connects these two points.} $\mathcal{T}_{+}(\v t^{(2)})$ and $\mathcal{T}_{+}(\v t^{(3)})$. Furthermore, the construction is limited only to convex combinations of futures for consecutive tangent vectors since mixtures of two non-successive ones, for instance $\v t^{(1)}$ and $\v t^{(4)}$, do not give any additions to the incomparable region. Such combinations belong to the past cone of $\v p$ as every point of $\lambda \v{t}^{(1)} + (1-\lambda)\v{t}^{(4)}$ would majorise $\v{p}$.

\newpage
\subsection{Links to other resource theories}~\label{Subsec:other-resource-theories}

The two well-known examples of majorisation-based resource theories, where our results are also applicable, include the resource theories of entanglement and coherence. These are defined via the appropriate sets of free operations and free states: local operations and classical communication (LOCC) and separable states in entanglement theory~\cite{Horodecki2009}, and incoherent operations (IO) and incoherent states in coherence theory~\cite{Plenio2014}. Within each of these theories, there exists a representation of quantum states via probability distributions that is relevant for formulating state interconversion conditions under free operations. In entanglement theory, a pure bipartite state \mbox{$\rho = \ketbra{\Psi}{\Psi}$} can be written in terms of the Schmidt decomposition given by $\ket{\Psi} = \sum_i a_i \ket{\psi_i, \psi'_i}$, and represented by a probability vector $\v{p}$ with $p_i = |a_i|^2$. Then, Nielsen’s theorem~\cite{nielsen1999conditions} states that an initial state $\v p$ can be transformed under LOCC into a target state $\v q$ if and only if $\v p \prec \v q$. Similarly, in the resource theory of coherence with respect to a fixed basis $\{ \ket{i}\}$, one can represent a pure state $\rho = \ketbra{\psi}{\psi}$ by a probability vector $\v p$ with $p_i = |\langle i | \psi \rangle|^2$. Then, a given initial state $\v p$ can be transformed into $\v q$ via incoherent operations if and only if $ \v p \prec \v q$~\cite{Du2015}.
\begin{marginfigure}[0cm]
	\includegraphics[width=4.859cm]{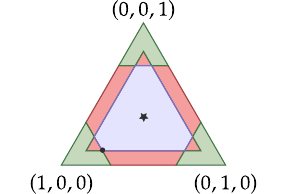}
	\caption[fig-entanglementcone]{\emph{Achiviebility under LOCC}. Entanglement cone in the simplex of the Schmidt coefficients of a $3\times 3$ system. Conversely to thermodynamics, the past of entanglement transformations is the thermodynamic future and vice-versa. The black dot $\bullet$ indicates the Schmidt vector of the initial state \mbox{$\v p = (0.7,0.2,0.1)$}, whereas the black star $\bigstar$ represents the maximally entangled state \mbox{$\v \eta = (1/3,1/3,1/3)$}}
	\label{fig-entanglementcone}
\end{marginfigure}
Therefore, we observe that the partial order emerging in the two cases is precisely the opposite to the thermodynamic order in the infinite temperature limit (for more details see Ref.~\cite{Kollas_master}). Consequently, the thermodynamic past and future become the future and past for entanglement and coherence, while the incomparable region remains unchanged~(see Fig.~\ref{fig-entanglementcone}). Note that, for entanglement and coherence, sharp states $\v s$ are in the future cone of any given state, while for thermodynamics (at $\beta = 0$), they are in the past. The flat distribution $\v \eta$ is in the past of any state in entanglement and coherence theories, whereas in thermodynamics it is in the future.

One can make a general remark concerning resource monotones, applying to the entanglement, coherence and thermodynamic scenarios alike. Consider an entangled state \mbox{$\ket\psi \in \mathcal{H}_N\otimes \mathcal{H}_N$} with the associated Schmidt coefficients $\v{p}$ and concurrence $\mathcal{C}(\ket{\psi})$ as an example of a resource monotone~\cite{HW97}. If another state $\ket{\phi}$ with Schmidt coefficients $\v{q}$ is in the future cone of $\ket{\psi}$, $\v{q}\in\mathcal{T}_+(\v{p})$, then $\mathcal{C}(\ket{\phi}) \leq \mathcal{C}(\ket{\psi})$. Otherwise, if it lies in its past cone, $\v{q}\in\mathcal{T}_-(\v{p})$ we know that $\mathcal{C}(\ket{\phi}) \geq \mathcal{C}(\ket{\psi})$. However, if the two states are incomparable,  $\v{q}\in\mathcal{T}_\emptyset(\v{p})$, nothing can be said about the relation between both concurrences.

\subsection{Probabilistic majorisation cones}~\label{Subsec:probabilistic-cones}

Finally, it should be observed that the notion of majorisation cones, as presented until now, deals with deterministic transformations. However, this approach can be extended to probabilistic transformations using Vidal's criterion for entanglement~\cite{vidal1999entanglement,vidal2000approximate} and coherence transformations~\cite{ZhuEtAl2017coherence} under LOCC and IO, respectively. In the case of probabilistic transformations of bipartite entangled states under LOCC, this is captured by the following theorem~\cite{vidal1999entanglement}. 

\begin{restatable}{theorem}{vidaltheorem}\label{Thm:vidal}
    Consider two bipartite pure states $\ket{\psi}$ and $\ket{\phi}$, whose Schmidt decompositions are described by probability vectors $\v p$ and $\v q$, respectively. The maximal transformation probability from $\ket{\psi}$ to $\ket{\phi}$ under LOCC is given by
    \begin{equation} \label{eq:probal_crit_IO}
        \mathcal{P}(\v{p}, \v{q}) = \min_{1\leq k \leq d} \frac{\sum_{j = k}^d p^{\downarrow}_j}{\sum_{j = k}^d q
        ^{\downarrow}_j}.
    \end{equation}
\end{restatable}

\begin{figure*}
\includegraphics[width=16.107cm]{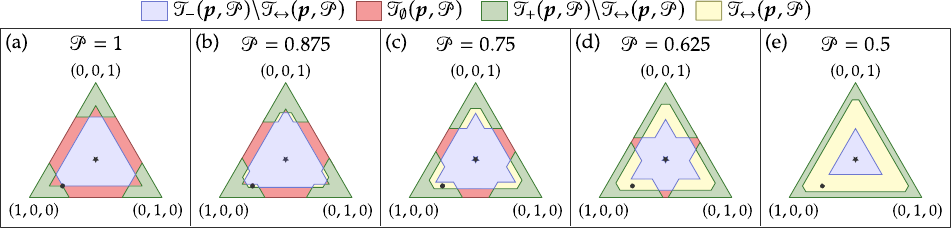}
    \caption{\emph{Probabilistic majorisation cones for $d=3$}. For a three-level system with a state given by $\v{p} = (0.7,0.2,0.1)$ represented by a black dot~$\bullet$ and a maximally entangled state $\v \eta = (1/3,1/3,1/3)$ represented by a black star $\bigstar$, we plot its probabilistic majorisation cone for probabilities of transformation $\mathcal{P}$ decreasing from  $1$ to $0.5$ with $0.125$ steps (a-e), respectively. Observe that, for $\mathcal{P} = 1$, we recover the structure of the standard majorisation cones, while as $\mathcal{P}$ decreases the interconvertible region $\mathcal{T}_{\leftrightarrow}(\v{p},\mathcal{P})$ expands and the incomparable region $\mathcal{T}_{\emptyset}(\v{p},\mathcal{P})$ shrinks, disappearing altogether between panel (d)~and~(e)}.
    \label{fig-probabilistic-cones-examples} 
\end{figure*}

In Section~\ref{app:probal_deriv}, we discuss the extension of majorisation cones to probabilistic ones, denoted as $\mathcal{T}_i(\v p; \mathcal{P})$, with \mbox{$i \in \{-, \emptyset, + \}$} and $\mathcal{P}$ being the minimal probability of transformation. Here we will limit ourselves to a brief qualitative discussion
about the behaviour of the probabilistic majorisation cones as the transformation probability changes from $\mathcal{P}(\cdot,\cdot) = 1$ to $\mathcal{P}(\cdot,\cdot) < 1$ (see Fig.~\ref{fig-probabilistic-cones-examples}). Note that the only common points of the future and past for $\mathcal{P} = 1$ are the current state of the system $\v p$ and its permutations. Conversely, for $\mathcal{P} < 1$ this is not the only case; consequently, we may define the \textit{interconvertible region} of $\v p$ at the probability level $\mathcal{P}$ as the intersection between the probabilistic past and probabilistic future, $\mathcal{T}_{\leftrightarrow}(\v p, \mathcal{P}) \equiv \mathcal{T}_+(\v p, \mathcal{P}) \cap \mathcal{T}_-(\v p, \mathcal{P})$. This region is non-empty for every $\mathcal{P} < 1$. It is easily shown that the future and the past cones grow as the probability of transformation decreases, $\mathcal{T}_+(\v p, \mathcal{P}') \subset \mathcal{T}_+(\v p, \mathcal{P}) $ and $\mathcal{T}_-(\v p, \mathcal{P}') \subset \mathcal{T}_-(\v p, \mathcal{P})$ for $\mathcal{P}' > \mathcal{P}$. Therefore, the only region that decreases together with $\mathcal{P}$ is the incomparable region, $\mathcal{T}_\emptyset(\v p, \mathcal{P}') \supset \mathcal{T}_\emptyset(\v p, \mathcal{P})$. Interestingly, for every state~$\v p$, we observe that there is a critical value $\mathcal{P}^*$, at which no two states are incomparable, i.e., $\mathcal{T}_\emptyset(\v p, \mathcal{P}) = \emptyset$~[see Figs.~\hyperref[fig-probabilistic-cones-examples]{\ref{fig-probabilistic-cones-examples}a,~e}]. Analogous results hold in the context of coherence, as Theorem~\ref{Thm:vidal} has its counterpart when considering pure state transformations under IO operations~\cite{ZhuEtAl2017coherence,Cunden2021}.

Finally, it is worth mentioning, that a criterion similar to the Vidal's criterion was established for probabilistic transformation in the context of thermal operations~\cite{AOP16}. In this case, the construction of the probabilistic cones for majorisation generalises directly to the thermomajorisation by using the construction of thermomajorisation cones which we will introduce in the next section. 

\section{Thermal cones}~\label{Sec:thermal-cones}

Let us turn our attention to a more general scenario, assuming that the temperature is finite, $\beta> 0$. In this case, the reachability of energy-incoherent states under Gibbs-preserving matrices can be studied by introducing the notion of \emph{thermal cones}, defined as follows:

\begin{definition}[Thermal cones]\label{def_Thermal_cones}
The set of states that an energy-incoherent state $\v p$ can be mapped into by Gibbs-preserving matrices is called the \textbf{future thermal cone} $\T^{\, \beta}_+(\v{p})$. Similarly, the set of states that can be mapped to $\v p$ by Gibbs-preserving matrices is called the \textbf{past thermal cone} $\T^{\, \beta}_-(\v{p})$. Finally, the set of states that are neither in the past nor in the future of $\v p$ is called the \textbf{incomparable thermal region} $\T^{\, \beta}_{\emptyset}(\v{p})$.
\end{definition}

Despite apparent similarities, the case of $\beta > 0$ turns out to be significantly harder than $\beta=0$. Difficulties stem mostly from a simple fact demonstrated in Ref.~\cite{Korzekwa2017} -- even though thermomajorisation forms a lattice in each $\beta$-order, it does not provide a common lattice for the entire probability simplex. Thus, before extending Lemma~\ref{Lem:Incomparable_region0} and Theorem~\ref{Thm:past_cone} to the thermal setting (proofs of which rely heavily on the existence of a join), we will introduce an \emph{embedding lattice} -- a structure which encompasses thermomajorisation order as its subset -- and we will demonstrate operations shifting to and from the newly introduced picture.

\subsection{Embedding lattice} 
\label{Subsec:embedding-lattice}

To define a $\beta$-dependent embedding $\mathfrak{M}$ of the simplex $\Delta_d$ into a subspace $\Delta^{\mathfrak{M}}_d\subset\Delta_{2^d - 1}$, illustrated by a graphical example in Fig.~\ref{fig-embeddingscheme}, we first introduce the vector $\v \Gamma^{\mathfrak{M}}$ with entries given by all possible partial sums of the Gibbs distribution,
\begin{equation}
    \v \Gamma^{\mathfrak{M}} = \left\{\sum_{i\in I}\gamma_i:\,I\in2^{\{1,\hdots,d\}}\right\},
\end{equation}
with $2^{\{1,\hdots,d\}}$ denoting the power set of $d$ indices. Moreover, we enforce that it is ordered non-decreasingly, i.e., for \mbox{$i > j$} we have \mbox{$\Gamma^{\mathfrak{M}}_i \geq \Gamma^{\mathfrak{M}}_j$}. Then, the embedded probability vector, \mbox{$\v{p}^{\mathfrak{M}}:=\mathfrak{M}(\v p)$}, is defined by 
\begin{align}
    p^{\mathfrak{M}}_i & = f^{\, \beta}_{\v p}\left(\Gamma^{\mathfrak{M}}_i\right) -f^{\, \beta}_{\v p}\left(\Gamma^{\mathfrak{M}}_{i-1}\right) ,
\end{align}
where $f^{\, \beta}_{\v p}$ is the thermomajorisation curve of $\v{p}$.

Within this embedding, the thermomajorisation indeed proves to be almost-standard majorisation relation between the embedded distributions,  
\begin{equation} \label{eq_embedSpace_majorisationDefinition}
    \v p \succ_\beta \v q ~\Longleftrightarrow~ \forall j:~\sum_{i=1}^j p^{\mathfrak{M}}_i \geq \sum_{i=1}^j q^{\mathfrak{M}}_i ~\overset{\text{def}}{\Longleftrightarrow}~\v {p}^{\mathfrak{M}}\succ_{\mathfrak{M}}\v{q}^{\mathfrak{M}},
\end{equation}
where the last symbol $\succ_{\mathfrak{M}}$ denotes the majorisation variant related to the embedding lattice. Finally, we note that the only deviation from the standard majorisation lies in the convexity condition in the embedding space, which is imposed not on the probabilities $p^{\mathfrak{M}}_i$ themselves, but on their rescaled versions, 

\begin{equation} \label{eq_embedSpace_orderingRule}
    i \leq j \Rightarrow \frac{p^{\mathfrak{M}}_i}{\gamma^{\mathfrak{M}}_i} \geq \frac{p^{\mathfrak{M}}_j}{\gamma^{\mathfrak{M}}_j} ,
\end{equation}
with scaling factors directly related to the embedded majorisation curve of the Gibbs state $\v{\gamma}^{\mathfrak{M}}$. This ordering should be compared (but not confused) with the $\beta$-ordering introduced in Eq.~\eqref{Eq:beta-ordering}, pointing to a relation with thermomajorisation which we will use to show the lattice structure of the introduced space.
\begin{figure*}
    \centering
    \includegraphics{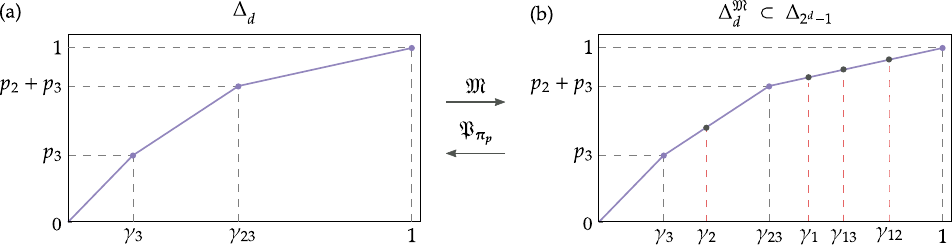}
    \caption{\label{fig-embeddingscheme} \emph{Embedded majorisation curve}. Embedding $\mathfrak{M}:\Delta_d\rightarrow\Delta_{2^d-1}$ of $d$-dimensional probability distribution $\v{p}$ (here $d=3$) into a $2^d-1 = 7$-dimensional space is most easily understood by noting that each thermomajorisation curve (left) has elbows corresponding to a subset of $\v{\Gamma}^{\mathfrak{M}}$ on the horizontal axis. The embedding includes all entries of $\v{\Gamma}^{\mathfrak{M}}$ by subdividing the Lorenz curve into $2^d - 1$ fragments. Conversely, the projection $\mathfrak{P}_{\v \pi_{\v p}}$ corresponds to selecting only a subset of elbows that correspond to a selected order, in this case the original order $\v{\pi}_{\v p}$. In the $x$-axes we used the shorthand notation \mbox{$\gamma_{ij} = \gamma_i + \gamma_j$}. }
\end{figure*}
The projection $\mathfrak{P}_{\v \pi}$ of an arbitrary probability vector \mbox{$\v{q}\in\Delta_{2^d-1}$} satisfying Eq.~\eqref{eq_embedSpace_orderingRule} onto a selected $\beta$-order $\v \pi$ in the original space can be defined descriptively as taking only those elbows of the embedded majorisation curve that match the values of cumulative Gibbs distribution for the selected permutation. Formally, the projected vector, \mbox{$\v q^{\mathfrak{P}}_{\v{\pi}}:=\left(\mathfrak{P}_{\v \pi}\left(\v q\right)\right)^{\, \beta}$}, is entry-wise defined by 
\begin{equation} \label{eq_embedding_porj_constr}
   \left(q^\mathfrak{P}_{\v{\pi}}\right)_i = \sum^{k(i)}_{j = k(i-1)} q_j ,
\end{equation}
with the indices $k(i)$ defined by the requirement that \mbox{$\Gamma^\mathfrak{M}_{k(i)} = \sum_{j=1}^i \gamma_{\v \pi^{-1}(j)}$}. 

In particular, it is worth noting two properties of the embedding $\mathfrak{M}$ and projections $\mathfrak{P}_{\v \pi}$. First, given a vector $\v p \in\Delta_d$ with a $\beta$-order $\v \pi_{\v p}$, we find that by construction $\mathfrak{P}_{\v \pi_{\v p}}\left(\mathfrak{M}\left(\v p\right)\right) = \v p$, which follows directly from Eq.~\eqref{eq_embedding_porj_constr}. On the other hand, for $\v \pi \neq \v \pi_{\v p}$  we find that $\v p \succ_\beta \mathfrak{P}_{\v \pi}\left(\mathfrak{M}\left(\v p\right)\right)$. The statement is easily shown by observing that the Lorenz curve of $\mathfrak{P}_{\v \pi}\left(\mathfrak{M}\left(\v p\right)\right)$ connects by line segments $d+1$ points of the Lorenz curve corresponding $f^\beta_{\v p}\left(\sum_i \left(\mathfrak{M}(p)^\mathfrak{P}_{\v{\pi}}\right)_i\right)$ and therefore majorisation is resolved by linear approximation of a convex function, $f[(1-t)x + t y] \geq (1-t)f(x) + tf(y)$ for any $x, y$. The second property is concerned with $\v q \in\Delta_{2^d-1}$ satisfying Eq.~\eqref{eq_embedSpace_orderingRule} and can be summarised as the fact that projecting and re-embedding the vector will always give the object majorised by the original vector: $\v q\succ_\mathfrak{M} \mathfrak{M}(\mathfrak{P}_{\v \pi}(\v q))$. It follows similarly to the prior majorisation by the argument of linear approximation of a convex function.

It is necessary to stress that the introduced embedding structure is distinct from the one used in the usual method of reducing thermomajorisation to majorisation~\cite{horodecki2013fundamental}. Most importantly, the standard approach requires going to the limit of infinite embedding dimension in order to recover thermomajorisation for arbitrary $\beta$ as a special case of standard majorisation. In our proposition, a thermomajorisation curve in dimension $d$ is embedded within a $2^d-1$-dimensional space. The main difference lies in the non-constant widths of the segments of the Lorenz curve and the fact that once the embedded vector $\v{p}^{\mathfrak{M}}$ is constructed, its entries should not be subject to reordering. Finally, we present the argument proving that the embedding together with embedded majorisation indeed provide a lattice structure.
\begin{corollary}
    The subset of the probability simplex $\Delta_{2^d - 1}$ satisfying Eq.~\eqref{eq_embedSpace_orderingRule} and subject to the embedded majorisation $\succ_\mathfrak{M}$ defined in Eq.~\eqref{eq_embedSpace_majorisationDefinition} forms a lattice.
\end{corollary}

\begin{proof}
The embedded majorisation $\succ_\mathfrak{M}$ may be reinterpretted as thermomajorisation defined for a specific Gibbs state $\v{\gamma}^\mathfrak{M}$ and restricted to a particular Weyl chamber of the probability space $\Delta_{2^d - 1}$ by comparing Eq.~\eqref{eq_embedSpace_orderingRule} with an analogous sorting rule from Eq.~\eqref{Eq:beta-ordering}. This direct isomorphism between thermomajorisation order $\succ_\beta$ restricted to a single Weyl chamber, known to provide a lattice structure, and the embedded majorisation $\succ_\mathfrak{M}$ proves the statement.
\end{proof}

\subsection{Geometry of thermal cones}~\label{Subsec:geometry-of-thermal-cones}

As already mentioned, for finite temperatures the rules underlying state transformations are no longer captured by a majorisation relation, but rather by its thermodynamic equivalent known as thermomajorisation~\cite{Rusch,horodecki2013fundamental}. As a result, Birkhoff’s theorem cannot be employed anymore, and the characterisation of the future thermal cone is no longer given by Theorem~\ref{Thm:Future_majoristion_cone}. However, the set of Gibbs-preserving matrices still forms a convex set~\cite{mazurek2018decomposability,mazurek2019thermal}, and the extreme points of the future thermal cone can be constructed by employing the following lemma:
\begin{lemma}[Lemma 12 of Ref.~\cite{Lostaglio2018elementarythermal}]
    \label{lem_extreme}
	Given $\v{p}$, consider the following distributions $\v{p}^{\, \v \pi}\in \T^{\, \beta}_{+}(\v{p})$ constructed for each permutation $\v \pi\in \S_d$. For $i\in\left\{1,\dots,d\right\}$:
	\begin{enumerate}
		\item Let $x_i^{\v \pi}=\sum_{j=0}^{i} e^{-\beta E_{\v \pi^{-1}\left(j\right)}}$ and $y_i^{\v \pi}=f^{\, \beta}_{\v{p}}\left(x_i^{\v \pi}\right)$. 
		\item Define $p^{\v \pi}_i:=y^{\v \pi}_{\v \pi(i)} - y^{\v \pi}_{\v \pi(i)-1}$, with $y_{0}:=0$.
	\end{enumerate}
	Then, all extreme points of $\T^{\, \beta}_{+}(\v{p})$ have the form $\v{p}^{\v \pi}$ for some ${\v \pi}$. In particular, this implies that $\T^{\, \beta}_{+}(\v{p})$ has at most $d!$ extremal points.
\end{lemma}

The above lemma allows one to characterise the future thermal cone of $\v p$ by constructing states $\v p^{\, \v \pi}$ for each $\v \pi \in \mathcal S_d$, and taking their convex hull. It is worth mentioning that $\T^{\, \beta}_{+}(\v p)$ can also be constructed by finding the whole set of extremal Gibbs-preserving matrices~\cite{Lostaglio2018elementarythermal, mazurek2018decomposability}. This follows the same spirit as in Sec.~\ref{Sec:majorisation-cones}, where the majorisation cone was characterised by employing Theorem~\ref{Thm:HLP}. However, this is a harder problem to solve, and so the extremal Gibbs-preserving matrices were characterised only for $d \leq 3$~\cite{mazurek2018decomposability,mazurek2019thermal}. Here, we provide the construction of $\T^{\, \beta}_{+}(\v p)$ as a simple corollary of Lemma~\ref{lem_extreme}.
\begin{figure*}
    \centering
    \includegraphics{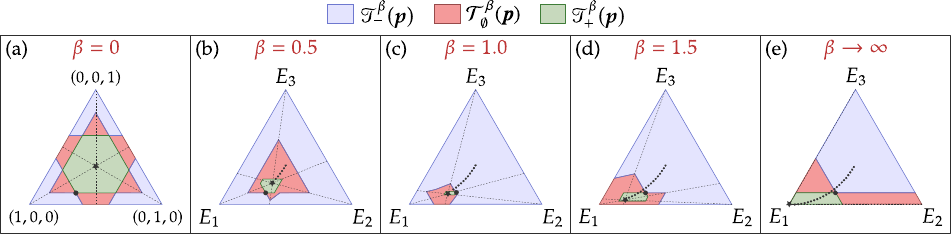}
    \caption{\emph{Thermal cones for $d=3$}. For a three-level system with population given by $\v p = (0.4, 0.36, 0.24)$, represented by a black dot $\bullet$, and energy spectrum $E_1 = 0, E_2 = 1$ and $E_3 = 2$, we plot its thermal cone for (a) $\beta = 0$, (b) $\beta = 0.5$, (c) $\beta = 1.0$ and (d) $\beta \to \infty$. By increasing $\beta$, the thermal state (black star $\bigstar$) tends toward the ground state $E_1$, and the past thermal cone becomes convex.} 
    \label{fig-thermal-cones-examples} 
\end{figure*}
\begin{corollary}[Future thermal cone]
\label{thm_futurefinite}
The future thermal cone of a $d$-dimensional energy-incoherent state $\v p$ is given by
\begin{equation}
\T^{\,\beta}_+(\v p) = \operatorname{conv}[\{\v p^{\v \pi}, \v \pi\in\S_d\}] \, .
\end{equation}
\end{corollary}

Furthermore, one can use the embedding lattice to provide an alternative formulation for the future thermal cone:
\begin{observation}
    Since $\v{p}^{\v{\pi}} = \mathfrak{P}_{\v{\pi}}(\mathfrak{M}(\v{p}))$, the future thermal cone of an energy-incoherent state $\v p$ can be expressed in terms of all possible projections from the related embedding,
    \begin{equation}
        \mathcal{T}^\beta_+(\v{p}) = \operatorname{conv}\left[\left\{\mathfrak{P}_{\v{\pi}}(\mathfrak{M}(\v{p}))\,, \v \pi\in\mathcal{S}_d\right\}\right].
    \end{equation}
\end{observation}
Our main technical contribution is captured by the following Lemma that generalises Lemma~\ref{Lem:Incomparable_region0} and provides the construction of the incomparable thermal region for finite temperatures. Its proof is based on the concept of embedding lattice that we introduced in Sec.~\ref{Subsec:embedding-lattice} and can be found in Appendix~\ref{app:finite_temp_derivs}

\begin{lemma}[Incomparable thermal region] 
\label{Lem:Incomparable_regionT}
Given an energy-incoherent state $\v p$ and a thermal state $\v \gamma$, consider distributions $\v t^{(n,\v \pi)}$ in their $\beta$-ordered form, constructed for each permutation $\v \pi \in \S_{d}$,
\begin{equation}
\label{eq_thermaltangentvectors}
\left(\v{t}^{(n,\v \pi)}\right)^{\, \beta} = \left(
    t^{(n,\v \pi)}_{\v \pi(1)}, 
    p^{\, \beta}_n\frac{\gamma_{\v \pi (2)}}{\gamma^{\, \beta}_n}, ..., 
    p^{\, \beta}_n\frac{\gamma_{\v \pi (d-1)}}{\gamma^{\, \beta}_n},
    t^{(n,\v \pi)}_{\v \pi(d)}\right) ,  
\end{equation}
with
\begin{subequations}
\begin{align}
    t^{(n,\v \pi)}_{\v \pi(1)} & = \sum_{i=1}^n p^{\, \beta}_i-\frac{p^{\, \beta}_n}{\gamma^{\,\beta}_n}\left(\sum_{i=1}^n \gamma^{\,\beta}_i-\gamma_{\v \pi(1)} \right), \\ 
    t^{(n,\v \pi)}_{\v \pi(d)} & = 1- t^{(n,\v \pi)}_{\v \pi(1)}-\frac{p^{\, \beta}_n}{\gamma^{\,\beta}_n}\sum_{i=2}^{d-1} \gamma_{\v \pi(i)} .
\end{align}
\end{subequations}
Defining the set
\begin{equation}
\label{eq_incomparableconebeta}
    \mathbb{T}^{\,\beta} = \bigcup_{\v \pi\in\mathcal{S}_d}\bigcup_{i=1}^{n-1}\operatorname{conv}\left[\mathcal{T}^{\,\beta}_+\left(\v{t}^{(i,\v \pi)}\right)\cup\mathcal{T}^{\,\beta}_+\left(\v{t}^{(i+1,\v \pi)}\right)\right], 
\end{equation}
the incomparable region of $\v{p}$ is given by 
\begin{align}
\label{eq_incomparableconeFT}
\T^{\, \beta}_{\emptyset}(\v{p}) = \left[\operatorname{int}\left(\mathbb{T}^{\, \beta}\right)\backslash\T^{\, \beta}_+(\v p)\right]\cap \v{\Delta}_d .
\end{align}
\end{lemma}

\begin{figure*}
    \centering
    \includegraphics{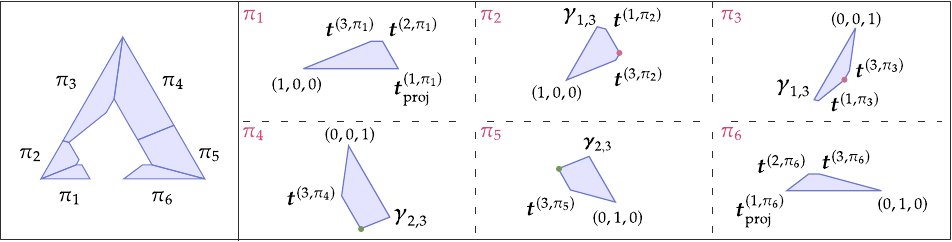}
    \caption{\label{fig-past_chamberextreme} \emph{Extreme points of $\mathcal{T}^{\,\beta}_-(\v p)$}. The past thermal cone for a three-level system with population give by $\v p = (0.7,0.2,0.1)$, energy spectrum $E_1 = 0, E_2 = 2$ and $E_3 = 3$, and the inverse temperature $\beta = 0.5$. At each chamber $\v \pi$, the non-trivial extreme points of the past are given by $\v t^{(n, \v \pi)}$ (for $n \in \{1, 2, 3\}$ and $\v \pi \in \mathcal{S}_d$) and by extreme points of \mbox{$\bigcup_{i=1}^{d-1} \operatorname{conv}[(\mathcal{T}^{\, \beta}_+(\v{t}^{(i,\v \pi)})\cup\mathcal{T}^{\, \beta}_+(\v{t}^{(i+1,\v \pi)})]$} (red dots). The remaining ones are sharp states, points in the boundary between chambers (green dots) and those which are on the edge of the probability simplex, i.e., $(\v \gamma_{i,j})_{k} = (\gamma_i \delta_{ik}+\gamma_j \delta_{ik})/(\gamma_i + \gamma_j)$, for $i\neq j$, and $k \in \{1, 2,3\}$.}
    \end{figure*}
Analogously to the infinite temperature case, Lemma~\ref{Lem:Incomparable_regionT} allows us to obtain the past thermal cone of $\v p$.
\begin{theorem}[Past thermal cone]
\label{thm_pasthermalconefinite}
The past thermal cone of $\v p$ is given by \begin{equation}
\T^{\, \beta}_{-}(\v p)=\v{\Delta}_d \backslash \operatorname{int} \left(\mathbb{T}^{\, \beta}\right) \, .
\end{equation}
\end{theorem}
\begin{proof}
Following the same reasoning as in the proof of Theorem~\ref{Thm:past_cone}, we only need to use the fact that \mbox{$\operatorname{int}(\mathbb{T}^\beta) = \mathcal{T}^\beta_+(p) \cup \mathcal{T}^\beta_\emptyset(p)$}.   
\end{proof}

In Fig.~\ref{fig-thermal-cones-examples}, we illustrate Lemma~\ref{Lem:Incomparable_regionT} and Theorem~\ref{thm_pasthermalconefinite} for a three-level system. We also provide a \texttt{Mathematica} code~\cite{adeoliveirajunior2022} that constructs the set of extreme points of the future and past thermal cones for arbitrary dimensions.

As before, the past thermal cone forms a convex polytope only when restricted to a single Weyl chamber, now defined as a set of probability vectors with common $\beta$-order. The extreme points of the past thermal cone correspond to tangent vectors $\v t^{(n,\v \pi)}$ or by their projection onto the boundary of the probability simplex. The exceptional points in comparison with the infinite-temperature case may appear when considering extreme points of \mbox{$\bigcup_{i=1}^{d-1} \operatorname{conv}[(\mathcal{T}^\beta_+(\v{t}^{(i,\v \pi)})\cup\mathcal{T}^\beta_+(\v{t}^{(i+1,\v \pi)})]$} for a given chamber $\v \pi$. Vertices arising in this way may not correspond to any tangent vector (see Fig.~\ref{fig-past_chamberextreme}). Moreover, the $\v t^{(n,\v \pi)}$ vectors are also responsible for the convexity of the past thermal cone, and the following observation illustrates this:
\begin{corollary}[Asymptotic temperature limit]
Approaching the limit of $\beta \to \infty$, the past thermal cone becomes convex.
\end{corollary}
\begin{proof}
By dividing the probability simplex into equal chambers with the thermal state in the barycenter, the past thermal cone is the union of $d!$ convex pieces. As $\beta \to \infty$, the thermal state collapses to the ground state. There is only one chamber in this limit, and therefore, the past thermal cone is a single convex piece [see Fig.\hyperref[fig-thermal-cones-examples]{\ref{fig-thermal-cones-examples}~e} for an example in the particular case of $d=3$].
\end{proof}

The intuition behind the above observation can be understood by studying the behaviour of a three-level system and the tangent vector $\v t^{(3,(132))}$. By a decreasing temperature, this extremal point tends towards the edge of the simplex and reaches the edge at $\beta \to \infty$ [see Figs.~\hyperref[fig-thermal-cones-examples]{\ref{fig-thermal-cones-examples}b,~e}].

To wrap up the considerations of the geometry of thermal cones, let us go back, once again, to the analogy between the thermal cones and special relativity. Consider that given a specific division of space-time into future, past and space-like regions, one is able to recover the specific event generating it, which we may refer to as present. Concisely -- there is a one-to-one relation between events and the divisions of space-time they generate. The situation is exactly reflected for thermal cones with $\beta > 0$ -- given a specific arrangement of incomparable region and future and past thermal cones, one can exactly recover the current state of the system [see the black dot~$\bullet$ in Fig.~\hyperref[fig-thermal-cones-examples]{\ref{fig-thermal-cones-examples}b-d}]. It is in stark contrast with the majorisation cones for $\beta = 0$, where every division into past, future and incomparable possesses $d!$-fold symmetry [see Fig.~\hyperref[fig-thermal-cones-examples]{\ref{fig-thermal-cones-examples}a}] and hence, the present state of the system cannot be recovered solely on its basis unless provided with additional information like the permutation which sorts the probabilities in non-decreasing order.

\section{Volume of thermal cones}~\label{Sec:volume-thermal-cones}

In the previous sections, we characterised and discussed the behaviour of thermal cones by introducing explicit constructions of the past, the future and the incomparable region. It is then natural to ask what is the role played by their volumes in quantifying the resourcefulness of different states~\cite{Cunden_2020,Cunden2021}. Thus, for a $d$-dimensional probability vector $\v p$ we define the relative volumes of its thermal cones as
\begin{equation}
    \mathcal{V}^{\, \beta}_{i}(\v p) := \frac{V\left[\mathcal{T}^{\, \beta}_i (\v p)\right]}{V(\Delta_d)} \quad  \text{with} \quad i \in \{\emptyset, -, +\},
\end{equation}
where $V$ denotes the volume measured using the Euclidean metric. 

We start our analysis of volumes of thermal cones by presenting an operational interpretation in terms of guessing probabilities for the future and past of a given state subject to a thermal evolution; such interpretation provides a solid basis for presenting the aforementioned volumes as resource-theoretic monotones. Subsequently, we proceed to an in-depth analysis of the volumes with a particular focus on their behaviour as a function of the inverse temperature $\beta$. Finally, we explain how to modify the analysis to obtain meaningful volumes of entanglement cones.

\subsection{Interpretation}~\label{Subsec:interpretation}

Consider a task of predicting the future, which roughly translates to guessing a state $\v{q}$ by having knowledge that it has originated from a given prior state $\v p$. In this case, the probability of correctly guessing a state that is $\epsilon$-distant from ${\v q}$ is given by
\begin{equation}
    \text{Pr}\left( {\v q}, \epsilon\,|\,\v q\in\mathcal{T}_+(\v p)\right) = \frac{V(\mathcal{B}_\epsilon(\v q))}{\mathcal{V}^{\, \beta}_+(\v{p})},
\end{equation}
where $\mathcal{B}_\epsilon(\v q)$ is an $\epsilon$-ball centred at $\v q$. We may get rid of the dependence on $\epsilon$ and $\v q$ by taking the ratio for two different states $\v p_1$ and $\v p_2$,
\begin{equation}
    \frac
    {\text{Pr}\left[\v q, \epsilon\,|\,\v q\in\mathcal{T}^{\, \beta}_+(\v p_2)\right]}
    {\text{Pr}\left[\v q, \epsilon\,|\,\v q\in\mathcal{T}^{\, \beta}_+(\v p_1)\right]} = \frac{\mathcal{V}^{\, \beta}_+(\v{p}_1)}{\mathcal{V}^{\, \beta}_+(\v{p}_2)} .
\end{equation}
Thus, the ratio of volumes yields a relative probability of guessing the future of two different states. In particular, if $\v p_2 \in \mathcal{T}^{\, \beta}_+(\v p_1)$, then also $\mathcal{T}^{\,\beta}_+(\v p_2)\subset\mathcal{T}^{\, \beta}_+(\v p_1)$ and in consequence
\begin{equation}
    \frac{\text{Pr}\left[\v q, \epsilon\,|\,\v q\in\mathcal{T}^{\, \beta}_+(\v p_2)\right]}
    {\text{Pr}\left[\v q, \epsilon\,|\,\v q\in\mathcal{T}^{\, \beta}_+(\v p_1)\right]} > 1. 
\end{equation}
This can be understood as follows -- as the evolution of a system progresses, the future becomes easier to guess or, in other words, more predictable.

In complete analogy, we may define a game in which, instead of guessing -- or predicting -- the future of a state $\v p$, one has to guess its past. For such a game, it is easy to show that 
\begin{equation}
    \frac
    {\text{Pr}\left[\v q, \epsilon\,|\,\v q\in\mathcal{T}^{\, \beta}_-(\v p_2)\right]}
    {\text{Pr}\left[\v q, \epsilon\,|\,\v q\in\mathcal{T}^{\, \beta}_-(\v p_1)\right]} = \frac{\mathcal{V}^{\, \beta}_-(\v{p}_1)}{\mathcal{V}^{\, \beta}_-(\v{p}_2)},
\end{equation}
so that the ratio of volumes yields the relative probability of guessing the past. Given $\v p_1\in \mathcal{T}^{\, \beta}_-(\v p_2)$ we have ${\mathcal{V}^{\, \beta}_-(\v{p}_1)}/{\mathcal{V}^{\, \beta}_-(\v{p}_2)}<1$, which simply means that as the evolution towards equilibrium progresses, one finds the past of a given state harder and harder to guess correctly.

\subsection{Properties}~\label{Subsec:properties}
We now show that the volumes of the thermal cones are thermodynamic monotones, i.e., functions of a state that decrease under thermal operations.
\begin{theorem}[Thermodynamic monotones]\label{thm_monotne}
    The relative volumes $\mathcal V^{\, \beta}_{+}$ and $1-\mathcal V^{\, \beta}_{-} = \mathcal V^{\, \beta}_{+} + \mathcal V^{\, \beta}_{\emptyset}$ are thermodynamic monotones. Moreover, both monotones are faithful, taking the value 0 only when applied to the Gibbs state $\v \gamma$.
\end{theorem}

\begin{proof}
One can straightforwardly show that both quantities decrease monotonically under thermodynamic operations. This is a simple consequence of the fact that for $\v p$ and $\v q$ connected via a thermal operation, we have \mbox{$\mathcal{T}^{\, \beta}_{+}(\v{q}) \subset \mathcal{T}^{\, \beta}_{+}(\v{p})$} and \mbox{$\mathcal{T}^{\, \beta}_{-}(\v{p}) \subset \mathcal{T}^{\, \beta}_{-}(\v{q})$} which automatically implies $\mathcal V^{\, \beta}_{+}(\v q)\leq \mathcal V^{\, \beta}_{+}(\v p)$ and $\mathcal V^{\, \beta}_{-}(\v q)\geq \mathcal V^{\, \beta}_{-}(\v p)$. 

In order to demonstrate its faithfulness, first note that every state that is not thermal can be mapped to a thermal state, thus showing that $1 - \mathcal{V}_-(\v \gamma) = 0$. Similarly, the Gibbs state cannot be mapped via thermal operations to anything else than itself, thus $\mathcal{V}_+(\v \gamma) = 0$. Now, in order to show that both monotones are non-zero for any state different from the Gibbs state, it suffices to demonstrate that $\mathcal{V}_+(\v p) > 0$ for any $\v p\neq \v \gamma$. It is enough to consider, without loss of generality, a state $\v p$ thermalised within a $(d-1)$-dimensional subspace, i.e., \mbox{${p_i}/{p_j} = {\gamma_i}/{\gamma_j}$} for all $i\neq j\neq1$. For $k\in\{2,\hdots,d\}$, using $\beta$-swaps between levels $1$ and $k$ defined as~\cite{Lostaglio2018elementarythermal}
$$
    \left\{p_1,\,p_k\right\} \longmapsto \left\{\frac{\gamma_1 - \gamma_k}{\gamma_1}p_1  + p_k,\,\frac{\gamma_k}{\gamma_1}p_1\right\}, 
$$
we generate $d-1$ new points shifted from the original state by displacements $(\v \delta_k)_i = \delta_k \epsilon_{ik}$ defined by the Levi-Civita symbol $\epsilon_{ik}$ and $\delta_k \neq 0$. The entire set $\{\v \delta_k\}_{k=2}^d$ of displacements is linearly independent; therefore, they define a $(d-1)$-dimensional simplex of non-zero volume. Conversely, if we assume that $\delta_k = 0$ for any $k$, we are led to a conclusion that the system is thermalised between levels $1$ and $k$ and, by transitivity, it must be equal to the Gibbs state $\gamma$, which leads to a contradiction with the initial assumption. Finally, we restore the full generality by noticing that any state $\v q$ contains in its future cone states thermalised in any of the $(d-1)$-dimensional subspaces and, consequently, their entire future cone. Therefore, the non-zero volume of the latter implies the non-zero volume of the former. 
\end{proof}

The behaviour of $\mathcal V_+$ and $\mathcal V_-$ as a function of $\beta$ strongly depends on the $\beta$-ordering of the state under consideration. Among all $\beta$-orderings, there are two extreme cases, namely the one where the population and energies are arranged in non-increasing order ($p_i \leq p_j$ for $E_i \leq E_j$), and the other in which the populations are `anti-ordered' with respect to the energies ($E_i \leq E_j$ implies $p_i \geq p_j$). These two distinct $\beta$-orderings characterise \emph{passive} and \emph{maximally active} states, respectively~\cite{Pusz1978,Lenard1978}. At the bottom left of Fig.~\ref{fig-volumeasfunctionofbeta}, for a three-level system, we depict each Weyl chamber with different colours: passive states lie in the black chamber, whereas maximally active states lie in the red one.
Let us first focus on the volume of the future thermal cone and analyse how it changes by increasing the inverse temperature $\beta$ from $\beta = 0$ to $\beta \to \infty$. In the present analysis, the initial state is kept fixed, while the thermal state is taken to be a function of temperature $\v \gamma = \v \gamma(\beta)$, and it follows a trajectory from the centre of the probability simplex to the ground state (see Fig.~\ref{fig-thermal-cones-examples} for an example considering a three-level system). 
\begin{marginfigure}[0.1pt]
    \centering
    \includegraphics{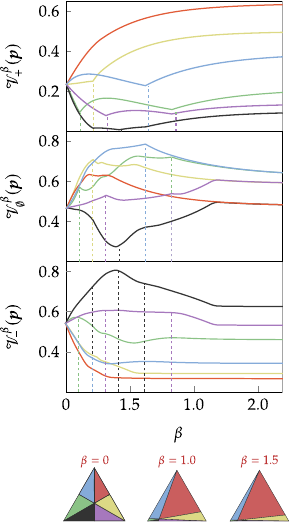}
    \caption{\emph{Thermal cone volumes}. For all permutations of the state $\v p = (0.52, 0.12, 0.36)$, we plot the volume of the future thermal cone $\mathcal{V}^{\, \beta}_{+}$ (top), the incomparable thermal region $\mathcal{V}^{\, \beta}_{\emptyset}$ (centre) and the past thermal cone $\mathcal{V}^{\, \beta}_{-}$ (bottom). Each colour corresponds to a permutation associated with a given chamber of the probability simplex. Among all states, two are distinct: the maximally active (red curve) $\v p_{\text{max}} = (0.12, 0.36, 0.52)$ and the passive (black curve) $\v p_p = (0.52, 0.36, 0.12)$ states. Any other permutation of the initial state characterises a different active state. The kinks in $\mathcal{V}^{\, \beta}_{+}$ match with the inverse temperatures at which $\v p$ changes its $\beta$-ordering (vertical lines of matching colors). The three different simplices at the bottom show how the Weyl chambers change with~$\beta$.}
    \label{fig-volumeasfunctionofbeta} 
\end{marginfigure}
If at $\beta = 0$ the initial state $\v p$ is passive, the volume of its future thermal cone first decreases with $\beta$, and then starts to increase when $\v \gamma$ passes $\v p$ (i.e., when $\v{p}$ changes its $\beta$-ordering), tending asymptotically to a constant value (see black curves in Fig.~\ref{fig-volumeasfunctionofbeta}). However, if at $\beta = 0$ the initial state $\v{p}$ is maximally active, the behaviour of the volume of $\T_+^\beta(\v{p})$ differs from the previous case. As the thermal state approaches the ground state with increasing $\beta$, the distance between maximally active states and $\v \gamma$ increases with $\beta$, because the ground state and $\v{p}$ are located in opposite chambers. Consequently, the volume increases asymptotically to a constant value (see red curves in Fig.~\ref{fig-volumeasfunctionofbeta}). For general states, one can provide a qualitative explanation of the behaviour of their volumes based on their $\beta$-orderings. The inverse temperatures $\beta$ for which one finds kinks in the future volume $\mathcal{V}_+^\beta(\v{p})$ (vertical lines in Fig.~\ref{fig-volumeasfunctionofbeta}) match with the transitions from a given $\beta$-ordering to another one. This should be compared with the isovolumetric level sets in Fig.~\ref{fig-equivolume}, where a matching non-smooth behaviour is found. Observe that similar kinks are not encountered for the past volume $\mathcal{V}^\beta_-(\v{p})$, related to the smooth behaviour of the corresponding level sets. However, by considering a passive state $\v p$ and all its permutations, we can demonstrate that maximally active and passive states have maximum and minimum future volumes, respectively: 

\begin{figure*}
    \centering
\includegraphics{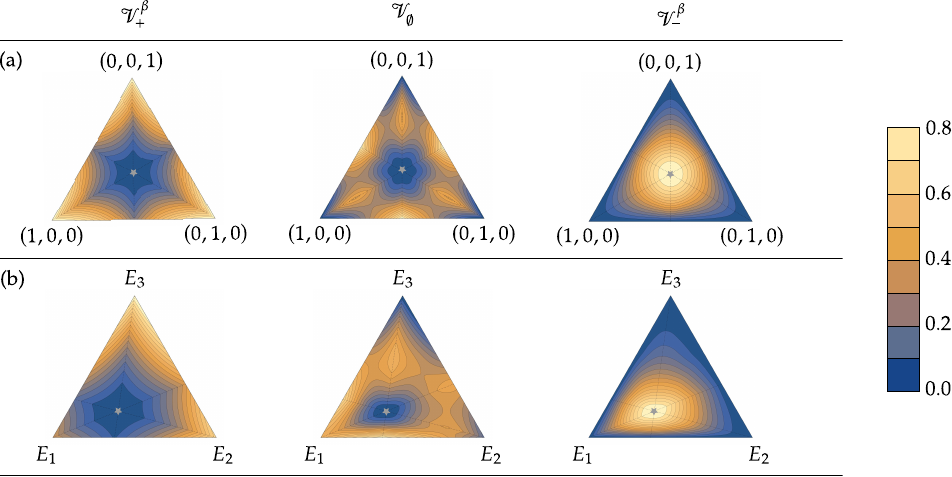}
	\caption[fig-equivolume]{\emph{Isovolumetric sets for thermal cones}. The volume of the future $\mathcal V^{\, \beta}_{+}$, incomparable $\mathcal V^{\, \beta}_{\emptyset}$, and past thermal regions $\mathcal V^{\, \beta}_{-}$, in the space of three-dimensional probability distributions for an equidistant energy spectrum $E_1 =0$, $E_2=1$ and $E_3 = 2$ and inverse temperature (a) $\beta =0$ and (b) $\beta = 0.5$. The thermal state is depicted by a grey star $\color{gray}\bigstar$.}
	\label{fig-equivolume}
    \end{figure*}

\begin{proposition}[Max and min volumes]
\label{obs_passive_active_obs}
For a $d$-dimensional energy-incoherent state $\v p$ with Hamiltonian $H$, and all states defined by permuting its population, the future thermal cone of the permutation resulting in the maximally active and passive states achieve maximum and minimum volumes, respectively.
\end{proposition} 
\begin{proof}
The Corollary is proven by noting that all permutations of $\v p$ are thermomajorised by the one corresponding to the maximally active state $\v p_{\text{max}}$, while the associated passive state $\v p_{p}$ is thermomajorised by all the other permutations. Consequently, $ \mathcal{T}_{+}(\v p_p) \subset \mathcal{T}_{+}(\Pi \v p)$, while $\mathcal{T}_{+}(\Pi \v p) \subset \mathcal{T}_{+}(
\v p_{\text{max}})$ for any permutation matrix $\Pi$.
\end{proof}
Corollary~\ref{obs_passive_active_obs}, also implies that the volume of the past thermal cone is minimum and maximum for the one corresponding to the maximally active and passive states, respectively (see~Fig.~\ref{fig-volumeasfunctionofbeta}). To provide further characterisation of the volumes, we apply Lemma~\ref{Lem:Incomparable_regionT} to a non-full rank state, which allows us to derive the following result: 
\begin{proposition}[Zero volume]
\label{cor_thermalvolumepast}
The past thermal cone of a non-full rank state has volume zero despite being non-empty. 
\end{proposition} 
\begin{proof}
Without loss of generality, consider a state of non-full rank $\v{p} = (p_1, ..., p_{d-1},0)$. Applying Eq.~\eqref{eq_thermaltangentvectors} yields $\v{t}^{(d,\v \pi)} = (1,0,...,0)$, for all $\v \pi$. Consequently, the incomparable region is given by all points in the interior of the probability simplex, except those that are in the future of $\v{p}$. Then, according to Theorem \eqref{eq_incomparableconeFT}, all the points of the past will be located at the edge, and therefore the volume of $\T^{\, \beta}_{-}(\v{p})$ is zero.
\end{proof}

Understanding the behaviour of the thermal incomparable region is not directly straightforward. However, Corollary~\ref{cor_thermalvolumepast} helps us to find the state of non-full rank with the largest incomparable region:
\begin{proposition}[Largest incomparagion region]
\label{cor_full_rank_largest}
The non-full rank state with the largest thermal incomparable region is given by
\begin{equation}
\label{eq_nonfullrankstatethermalinc}
\v g = \frac{1}{Z}\left(e^{-\beta E_1}, ...,e^{-\beta E_{d-1}},0 \right) \quad \text{where} \quad Z = \sum_{i=1}^{d-1}e^{-\beta E_i}. 
\end{equation}
\end{proposition}
\begin{proof}
Consider an arbitrary non-full rank state $\v{p}$. According to Corollary~\ref{cor_thermalvolumepast}, the volume of the thermal incomparable region can be written as $\mathcal V^{\, \beta}_{\emptyset}(\v p) = 1-\mathcal V^{\, \beta}_+(\v p)$. Now note that, $\v p \succ_{\beta} \v g$, and $\mathcal T_{+}(\v g) \subset \mathcal T_{+}(\v p)$. This implies that $\v g$ is the non-full rank state with the smallest future
thermal cone and, therefore, with the largest incomparable region. 
\end{proof}

So far, the behaviour and properties of the volumes have been analysed and discussed without explicitly calculating them. There are several known algorithms for computing volumes of convex polytopes, such as triangulation, signed decomposition methods, or even direct integration~\cite{iwata1962, buller1998}. These algorithms can be employed to obtain the volumes of past and future cones; the volume of the incomparable region can be calculated using the fact that the total volume of the probability simplex $\Delta_d$ is equal to one. For a three-dimensional energy-incoherent state $\v p$, expressions for the volumes can be easily derived. The starting point is to consider the Gauss area formula~\cite{braden1986}, which allows us to determine the area of any polygon with vertices described by Cartesian coordinates. Taking into account a polygon $P$ whose vertices, assumed to be arranged along the boundary in a clockwise manner, are denoted by $P_i = (x_i, y_i)$, with $i= \{1, ..., n\}$, the Euclidean volume can be expressed as 
\begin{equation}
V = \frac{1}{2} \left| \sum_{i=1}^n \operatorname{det} \begin{pmatrix} x_i & x_{i+1} \\ y_i & y_{i+1} \end{pmatrix} \right| ,   \end{equation}
where $x_{n+1} = x_1$ and $y_{n+i} = y_1$. For $\beta = 0$, deriving the volume of the thermal cones is straightforward, since the vertices, or extreme points, are permutation of $\v p$. In this case, we arrive at the following closed-form expressions:
\begin{align}
    \mathcal V^{\, 0}_{+}(\v p) =& \: (3p^{\downarrow}_1-1)^2-3(p^{\downarrow}_2-p^{\downarrow}_1)^2 , \label{eq_futurevolume} \\ 
    \mathcal V^{\, 0}_{\emptyset}(\v p) =& \: 1-3(1-p^\downarrow_1)^2+ (1-3p^\downarrow_3)^2 - 2 \mathcal V^{\, 0}_{+}(\v p) \\
    &\:+ 3\theta(1/2 - p^\downarrow_1)(1-2p^\downarrow_1)^2 \nonumber \\
    \mathcal{V}^{\, 0}_{-}(\v p) =& \: 12p^\downarrow_2p^\downarrow_3 - 3\theta(1/2 - p^\downarrow_1)(1-2p^\downarrow_1)^2 , \label{eq_pastvolume} 
\end{align}
where $\theta$ is the Heaviside step function. The situation involving a finite $\beta$ is not as simple as before. Now, the extreme points are obtained by applying Lemma~\ref{lem_extreme}, and, although computationally it is an easy task to calculate them, a neat and concise closed-form expression cannot be derived. 

Finally, let us look at iso-volumetric curves for different values of $\beta$. In Fig.~\ref{fig-equivolume}, we plot these curves for a three-level system and four different temperatures. As expected, the symmetry is broken for any $\beta > 0$ as $E_3 > E_2 > E_1$. A simple fact worth mentioning is that the volume of the future thermal cone of the highest excited state $\v s_d = (0,\hdots,0,1)$ is always independent of $\beta$, maximum, and equal to unity. Conversely, the past thermal cone volume is maximum for a Gibbs state and equal to unity. Moreover, these curves give insight into how resourceful states are distributed within the space of states. 

\subsection{The volumes of entanglement cones}~\label{Subsec:volume-entanglement-cones}

\begin{figure*}[t]
    \centering
\includegraphics{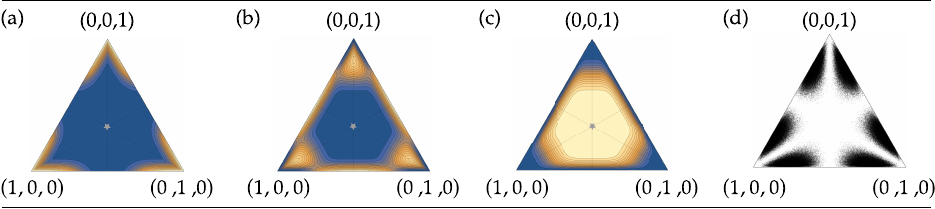}
    \caption{\label{fig-level-sets-entanglement}\emph{Isovolumetric sets for entanglement cones}. In order to calculate the isovolumetric lines for (a) the past cone, (b) the incomparable region and (c) the future cone for entanglement resource theory of states in $\mathcal{H}^{3\times3}$ one considers the density of the Schmidt coefficients $\v{p}$ induced by the Haar measure in the space $\mathcal{H}_9$ of pure bipartite states, as depicted in panel (d) by showing a sample of $5\cdot10^4$ points in $\Delta_3$. Observe the scarcity of points in the central region, resulting in characteristic concentration of states with large future in the centre of the simplex and large past in the vicinity of the vertices.} 
\end{figure*}
Finally, we will briefly discuss the general qualitative aspects of the volumes of entanglement cones based on the numerical considerations. Detailed formal methods used in order to obtain them are discussed in Section~\ref{app:entanglement_vols}.

Naturally, depending on the context, the probability distribution $\v{p}$ may pertain to the Schmidt coefficients of a pure entangled state or the coefficients resulting from decomposing the state in a distinguished basis in the context of coherence~\cite{Streltsov2016,Streltsov2017}. 
Despite the close connection between the resource theory of thermodynamics (at $\beta = 0$) with the resource theories of entanglement and coherence, crucial differences appear already at the level of a single state as the order is reversed~\cite{zyczkowski2002}, hence, interchanging the future $\mathcal{T}_+$ with the past $\mathcal{T}_-$. The difference is even more pronounced within the context of volumes for the entanglement of pure states under LOCC operations.

The distribution of Schmidt coefficients induced by the uniform Haar measure in the space of pure bipartite states is significantly different from the flat distribution in the probability simplex $\Delta_d$ \cite{zyczkowski2001induced}. In particular, one observes a repulsion from the centre and, in the case of entangled systems of unequal dimension, from the facets of the simplex. Consequently, this implies a significant difference in the qualitative features of the isovolumetric curves. Fig.~\ref{fig-level-sets-entanglement} shows the isovolumetric curves for $d=3$ with equal-sized systems. Observe that states with large future volumes $\mathcal{V}_+$ are concentrated around the centre of the simplex, which is explained by the repelling property. Inverse effects can be seen for the states with large past volumes $\mathcal{V}_-$, which concentrate at the boundaries of the simplex. The differences become even more pronounced for systems of unequal dimensions, as we will demonstrate in qualitative figures in Section~\ref{app:entanglement_vols}~(see Fig.~\ref{fig-entanglementvolume-2}).

\section{Coherent thermal cones for a two-level system}\label{app:coherent_thermal}

\begin{figure}[t]
\includegraphics[width=10.774cm]{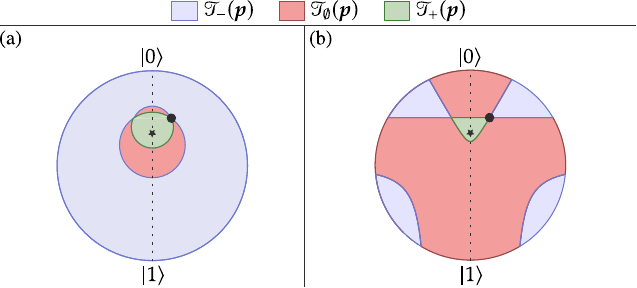}
	\caption{For a two-level system with initial state $\rho$, represented by a black dot $\bullet$, and thermal state represented by a black star $\bigstar$  with Bloch vectors \mbox{$\v r_{\rho} = (0.2,0,0.5)$} and \mbox{$\v r_{\gamma} = (0,0,1/3)$}, respectively, we depict in the real cross-section of the Bloch ball, the coherent thermal cone (a) under Gibbs-preserving operations (b) under thermal operations}
	\label{Fig:coherent_thermal_cone}
\end{figure}

Finally, let us now explain how to use the results from Refs.~\cite{LostaglioKorzekwaCoherencePRX} and~\cite{Korzekwa2017} to obtain the future and past thermal cones beyond energy-incoherent states. Specifically, we will explicitly construct the thermal cones for a qubit under thermal and Gibbs-preserving operations.

\begin{center} 
\emph{\textbf{Coherent thermal cones for thermal operations}}
\end{center}

Consider initial and target states of a two level system, $\rho$ and $\sigma$, with both written in the energy eigenbasis as
\begin{equation}
\rho = \begin{pmatrix}
p & c \\ 
c & 1-p 
\end{pmatrix} \quad , \quad \sigma = \begin{pmatrix}
q & d \\ 
d & 1-q 
\end{pmatrix},
\end{equation}
where $c$ and $d$ are assumed to be real without loss of generality, which amounts to considering a cross-section of the Bloch ball in the $XZ$ plane. Moreover, the thermal ground state occupation of the considered two-level system will be denoted by $\gamma$. It has been shown that for thermal operations, the coherences of the initial and target states have to satisfy the following inequality \cite{LostaglioKorzekwaCoherencePRX},
\begin{equation}\label{eq_condition_coherence}
    d \leq c \frac{\sqrt{[q(1-\gamma)-\gamma(1-p)][p(1-\gamma) - \gamma(1-q)]}}{|p-\gamma|}.
\end{equation}
Thus, we find the boundary of the future thermal cone by saturating Eq.~\eqref{eq_condition_coherence}, and solving it for $q$ we obtain the achievable ground state occupation as a function of target coherence $d$, 
\begin{align}\label{eq_solution_coherence_fut}
q_1(d) &= \frac{(\gamma -p) \sqrt{c^2 (1-2 \gamma )^2+4\gamma d^2 (1-\gamma) }}{2\gamma  c (\gamma -1) } \nonumber \\ & \quad+\frac{(p-\gamma)-2\gamma  p(1-\gamma)}{2\gamma c (\gamma -1)}.
\end{align}
Therefore, the coherent future thermal cone is given by the region delimited by Eq.~\eqref{eq_solution_coherence_fut} from one side and a line segment connecting $(-c,p)$ and $(c,p)$. To characterise the coherent past thermal cone, it will be convenient to introduce a number $d_{\text{cross}} \geq 0$ defined by the relation \mbox{$d_{\text{cross}}^2 + q(d_{\text{cross}})^2 = 1$}.

The coherent past thermal cone is generically composed of two disjoint regions. The first region is contained between a line segment connecting the points $(c,p)$ and $(d_{\text{cross}},p)$, the curve $q_1(d)$ for $d\in[c, d_{\text{cross}}]$ and the boundary of the Bloch ball, together with its reflection with respect to the $Z$-axis.
The second one is obtained in a similar manner by focussing on the past state rather than the target, and thus by solving Eq.~\eqref{eq_condition_coherence} with interchanges $p\leftrightarrow q$ and $c\leftrightarrow d$. This results in
\begin{align}\label{eq_solution_coherence_pas}
q_2(d) &= \frac{2 \gamma c^2+\sqrt{c^2 (p-\gamma )^2 \left[(1-2 \gamma )^2 d^2-4 c^2 (\gamma -1) \gamma \right]}}{2 \left[d^2+(\gamma -1) \gamma  c^2\right]} \nonumber \\ & \:\:\:\:+ \frac{d^2 [p-\gamma -2\gamma  p (1-\gamma)]}{2 \left[d^2-(1-\gamma) \gamma  c^2\right]},
\end{align}
with $d \in [d_{\textrm{min}},d_{\textrm{max}}]$, where $d_{\textrm{min}}$ and $d_{\textrm{max}}$ are real positive solutions of equation $q(d)^2+d^2 = 1$, such that \mbox{$d_{\text{cross}} \leq d_{\text{min}} \leq d_{\text{max}}$}. However, for large values of coherence $c$, we note that this second region may not appear at all. Finally, the incomparable region $\mathcal{T}_{\emptyset}(\rho)$ is obtained by subtracting the past and future cones from the entire Bloch ball.

\begin{center}
\emph{\textbf{Coherent thermal cones for Gibbs-preserving operations}}
\end{center}

Consider a parametrisation of qubit states $\rho$ in the Bloch sphere representation,
\begin{equation}
\label{eq:bloch_state}
\rho=\frac{\iden+\v{r}_\rho\cdot\v{\sigma}}{2},
\end{equation}
where \mbox{$\v{\sigma}=(\sigma_x,\sigma_y,\sigma_z)$} denotes the vector of Pauli matrices. The Bloch vectors of the starting state $\rho$, target state $\rho'$ and the Gibbs state $\gamma$ are given by:
\begin{equation}
\label{eq:bloch_param}
\v{r}_\rho=(x,y,z),\quad\v{r}_{\rho'}=(x',y',z'),\quad\v{r}_\gamma=(0,0,\zeta), 
\end{equation}
where the $z$ coordinate of the Gibbs state can be related to the partition function $Z$ by $\zeta=2Z^{-1}-1\geq 0$.

According to Ref.~\cite{Korzekwa2017}, there exists a GP quantum channel $\E$ such that $\E(\rho)=\rho'$ if and only if \mbox{$R_{\pm}(\rho)\geq R_{\pm}(\rho')$} for both signs, where \mbox{$R_{\pm}(\rho)=\delta(\rho)\pm\zeta z$} and 
\begin{equation}
	\label{eq:delta_lattice}
	\delta(\rho):=\sqrt{(z-\zeta)^2+(x^2+y^2)(1-\zeta^2)}.
	\end{equation}
Consequently, the future thermal cone $\T_+(\rho)$ of any qubit state $\rho$ under GP operations can be directly constructed from the above result. For a generic qubit state $\rho$, we first orient the Bloch sphere so that its $XZ$ plane coincides with the plane containing $\rho$ and a thermal state $\gamma$, i.e., \mbox{$\v{r}_\rho=(x,0,z)$}. Then, define two disks, $D_1(\rho)$ and $D_2(\rho)$ with corresponding circles $C_1(\rho)$ and $C_2(\rho)$, of radii 	
	\begin{equation}
	\label{eq:radii}
	R_1(\rho)=\frac{R_-(\rho)+\zeta^2}{1-\zeta^2},\quad R_2(\rho)=\frac{R_+(\rho)-\zeta^2}{1-\zeta^2},
	\end{equation}
centred at 
	\begin{equation}
	\label{eq:centres}
	\begin{array}{ccc}
		\v{z}_1(\rho)&=&[0,0,\zeta(1+R_1(\rho))],\\ \v{z}_2(\rho)&=&[0,0,\zeta(1-R_2(\rho))].
	\end{array}
	\end{equation}
Therefore, the future thermal cone under GP quantum channels is given by the intersection of two disks of radii $R_1(\rho)$ and $R_2(\rho)$ centred at $\v{z}_1(\rho)$ and $\v{z}_2(\rho)$, $\mathcal{T}_+(\rho) = D_1(\rho)\cap D_2(\rho)$.

The incomparable region is given by mixed conditions, i.e., $\rho'\in\mathcal{T}_\emptyset(\rho)$ if and only if $R_{\pm}(\rho)\geq R_{\pm}(\rho')$ and $R_{\mp}(\rho)< R_{\mp}(\rho')$, or in terms of the disks given beforehand, $\mathcal{T}_\emptyset(\rho) = D_1(\rho)\cap D_2(\rho) \backslash \mathcal{T}_+(\rho)$. Finally, the past cone $\mathcal{T}_-(\rho)$ can be easily given by subtracting the future cone and the incomparable region from the entire Bloch ball.

\section{Derivation of the results}\label{sec_appendix}

In this section, we derive the two most important results of this chapter, namely, the incomparable thermal region for $\beta = 0$ and $\beta \geq 0$. We proceed by first considering the case of $\beta = 0$, covered by Lemma~\ref{Lem:Incomparable_region0}, and develop the notion of the so-called tangent vectors $\v t^{(n)}$ that comprise the boundary of the incomparable region. Then, we generalize the notion of tangent vectors and use it in the proof of Lemma~\ref{Lem:Incomparable_regionT}, which concerns the incomparable region for $\beta > 0$.

\subsection{Infinite-temperature}

\begin{center}
\emph{\textbf{a. Tangent vector}}
\end{center}

The first concept we will need in the proof is that of a \textbf{tangent probability vector}, referred to as the tangent vector for brevity, which will prove to be an essential ingredient.

Let $\Delta_d$ be the set of probability vectors of dimension $d$ with real entries, $\Delta_d= \left\{(p_1, ..., p_d) \in \mathbb{R}^d : \sum_i p_i = 1\right\}$, and let us restrict ourselves to vectors ordered in a non-increasing order, i.e., $p_i \geq p_{i+1}$. To avoid complications, we assume, without loss of generality, that all $p_i \neq p_j$. Following the notation established in Chapter~\ref{C:mathematical_preliminaries}, we will denote probability vectors by bold lowercase letters $\v p \in \Delta_d$ and their corresponding cumulative counterparts by bold uppercase vectors $\v p: P_i = \sum_{j=1}^i p_i$ with $i\in{0,\hdots,d}$. For any vector $\v p$, we introduce a tangent vector $\v{t}(\v p) \equiv \v{t}\in\Delta_d$ by imposing that all its components except the first and the last are equal, $t_i = t_j$ for all $1 < i < j < d$. Additionally, we require that the cumulative vector $\v{T}$ agrees with the vector $\v p$ in at most two consecutive points, i.e., $T_i = P_i$ and $T_j > P_j$ for all $j\in\left\{1,\hdots,d-1\right\}\backslash{i}$, or $T_j = P_j$ for $j \in\left\{i, i+1\right\}$ and $T_j > P_j$ elsewhere. The two imposed conditions follow the intuition of tangency and, by construction, satisfy the majorisation relation $\v{t} \prec \v p$. Furthermore, note that $T_0 = P_0 = 0$ and $T_d = P_d = 1$ and thus are naturally excluded from the considerations.

\begin{figure*}
    \centering
    \includegraphics{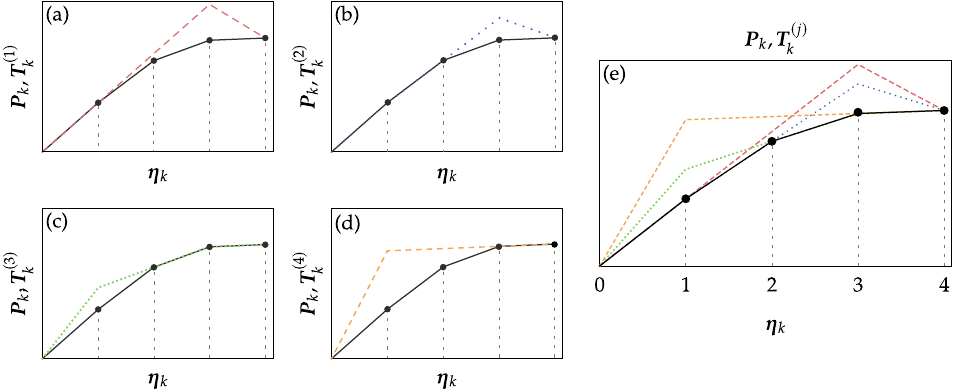}
    \caption{\label{fig-majorisationd4-examples2} \emph{Lorenz curves of the tangent vectors $\v t^{(n)}$}. Majorisation curves of the state $\v p = (0.43, 0.37, 0.18, 0.02)$ (black) and (a) $\v{t}^{(1)}$ (b)  $\v{t}^{(2)}$ (c) $\v{t}^{(3)}$ (d) $\v{t}^{(4)}$ and (e) all tangent vectors $j \in \{1,2,3,4\}$, respectively.} 
\end{figure*}

Indeed, assuming equality between $\v T$ and $\v P$ at exactly two consecutive points restricts the tangent vectors to a set of $d$ unique probability vectors $\v{t}^{(n)}$ defined as 
\begin{equation}
\label{A:tvector}
    \v{t}^{(n)}(\v{p})\equiv \v{t}^{(n)}= \left(t^{(n)}_1, p_n, ..., p_n, t^{(n)}_d\right) , 
\end{equation}
for $1\leq n \leq d$ with the first and last components given by 
\begin{equation}
    t_1^{(n)}  = \sum_i^{n-1} p_i - (n-2) p_n\quad, \quad
    t_d^{(n)}  = 1 - t^{(n)}_1 - (d-2)p_n.
\end{equation}

Observe that the tangent vectors $\v{t}^{(n)}$ that agree with the Lorenz curve of $\v p$ at two successive points can be used to construct all possible tangent vectors $\v t$ that satisfy the condition of agreement at at least a single point, $T_i = P_i$. Indeed, consider a vector $\v t (\lambda, i) \equiv \v t= (1 - \lambda)\v{t}^{(i)} + \lambda \v{t}^{(i+1)}$. Direct calculation shows that it is tangent at just one point, $T_i = \lambda T_i^{(i)} + (1-\lambda)T_i^{(i+1)} = P_i$. Similarly, we arrive at $T_j > P_j$ for $j \neq i$ and $0<\lambda<1$. 

Thus, starting with a discrete set of $d$ tangent vectors $\v{t}^{(n)}$, we recover the entire continuous family of tangent vectors $\v{t}(\lambda, i)$. This argument can be further formalised by considering the Lorenz curves $f_{\v p}(x)$ and $f_{\v t}(x)$. Assuming the left and right derivatives of the former, $\lim_{x \rightarrow \frac{i}{d}-} f'{\v p}(x) = d\cdot p_i$ and $\lim_{x \rightarrow \frac{i}{d}+} f'{\v p}(x) = d\cdot p_{i+1}$, we obtain the extremal slope values for the tangent lines at the $i$-th elbow. Now, considering the second Lorenz curve by construction we have that $f_{\v t}(i/d) = f_{\v p}(i/d)$ and its derivative at this point, $f'{\v t}(i/d) = d\left[(1-\lambda)p_i + \lambda p{i+1}\right]$, span all values between the extremal slope values $d\cdot p_i$ and $d\cdot p_{i+1}$, therefore exhausting the family of possible tangent lines at the $i$-th elbow.

\begin{center}
\emph{\textbf{b. Lattice}}
\end{center}

Lattices provide a setting in which it is natural to represent the precedence or succession of elements within a given set. In particular, they can be used to establish a time-like structure for a given set, as defined by a lattice (see Definition~\ref{def_Lattice}).

It is well known that the partially ordered set $(\Delta_d, \succ)$ of $d$-dimensional probability vectors with real entries in non-increasing order under majorisation forms a lattice~\cite{cicalese2002,Korzekwa2017}. In this setting, the join $\v p \vee \v q$ can be interpreted as the last common past point of $\v p$ and $\v q$, while the meet $\v p \wedge \v q$ can be seen as the first common future point of the pair $\v p, \v q$. The procedure to obtain the join and meet has been illustrated in Ref.~\cite{Korzekwa2017}, and since part of our proof relies on the existence of the join, we will now review the algorithm used to construct it. 

To construct the join of $\v p$ and $\v q$, we start with a probability vector $\v r^{(0)}$ with elements defined by
\begin{equation}
    r_i^{(0)} = \max\left\{\v P_i, \v Q_i\right\} - \max\left\{\v P_{i-1}, \v Q_{i-1}\right\} .
\end{equation}
At this stage, it is possible that the entries of $\v r^{(0)}$ are not ordered in a non-increasing manner. However, we can obtain a properly ordered probability vector $\v r = \v p \vee \v q$, defining the actual join, in no more than $d - 1$ steps. In each step $k\geq 0$, we define $N \geq 2$ as the smallest index where there is an increase between two consecutive components of the probability vector $r^{(k)}$, i.e., $r^{(k)}N > r^{(k)}{N-1}$. Next, we define $M \leq N - 1$ in a way that introducing constant probabilities for the entries with $i\in{M,\hdots,N}$ eliminates the growth. This is done by ensuring that
\begin{equation} \label{eq:joinConstruct_aCoeffs}
    r^{(k)}_M \geq \frac{\sum_{i = M}^N r_i^{(k)}}{N - M + 1} =: a_k.
\end{equation}
Thus, the next iterative step $\v r^{(k+1)}$ is defined by setting its components as
\begin{equation}\label{eq:joinConstruct_flatFragments}
    r_i^{(k+1)} = 
    \begin{cases}
        a_k & \text{for } i\in\{M,\hdots,N\} \\
        r_i^{(k)} & \text{otherwise}
    \end{cases}.
\end{equation}
This construction is repeated until for some $k'$ the probability vector $\v r^{(k')} \equiv \v r$ is ordered non-increasingly; in this way we get the proper join.

\begin{center}
\emph{\textbf{c. Incomparable region and the boundaries}}
\end{center}

As a final piece of information needed to understand the proofs, we introduce the definition of the boundary of the past cone.

\begin{definition}[Boundary of the past thermal cone] \label{def_pastBoundary}
Consider a $d$-dimensional energy incoherent state $\v p\in\mathcal{P}_d$ with $d \geq 3$. We define the boundary of the past thermal cone as the set of probability vectors $\v q \succ \v p$ for which the cumulative vector $\v Q$ is equal to the cumulative vector $\v P$ at least one point. In other words, $\v P_j < \v Q_j$ for some proper subset of the indices $j$ and $\v P_i = \v Q_i$ for all other indices.
\end{definition}

One can define the boundary of the future cone in a similar manner by changing the direction of the inequalities. By considering the common part of the boundary between the future and past cones, we arrive at a simple observation:

\begin{observation}[Common point of future and past cones]
    A point that lies simultaneously at the common part of the boundary between the future and the past must fulfil $\forall_i P_i = Q_i$. Therefore, for $\beta = 0$ we have the equality of $\v p$ and $\v q$ up to a permutation, giving a total of $d!$ common points between the future and the past of any vector $\v p$.
\end{observation}

Equipped with the notion of tangent vectors $\v t^{(n)}$, the join $\v p \vee \v q$ and the boundary of the past cone, we are now prepared to tackle the Lemma~\ref{Lem:Incomparable_region0} concerning the incomparable region, and this is done by proving the following result:
\begin{lemma}
    \label{lem-incomp}
      Consider $\v{p}, \v{q} \in \Delta_d$ and assume that $\v p \nsucc \v q$. Then, $\v{q}$ belongs to the incomparable region of $\v{p}$, $\v q \in \mathcal{T}_\emptyset(\v p)$, if and only if it belongs to the future majorisation cone of some vector $\v t$ tangent to $\v p$, $\v{q}\in\T_+(\v{t})$, with
        $$
            \v{t} \equiv \v{t}(\v{p};\lambda, n)=  \lambda \v{t}^{(n)}(\v{p}) + (1-\lambda) \v{t}^{(n+1)}(\v{p}) ,
        $$
        for some $n\in \{1,d-1\}$ and $\lambda \in[0,1]$.
\end{lemma}

\begin{proof}
    
To prove the ``if" direction, we take a probability vector lying in the interior of the future of the tangent vector, $\v{q} \in \operatorname{int} [\T_{+}(\v{t})]$. Then, to prove that $\v q \in \T_{\emptyset}(\v p)$, one needs to show that $\v{q} \nsucc \v{p}$. By construction, we have $T_k \geq  P_k$, for every $k\neq n$, with equality when $k = n$, so $P_n = T_n > Q_n$ with the equality excluded by putting $\v q$ in the interior of the future.  Consequently, $P_n > Q_n$, and thus $\v{q} \nsucc \v{p}$ and by the initial assumption, $\v p \nsucc \v q$. Therefore, $\v{q} \in \T_{\emptyset}(\v{p})$.

In order to demonstrate the ``only if" direction, let us take an arbitrary $\v q \in \T_{\emptyset}(\v p)$ and recall that there always exists the last common past point for $\v p$ and $\v q$ called join, $\v r = \v p \vee \v q$. From the construction of the join $\v r$, it is found that the entries of the cumulative distribution $\v R$ will be divided into three subsets, namely points common with $\v P$, common with $\v Q$ and the ones lying above either, that is $I_p := \left\{i:0<i<d,\,R_i = P_i\right\}$, $I_q := \left\{i:0<i<d,\,R_i = Q_i\right\}$ and $J := \left\{0<j<d,\,:R_j > \max(P_j,Q_j)\right\}$, respectively. In particular, looking at equation \eqref{eq:joinConstruct_flatFragments} one can see that it is either the case that $M-1\in I_p$ and $N+1\in I_q$, or $M-1\in I_p$ and $N+1\in I_q$ or by invoking geometric intuition, endpoints of the flat fragments of $\v R$ will join $\v P$ and $\v Q$. Thus, at each step, the sets $I_p$ and $I_q$ will be non-empty. 
By this argument, we may choose any index $i\in I_p$ and construct a tangent vector $\v t'$ for the join $\v{r}$ at the $i$-th elbow.
\begin{equation}
    \v t' \equiv \v t(\v{r};\mu,i) = \mu \v t^{(i)}(\v r) + (1 - \mu) \v t^{(i+1)}(\v r)
\end{equation}
for any $\mu\in[0,1]$. This vector obeys, by construction, the majorisation relation $\v t' \succ \v r \succ \v q$. Furthermore, due to the choice $i\in I_p$ it is also a tangent vector for the $\v p$,
\begin{equation}
    \v t' = \v t \equiv \v t(\v{p};\lambda(\mu),i) =  \lambda(\mu) \v t^{(i)}(\v p) + (1 - \lambda(\mu)) \v t^{(i+1)}(\v p)
\end{equation}
for $\lambda(\mu) \in [0,1]$. Therefore, by the properties of the tangent vector $\v t$ it follows that $\v q \in \T_{+}(\v t)$.  
\end{proof}

The last step necessary in order to demonstrate the Lemma~\ref{Lem:Incomparable_region0} is to notice that the tangent vectors $\v{t}(\v{p};\lambda(\mu),i)$ are convex combinations of $\v t^{(i)}(\v p)$ and $\v t^{(i+1)}(\v p)$  and their future majorisation cones are convex, therefore, union of their future cones corresponds to the convex hull of the future cones of the extreme points
\begin{equation}
     \bigcup_{\lambda\in[0,1]}\mathcal{T}_+(\v{t}(\v{p};\lambda,i)) = \text{conv}\left[\mathcal{T}_+\left(\v t^{(i)}\right)\cup\mathcal{T}_+\left(\v t^{(i+1)}\right)\right]
\end{equation} 
which completes the statement of Lemma~\ref{Lem:Incomparable_region0}.\qed

\subsection{Finite temperatures}
\label{app:finite_temp_derivs}

The proof of Lemma~\ref{Lem:Incomparable_regionT} for $\beta > 0$ is developed in the simplest way by considering the embedding $\mathfrak{M}$ introduced in Section~\ref{Subsec:embedding-lattice}, which takes $d$-dimensional probability distributions $\v p\in\Delta_d$ to its higher-dimensional image $\mathfrak{M}(\v p) \in \Delta^\mathfrak{M}_{d}$
which allows us to closely follow the steps of the proof for $\beta = 0$.

After shifting our focus to the embedded space, we construct the corresponding tangent vectors $\v{t}^{\mathfrak{M}}(\v{p};\lambda,n) = \lambda \v t^{\mathfrak{M}(n)}(\v p) + (1 - \lambda) \v t^{\mathfrak{M}(n+1)}(\v p)$ that respect the rules for constructing the majorisation curve within the embedded space. In particular, in full analogy to the $\beta = 0$ case, the full family can be given in terms of vectors tangent to the $n$-th linear fragment of the embedded majorisation curve,

\begin{align}
    \v{t}^{\mathfrak{M}(n)}(\v{p}) = \left(t_1^{\mathfrak{M}(n)}, \frac{p_n}{\gamma_n}\gamma^{\mathfrak{M}}_2, \hdots, \frac{p_n}{\gamma_n}\gamma^{\mathfrak{M}}_{2^d-1},t_{2^d}^{\mathfrak{M}(n)}\right) ,
\end{align}
where the first entry is defined in such a way that the majorisation curves, defined as the piecewise-linear functions given by their elbows $\left\{\left(\Gamma^{\mathfrak{M}}_i, P^{\mathfrak{M}}_i\right)\right\}_{i=0}^{2d-1}$, agree in at least one point, $f^{\mathfrak{M}}_{\v{t}^{(n)}}(\Gamma^{\mathfrak{M}}_n) = f_{\v{p}}^{\mathfrak{M}}(\Gamma^{\mathfrak{M}}_n)$ and the last one guarantees that $\sum_i t_i^{\mathfrak{M}(n)} = 1$. Observe that vectors $\v{t}^{\mathfrak{M}(n)}$ constructed in this way are tangent with respect to the embedded Lorenz curve, i.e., $t^{\mathfrak{M}(n)}_i/\gamma^{\mathfrak{M}}_i = t^{\mathfrak{M}(n)}_j/\gamma^{\mathfrak{M}}_j$, thus taking into account the varying intervals on the horizontal axis. Equipped with these, we pose a technical lemma similar to Lemma~\ref{lem-incomp},

As a preliminary step, we give the algorithm for the construction of the join in the embedding space by modifying the crucial steps \eqref{eq:joinConstruct_aCoeffs} and \eqref{eq:joinConstruct_flatFragments} to take into account the varying widths and redefine the point $N$ of increase by requiring $r_N^{\mathfrak{M}(k)}/\gamma_N^{\mathfrak{M}} > r_{N-1}^{\mathfrak{M}(k)}/\gamma_{N-1}^{\mathfrak{M}}$. Consequently, we redefine $M$ by a condition similar to \eqref{eq:joinConstruct_aCoeffs} that incorporates the scaling,
\begin{align} \label{finTemp_join_akSlope}
    \frac{r^{\mathfrak{M}(k)}_M}{\gamma^{\mathfrak{M}}_M} \geq \frac{\sum_{i = M}^N r_i^{\mathfrak{M}(k)}}{\sum_{i = M}^N \gamma_i^{\mathfrak{M}}} =: a^{\mathfrak{M}}_k.
\end{align}
Finally, we define the join candidate in the $k$-th step in full analogy to \eqref{eq:joinConstruct_flatFragments} as

\begin{equation} \label{finTemp_join_kStep}
    r_i^{\mathfrak{M}(k+1)} = 
    \begin{cases}
        a_k^{\mathfrak{M}}\gamma_i^{\mathfrak{M}} & \text{for } i\in\{M,\hdots,N\} \\
        r_i^{\mathfrak{M}(k)} & \text{otherwise}
    \end{cases}.
\end{equation}
The algorithm defined in this way follows precisely the same logic as the one given in Ref. \cite{Korzekwa2017} and thus it always terminates, in this in no more than $ d^\mathfrak{M}(\v p, \v q) - 1 = \overline{\left\{\sum_{i=1}^j \gamma_{\v \pi^{-1}_{\v{p}}(i)}\right\}_{j=1}^{d}\cup\left\{\sum_{i=1}^j \gamma_{\v \pi^{-1}_{\v{q}}(i)}\right\}_{j=1}^{d}} - 1$ steps.

As a final remark, one has to note that the family of tangent vectors $\v{t}^{\mathfrak{M}(n)}(\v{r}^\mathfrak{M})$ should be indexed by $n\in\{1,\hdots,d'\}$, where the number $d'$ of constant-slope fragments of the join, even though bounded, $d^\mathfrak{M}(\v p, \v q) \geq d' \geq d$, is \textit{a priori} not well defined due to many possible ways of disagreement between the $\beta$-orders of $\v{p}$ and $\v{q}$. With these tools, we are ready to present the technical lemma needed for constructing the incomparable region for $\beta > 0$.

\begin{lemma}
\label{lem-incomp_nonZeroBeta}
$\mathfrak{M}(\v{q})\equiv\v{q}^\mathfrak{M}$ belongs to the incomparable region of $\mathfrak{M}(\v{p})\equiv \v{p}^\mathfrak{M}$, $\v{q}^\mathfrak{M} \nsucc_\mathfrak{M} \v{p}^\mathfrak{M}$, if and only if it belongs to the future thermal cone of some vector $\v{t}^\mathfrak{M}$ tangent to $\v{p}^\mathfrak{M}$, $\v{q}^\mathfrak{M}\succ\v{t}^\mathfrak{M}$, with
    $$
        \v{t}^\mathfrak{M} \equiv \v{t}^\mathfrak{M}(\v{p};\lambda, n)=  \lambda \v{t}^{\mathfrak{M}(n)}(\v{p}) + (1-\lambda) \v{t}^{\mathfrak{M}(n+1)}(\v{p}) ,
    $$
    for some $n\in \{1,d-1\}$ and $\lambda \in[0,1]$.
\end{lemma}
\begin{proof}
    The proof follows in complete analogy with the standard majorisation case as presented in the proof for Lemma \ref{lem-incomp} by replacing the standard majorisation $\succ$ in every statement with the majorisation variant $\succ_\mathfrak{M}$ given for the embedding space and employing the adjusted join construction, summarised in equations \eqref{finTemp_join_akSlope} and \eqref{finTemp_join_kStep}. 
\end{proof}
In order to go back from the embedded space $\Delta_{d}^{\mathfrak{M}}$ to the formulation of Lemma \ref{Lem:Incomparable_regionT} in the original space $\Delta_d$ we combine two observations following from embedding and projection operations. First, note that majorisation in embedded space implies majorisation between projections, thus $\v{t}^\mathfrak{M} \succ_\mathfrak{M} \v{q} \Rightarrow \mathfrak{P}_{\v \pi}(\v{t}^\mathfrak{M}) \succ_\beta \mathfrak{P}_{\v \pi}(\v{q})$ for every order $\v \pi$. Second, observe that embedding preserves the majorisation relations between the vectors, therefore $\v{q}\in\mathcal{T}_\emptyset(\v p) \Leftrightarrow \mathfrak{M}(\v q) \in \mathcal{T}_\emptyset(\mathfrak{M}(\v p))$. These two statements show that it is enough to consider vectors $\v{t}^{(n,\v \pi)} = \mathfrak{P}_{\v \pi}(\v{t}^{\mathfrak{M}(n)})$ and convex combinations of their future thermal cones, thus proving Lemma~\ref{Lem:Incomparable_regionT}. \qed

\subsection{Construction of probabilistic majorisation cones}\label{app:probal_deriv}

In this section, we derive the results presented in Section~\ref{Subsec:probabilistic-cones}. Specifically, we explore how to extend the notion of majorisation cones to probabilistic ones. For the convenience of the reader, we restate the theorem concerning probabilistic transformations:

\vidaltheorem*

To establish a direct connection to majorization, we reformulate the above theorem as follows:
\begin{equation}
    \forall_{1\leq k\leq d}:\mathcal{P}(\v p, \v q) \leq \frac{\sum_{j = k}^d p^{\downarrow}_j}{\sum_{j = k}^d q^{\downarrow}_j} = \frac{1 - \sum_{j = 1}^{k-1} p^\downarrow_j}{1 - \sum_{j = 1}^{k-1} p^\downarrow_j} = \frac{1 - P_k}{1 - Q_k}.
\end{equation}
By setting $\mathcal{P}(\v p , \v q) = 1$, we recover the standard majorisation condition on deterministic convertibility,
\begin{equation}
    \forall_{1\leq k\leq d}: 1 \leq \frac{\sum_{j = k}^d p^{\downarrow}_j}{\sum_{j = k}^d q^{\downarrow}_j} \Leftrightarrow \v p \prec \v q.
\end{equation}
To determine the probabilistic past cone $\mathcal{T}_-(\v p, \mathcal{P})$ at probability $\mathcal{P}$, we consider 
\begin{equation}
    \begin{aligned}
       \forall_{1\leq k\leq d}: \mathcal{P} \leq \frac{1 - Q_k}{1 - P_k} 
        & \Rightarrow \mathcal{P} - \mathcal{P} P_k \leq 1 - Q_k  \Rightarrow Q_k \leq \mathcal{P} P_k + (1 - \mathcal{P}) \\
        & \Rightarrow \v{q}\prec\tilde{\v{p}},
    \end{aligned}
\end{equation}
with an auxiliary distribution
\begin{equation}
	\tilde{p}_i = \begin{cases}
		\mathcal{P} p^{\downarrow}_1 + (1 - \mathcal{P}) & \text{for } i = 1, \\
		\mathcal{P} p^{\downarrow}_i & \text{otherwise},
	\end{cases}
\end{equation}
which is always a proper probability distribution ordered non-increasingly, $\tilde{\v p}  = \tilde{\v p}^\downarrow$, therefore providing a proper Lorenz curve. 

Following a similar procedure for the future cone $\mathcal{T}_+(\v p, \mathcal{P})$ leads to
\begin{equation}
    \begin{aligned}
        \mathcal{P} \leq \frac{1 - P_k}{1 - Q_k} 
        & \Rightarrow \mathcal{P}^{-1} - \mathcal{P}^{-1} P_k \geq 1 - Q_k  \Rightarrow Q_k \geq \mathcal{P}^{-1} P_k + (1 - \mathcal{P}^{-1}) \\
        & \Rightarrow \v{q}\succ\hat{\v{p}},
    \end{aligned}
\end{equation}
with the second auxiliary distribution
\begin{equation}
	\hat{p}_i = \begin{cases}
		\mathcal{P}^{-1} p^{\downarrow}_1 + (1 - \mathcal{P}^{-1}) & \text{for } i = 1, \\
		\mathcal{P}^{-1} p^{\downarrow}_i & \text{otherwise}.
	\end{cases}
\end{equation}
In contrast to the case of the past cone, the distribution $\hat{\v p}$ in this formulation is not ordered in a non-increasing manner beyond a certain value of $\mathcal{P}$. At first glance, one might think that reordering should solve the problem; however, it would be equivalent to a decrease in the probabilistic future with decreasing $\mathcal{P}$, which creates a contradiction. The solution is provided by noting that Vidal's criterion deals with rescaled entries of the Lorenz curve rather than the probabilities themselves. Therefore, the Lorenz curve for $\hat{\v p}$ should remain convex for all values of $\mathcal{P}$ without the need for reordering.

Consider the following critical values of $\mathcal{P}$, namely,
\begin{equation}
    \mathcal{P}_n = (n-1)p^{\downarrow}_n - \sum_{i=1}^{n-1} p^{\downarrow}_i + 1,
\end{equation}
for which the first $n$ entries of the distribution $\hat{\v{p}}$ will not be ordered non-increasingly, resulting in an improper Lorenz curve. The resulting non-convexity is controlled by replacing
\begin{equation}
    \left\{\hat{p}_1,\hdots,\hat{p}_n\right\} \rightarrow \frac{1}{n}\sum_{i=1}^n\hat{p}_i \left\{1,\hdots1\right\},
\end{equation}
which ensures that $\hat{\v p}  = \hat{\v p}^\downarrow$.

This way, the auxiliary ordered distributions $\tilde{\v{p}}$ and $\hat{\v{p}}$, together with the construction for the deterministic majorisation cones provide the full construction of the probabilistic cones. However, it is important to note that when considering entanglement and coherence theories, the roles of the future and past majorisation cones are reversed, and this should be taken into account.

\subsection{Volumes of entanglement majorisation cones} \label{app:entanglement_vols}

Finally, here we present the methods used to obtain the isovolumetric sets for entanglement cones, as discussed in Section~\ref{Subsec:volume-entanglement-cones}. Additionally, we plot these curves for systems of unequal dimensions to compare with the case of equal dimension.

Consider a uniform Haar distribution of pure states in a composed space \mbox{$\ket{\psi}\in\mathcal{H}^{N}\otimes\mathcal{H}^M$} with $N \leq M$. The partial tracing induces a measure in the space of reduced states, $\rho = \operatorname{Tr}_2 \ket{\psi}\bra{\psi}$, characterised by the distribution of eigenvalues \mbox{$\Lambda = \left\{\lambda_1,\hdots,\lambda_N\right\}$} of the reduced state~\cite{zyczkowski2001induced}:
    \begin{equation}\label{eq_haardistribution}
        P_{N,M}(\Lambda) = C_{N,M} \delta\left(1 - \sum_i \lambda_i\right) \prod_i \lambda_i^{M-N}\theta(\lambda_i) \prod_{i<j} (\lambda_i - \lambda_j)^2,
    \end{equation}
where $\delta$ and $\theta$ are the Dirac delta and Heavyside step functions, respectively. The normalisation constant is given by
\begin{equation}
      C_{N,M} = \frac{\Gamma(NM)}{\prod_{j=0}^{N-1}\Gamma(M-j) \Gamma(N-j+1)},
\end{equation}
where $\Gamma$ is the Gamma function. Before we continue with our discussion, let us first understand the role played by all factors in Eq.~\eqref{eq_haardistribution}. As previously mentioned, $C_{N,M}$ is the normalisation. The delta function ensures that the spectrum sums to one (is normalised), whereas the step function will guarantee that it is positive. Now, notice that the first product (capital pi notation) essentially does not influence the behaviour of the distribution for \mbox{$N=M$} and introduces the repelling of the faces of the probability simplex otherwise, as it goes to zero whenever \mbox{$\lambda_i = 0$} for any~$i$. The second product is responsible for the repelling from distributions with any two entries equal, since it goes to zero whenever \mbox{$\lambda_i = \lambda_j$} for any~$i\neq j$.

\begin{figure*}
    \centering
    \includegraphics{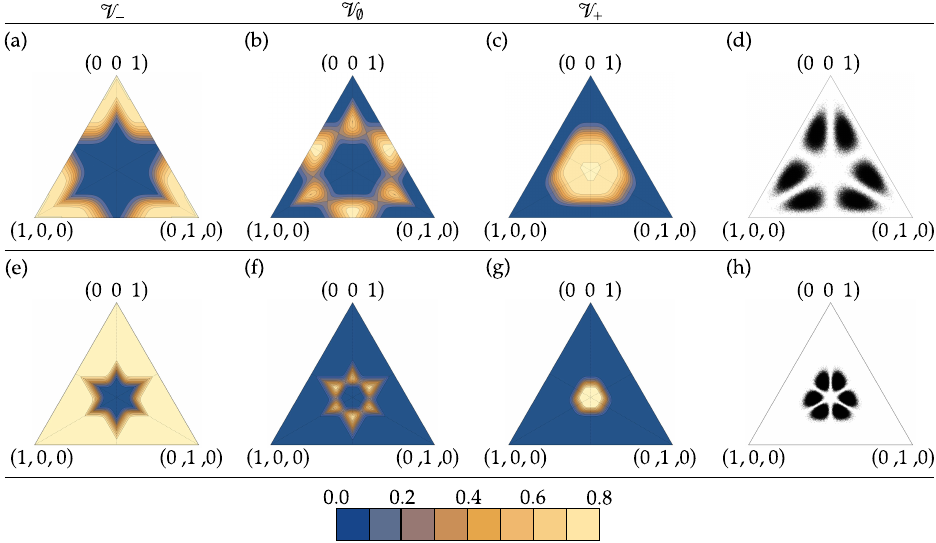}
    \caption{\label{fig-entanglementvolume-2} \emph{Isovolumetric sets for entanglement $3 \times M$ bipartite systems}. The density $P_{3,M}(\Lambda)$ of Schmidt coefficients of pure states for qutrit-quMit systems depends heavily on the dimension $M$ of the second system. Panels (a-d) and (e-h) present the isovolumetric lines for past, incomparable and future regions and the density of the states for $M = 6$ and $M = 30$, respectively. Note that for larger $M$ the density $P_{3,M}$ is more and more concentrated around the regions close to the centre [compare (d and (h)]. This affects the subset of states with large future volume, making it smaller [(c) and (g)] as well as the set of states with large past volume, enlarging it [(a) with (e)].}
\end{figure*} 

Sampling from the $P_{N,M}$ distribution, ordinarily done by generating state vectors $\ket{\psi}\in\mathcal{H}_{NM}$ which would be computationally prohibitive for large $M$, can be achieved using only $O(N)$ random numbers for any dimension of the secondary system. It has been demonstrated in Ref.~\cite{Cunden_2020} that the distribution $P_{N,M}$ is precisely the Laguerre unitary ensemble generated by Wishart matrices of size $N$ and parameter $M$~ \cite{RandomMatrixBook} and, in turn, generated using a tridiagonal method containing only $O(N)$ random real numbers~ \cite{DE02}, which indeed allows one to study the $P_{N,M}$ distributions for arbitrary high-dimensional ancillary systems. 

The procedure for generating the isovolumetric lines for the majorisation cones with a given distribution, $P_{N,M}$, proceeds as follows:
    \begin{enumerate}
        \item Generate a sample of $n$ sets of eigenvalues $\{\Lambda_1,\hdots,\Lambda_n\}$ taken from the distribution $P_{N,M}$ using the tridiagonal method.
        \item Consider regularly spaced grid of points $S$ in a single chamber of the full probability simplex $\Delta_N$ (e.g., \mbox{$p_1 \geq p_2, ..., p_d)$} in order to avoid repeated counting (achieving $N!$ decrease in operations)
        \item For each $\v{p} \in S$ consider its majorisation cones $\mathcal{T}_i(\v p)$ and divide $S$ into $S_i \equiv \left\{\Lambda_i \in \mathcal{T}_i(\v{p})\right\}$.
        \item This way we arrive at the approximations of the volumes of the three regions, $$\mathcal{V}_i \approx \frac{|S_i|}{n},$$ where $|X|$ denotes the number of elements in a set $X$.
    \end{enumerate}

    We applied this method for $N=3$ with $M = 3$, displayed in the main text of Fig.~\ref{fig-level-sets-entanglement}, and additionally with $M = 6$ and $30$, as shown in Fig.~\ref{fig-entanglementvolume-2}. These two cases show the significant dependence of the isovolumetric lines on the size of the environment.

\section{Concluding remarks}~\label{Sec:concluding_remarks}

In this chapter, we investigated the structure of the thermodynamic arrow of time by analysing thermal cones and their behavior under different conditions. By dividing the space of energy-incoherent states into the past, the future, and the incomparable region, in analogy with the future, past, and spacelike regions of Minkowski spacetime, we identified thermal cones as regions of the probability simplex that encode the achievability of state transformations under thermal operations. We specifically focused on energy-incoherent states in the presence of a thermal bath at a finite temperature and in the limit of temperature going to infinity, fully characterising and carefully analysing the incomparable and past thermal cones in both regimes. Additionally, we identified the volumes of the thermal cones as thermodynamic monotones and performed a detailed analysis of their behavior. Our results can be applied directly to the study of entanglement, as the order defined on the set of bipartite pure entangled states by local operations and classical communication is the opposite of the thermodynamic order in the limit of infinite temperature. In this context, the future thermal cone becomes the past for entanglement, and the past becomes the future. Furthermore, a similar extension can be drawn to coherence resource theory.

There are several potential research directions for generalising and extending the results presented in this chapter. One possible avenue is to expand the analysis beyond energy-incoherent states to encompass the full space of quantum states. While available tools for this purpose are comparatively scarce~\cite{Korzekwa2017, LostaglioKorzekwaCoherencePRX, Gour_2018}, various techniques can be employed to construct coherent thermal cones under both Gibbs-preserving and thermal operations for qubit systems, as shown in Fig.~\ref{Fig:coherent_thermal_cone} and explained in Section~\ref{app:coherent_thermal}. Additionally, it would be interesting to explore an equivalent construction of the past and incomparable regions for continuous thermomajorisation. Extending our analysis to this setting could provide further insight into the structure of the thermodynamic arrow of time for memoreless processes. Futhermore, an extension that includes many non-interacting subsystems (possibly independent and identically distributed) could be done by defining an appropriate function to analyse the behavior of the thermal cones. While our investigation focused on single subsystems, such an extension could be useful in exploring the behavior of thermodynamic quantities in a broader range of physical systems.

Recently, the framework outlined in this section was applied in Ref.~\cite{de2024ent} for exploring the thermodynamic constraints on the pivotal task of generating entanglement using non-equilibrium states. In that work, the authors present a detailed construction of the future thermal cone of entanglement -- the set of entangled states that a given separable state can thermodynamically evolve to. In a similar spirit to the discussion presented here, the properties of the future thermal cone of entanglement were studied and related to the ability to generate entanglement. 

Finally, the fundamental limits and advantages of using a catalyst to aid thermodynamic transformations within the framework of thermal cones were addressed in Ref.~\cite{catalysiskuba2024}. More precisely, the set of states to which a given initial state can thermodynamically evolve (the catalyzable future) or from which it can evolve (the catalyzable past) with the help of a strict catalyst was characterized. Secondly, lower bounds on the dimensionality required for the existence of catalysts under thermal processes, along with bounds on the catalyst's state preparation, were found.
\chapter{Memory-assisted Markovian thermal processes}\label{C:memory-MTP}

Information has become ubiquitous in thermodynamics. It all started with Maxwell's seminal inquiry~\cite{maxwell1872theory}: \emph{what would happen if we had knowledge of a system's state?} The ramifications of this hypothesis led to potential violations of the second law of thermodynamics and a century-long puzzle~\cite{Szilard1929, brillouin1951maxwell}. Ultimately, it was found that thermodynamics imposes physical restrictions on information processing~\cite{Landauer1961, Bennett1982}, resulting in the development of frameworks devoted to incorporating information into thermodynamics~\cite{maruyama2009colloquium,seifert2012stochastic,sagawa2012thermodynamics,Goold2016,binder2018thermodynamics, Deffner2019}. A crucial concept at the intersection between these two fields is \emph{memory}, a thermodynamic resource for storing, processing, and erasing information. In particular, memory effects can bring numerous advantages, including enhanced cooling~\cite{taranto2020exponential}, generation of entanglement~\cite{mirkin2019entangling,mirkin2019information} or improved performance of heat engines and refrigerators~\cite{PhysRevA.99.052106,PhysRevA.102.012217,PhysRevE.106.014114}. However, realistic quantum mechanical systems are open and governed by non-unitary time evolution, which encompasses the irreversible phenomena such as energy dissipation, relaxation to thermal equilibrium or stationary non-equilibrium states, and the decay of correlations~\cite{breuer2002theory,rivas2012open}. Hence, assumptions like weak coupling, large bath size, and fast decaying correlations are commonly made in modelling such systems, thus neglecting memory effects. This raises the question of how memoryless processes get modified when system-bath memory effects become non-negligible, i.e., how to assess and quantify the role of memory in thermodynamic processes~\cite{rivas2014quantum}. 

As we have seen in Chapter~\ref{C:resource_theory_of_thermodynamics}, the resource theory of thermodynamics is a relatively recent framework allowing one to address foundational questions in thermodynamics. By relying on the notion of thermal operations it offers a complete set of laws for characterising general state transformations under thermodynamic constraints. The downsides of this formalism are twofold. Firstly, it focuses only on snapshots of the evolution, making it hard to discuss how the processes are realised in time. Secondly, it may require precise control over the system and the bath. The first problem was addressed by developing a hybrid framework that reconciles resource theory and master equation approaches~\cite{lostaglio2021continuous,Korzekwa2022}, where the concept of a Markovian thermal process was introduced. This new set of operations refines the notion of thermal operations by encoding all relevant constraints of a Markovian evolution. The second problem was partially addressed in Ref.~\cite{Lostaglio2018elementarythermal} by introducing the concept of elementary thermal operations, i.e., a subset of transformations that can be decomposed into a series of thermal operations, each acting only on two energy levels of the system. Such decompositions offer a method to bypass the need for a complete control over interactions between the system and the environment, and the approach was recently generalised to also include catalytic transformations~\cite{Jeongrak2022}. While elementary operations require only a limited control, they still rely on non-Markovian effects, and so the question of quantifying memory effects in the resource theory of thermodynamics remains open.

In this chapter, we make a step towards bridging the gap between thermal operations and Markovian thermal processes for energy-incoherent states by introducing and investigating \emph{memory-assisted Markovian thermal processes} (MeMTPs). These are defined by extending the Markovian thermal processes framework with ancillary memory systems, allowing one to interpolate between memoryless dynamics and the one with full control. More specifically, we demonstrate that energy-incoherent states achievable from a given initial state via thermal operations can be approached arbitrarily well by repeatedly interacting the main system with a memory that is initialised in a thermal equilibrium state (and therefore is thermodynamically resourceless) via an algorithmic procedure composed of Markovian thermal processes. Physically, this can be seen as a partial control over the bath degrees of freedom, where the bath can be thought of as a large, discrete, collection of smaller thermal units, and one can control the interactions of the main system with a small number of thermal subsystems~(see Fig.~\ref{Fig-schematic_representation.pdf}). 

\begin{figure*}
    \centering
    \includegraphics{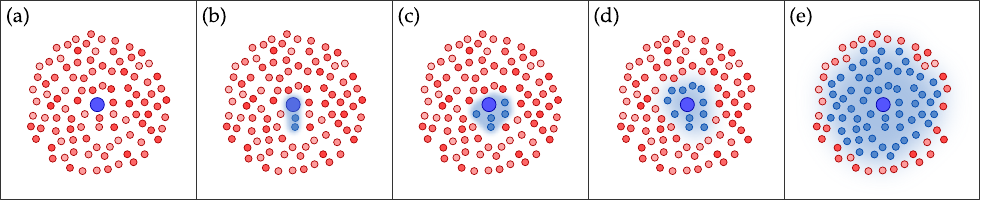}
    \caption{\emph{Memory-assisted Markovian thermal processes}. Schematic representation of the general setting. (a) Initially, the main system (large blue circle) is coupled to a heat bath at inverse temperature $\beta$ (small red circles) and their interaction is Markovian, so that the bath is in thermal equilibrium at each moment in time. (b)-(e) Then, the control is extended to parts of the environment (small blue circles with blue background) that do not instantaneously thermalise to equilibrium after interactions with the system, and can thus lead to non-Markovian dynamics of the main system.}
    \label{Fig-schematic_representation.pdf}
\end{figure*}

Following this idea, we introduce a family of memory-extended Markovian thermodynamic protocols that require minimal control and are valid for any temperature regime. In the infinite temperature limit, we prove that our protocol can arbitrarily well simulate any state transition that can be achieved via thermal operations. More precisely, our first main result states that when memory grows and the number of interactions goes to infinity, the full set of states achievable by thermal operations can be reached by MeMTPs. We also provide analytic expressions for the convergence rates, which scale either polynomially with the number of degrees of freedom or exponentially with the number of memory subsystems. Moreover, based on strong numerical evidence, we propose a conjecture for a more accurate approximation of arbitrary state transformation (i.e., converging faster with the growing size of the memory) through sequences of truncated versions of our protocol. Our second main result extends these considerations to the finite-temperature regime. Here, we first show analytic convergence of our MeMTP protocol to a particular subset of state transformations that can be achieved by thermal operations. These include all the so-called $\beta$-swaps, as well as $\beta$-cycles, which form a thermodynamic equivalent of cyclic permutations. Then, based on numerical simulations, we conjecture that actually an arbitrary state reachable via thermal operations can be obtained using a proper sequence of truncated protocols.

With these results at hand, we then proceed to discussing their applicability. First, we explain how to assess the role played by memory in the performance of thermodynamic protocols by investigating work extraction in the intermediate regime of limited memory. We thus interpolate between the two extremes of no memory and complete control, and quantify environmental memory effects with the amount of extractable work from a given non-equilibrium state. Second, we consider the task of cooling a two-level system using a two-dimensional memory characterised by a non-trivial Hamiltonian. This example represents a minimal model requiring the manipulation of two two-level systems. Various experimental proposals are available across distinct platforms suitable for realising this specific model. Such platforms encompass quantum dots~\cite{PhysRevLett.110.256801,PhysRevB.98.045433}, superconducting circuits~\cite{PhysRevB.94.235420,chen2012quantum}, and atom-cavity systems~\cite{mitchison2016realising}. We then clarify that all transitions achievable via thermal operations can be performed using a subset of thermal operations that only affect two energy levels at the same time. This may seem to contradict the results of Refs.~\cite{Lostaglio2018elementarythermal,mazurek2018decomposability}, where it was proven that thermal operations constrained to just two energy levels of the system are not able to generate all thermodynamically allowed transitions. We resolve this apparent contradiction by noting that in our case we require the control over two levels of the joint system-memory state, and not just the system state. Finally, we discuss the behaviour of free energy during a non-Markovian evolution, explaining the role of the memory as a free energy storage.  

The chapter is structured as follows. First, in Sec.~\ref{sec:framework} we compare the frameworks of thermal operations and Markovian thermal processes, and then introduce the central notion of this chapter, the memory-assisted Markovian thermal processes. Next, in Sec.~\ref{sec:bridging}, we describe the protocol that employs thermal memory states to approximate non-Markovian thermodynamic state transitions with Markovian thermal processes. We then explain how this approximation convergences to the full set of transitions achievable via thermal operations as the size of the memory grows. Section~\ref{sec:discussion} contains discussion and application of our results. Finally, in Sec.~\ref{sec:outlook}, we conclude and provide outlook for future research.

\section{Thermal operations vs Markovian thermal processes}\label{sec:framework}

\emph{Thermal operations} (TOs) framework as discussed in Chapter~\ref{C:resource_theory_of_thermodynamics} uses minimal assumptions on the joint system-bath dynamics by only assuming that the joint system is closed and thus evolves unitarily, and that this unitary evolution is energy-preserving. Since there are no further constraints on $U$, arbitrarily strong correlations can build up between the system and the bath, and one can expect non-Markovian memory effects to come into play. At the same time, from the perspective of control theory, generating an arbitrary TO may require very complex and fine-tuned control over system-bath interactions~\cite{Lostaglio2018elementarythermal}.

\emph{Markovian thermal processes} (MTPs) framework~\cite{lostaglio2021continuous,Korzekwa2022}, on the other hand, uses typical assumptions of the theory of open quantum systems (weak coupling, large bath size, quickly decaying correlations, etc.)~\cite{breuer2002theory}, to argue that the system undergoes an open dynamics described by a Lindblad master equation~\cite{kossakowski1972quantum,gorini1976completely,lindblad1976generators}. Since the dynamics generated by an MTP arises explicitly from a Markovian model, there are no memory effects. Also, as shown in Ref.~\cite{lostaglio2021continuous}, the universal set of controls that allows one to generate any incoherent state transformation achievable via MTPs consists only of two-level partial thermalisations. These transform the populations of two energy levels, $i$ and $j$, in the following way 
\begin{subequations}
\begin{align}
\label{eq:partial1}
    p_i &\rightarrow (1-\lambda) p_i + \lambda \frac{p_i+p_j}{\gamma_i+\gamma_j} \gamma_i,\\
    \label{eq:partial2}
    p_j &\rightarrow (1-\lambda) p_j + \lambda \frac{p_i+p_j}{\gamma_i+\gamma_j} \gamma_j ,
\end{align}
\end{subequations}
where $\lambda\in[0,1]$.

\begin{figure}[t]
    \centering
    \includegraphics[width=10.7cm]{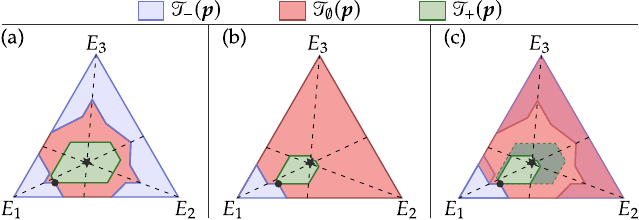}
    \caption{\emph{Thermal operations vs Markovian thermal processes}.
    Sets of states that a three-level system with an equidistant energy spectrum $\v{E} = (0,1,2)$ and prepared in an energy-incoherent state $\v p = (0.7,0.2,0.1)$ (depicted by a black dot $\bullet$) can be transformed to [green region $\T_+(\v p)$] or transformed from [blue region $C_-(\v p)$] by (a) thermal operations and (b) Markovian thermal processes with respect to inverse temperature $\beta = 0.3$. In (c) we show the overlap of sets of achievable states via TOs (dark green) and MTPs (light green). The thermal state of the system is depicted by a black star at the intersection of the dashed lines.}
    \label{Fig:thermalconevsmarkovianthermalcone}
\end{figure}
The set of states $\mathcal{T}^{\textrm{TO}}_{+}(\v p)$ achievable via thermal operations from a given incoherent initial state $\v{p}$ can be fully characterised using the notion of thermomajorisation. The so-called future thermal cone $\mathcal{T}^{\textrm{TO}}_{+}(\v p)$ (see Chapter~\ref{C:thermal_cones}) is a convex set that consists of at most $d!$ extreme points, the construction of which was given in Corollary~\ref{Thm:Future_majoristion_cone} (see Fig.~\hyperref[Fig:thermalconevsmarkovianthermalcone]{\ref{Fig:thermalconevsmarkovianthermalcone}a} for an example with a three-level system). On the other hand, the set of states $\T^{\textrm{MTP}}_{+}(\v p)$ achievable via Markovian thermal processes from a state $\v{p}$ was recently characterised using the notion of continuous thermomajorisation. The future Markovian thermal cone $\mathcal{T}^{\textrm{MTP}}_{+}(\v p)$ is not convex (as illustrated in Fig.~\hyperref[Fig:thermalconevsmarkovianthermalcone]{\ref{Fig:thermalconevsmarkovianthermalcone}b} for a three-level system), but Theorem~\ref{thm:universality} provides a construction of its extreme points using sequences of two-level full thermalisations (i.e., transformations from Eqs.~\eqref{eq:partial1}-\eqref{eq:partial2} with $\lambda=1$). As can be seen in Fig.~\hyperref[Fig:thermalconevsmarkovianthermalcone]{\ref{Fig:thermalconevsmarkovianthermalcone}c}, $\T^{\textrm{MTP}}_{+}(\v p)\subset \T^{\textrm{TO}}_{+}(\v p)$ and the difference between these two sets of thermodynamically accessible states arises purely from memory effects.

\section{Bridging the gap with memory}
\label{sec:bridging}

We begin this section by explaining the main building block of this work, namely the notion of memory-assisted Markovian thermal processes. Then, we demonstrate how energy-incoherent states achievable from a given initial state $\v{p}$ via thermal operations [i.e., any $\v{q}\in C_+^{\mathrm{TO}}(\v{p})$] can be approached arbitrarily well using memory-assisted Markovian thermal processes with large enough memory. We will start by simplifying the problem and showing that it is sufficient to only consider the achievability of the extreme points of $C_+^{\mathrm{TO}}(\v{p})$. Next, we will introduce MeMTP protocols that will serve us to approach extreme points of $C^{\textrm{TO}}_{+}(\v p)$ using MTPs acting on the system and memory state, $\v{p}\otimes \v{\gamma}_M$. Finally, we will analyse the performance of these protocols, i.e., we will show how well they approximate the desired transformations as the size of the memory $N$ grows. Due to structural differences, we will do this separately for the infinite temperature limit and the case of finite temperatures.

\subsection{Memory-assisted Markovian thermal processes}

In this chapter, our objective is to interpolate between the two extreme regimes characterised by arbitrarily strong and no memory effects, as described within the TO and MTP frameworks. To accomplish this, we will concentrate on the more restrictive MTP framework and extending it by explicitly modeling memory effects by bringing ancillary systems in thermal states that will be discarded at the end. More precisely, we consider MTPs acting on a composite system consisting of the main $d$-dimensional system in a state $\rho$ and an $N$-dimensional memory system prepared in its thermal state $\gamma_M$ (i.e., given by Eq.~\eqref{Eq:thermal_state} with $H_E$ replaced by the Hamiltonian $H_M$ of the memory system, which can be arbitrary). The thermality of the ancillary system $M$ is crucial, as this way we ensure that no extra thermodynamic resources are brought in unaccounted, and the only role played by $M$ is to bring extra dimensions that can act as a memory. As already explained in the introduction, this can also be viewed as having control over the small $N$-dimensional part of the bath. Formally, we define the following set of quantum channels.

\begin{definition}[Memory-assisted MTPs]
    A quantum channel $\E$ is called a memory-assisted Markovian thermal process (MeMTP) with memory of size $N$, if it can be written as
    \begin{equation}
        \E(\rho) = \textrm{Tr}_{M}[\E_{\mathrm{MTP}}(\rho\otimes \gamma_M)],
    \end{equation}    
    where $\E_{\mathrm{MTP}}$ is a Markovian thermal process acting on the original system extended by an $N$-dimensional ancillary system $M$ prepared in a thermal state $\gamma_M$.
\end{definition}

As already mentioned, in this chapter we will focus on transformations between energy-incoherent states. Our aim is to show that the sets of states achievable from a given $\v{p}$ via memory-assisted MTPs interpolate between $\T^{\textrm{MTP}}_{+}(\v p)$ (for $N=1$) and $\T^{\textrm{TO}}_{+}(\v p)$ (for $N\to\infty$). As a final note, observe that this framework can be formally related to a particular kind of catalytic transformations~\cite{jonathan1999entanglement,brandao2015second}. This is because the ancillary memory system can always be thermalised at the end of the process and this way be brought to the initial state.  

\begin{kaobox}[frametitle= Interpolating between extreme regimes]
Arrows between distributions represent the existence of specific thermodynamic transformations, whose existence is determined by partial-order relations: thermomajorisation $\succ_\beta$ for thermal operations and continuous thermomajorisation $\ggcurly_\beta$ for Markovian thermal processes. Our results demonstrate that the gap between these two frameworks can be bridged with the use memory-assisted MTPs employing ancillary memory systems of growing dimension $N$ and prepared in thermal states.
\begin{equation*}
 \begin{matrix}
\text{Memoryless} & \v p \xrightarrow[]{\text{MTP}} \v q & \Longleftrightarrow & \v p \ggcurly_{\beta} \v q  \\
\biggl| &   & \vdots \\
\text{Finite memory} & \v p \xrightarrow[]{\text{MeMTP}} \v q  & \Longleftrightarrow & \v p \otimes \v \gamma_M \ggcurly_{\beta} \v q \otimes \v \gamma_M\\
\bigg\downarrow &   &\:\:\begin{matrix}
\vdots \\
_{N\to \infty}
 \\
\downarrow
\end{matrix}  \\
\text{Full control} & \v p \xrightarrow[]{\text{TO}} \v q & \Longleftrightarrow & \v p \succ_{\beta} \v q
\end{matrix}
\end{equation*}
\end{kaobox}
\subsection{Simplification to extreme points}

    We start by recalling the following notion that is crucial for our analysis.
   \betaordering* 

The $d$-dimensional matrix representation of $\pi_{\v p}$ will be denoted by $\Pi_{\v p}$, i.e., $\Pi_{\v p}\v{p}=\v{p}^\beta$.

    The future thermal cone $\mathcal{T}_+^{\mathrm{TO}}(\v{p})$ is a polytope with at most~$d!$ extreme points, one for each possible $\beta$-order $\pi$. We will denote them by $\v{p}^{\pi}$ (in particular it means that \mbox{$\v{p}^{\pi_{\v p}}=\v{p}$}). Now, we will use two crucial observations. First, in Ref.~\cite{Lostaglio2018elementarythermal} it was shown that
    \begin{equation}
        \label{eq:extremal_enough}
        \v{q} \in \mathcal{T}_+^{\mathrm{TO}}(\v{p}) \Rightarrow \v{q} \in \T_+^{\mathrm{TO}}(\v{p}^{\pi_{\v{q}}}),
    \end{equation}
    meaning that all states with a $\beta$-order $\pi$ that can be achieved from $\v{p}$ via thermal operations can also be achieved starting from $\v{p}^\pi$. And second, it was shown in Ref.~\cite{lostaglio2021continuous} that
    \begin{equation}
        \label{eq:same_beta_order}
        \left[\v{q} \in \mathcal{T}_+^{\mathrm{TO}}(\v{p})~~\mathrm{and}~~\pi_{\v{q}}=\pi_{\v{p}}\right] \Rightarrow \v{q} \in \T_+^{\mathrm{MTP}}(\v{p}),
    \end{equation}
    meaning that within the same $\beta$-order as the initial state, the subsets of states achievable via TOs and via MTPs do coincide. As a result, if one can construct memory-assisted MTPs that reach all the extreme points of $\mathcal{T}_+^{\mathrm{TO}}(\v{p})$, then one can also get to every state in $\mathcal{T}_+^{\mathrm{TO}}(\v{p})$ via MeMTPs. This is done by simply first transforming $\v{p}$ to a given extreme point $\v{p}^\pi$ of $\mathcal{T}_+^{\mathrm{TO}}(\v{p})$, and then using MTPs to get from $\v{p}^\pi$ to every state with a $\beta$-order $\pi$ in $\mathcal{T}_+^{\mathrm{TO}}(\v{p})$.

    In order to quantify how well a given state in $\mathcal{T}_+^{\mathrm{TO}}(\v{p})$ can be approximated, we will use the total variation distance defined by
    \begin{equation}
        \delta(\v{p},\v{q}):=\frac{1}{2}\sum_{i=1}^d |p_i-q_i|.
    \end{equation}
    From the discussion above, it should be clear that if we can construct MeMTP protocols approximating every extreme point with an error at most $\epsilon$, then
    \begin{equation}
          \forall \v{q} \in \mathcal{T}_+^{\mathrm{TO}}(\v{p}):\quad \min_{\P\in \mathrm{MeMTP}} \delta(\P(\v{p}),\v{q})\leq \epsilon.
    \end{equation}

    A particular subset of extreme points of $\mathcal{T}_+^{\mathrm{TO}}(\v{p})$ that we will investigate in more detail is given by those extreme states that can be achieved via sequences of $\beta$-swaps. A $\beta$-swap $\Pi_{ij}^\beta$ can be seen as a thermodynamic analogue of a population swap between levels~$i$ and~$j$~\cite{Lostaglio2018elementarythermal}:
    \begin{equation}
    \label{Eq:beta-swap}
        \Pi^{\, \beta}_{ij} := \begin{bmatrix}
        1-e^{-\beta {(E_j - E_i)}}  &1 \\ 
        e^{-\beta {(E_j - E_i)}} &0 
        \end{bmatrix}\oplus \mathbbm{1}_{{\backslash}(ij)},
    \end{equation}
    with $E_i \leq E_j$ and $\mathbbm{1}_{{\backslash}(ij)}$ denoting the $(d -2) \times (d-2)$ identity matrix on the subspace of all energy levels except $i, j$. Note that in the infinite temperature limit ($\beta=0$), the above recovers a transposition on levels $i$ and $j$, which we will simply denote by~$\Pi_{ij}$. In this limiting case, all extreme points of $\mathcal{T}_+^{\mathrm{TO}}(\v{p})$ can be obtained by sequences of transpositions (that is because an extreme point in that case is of the form $\Pi \v{p}$ for a permutation matrix $\Pi$, and every $\Pi$ can be constructed from transpositions). 

    For finite temperatures, a $\beta$-swap transforms $\v{p}$ into an extreme point of $\v{p}^\pi$ if the $\beta$-orders $\pi_{\v{p}}$ and $\pi$ differ only by a transposition of adjacent elements~\cite{Lostaglio2018elementarythermal,deoliveirajunior2022}. In other words, it happens for $\Pi_{ij}^\beta$ when $\pi_{\v{p}}(i)=\pi_{\v{p}}(j) \pm 1$. More generally, a sequence of such non-overlapping $\beta$-swaps will also produce an extreme point of $\T^{\mathrm{TO}}_+(\v{p})$, and so a total number of 
    extreme points that can be achieved by sequences of $\beta$-swaps for dimension $d$ (including the starting point) 
    is given by $F(d+1)$, where $F(k)$ is the $k$-th Fibonacci number~\cite{white_1983}. Finally, we will also make use of the notion of $\beta$-cycles that we now define. For a state $\v{p}$, consider a $k$-dimensional subset of energy levels $ \{i_1,\hdots,i_{k}\}$ neighbouring in the $\beta$-order, i.e., $\pi_{\v{p}}(i_{j+1}) = \pi_{\v{p}}(i_j) + 1$. Denote by $\pi$ a cyclic permutation on this subset, i.e., either $\pi(i_j) = i_{j+1 \text{ mod } k}$, or $\pi(i_j) = i_{j-1 \text{ mod } k}$. Then, a thermal operation mapping $\v{p}$ to its extreme point $\v{p}^{\pi'}$ is called a $\beta$-cycle, if $\Pi'=\Pi \Pi_{\v{p}}$ (here $\Pi, \Pi'$ and $\Pi_{\v{p}}$ denote matrix representations of permutations $\pi,\pi'$ and $\pi_{\v{p}}$). To emphasise that a given $\beta$-cycle acts on $k$ levels, we will sometimes refer to it as a $\beta$-$k$-cycle.
    
    \subsection{Memory-assisted protocols}
\label{Sec:Memory-assisted protocols}

The basic building blocks of all our protocols are given by two-level elementary thermalisations that are formally defined as follows. 
\begin{definition}[Two-level thermalisations\label{def:neighbour-thermalisations}] 
Consider a system in a state $\v{p}$ with the corresponding thermal state $\v{\gamma}$. Then, an MTP transformation
\begin{equation}\label{eq:neighbour_thermalisations}
    \{p_i, p_{j} \} \to \left\{ \frac{p_i+p_j}{\gamma_i+\gamma_j}\gamma_i,\frac{p_i+p_j}{\gamma_i+\gamma_j}\gamma_j \right\}
\end{equation}
is called a two-level thermalisation between levels $i$ and $j$, and the corresponding matrix acting on probability vectors will be denoted by $T_{ij}$. Moreover, if $\pi_{\v{p}}(i)=\pi_{\v{p}}(j) \pm 1$, then $T_{ij}$ is called a neighbour thermalisation.
\end{definition}

Let us note that the importance of neighbour thermalisations and the reason we employ them in our protocols stems from the fact that their sequences produce the extreme points of the Markovian thermal cone $\T_+^{\mathrm{MTP}}$~\cite{lostaglio2021continuous}. Intuitively, one can expect that in order to approximate extreme states of $\T_+^{\mathrm{TO}}(\v{p})$ using MeMTPs, one should get to the extreme points of $\mathcal{T}_+^{\mathrm{MTP}}(\v{p}\otimes\v{\gamma}_M)$, and these can be achieved by neighbour thermalisations of the composite system-memory state.

Before delving into the full details of our protocol for approximating $\beta$-swaps, let us start with a high-level description to provide some insight into our investigation. We begin with the simplest case of a two-level system and a two-dimensional memory, drawing an analogy between continuous (thermo)majorisation and connected vessels~(see Fig.~\ref{Fig:vessels} for a schematic representation). 
Considering two vessels—one filled with liquid and one empty—the most one can do when they are connected is to equalise the levels of the liquid between them. However, by adhering to the simple schematic provided in Fig. \ref{Fig:vessels}, it is possible to exceed this intuitively unbeatable limit.
\begin{marginfigure}[0.34cm]
	\includegraphics[width=4.718cm]{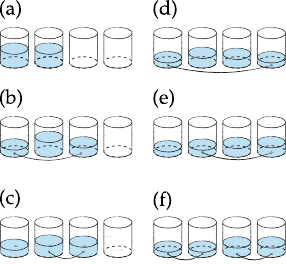}
	\caption{\emph{Simplest example using connected vessels analogy}. Continuous (thermo)majorisation on $d$-level probability vectors is equivalent to a task of redistributing the content of $d$-ordered vessels that are connected pairwise. Adding a memory in a Gibbs state is then akin to multiplying glasses -- one empty and one full glass become $N$ pairs of full and empty glasses. The process involves five steps, read from top to bottom and left to right, which represent the simplest protocol that allows shifting more than half of the liquid from full to empty glasses. The final distribution after the protocol is applied is $(3/8, 5/8)$.}
	\label{Fig:vessels}
\end{marginfigure}

The sequence in Fig.~\ref{Fig:vessels} should be read from panels (a) to (f), in a top-down and left-to-right manner. We begin with four vessels: two are half-filled and two are empty. We can connect these vessels in pairs, thus equalising the fluid levels. The first two steps involve connecting the half-filled vessels sequentially to the first empty vessel. Likewise, the next two steps connect both initially half-filled vessels to the second empty one. The final step, which involves equalising the fluid levels in pairs, is analogous to thermalising the memory. An astute reader can confirm that $5/8$ of the total fluid ends up in the vessels that were initially empty, thereby surpassing the $1/2$ limit. 

We now describe our proposition for a MeMTP protocol approximating the $\beta$-swap $\Pi^{\, \beta}_{ij}$ between the $i$-th and $j$-th energy levels of the main system. It involves a sequence of $N^2$ two-level thermalisations of the state of the composite system (see Fig.~\ref{Fig:thermodynamic_protocol}), which includes the main system and a memory starting at thermal equilibrium, i.e., a state \mbox{$\v p \otimes \v \gamma_M$}. In particular, we focus on the populations $[\v p \, \otimes \, \v \gamma_M]_{N(i-1)+1},\hdots,[\v p \, \otimes\,  \v \gamma_M]_{Ni}$ corresponding to the $i$-th level of the main system and similarly for the $j$-th level. The protocol can be split into a sequence of $N$ rounds $\mathcal{R}_k^{(ij)}$ with $k = 1,\hdots,N$ consisting of $N$ steps each (shaded area in Fig.~\ref{Fig:thermodynamic_protocol}). In the $k$-th round, we select the entry $[\v p \otimes \v \gamma_M]_{N(i-1)+k}$ and thermalise it sequentially with all the levels corresponding to the level $j$ of the main system:
\begin{equation}
    \mathcal{R}^{(ij)}_k(\v p \otimes \v \gamma_M) := \qty(\prod_{l=1}^N T_{(i-1)N + k,\,(j-1)N + l}) (\v p \otimes \v \gamma_M).
\end{equation}
Note that if $\pi_{\v{p}}(i)=\pi_{\v{p}}(j) \pm 1$ (i.e., the $\beta$-orders of $\v{p}$ and $\Pi^\beta_{ij}\v{p}$ differ by a transposition of adjacent elements), then all thermalisations performed are neighbour thermalisations. Using the above, we can now define the action of the truncated protocol $\widetilde{\mathcal{P}}^{(ij)}$:
\begin{equation}\label{eq:trunc_protocol}
    \widetilde{\mathcal{P}}^{(ij)}(\v p \otimes \v \gamma_M) := \mathcal{R}^{(ij)}_N\circ\hdots\circ\mathcal{R}^{(ij)}_1(\v p \otimes \v \gamma_M).
\end{equation}
The final step is to decouple the main system from the memory using a full thermalisation $\T$ of the memory system $M$, which acts on a general joint state $\v{Q}$ as: 
\begin{equation}
    \TT(\v Q) = \v q\otimes \v{\gamma}_M,\qquad q_i = \sum_{j=1}^N Q_{N(i-1)+j}.
\end{equation}
Thus, the full protocol approximating a $\beta$-swap $\Pi_{ij}^\beta$ is given by
\begin{equation}
    \mathcal{P}^{(ij)}(\v p \otimes \v \gamma_M) = \TT\circ\widetilde{\mathcal{P}}^{(ij)}(\v p \otimes \v \gamma_M).
\end{equation}
\begin{figure}[t]
\centering
\includegraphics[width=8.488cm]{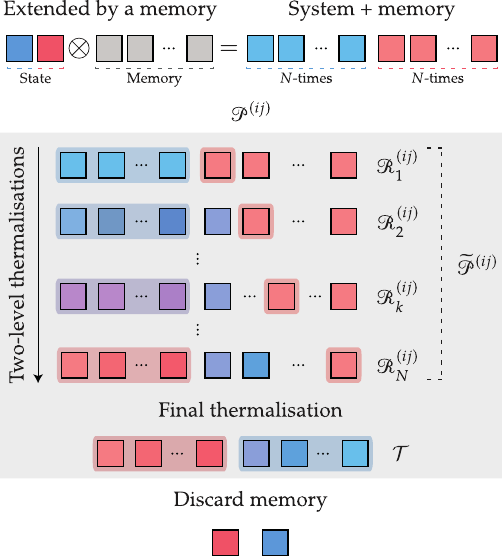}
\caption{\emph{$\beta$-swap protocol}. A two-level subsystem of a generically $d$-dimensional system, represented by blue and red squares, is extended by an $N$-dimensional memory represented by grey squares. The composite $2N$-dimensional system undergoes $N$ rounds of processing, where the $k$-th round involves $N$ sequential two-level thermalisations of the first $N$ entries with the $(N + k)$-th entry (represented by the shaded colour around the squares). After the final thermalisation step, the memory can be discarded.}
\label{Fig:thermodynamic_protocol}
\end{figure}

We will also employ more general protocols that aim at approximating the transformation of the initial state $\v{p}$ into an extreme state of $\T_+^{\mathrm{TO}}(\v{p})$ given by $\v{p}^{\pi'}$. Denote matrix representations of $\beta$-orders of these states by $\Pi_{\v{p}}$ and $\Pi'=\Pi \Pi_{\v{p}}$ for some permutation matrix $\Pi$. Moreover, let us decompose $\Pi$ into neighbour transpositions with respect to $\v{p}$, i.e., we write \mbox{$\Pi = \Pi_{i_mj_m}\dots \Pi_{i_1j_1}$} with every consecutive transposition $\Pi_{i_kj_k}$ changing the $\beta$-order of the state $\Pi_{i_{k-1}j_{k-1}}\dots \Pi_{i_1j_1} \v{p}$ only by a transposition of adjacent elements. Then, we define the following two protocols to approximate $\v{p}^{\pi'}$:
\begin{subequations}
\begin{align}
    \label{eq:protocol}
    \mathcal{P}^{{\Pi}} &:= \mathcal{P}^{(i_mj_m)} \circ\dots\circ \mathcal{P}^{(i_1j_1)},
     \\
     \label{eq:truncated}
    \widetilde{\mathcal{P}}^{{\Pi}} &:= \T\circ \widetilde{\mathcal{P}}^{(i_mj_m)} \circ\dots\circ \widetilde{\mathcal{P}}^{(i_1j_1)}.
\end{align}
\end{subequations}
Note that, by construction, all two-level thermalisations performed in the above protocols are neighbour thermalisations.

\subsection[Achieving extreme points of the future thermal cone for infinite temperature]{\texorpdfstring{Achieving extreme points of $\T_+^{\mathrm{TO}}$ for $\beta= 0$}{Achieving extreme points of the future thermal cone for infinite temperature}}

We are now ready to state our main results concerning the power of memory-assisted Markovian thermal processes in the infinite temperature limit. Let us recall that we focus on a $d$-level system and an $N$-level memory in energy-incoherent states represented by probability distributions $\v p$ and $\v \gamma_M$. Since $\beta = 0$, the thermal state of the memory is described by a uniform distribution $\v{\eta}_M$ with every entry equal to $1/N$. We start with the following lemma, the proof of which can be found in Section~\ref{app:neighbour} (with the necessary background on the mathematical tools used presented in Section~\ref{app:beta}). 

\begin{lemma}[Memory-assisted transposition]
    \label{Lem:transposition} 
    In the infinite temperature limit, $\beta=0$, and for an $N$-dimensional memory, the MeMTP protocol $\mathcal{P}^{(ij)}$ acts as
    \begin{equation}
        \mathcal{P}^{(ij)} (\v{p}\otimes\v{\eta}_M)=  \v{q} \otimes \v \eta_M ,
    \end{equation}
    with
    \begin{equation}
        \v{q} = \left(\Pi_{ij} + \epsilon\qty(\mathbbm{1} - \Pi_{ij})\right) \v{p},
    \end{equation}
    and $\epsilon$ given by
    \begin{equation}
        \epsilon=(\pi N)^{-1/2} + o\qty(N^{-1/2}) \overset{N\rightarrow\infty}{\longrightarrow} 0.
    \end{equation}
\end{lemma}

It is well known that any permutation of $d$ elements can be decomposed into a product of at most $d$ transpositions or $\binom{d}{2}$ neighbour transpositions. Therefore, by employing Lemma~\ref{Lem:transposition}, we can demonstrate that an arbitrary permutation can be achieved using a composition of our approximate protocols.

\begin{theorem}[Memory-assisted permutation]
\label{Thm:permutation} 

In the infinite temperature limit, $\beta=0$, and for an $N$-dimensional memory, $\Pi$ can be approximated by the MeMTP protocol $\mathcal{P}^{{\Pi}}$ as follows: 
\begin{equation}
    \mathcal{P}^{{\Pi}}(\v{p}\otimes \v{\eta}_M) = \v{q}\otimes \v{\eta}_M,
\end{equation}
where
\begin{equation}
    \v{q}=\left(\Pi + \epsilon \v\Delta + o\qty(N^{-1/2})\right)\v{p}
    \overset{N\rightarrow\infty}{\rightarrow} \Pi \v {p},
\end{equation}
with $\epsilon = (\pi N)^{-1/2}$ and the operator $\v\Delta$ defined in terms of transpositions appearing in the definition of $\P^\Pi$ in Eq.~\eqref{eq:protocol}:
\begin{equation} \label{eq:Lambda_correction_op}
     \v\Delta = \sum_{l=1}^{m}\qty(\prod_{k=l+1}^{m} \Pi_{i_kj_k})\qty(\mathbbm{1}-\Pi_{i_lj_l})\qty(\prod_{k=1}^{l-1} \Pi_{i_kj_k}).
\end{equation}

\end{theorem}
\begin{proof}

From the definition of $\mathcal{P}^\Pi$ and Lemma~\ref{Lem:transposition} we get
\begin{equation}
    \v{q} =  \qty[\Pi_{i_mj_m} + \epsilon \qty(\mathbbm{1} - \Pi_{i_mj_m})]\dots \qty[\Pi_{i_1j_1} + \epsilon \qty(\mathbbm{1} - \Pi_{i_1j_1})]\v{p}.
\end{equation}
Clearly, the leading term is given by $\Pi\v{p}$, whereas the next leading term, proportional to $\epsilon$, is given by $\epsilon\Delta\v{p}$. All higher order terms scale at least as $\epsilon^2$, so are of the order $o(N^{-1/2})$.
\end{proof}

The above theorem can then be directly used to obtain the bound on how close one can get from a given $\v{p}$ to any state $\v{q}\in \T_+^{\mathrm{TO}}(\v{p})$ using MeMTPs with $N$-dimensional memory. We explain how to derive such a bound for a given $\v{p}$ in Section~\ref{app:bound}, whereas below we present a weaker, but much simpler, bound that is independent of $\v{p}$.

\begin{corollary} \label{corr:general_bound}
    Consider states $\v{p}$ and $\v q\in C_+^{TO}(\v{p})$. Then, in the infinite temperature limit, $\beta=0$, and for an $N$-dimensional memory, there exists a MeMTP protocol $\mathcal{P}$ such that
    \begin{equation}
        \label{eq:cor1}
        \P(\v{p}\otimes \v{\eta}_M) = \v{q}'\otimes \v{\eta}_M,
    \end{equation}
    with
    \begin{equation}
        \label{eq:cor2}
        \delta(\v{q}',\v{q})\leq \frac{d(d-1)}{2\sqrt{\pi N}} + o\qty(N^{-1/2}).
    \end{equation}   
\end{corollary}
\begin{proof}
    First, define $\Pi$ as a permutation that changes the $\beta$-order of $\v{p}$ to that of $\v{q}$. In other words, the $\beta$-order of $\Pi \v{p}$ is~$\pi_{\v{q}}$. Then, using Theorem~\ref{Thm:permutation}, we have that
    \begin{equation}
        \P^{\Pi}(\v{p}\otimes \v{\eta}_M)=\v{r}\otimes \v{\eta}_M
    \end{equation}
    with
    \begin{equation}
        \delta(\v{r},\Pi \v{p})\simeq\frac{1}{2\sqrt{\pi N}} \sum_{i=1}^d |(\Delta \v{p})_i|   \lesssim  \frac{m}{\sqrt{\pi N}} \lesssim  \frac{d(d-1)}{2\sqrt{\pi N}},
    \end{equation}
    where $\simeq$ and $\lesssim$ denote the equalities and inequalities up to $o(N^{-1/2})$. In the above, we have used the triangle inequality and the fact that one can always decompose $\Pi$ into at most $d(d-1)/2$ neighbour transpositions. Next, from Eq.~\eqref{eq:same_beta_order}, we know that there exists an MTP protocol $\P'$ mapping $\Pi\v{p}$ to~$\v{q}$. Using the contractiveness of the total variation distance under stochastic processing, we then have
    \begin{align}
        \delta(\P'(\v{r}),\v{q})&=\delta(\P'(\v{r}),\P'(\Pi\v{p}))\leq \delta(\v{r},\Pi\v{p})\lesssim \frac{d(d-1)}{2\sqrt{\pi N}}.
    \end{align}
    We thus conclude that by choosing $\P=(\P'\otimes \mathcal{I}_M)\circ \P^\Pi$, Eqs.~\eqref{eq:cor1}-\eqref{eq:cor2} are satisfied.
\end{proof}

Furthermore, we present the following conjecture for a better approximation of arbitrary permutations.

\begin{conjecture}[Improved convergence]
\label{Conj:convergence}
In the infinite temperature limit, $\beta=0$, and for an $N$-dimensional memory, $ \widetilde{\mathcal{P}}^{{\Pi}}$ gives a better approximation of a permutation $\Pi$ than  $\mathcal{P}^{\Pi}$:
    \begin{equation}
        \delta\left( \Pi \v p,\tilde{\v q}\right) \leq  \delta\left( \Pi \v p, \v{q}\right),
    \end{equation}
    where $\tilde{\v{q}}$ and $\v{q}$ are defined via
    \begin{equation}
        \widetilde{\mathcal{P}}^{{\Pi}}(\v{p}\otimes \v{\eta}_M) = \tilde{\v{q}}\otimes \v{\eta}_M,\qquad \mathcal{P}^{{\Pi}}(\v{p}\otimes \v{\eta}_M) = \v{q}\otimes \v{\eta}_M.
\end{equation}
\end{conjecture}
The conjecture is solidified by strong numerical evidence (see Fig.~\ref{fig:beta0_convergence} for an example considering $d = 6$). We note that the convergence is better, but the overall character of $O(N^{-1/2})$ is still preserved. More specifically, we observe that for permutations given by $\beta$-$k$-cycles with $k \leq d$, there is no advantage to removing the intermediate thermalisations (i.e., no advantage of $\widetilde{\P}^\Pi$ over ${\P}^\Pi$). The advantage already appears for a composition of $\beta$-$d$-cycle with $\beta$-$(d-1)$-cycle, leading to the $\beta$-order $(d, d-1,1,\hdots,d-2)$ (here, without loss of generality, we assumed that the initial $\beta$-order is given by $(1,2,\dots,d)$). In general, the advantage grows with the number of composed $\beta$-cycles (see Fig.~\ref{fig:beta0_convergence}, where different colours and markers correspond to different length compositions of $\beta$-cycles). In particular, we verified that for a permutation $(16)(25)(34)$, which is composed of $\beta$-cycles of length $6$ through $2$ (or $15 = \binom{6}{2}$ neighbour transpositions), both $\mathcal{P}^{\Pi}$ and $\widetilde{\mathcal{P}}^{\Pi}$ converge to the actual extreme point $\Pi \v{p}$. Surprisingly, we find that all the other possible permutations fall within the convergence advantage class of one of the aforementioned $\beta$-cycle compositions. This includes, in particular, the cases when the last $\beta$-cycle in the sequence is incomplete, i.e., it is shortened by the final subsequence of $\beta$-swaps of any length.
\begin{figure}[t]
    \centering \includegraphics[width=10.765cm]{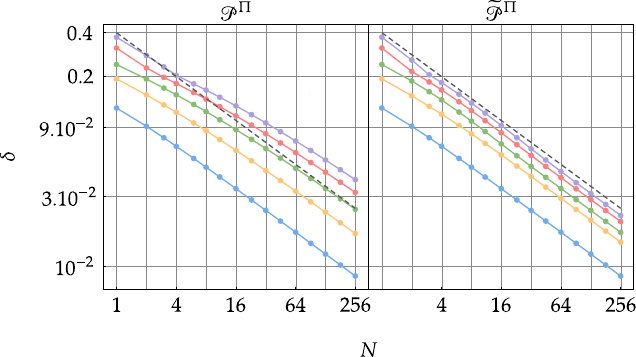}
    \caption{\emph{Convergence rates at infinite temperature}. Log-log plot of the total variation distance $\delta$ between the extreme points \mbox{$\v{p}^\pi\in \T_+^{\mathrm{(TO)}}(\v{p})$} and the states obtained from $\v{p}$ via the algorithm $\mathcal{P}^{\Pi}$ (left panel) and $\widetilde{\mathcal{P}}^{ \Pi}$ (right panel), as a function of the memory size~$N$. Here, \mbox{$\v p = \qty(0.37, 0.24, 0.16, 0.11, 0.07, 0.05)$}, $\beta = 0$, and different colours correspond to families of extreme points $\v{p}^\pi$ with different convergence rates (from bottom to top $\v{p}^\pi$ is obtained from $\v{p}$ via a $\beta$-6-cycle, a composition of a $\beta$-6-cycle with a $\beta$-5-cycle, and so on). All convergences behave as $O(N^{-1/2})$, which can be seen by the comparison with the function $0.4/\sqrt{N}$ (dashed black line), with multiplicative advantage for $\widetilde{\mathcal{P}}^{ \Pi}$ over $\mathcal{P}^{ \Pi}$. }
    \label{fig:beta0_convergence}
\end{figure}

\subsection[Achieving extreme points of the future thermal cone for finite temperatures]{\texorpdfstring{Achieving extreme points of $\T_+^{\mathrm{TO}}$ for $\beta\neq 0$}{Achieving extreme points of the future thermal cone for finite temperatures}}

Our second main result concerns the power of memory-assisted Markovian thermal processes at finite temperatures. 
We start with the following generalisation of Lemma~\ref{Lem:transposition}, the proof of which can be found in Section~\ref{app:neighbour}.  

\begin{theorem}[Memory-assisted $\beta$-swap]
\label{Thm:beta-swap} 
    For a finite temperature, $\beta\neq 0$, and for an $N$-dimensional memory described by a trivial Hamiltonian (so that its thermal state is $\v{\eta}_M$), $\Pi_{ij}^\beta$ can be approximated by the MeMTP protocol $\P^{(ij)}$ as follows:    
    \begin{equation}
        \P^{(ij)} (\v{p}\otimes \v \eta_M)=\v{q}\otimes \v \eta_M,
    \end{equation}
    with
    \begin{equation}
        \delta(\v q,\Pi_{ij}^\beta \v{p}) = \frac{\qty(4\Gamma_i\Gamma_j)^N}{\qty(\Gamma_i - \Gamma_j)^2}\qty[\frac{\abs{p_i \Gamma_j - p_j\Gamma_i}}{(N+1)\sqrt{\pi N}} + o\qty(N^{-3/2})]
    \end{equation}
    where we have used $\Gamma_i = \gamma_i/(\gamma_i + \gamma_j)$ and likewise for $\Gamma_j$.
\end{theorem}

By using the above theorem, one can approximate with arbitrary precision a total of $F(d+1)$ extreme points achievable by a composition of non-overlapping $\beta$-swaps (recall that $F(k)$ is the $k$-th Fibonacci number). However, since $F(d+1)\leq d!$ for $d\geq 3$, not all extreme points of $\T_+^{\mathrm{TO}}(\v{p})$ can be obtained this way. Nevertheless, we conjecture that using MeMTP protocols $\widetilde{\P}^\Pi$ that are composed of blocks imitating $\beta$-swaps, just without intermediate thermalisations, one can reach all the extreme points of $\T_+^{\mathrm{TO}}(\v{p})$.

\begin{conjecture}[Extreme points of $\T_+^{\mathrm{TO}}$]\label{Conj:beta_permutations}
    Consider a state $\v{p}$ and the extreme point of $\T_+^{\mathrm{TO}}(\v{p})$ given by $\v{p}^{\pi'}$, with matrix representations of $\beta$-orders of these states satisfying \mbox{$\Pi'=\Pi\Pi_{\v{p}}$} for some permutation $\Pi$. Then, for a finite temperature, $\beta\neq 0$, and for an $N$-dimensional memory described by a trivial Hamiltonian (so that its thermal state is $\v{\eta}_M$), the MeMTP protocol $\widetilde{\P}^{\Pi}$ acts as
    \begin{equation}
        \widetilde{\P}^{\Pi}(\v{p}\otimes \v{\eta}_M)=\v{q}\otimes \v{\eta}_M,
    \end{equation}
    with
    \begin{equation} \label{eq:beta_nonzero_conv}
        \delta(\v{q},\v{p}^{\pi'})\overset{N\rightarrow\infty}{\longrightarrow} 0.
    \end{equation}
\end{conjecture}

The conjecture is solidified by the following two evidences. First, we can actually prove it in the following special case.

\begin{theorem}[Memory-assisted $\beta$-3-cycle]
\label{Thm:beta-3-cycle} 

Consider a state~$\v{p}$ with entries $i_1,i_2,i_3$ being neighbours in the $\beta$-order (i.e., $\pi_{\v{p}}(i_1)=\pi_{\v{p}}(i_2)+1=\pi_{\v{p}}(i_3)+2$), and the extreme point of $\T_+^{\mathrm{TO}}(\v{p})$ given by $\v{p}^{\pi'}$, with matrix representations of $\beta$-orders of these states satisfying \mbox{$\Pi'=\Pi\Pi_{\v{p}}$} for $\Pi=\Pi_{i_1i_3}\Pi_{i_2i_3}$. Then, for a finite temperature, $\beta\neq 0$, and for an $N$-dimensional memory described by a trivial Hamiltonian (so that its thermal state is $\v{\eta}_M$), the MeMTP protocol $\widetilde{\P}^{\Pi}$ acts as
\begin{equation}
    \widetilde{\P}^{\Pi}(\v{p}\otimes \v{\eta}_M)= \v{q}\otimes \v{\eta}_M
\end{equation}
with
\begin{equation}
    \delta(\v{q},\v{p}^{\pi'})\overset{N\rightarrow\infty}{\rightarrow} 0.
\end{equation}
\end{theorem}
The proof of the above theorem can be found in Section~\ref{app:beta-3-cycle}, and potentially the same proving techniques can be applied to higher-order cycles. This would then provide a general method for simulating $\beta$-cycles, as well as any combinations of non-overlapping $\beta$-cycles, with arbitrary precision through MeMTPs. 

The second evidence is based on extensive numerical simulations. Let us consider any state $\v{p}$ and a set of permutations defined by a recurrence formula
\begin{equation}
    \widetilde{\Pi}_{i j} := \qty(\prod_{k=1}^{j} \Pi_{\pi_{\v p}(k),\, \pi_{\v p}(k+1)})\widetilde{\Pi}_{i-1,\,d+1-i}
\end{equation}
with $1 \leq j \leq d-i$ and assuming the starting condition $\widetilde{\Pi}_{0d} = \mathbbm{1}$. Note that $\widetilde{\Pi}_{1,d-1}$ represents a full $\beta$-$d$-cycle, $\widetilde{\Pi}_{i,d+1-i}$ a composition of $i$ $\beta$-cycles of length from $d$ to $d+1-i$, and finally $\widetilde{\Pi}_{d,1}$ is a permutation which fully reverses the $\beta$-order of $\v p$. For each such permutation, we have considered the action of the protocol $\widetilde{\mathcal{P}}^{\widetilde{\Pi}_{ij}}\qty(\v p)$ and its convergence to the respective extreme point~$\v{p}^{\pi}$ with $\Pi=\widetilde{\Pi}_{ij} \Pi_{\v{p}}$ (recall that $\Pi$ is a matrix representation of $\pi$). In each case, we have observed the convergence of the form from Eq.~\eqref{eq:beta_nonzero_conv} that is better than $N^{-1/2}$. Results for an exemplary state in dimension $d = 6$ are presented in Fig.~\ref{fig:beta_nonzero_convergence}, where a total of $15$ different curves are shown to lie below the $N^{-1/2}$ limit and diverging from it.

Finally, based on Theorem~\ref{Thm:beta-swap}, the proof of Theorem~\ref{Thm:beta-3-cycle}, and numerical evidence, one can reasonably strengthen Conjecture~\ref{Conj:beta_permutations} to make the following statement on the convergence:
\begin{equation}
    \delta(\v{q},\v{p}^{\pi'})=O\qty(\frac{e^{-A({\Pi}) N}}{N^{3/2}}),
\end{equation}
where $A(\Pi) = O(1)$ is a permutation-dependent exponent. 

\begin{figure}[t]
    \centering
    \includegraphics[width=10.478cm]{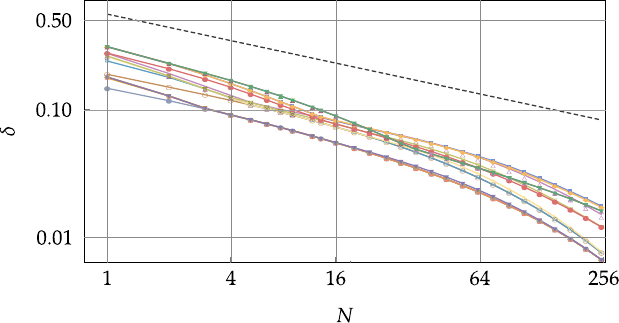}
    \caption{\emph{Convergence rates at finite temperature}. Log-log plot of the total variation distance $\delta$ between the states obtained from $\v{p}$ via the algorithm $\widetilde{\mathcal{P}}^{ \widetilde{\Pi}_{ij}}$ and the corresponding extreme points \mbox{$\v{p}^\pi\in \T_+^{\mathrm{(TO)}}(\v{p})$}, as a function of the memory size $N$. Here, \mbox{$\v p = \qty(0.37, 0.24, 0.16, 0.11, 0.07, 0.05)$}, $\beta = 0.1$, and different colours correspond to extreme points $\v{p}^\pi$ with matrix representation of the $\beta$-order $\pi$ given by $\Pi=\widetilde{\Pi}_{ij}\Pi_{\v{p}}$.
    For all curves, the convergence is better than $O(N^{-1/2})$, as can be seen by the comparison with the limiting line $1/\sqrt{\pi N}$ for $\beta = 0$ (dot-dashed black line).}
    \label{fig:beta_nonzero_convergence}
\end{figure}

\section{Discussion and applications}
\label{sec:discussion}

In Section~\ref{sec:bridging}, we demonstrated a method of achieving an arbitrary state from the future cone of TO using MeMTPs through MTP operations acting upon the system extended by memory, initiated in the thermal state $\gamma$. In the following sections, we will apply our protocol to study information-based quantum thermodynamic processes, such as work extraction and cooling. Next, we revisit the question of the sufficiency of two-level control for TOs. Finally, we provide a brief discussion of the behaviour of the free energy and correlations with the progression of our protocol. This sheds light on how non-Markovian effects arise in the memory-assisted protocol.


\subsection{Work extraction}

\begin{figure}[t]
    \centering
    \includegraphics{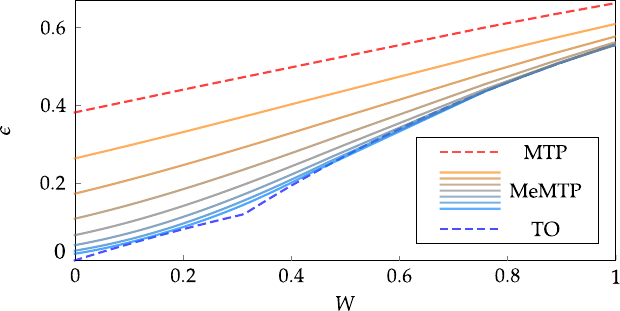}
    \caption{\emph{$\epsilon$-deterministic work extraction with MeMTPS}. Transformation error $\epsilon$ as a function of the work $W$ extracted from a two-level system with energy splitting $\Delta$ prepared in a thermal state at temperature $1/\beta_S$ smaller than the environmental temperature $1/\beta$ with parameters 
    \mbox{$\beta_S \Delta=2$} and \mbox{$\beta \Delta=1$}. System-environment interactions are modelled by TOs (dashed black curve),  MTPs (dashed red curve) and memory-assisted Markovian thermal process with a memory of size $2$, $4$ , $8$, $16$, $32$, $64$ and $128$, respectively.}
    \label{Fig:work-extractionMTP}
\end{figure}
We start our discussion with $\epsilon$-deterministic work extraction, which typically involves an out-of-equilibrium system~$S$, a thermal bath at inverse temperature $\beta$, and a battery $B$ initially in an energy eigenstate $E_0$~\cite{Aberg2013,horodecki2013fundamental}. The aim is to increase the energy of $B$ by an amount $W$ by exciting it from $E_0$ to $E_1=E_0+W$ with a success probability $1-\epsilon$. The optimal error $\epsilon$ for a given $W$ can be obtained via thermomajorisation condition for transformations given by thermal operations~\cite{horodecki2013fundamental} and through continuous thermomajorisation relations when transformations are given by Markovian thermal processes~\cite{Korzekwa2022}. Our framework allows one to interpolate between the two extremes by including a memory system with varying dimension $N$. 

Consider a two-level system $S$ and a two-level battery $B$ with energy levels $(0,\Delta)$ and $(0,W)$, respectively. Assume that the initial state of the joint system is given by $\v{p}_{SB} = \v{p}\otimes(1,0)$. One can then select the extreme point $\v{p}^{\pi'}_{SB}\in {\T}^{\mathrm{TO}}_+\qty(\v{p}_{SB})$ from the future thermal cone of the composite system for which the following relation is satisfied with the minimum value of $\epsilon_{\mathrm{TO}}$:
\begin{equation}
   \v{\gamma}\otimes(\epsilon_{\mathrm{TO}},1-\epsilon_{\mathrm{TO}}) \in 
   {\T}^{\mathrm{TO}}_+\qty(\v{p}^{\pi'}_{SB}).
\end{equation}
In other words, $\v{p}^{\pi'}_{SB}$ is an intermediate state from which one can achieve minimal error for extracting $W$ work from $\v{p}$ via any thermal operation. We can now define $\Pi$ as a permutation that maps the matrix representation of the initial $\beta$-order of $\v{p}_{SB}$ to the final $\beta$-order of $\v{p}_{SB}^{\pi'}$. Then, by using the algorithm $\widetilde{\mathcal{P}}^\Pi$, we can transform $\v{p}_{SB}$ into a state $\v{q}_{SB}$ that approximates~$\v{p}^{\pi'}_{SB}$. Finally, due to Eq.~\eqref{eq:same_beta_order}, we can use standard thermomajorisation to find the minimal value of $\epsilon_N$ for which the state $\v{q}_{SB}$ can be transformed to $\v{\gamma}\otimes(\epsilon_N,1-\epsilon_N)$ via MTPs. Note that $\epsilon_N$ then corresponds to the probability of failure of extracting work $W$ from $\v{p}$ using a memory of size $N$. Numerical simulations of this procedure (see Fig.~\ref{Fig:work-extractionMTP}) show that as $N$ grows, $\epsilon_N$ decreases, allowing us to conjecture that $\lim_{N\rightarrow\infty} \epsilon_N = \epsilon_{\mathrm{TO}}$. However, note that the convergence is not uniform: it is the slowest around $W = 0$ and the kink at $W=1/\beta \log (1+e^{-\beta \Delta})$. Nevertheless, Fig.~\ref{Fig:work-extractionMTP} clearly shows that even a small size memory can significantly improve the quality of the extracted work.
 
\subsection{Cooling a two-level system using a two-dimensional memory with nontrivial Hamiltonian}

\begin{figure}[t]
    \centering
    \includegraphics{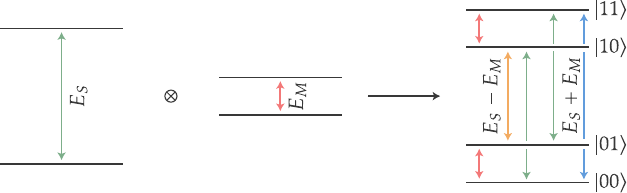}
     \caption{\emph{Energy level structure.} Schematic diagram of a two-level system, consisting of a main system and a memory with energy gaps $E_S$ and $E_M$, respectively. The energy gap of the composite system is such that it can be selectively coupled to the thermal bath. }
    \label{fig-cooling}
\end{figure}

As a second application of our findings, we consider the task of cooling a two-level system with the aid of a two-dimensional memory. The setup involves a two-level system with energy gap $E_S$, extended by a memory system with energy gap $E_M$. The joint system's energy level structure is depicted in Fig.~\ref{fig-cooling}.  We assume that the difference between energy gaps is such that it allows one to selectively couple with the bath, i.e., $E_S - E_M \neq E_M$. This enables us to separately address transitions $\ket{01}\leftrightarrow\ket{10}$, $\ket{00}\leftrightarrow\ket{11}$ together with two coupled pairs of the form $\ket{0i}\leftrightarrow\ket{1i}$ and $\ket{i0}\leftrightarrow\ket{i1}$. We will refer to these operations thermalise these levels as operation 1, 2, 3 and 4, respectively.


Let us now assume, for simplicity, that the system starts in an excited state extended by Gibbs memory, $\v{p}\otimes\v \gamma_M$ with $p_i = \delta_{i1}$. If we consider the main system alone with access only to MTPs, one can cool it down only to the ambient temperature. In this case, the system will reach thermal equilibrium, and the resulting distribution is given by
\begin{equation}
  \v \gamma_S = \qty(\frac{1}{1+e^{-\beta E_S}}, \frac{e^{-\beta E_S}}{1 + e^{-\beta E_S}}).  
\end{equation}
However, by implementing our protocol, which can be realised as a sequence of operations, $1\rightarrow2\rightarrow3\rightarrow4$ and discarding (thermalising) the memory, we arrive at \mbox{$\mathcal{P}(\v{p}\otimes\v \gamma_M) = \v{q}\otimes\v \gamma_M$} with
\begin{equation}
    \v{q} = \mqty[\frac{e^{\beta  E_M}+e^{\beta  \left(E_M+E_S\right)}+e^{\beta  \left(2 E_M+E_S\right)}+e^{\beta  \left(E_M+2 E_S\right)}+e^{\beta  E_S}}{\left(e^{\beta  E_S}+1\right) \left(e^{\beta 
   \left(E_M-E_S\right)}+1\right) \left(e^{\beta  \left(E_M+E_S\right)}+1\right)} \\
 \frac{e^{\beta  E_M}+e^{\beta  \left(2 E_M+E_S\right)}+e^{\beta  E_S}}{\left(e^{\beta  E_S}+1\right) \left(e^{\beta  E_M}+e^{\beta  E_S}\right) \left(e^{\beta 
   \left(E_M+E_S\right)}+1\right)}
    ].
\end{equation}
The distance of this state from the Gibbs state $\v{\gamma}_S$ at ambient temperature in terms of the 1-norm is given by
\begin{equation}
    \norm{\v{q} - \v \gamma_S}_1 = \frac{1}{\left(e^{-\beta  E_S}+1\right) \left[\cosh \left(\beta  E_M\right)+\cosh \left(\beta  E_S\right)\right]},
\end{equation}
which is positive for every non-zero value of $E_M$ and $E_S$. This means that despite non-triviality of the memory's spectrum, our simple memory-extended protocol achieves a cooling advantage over Markovian processes.


\subsection{Two-level control is sufficient for thermal operations}

In Ref.~\cite{Lostaglio2018elementarythermal} it has been proved that there exist thermodynamic state transformations that cannot be decomposed into the so-called \emph{elementary thermal operations}, i.e., thermal operations acting only on two levels of the system at the same time. Then, in Ref.~\cite{mazurek2018decomposability}, for any dimension $d$, an explicit final state $\v{q}\in \T_+^{\mathrm{TO}}(\v{p})$ was given such that it cannot be achieved (even approximately) starting from the ground state $\v{p}=(1,0,\dots,0)$ using convex combinations of sequences of elementary thermal operations. More precisely, given the energy spectrum of the system with $E_{i+1}\geq E_i$, this final state is given by
\begin{align}
    \label{eq:inaccessible}
    \v{q} & = \left(1 - \sum_{i=2}^{d}e^{-\beta E_i},    e^{-\beta E_2}, \dots, e^{-\beta E_{d}}\right)
\end{align}
with $\beta \geq \beta_{\text{crit}}$ such that $1 - \sum_{i=2}^{d}e^{-\beta_{\text{crit}} E_i} = 0$. It was then proven by the authors of Ref.~\cite{mazurek2018decomposability} that there exists $\epsilon>0$ such that any $\v{q}'$ achievable from $\v{p}$ satisfies $\delta(\v{q}, \v{q}') \geq \epsilon$. 

Given the above, one might conclude that being able to selectively couple to the bath just two energy levels at once is highly restrictive and does not allow one to induce all the transitions possible via general thermal operations. This conclusion, however, would be incorrect, as the restriction only arises when one is limited to coupling only two levels of the \emph{system} at a given time. When one is allowed to bring an auxiliary $N$-level system in a thermal equilibrium state $\gamma_M$, then the ability to selectively couple to the bath just two energy levels of the joint system allows one to induce all transitions of the main system possible via thermal operations as $N\to\infty$. Crucially, the operation 
\begin{equation}
    \E(\rho)=\rho\otimes \gamma_M
\end{equation}
is a thermal operation for every $N$. Thus, $\E$ followed by a sequence of elementary thermal operations on the joint system, followed by discarding the system $M$ at the end, can induce any energy-incoherent state transition of the system possible via general thermal operations. In other words, elementary control over two energy levels at a given time is sufficient to generate all thermodynamically possible transitions if we allow ancillary thermal systems.

\begin{figure}[t]
    \centering    
    \includegraphics{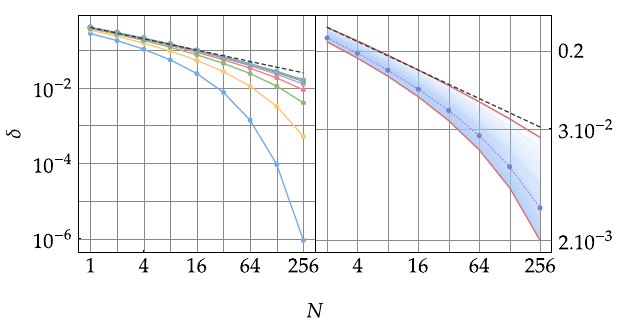}
    \caption{\emph{Convergence to states inaccessible via elementary thermal operations.} 
    Log-log plot of the total variation distance $\delta$ between the state $\v{q}$ from Eq.~\eqref{eq:inaccessible} and the state obtained from \mbox{$\v{p}=(1,0,\dots,0)$} via the algorithm $\widetilde{\mathcal{P}}^{ \Pi_{\v{q}}}$, as a function of the memory size~$N$. Left: systems with energy spectra $E_i=i$ with $d = 3,\hdots,10$ (bottom to top) and for \mbox{$\beta = 1.1 \log(2)>\beta_{\mathrm{crit}}$}. For all presented dimensions the convergence is better than $1/(2\sqrt{N})$ (black dashed line). Right: systems with energy spectra $(0,E_1,E_2,1)$ taken from a grid with interval $\Delta = 1/64$ (translucent blue lines) for $\beta = 1.1\cdot\beta_{\text{crit}}$. The red lines represent the extreme cases of convergence, while the thick blue line is the average convergence. Note that the top red line almost agrees with $1/(2\sqrt{N})$ (black dashed line), which well approximates the expected convergence for $\beta = 0$ and agrees with the fact that it is obtained for almost completely degenerate levels.}
    \label{fig:M_point_convergence}
\end{figure}

We illustrate the above with the following numerical examples, showing that our MeMTP protocol $\widetilde{\mathcal{P}}^{\Pi_{\v q}}$ (which consists of only two-level operations) is able to transform $\v{p}$ into $\v{q}'$ that approximates $\v{q}$ arbitrarily well (i.e., $\delta(\v{q}',\v{q})\to 0$ as $N\to\infty$). In order to focus attention, we chose a constant $\beta = 1.1\cdot \beta_{\text{crit}}$. First, we considered systems of varying dimension $d$, up to $d_{\text{max}} = 80$, and fixed the energy structure to $E_i = i$, corresponding to quantum harmonic oscillator. We observed that the convergence for all these dimensions scales according to the predictions from Conjecture~\ref{Conj:beta_permutations}, which can be seen in the left panel of Fig.~\ref{fig:M_point_convergence} for $d =3,\hdots,10$ and memory sizes up to $N = 2^8$. Moreover, in order to ascertain that the convergence does not depend on the energy structure of the system, we fixed $d = 4$ and considered energy levels $(0, E_1, E_2, 1)$ with $E_1 < E_2$ taken from a grid with spacing $\Delta E = 2^{-6}$, resulting in $1953$ uniformly distributed points. For each of these points, we have considered the protocol with memory size up to $N = 2^8$. As demonstrated in the right panel of Fig.~\ref{fig:M_point_convergence}, it turns out again that the convergence, independently from the energy structure, is better than $1/\sqrt{N}$, in accordance with Conjecture~\ref{Conj:beta_permutations}. 
\begin{figure*}
    \centering
    \includegraphics{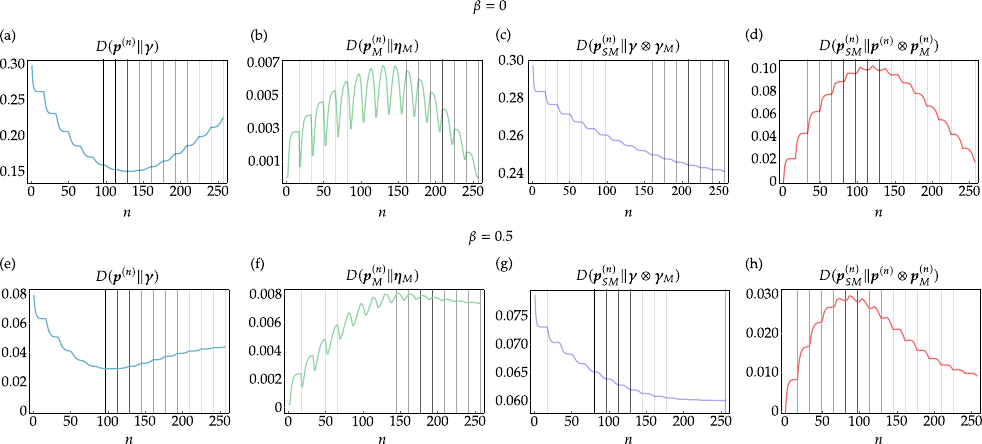}
    \caption{\emph{Evolution of non-equilibrium free energies and correlations}. Non-equilibrium free energies of the main system [(a) and (e)], the memory [(b) and (f)], and the joint system [(c) and (g)], as well as the mutual information between the system and memory [(d) and (h)], as a function of the step number $n$ of the protocol $\P^{(ij)}$. Here, the composite system consists of a three-level system initialised in a state $\v p = (0.7,0.2,0.1)$ and a 16-dimensional degenerate memory prepared in a thermal (maximally mixed) state, and the plots are presented for two inverse temperatures, $\beta =0$ and $\beta = 0.5$.     
    }   
    \label{Fig:Free-energy}
\end{figure*}

\subsection{Non-equilibrium free energy evolution}

To understand how non-Markovian effects arise in the memory-assisted protocol, we will now examine the evolution of the system and memory during the protocol $\mathcal{P}^{(ij)}$. More precisely, let us denote the joint state of the system and memory after the $n$-th two-level thermalisation step of the protocol by~$\v{p}_{SM}^{(n)}$. Similarly, let $\v{p}^{(n)}$ and $\v{p}_{M}^{(n)}$ denote the reduced states of the system and memory after the $n$-th step. Then, in the spirit of the analysis performed for elementary thermal operations in Ref.~\cite{Jeongrak2022}, we will examine the behaviour of the following entropic quantities. First, we will look at the relative entropy between $\v{p}^{(n)}$ and the thermal state of the system $\v \gamma$,
\begin{equation}
    D\left(\v p^{(n)}\|\v{\gamma}\right)= \sum^{d}_{i=1} p^{(n)}_i \log \frac{p^{(n)}_i}{\gamma_i},
\end{equation}
which is a thermodynamic monotone, as it decreases under (Markovian) thermal operations, and is directly related to the non-equilibrium free energy~\cite{brandao2013resource}. We will also look into the behaviour of the analogous quantities for the joint system and the memory system. Moreover, to track the correlations that build up between the system and memory, we will investigate the mutual information between them, which is given by \mbox{$D(\v{p}^{(n)}_{SM} \| \v{p}^{(n)} \otimes \v{p}^{(n)}_M)$}.

We use a three-level system and a $16$-dimensional memory as an illustrative example. We consider the joint system undergoing a $\beta$-swap protocol $\mathcal{P}^{(ij)}$ for $\beta=0$ and $\beta=0.5$. As shown in Figs.~\hyperref[Fig:Free-energy]{\ref{Fig:Free-energy}a} and~\hyperref[Fig:Free-energy]{\ref{Fig:Free-energy}e}, the non-equilibrium free energy of the main system initially decreases to a minimum, and then increases until it reaches a level that closely approximates the target state (a swap/$\beta$-swap). It is important to note that this observed increase is only possible because of the presence of the memory system. In contrast, note that the global non-equilibrium free energy decreases after each step, as depicted in Figs.~\hyperref[Fig:Free-energy]{\ref{Fig:Free-energy}c} and~\hyperref[Fig:Free-energy]{\ref{Fig:Free-energy}g}. However, during the process, a fraction of the main system's non-equilibrium free energy is transferred to the memory, which acts as a free energy storage. As such, it later enables the system to increase its local free energy again, hence allowing it to achieve the final state. More interestingly, the free energy of the memory, presented in Figs.~\hyperref[Fig:Free-energy]{\ref{Fig:Free-energy}b} and~\hyperref[Fig:Free-energy]{\ref{Fig:Free-energy}f}, exhibits a comb-like structure consisting of $d-1$ teeth with $(d+1)$ steps each. Specifically, within the $k$-th tooth, the first $d-k$ steps increase the free energy, while the remaining $k$ steps decrease it. Note that as long as the memory is not thermalised, its non-equilibrium free energy does not go to zero. However, for $\beta = 0$, it approaches a value very close to zero, but there are still correlations between the memory and the system, which are illustrated in Figs.~\hyperref[Fig:Free-energy]{\ref{Fig:Free-energy}d} and~\hyperref[Fig:Free-energy]{\ref{Fig:Free-energy}e}. 

\section{Alternative protocols realising equivalent $\beta$-swap approximation}\label{App:protocols}

Let us revisit the truncated protocol $\tilde{\mathcal{P}}^{(ij)}$ with $d$-dimensional memory as introduced in Section~\ref{sec:bridging}, where it was defined in Eq.~\eqref{eq:trunc_protocol}. This protocol is composed of $d^2$ two-level elementary thermalisations, denoted as $T_{kl}$. Each thermalisation $T_{kl}$ can be represented as a point $(k,l)$ on a plane, and an algorithm can be represented as an arrow pointing from the previous thermalisation to the next one. For instance, we present in Fig.~\hyperref[Fig:algorithms]{\ref{Fig:algorithms}a} the diagram of $\tilde{\mathcal{P}}$ for $d = 7$.  

\begin{figure}[b]
    \centering
    \includegraphics[width=10.629cm]{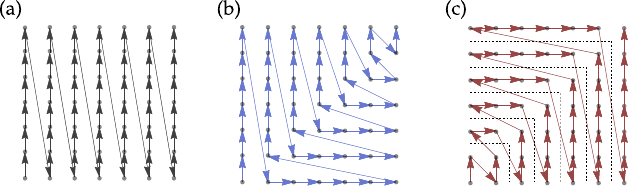}
    \caption{\emph{Alternative protocols}.
    Three different protocols are used to achieve equivalent $\beta$-swap approximations. Each thermalisation is represented by a point on a plane, with an arrow indicating the transition from one thermalisation to the next.}
    \label{Fig:algorithms}
\end{figure}

Visually, it is obvious that we iterate through an entire column before shifting to the next one. In other words, all the `filled' levels are used sequentially to fill up the first `empty' level, and the same process is repeated for all the subsequent `empty' levels. Now, let us look at the two algorithms depicted in Fig.~\hyperref[Fig:algorithms]{\ref{Fig:algorithms}b} and Fig.~\hyperref[Fig:algorithms]{\ref{Fig:algorithms}c}. The action of the blue algorithm can be summarised in the following way, starting with $i = 1$:
\begin{enumerate}
    \item Iterate through $i$-th column, starting from the first unvisited point.
    \item Iterate through $i$-th row, starting from the first unvisited point.
    \item Set $i\rightarrow i+1$ and go back to step 1 if any unvisited point remains.
\end{enumerate}
The red algorithm can be most easily understood as the reverse of the blue one. Instead of decreasing the length of vertical and horizontal stretches, they are gradually increased in the red algorithm. This allows the red algorithm to be recursively implemented, taking into account gradually more and more levels of memory, as indicated by dashed lines. The protocol for $d$-dimensional memory is implemented by extending the $d-1$-dimensional version with an additional row and column.

We furthermore investigated the cyan family, which mimics the blue algorithm and can be defined in the following manner:
\begin{enumerate}
    \item Iterate through row or column, starting from the first unvisited point.
    \item Repeat step 1. until no unvisited points remain.
\end{enumerate}
Moreover, we considered orange family related to the cyan family with an analogous reversal as between blue and red:
\begin{figure}[h]
    \includegraphics[width=10.7cm]{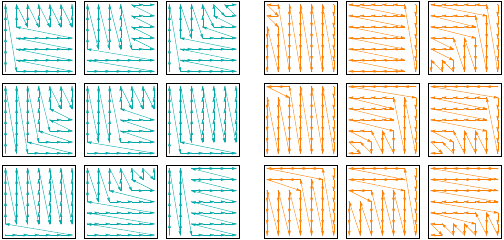}
\end{figure}

For these algorithms, we find the following properties we have observed from explicit implementation for a range of dimensions and inverse temperatures $\beta$, but we have not been able to prove them analytically:
\begin{itemize}
    \item All of the aforementioned algorithms acting on an initial state $\v{p}\otimes \v \gamma_d$ result in the same state as $\tilde{\mathcal{P}}(\v{p}\otimes \v \gamma_d)$.
    \item blue and red algorithms are slowest and fastest algorithms, respective, according to the convergence to the $\beta$-swap with respect to the 1-norm.
    \item Each algorithm in the Cyan and Orange family provide slower and faster convergence than the original algorithm $\tilde{\mathcal{P}}$, respectively.

The statement reinforces the claims mentioned earlier and explains that the 1-norm between the intermediate states of the system and the target state (the $\beta$-swapped counterpart of $\v{p}$ with $\beta = 0$) is plotted in Fig.~\ref{fig:convergence}. 

\end{itemize}
\begin{figure}[H]
    \centering
    \includegraphics{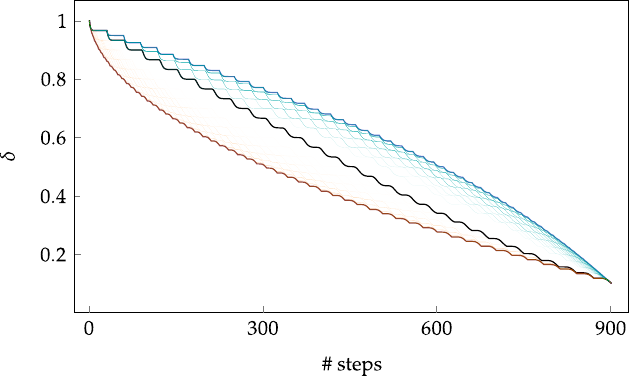}
    \caption{\emph{Algorithm convergence analysis}. Convergence to the target state $\v{q} = (0,1)$ of different algorithms acting on the state $\v{p} = (1,0)$ extended by memory with dimension $d = 30$. Note that all algorithms. Note that all algorithms finish at the same value of $\norm{\v{p} - \v{q}}_1$.}
    \label{fig:convergence}
\end{figure}

\newpage
\section{Derivation of the results}

\subsection{Regularised incomplete beta function} \label{app:beta}

The content of this section is based on Refs.~\cite{artin2015gamma,NIST:DLMF}.

\begin{center}
\emph{\textbf{Definition and properties}}
\end{center}

The finite-size corrections to a $\beta$-swap and its compositions are determined by the cumulative distribution function (CDF) known as the regularised incomplete beta function. This function is closely related to the well-known beta function $B(a,b)$ and is widely used in deriving our results. To provide the necessary background and present its key properties, we first recall the definition and properties of the beta function:
\begin{equation}\label{Eq:beta_function}
B(a,b) = \int_{0}^{1} t^{a-1} (1-t)^{b-1} dt ,
\end{equation}
with $a,b \in \mathbbm{C}$. The beta function relates to the gamma function in the following way
\begin{equation}\label{Eq:beta_and_gamma_function}
B(a,b) = \frac{\Gamma(a)\Gamma(b)}{\Gamma(a+b)}.
\end{equation}

The incomplete beta function $B_x(a,b)$ is defined by changing the upper limit of integration in Eq.~\eqref{Eq:beta_function} to an arbitrary variable, i.e.,
\begin{equation}\label{Eq:incomplete_beta_function}
    B_x(a,b) = \int_{0}^{x} t^{a-1} (1-t)^{b-1} dt.
\end{equation}
\begin{marginfigure}[-5.0cm]
	\includegraphics[width=4.754cm]{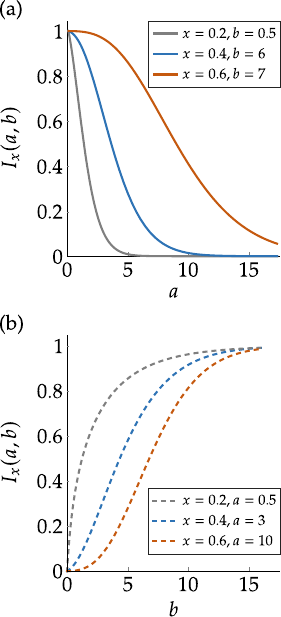}
    \caption{\emph{Regularised beta function}. Plots of the regularised incomplete beta function as a function of (a) $a$ and (b) $b$. }
	\label{Fig:regularised_beta_function}
\end{marginfigure}
Finally, we define the regularised incomplete beta function $I_x(a,b)$ (regularised beta function for short) by normalising the incomplete beta function,
\begin{equation}
    I_x(a,b) = \frac{B_x(a,b)}{B(a,b)}.
\end{equation}
We present plots of the regularised beta function for a few selected values of $x$ in Fig.~\eqref{Fig:regularised_beta_function}.

Throughout this chapter, we assume that $a, b >0$ and \mbox{$0\leq x \leq 1$}. It is easily noted that $I_0(a,b) = 0$, $I_1(a,b) = 1$, and $I_0(a,b) \leq I_x(a,b) \leq I_1(a,b)$, thus making it a proper CDF. Furthermore, for $a,\,b\in\mathbb{Z}$, $I_x(a,b)$ can be written in terms of a binomial function
\begin{equation}\label{Eq:incomplete_beta_function_binomial}
    I_x(a,b) = (1-x)^b\sum_{j=a}^{\infty}\binom{b+j-1}{j}x^{j}.
\end{equation}
From this equation, by using the geometric series and its derivatives, one can conclude that
\begin{equation}
\label{Eq:incomplete_beta_function_binomial_a_0}
    I_x(0,b) = 1.
\end{equation}
Moreover,
\begin{equation}\label{Eq:incompleta-beta-function-difference}
(1-x)^b\sum_{j=a}^{n}\binom{b+j-1}{j}x^{j} = I_x(a,b)-I_x(n+1,b).
\end{equation}
If $a =1$, then $I_x(1,b) = 1 - (1-x)^b$ and Eq.~\eqref{Eq:incomplete_beta_function_binomial} simplifies to 
\begin{equation}
(1-x)^b\sum_{j=1}^{n}\binom{b+j-1}{j}x^{j} = 1 - (1-x)^b - I_x(n+1,b).
\end{equation}
The next two useful properties of $I_x(a,b)$ are the symmetry relation 
\begin{equation}\label{Eq:incomplete_beta_function_symmetry}
    I_x(a,b) = 1- I_{1-x}(b,a),
\end{equation}
and the relation for an equal argument,
\begin{equation}\label{Eq:incomplete_beta_function_recurrence}
    I_x(a,a) = \frac{1}{2}I_{4x(1-x)}\qty(a,\frac{1}{2}),
\end{equation}
when $0\leq x \leq 1/2$.
Finally, there are two recurrence relations which allow one to shift either of the arguments of the function by one,
\begin{subequations} 
\begin{align}\label{Eq:incomplete_beta_function_recurrence2}
    I_x(a,b) &= I_x(a+1,b)+ \frac{x^a(1-x)^b}{a B(a,b)}, \\ \label{Eq:incomplete_beta_function_recurrence3}
    I_x(a,b) &= I_x(a,b+1)- \frac{x^a(1-x)^b}{a B(a,b)}, \\
    I_x(a,b) &= I_x(a+1,b-1)+\frac{x^a (1-x)^{b-1}}{aB(a,b)}  \label{Eq:incomplete_beta_function_recurrence4}, \\
     I_x(a,b) &= I_x(a-1,b+1)-\frac{x^{a-1} (1-x)^{b}}{aB(a,b)}.  \label{Eq:incomplete_beta_function_recurrence5}
\end{align}
\end{subequations}


\begin{center}
\emph{\textbf{General relations involving the gamma function}}
\end{center}

Next, let us present general relations and properties of the gamma function that are extensively used in our proofs. First, recall that for every positive integer $n$
\begin{equation}
    \Gamma(n) = (n-1)!. 
\end{equation}
Using the above, one can derive a simple formula for the following expression that appears when dealing with the regularised beta function:
\begin{equation}\label{Eq:beta-function-expansion-binomial}
    \frac{1}{n B(n,n)} = \frac{\Gamma(2n)}{n \Gamma(n)^2} = \frac{\Gamma(2n)}{n \Gamma(n)^2} \frac{2n}{2n} = \frac{1}{2}\binom{2n}{n}.
\end{equation}
Other important functional equation for the gamma function is the Legendre duplication formula:
\begin{equation}\label{Eq:Gamma-function-legendre}
\Gamma(n)\Gamma\qty(n+\frac{1}{2}) = 2^{1-2n}\sqrt{\pi}\Gamma(2n).    
\end{equation}
The factorial terms can be approximated using Stirling's approximation
\begin{equation}\label{Eq:Stirlings-approximation}
n! = \sqrt{2\pi n}\qty(\frac{n}{e})^n\left(1+O(n^{-1})\right).
\end{equation}


\begin{center}
\emph{\textbf{Asymptotic analysis}}
\end{center}

We will be interested in the asymptotic behaviour of the regularised beta function. So, in this section, we introduce important relations and identities that will be useful for proving our main theorems. Let us begin by considering the regularised beta function $I_x(a,b)$, where $x$ and $b$ fixed. For $a \to \infty$, we have the following asymptotic expansion:
\begin{align}\label{Eq:regularised-beta-asymptotic}
I_x(a,b) = \Gamma(a+b)x^a(1-x)^{b-1}&\Biggr[\sum_{k=0}^{n-1}\frac{1}{\Gamma(a+k+1)\Gamma(b-k)}\qty(\frac{x}{1-x})^k \nonumber\\ &\hspace{1.2cm}+ O\qty(\frac{1}{\Gamma(a+n+1)})\Biggr].
\end{align}
Note that $O$-term vanishes in the limit only if $n \geq b$. Furthermore, for each $n=0,1,2, ...$ If $b = 1, 2, 3, ...$, and $n>b$,  the $O$-term can be omitted, as the result is exact. In this work, Eq.~\eqref{Eq:regularised-beta-asymptotic} will be expanded up to the second order. Specifically, for the values of $a = N, b= 1/2$, we have the following equation
\begin{align}\label{Eq:regularised-beta-expansion-12}
   \! I_{x}\qty(N,\frac{1}{2}) &\simeq\! \frac{x^N}{\sqrt{1-x}}\qty[\frac{\Gamma\qty(N+\frac{1}{2})}{\Gamma\qty(N+1) \Gamma\qty(\frac{1}{2})}+\frac{\Gamma\qty(N+\frac{1}{2})}{\Gamma\qty(N+2) \Gamma\qty(-\frac{1}{2})}\qty(\frac{x}{1-x})] \nonumber \\
    &=\!\frac{x^N}{\sqrt{1-x}}\frac{(2N)!}{4^N (N!)^2}\qty[1-\frac{x}{2(N+1)(1-x)}],
\end{align}
where $\simeq$ symbol hides the terms of the order $O(x^N/N^2)$.

Next, we will consider sums of the regularised beta functions over the second argument and their limit as \mbox{$a \to \infty$}. We start with
\begin{align}
    \sum_{i=1}^{a}I_x(a,i+1) &= \sum_{i=1}^a \frac{B_x(a,i+1)\Gamma(a+1+i)}{\Gamma(a)\Gamma(i+1)} \nonumber \\
    &= \sum_{i=1}^a \frac{(a+i)!}{(a-1)! i!}\int_{0}^x dt \: t^{a-1}(1-t)^i  \nonumber \\
    &= \int_{0}^x dt \sum_{i=1}^a \binom{a+i}{i}t^{a-1}(1-t)^i,
\end{align}
where we have used definitions of the beta function and the gamma function for integer arguments. Using Eq.~\eqref{Eq:incompleta-beta-function-difference}, the above expression can be recast as
\begin{align}
 \sum_{i=1}^{a}I_x(a,i+1) &= \int_{0}^x \frac{dt}{t^2}\:[I_x(1,a-1)-I_{1-x}(a-1,a-1)] \nonumber \\
 &= \int_{0}^x \frac{dt}{t^2}\:\qty[\frac{1}{2}I_{4x(1-x)}\qty(a-1,\frac{1}{2})-I_x(a-1,1)],
\end{align}
for $0\leq x \leq 1/2$.

Finally, using the asymptotic expansion from~Eq.~\eqref{Eq:regularised-beta-asymptotic}, we get
\begin{align}
    \int_{0}^x dt\: \Biggl\{-t^{a-3}+\frac{1}{2x^2}\Biggl[&\frac{\Gamma(a-1/2)}{\Gamma(a+1)}\frac{[4x(1-x)]^{a-1}}{\sqrt{1-4x(1-x)}} + \nonumber\\ &\hspace{1.1cm} O\qty(\frac{(4x(1-x))^{a-1}\Gamma(a-1/2)}{\Gamma(a+1)}) \Biggl]\Biggl\}.
\end{align}
Thus, we see that this term vanishes in the limit of $a \to \infty$, and we find that the following sum vanishes in the limit:
\begin{equation}\label{Eq:beta-regularised-limit-sum}
    \lim_{a\to\infty} \sum_{i=1}^{a}I_x(a,i+1) = 0. 
\end{equation}

Moreover, we will need the following sum:
\begin{align}
\sum_{i=1}^N &I_{x}(N+1-i,N)= \sum_{i=1}^N\frac{B_{x}(N+1-i,N)}{B(N+1-i,N)}\nonumber\\
&=\sum_{i=1}^N\frac{N\Gamma(2N+1-i)}{N\Gamma(N+1-i)\Gamma(N)}\int_{0}^{x}dt\, t^{N-i}(1-t)^{N-1} \times \nonumber\\ 
&\hspace{1.7cm}\times \int_{0}^{x}dt N(1-t)^{N-1}\sum_{i=1}^N\binom{2N-i}{N-i}t^{N-i}\nonumber\\
&= \int_{0}^{x}dt N\frac{(1-t)^{N-1}}{(1-t)^{N+1}}(1-t)^{N+1}\sum_{j=0}^{N-1}\binom{(N+1)+j-1}{j}t^{j}\nonumber\\
&=\int_{0}^{x}dt \frac{N}{(1-t)^2}\Bigg(I_{1-x}(0,N+1)-I_{1-x}(N,N+1)\Bigg)\nonumber\\
&= \int_{0}^{x}dt \frac{N}{(1-t)^2}\Bigg(1-I_{1-x}(N,N)-\frac{t^N(1-t)^N}{NB(N,N)}\Bigg)\nonumber\\
&= N \int_{0}^{x}dt\frac{1}{(1-t)^2}=N \frac{x}{1-x}. \label{Eq:regularised-beta-function-sum-2}
\end{align}


\subsection{Proofs of Lemma~\ref{Lem:transposition} and Theorem~\ref{Thm:beta-swap}}~\label{app:neighbour}

To prove Lemma~\ref{Lem:transposition} and Theorem~\ref{Thm:beta-swap}, we will consider a composite system consisting of the main $d$-dimensional system and an $N$-dimensional memory system. Without loss of generality, we can assume that the main system is a two-level system with $i= 1$ and $j = 2$, whose state is described by unnormalised probability vector. We begin by deriving an expression that describes how the composite system evolves under the memory-assisted protocol $\mathcal{P}^{(12)}$. Next, to gain insight into the behaviour of the joint system as the memory size grows, we will prove the asymptotic result. Finally, in the last subsection, we will show how this result implies the desired convergence rates.


\begin{center}
    \emph{\textbf{Dynamics induced by two-level thermalisations}}
\end{center}

Consider a two-level system, described by a Hamiltonian $H = E_1 \ketbra{E_1}{E_1} + E_2 \ketbra{E_2}{E_2}$ and initially prepared in an energy-incoherent state \mbox{$\v p = (b,c)$}, together with a memory system, described by a trivial Hamiltonian $H_M = 0$ and prepared in a maximally mixed state $\v{\eta}_M = \qty(1, \hdots, 1)/N$. The joint state $\v r:= \v p \otimes \v \eta_M$ of the composite system is then given by
\begin{equation}\label{Eq:joint-state}
	\v{r}^{(0)} \equiv\v{r} = \frac{1}{N}\Big(\underbrace{b,\hdots,b}_{N\, \text{times}}|\underbrace{c ,\hdots,c }_{N\, \text{times}}\Big),
\end{equation}
whereas the joint thermal distribution \mbox{$\v \Gamma := \v \gamma \otimes \v \eta_M$} is given by
\begin{align}
   \v \Gamma &= \frac{1}{N(e^{-\beta E_1}+ e^{-\beta E_2})}\qty[e^{-\beta E_1}, \cdots , e^{-\beta E_1} |e^{-\beta E_2}, \cdots , e^{-\beta E_2}] \nonumber \\ &=\frac{1}{N}\qty(\gamma_1, ..., \gamma_1 | \gamma_2, ..., \gamma_2). 
\end{align}
For the sake of brevity, we introduce the rescaled Gibbs factors, which will be used extensively in the proofs:
\begin{equation}\label{Eq:rescaled-Gibbs-factor}
    \Gamma_{ij} =\frac{\gamma_i}{\gamma_i+\gamma_j}.
\end{equation}

Under a series of two-level thermalisations, the joint state of the composite system at the $k$-th round is given by
\begin{equation}
    \!\!\! \mathcal{R}_k\qty(\v{r}^{(k-1)})\equiv \v{r}^{(k)} \! =\! \qty(b_1^{(k)},\hdots, b_N^{(k)}\,|\,c_1^{(N)},\hdots,c^{(N)}_{k},c,\hdots,c),
\end{equation}
where $b^{(k)}_j$ and $c^{(N)}_j$ satisfy the following recurrence relations
\begin{align}\label{Eq:recurrence-relation-entries_bj}
b^{(k)}_{j} &= \Gamma^{j}_{21}\,\qty(\frac{\Gamma_{12}}{\Gamma_{21}}c+\Gamma_{12}\sum_{i=1}^j b^{(k-1)}_i \Gamma^{\, -i}_{21}), \\
c^{(N)}_j &= \Gamma_{21}^{\, N}c+\Gamma_{21}^{\, N+1}\sum_{i=1}^N b^{(j-1)}_i \Gamma^{\, -i}_{21},
\label{Eq:recurrence-relation-entries_cj}
\end{align}
with $b^{(0)}_j = b$. Equations \eqref{Eq:recurrence-relation-entries_bj} and \eqref{Eq:recurrence-relation-entries_cj} can be understood by noting that during the $k$-th round of two-level thermalisations, an additional $c$ is added to the previous $(k-1)$-th entry and the resulting state is again thermalised. By iterating this process for the first $k$ rounds, we can derive a closed-form expression for the entries $b^{(k)}_j$ and $c^{(N)}_j$:
\begin{align}\label{Eq:bj_entries}
\!\!\! b_j^{(k)} &=\Gamma_{12}\Gamma_{21}^{j-1}c\sum_{i=0}^{k-1}\binom{j\! +\! i\! -\! 1}{i}\Gamma_{12}^{i}+b\Gamma_{12}^{k}\sum_{i=1}^{j}\binom{j\! +\! k\! -\! 1\! -\! i}{k\! -\! 1}\Gamma_{21}^{j-i}, \\
\!\!\! c^{(N)}_j &= c\Gamma^N_{21} \sum_{i=0}^{j-1}\binom{N\! +\! i\! -\! 1}{i}\Gamma_{12}^{i}+b\Gamma_{21}\Gamma^j_{12}\sum_{i=0}^{N-1} \binom{i\! +\!  j}{j}\Gamma_{21}^{i}.\label{Eq:cj_entries}
\end{align}
After $N^2$ rounds of two-level thermalisations, the composite final state $\v r^{(N)}$ is given by
\begin{align}
	   \widetilde{\mathcal{P}}^{(12)}(\v r) \equiv\v r^{(N)} = \v q \otimes \v \eta_M =\frac{1}{N}{\small\qty[b^{(N)}_1,\hdots,b^{(N)}_N \, \bigg| \, c^{(N)}_{1},\hdots, c^{(N)}_{N}]}.
\end{align}

\begin{center}
\textbf{\emph{Infinite memory limit}}
\label{app:limit_proofs}
\end{center}

Up till now, we have obtained a closed-form expression describing the action of the truncated protocol $\widetilde{\mathcal{P}}^{(12)}$. In order to prove Lemma~\ref{Lem:transposition}~and~Theorem~\ref{Thm:beta-swap}, the next step consists of decoupling the main system from memory using a full thermalisation $\mathcal{T}$:
\begin{equation}
    \mathcal{P}^{(12)}(\v p \otimes \v \gamma_M) = \mathcal{T} \circ \widetilde{\mathcal{P}}^{(12)}(\v p \otimes \v \gamma_M). 
\end{equation}
Since we are initially focused on demonstrating the asymptotic results, our aim is to show that as $N$ approaches infinity, we achieve a $\beta$-swap:
\begin{subequations}
\begin{align}
    \lim_{N\to \infty} \frac{\sum_{i=1}^{N} b^{(N)}_i}{N} &= c+b(1-e^{-\beta (E_2-E_1)}),\label{Eq:limit_beta_swap1}\\ \lim_{N\to \infty} \,\frac{\sum_{i=1}^{N} c_i^{(N)}}{N} &= be^{-\beta (E_2-E_1)} .\label{Eq:limit_beta_swap2}
\end{align}
\end{subequations}
To prove the limits in Eqs.~\eqref{Eq:limit_beta_swap1} and \eqref{Eq:limit_beta_swap2}, we will begin by using the conservation of probability,
\begin{equation}\label{Eq:simple_fact}
\lim_{N\to \infty} \sum_{i=1}^{N}\frac{ b^{(N)}_i}{N}+ \lim_{N\to \infty}\sum_{i=1}^{N}\frac{ c_{i}^{(N)}}{N} = b+c.   
\end{equation}
This allows us to express the first limit as a function of the second:
\begin{equation}\label{Eq:first_limit}
    \lim_{N\to \infty} \sum_{i=1}^{N}\frac{ b^{(N)}_i}{N} = b+c - \lim_{N\to \infty}\sum_{i=1}^{N}\frac{ c^{(N)}_{i}}{N} .
\end{equation}

We will now examine the non-trivial term on the right hand side of Eq.~\eqref{Eq:first_limit}. We start by using Eq.~\eqref{Eq:incompleta-beta-function-difference} to re-write the first term of $c^{(N)}_j$ appearing in Eq.~\eqref{Eq:cj_entries} in terms of the regularised beta function
\begin{align}
\eval{c^{(N)}_j}_{\substack{b = 0\\c=1}} = \Gamma^N_{21} \sum_{i=0}^{j-1}\binom{N+i-1}{i}\Gamma_{12}^{i} &= [I_{\Gamma_{12}}(0,N)-I_{\Gamma_{12}}(j,N)] \nonumber \\ &= I_{\Gamma_{21}}(N,j),
\end{align}
where we also used Eq.~\eqref{Eq:incomplete_beta_function_symmetry} to invert the arguments of the regularised beta function. Thus, using \eqref{Eq:beta-regularised-limit-sum}, we conclude that
\begin{equation}
\lim_{N\to\infty}\frac{1}{N}\sum_{i=1}^{N}\eval{c^{(N)}_j}_{\substack{b = 0\\c=1}} = \lim_{N\to\infty}\frac{1}{N}\sum_{i=1}^{N}I_{\Gamma_{21}}(N,i) = 0.    
\end{equation}
Taking into account the second term of $c^{(N)}_j$ appearing in Eq.~\eqref{Eq:cj_entries}, one can immediately evaluate the sum
\begin{align}
   \lim_{N\rightarrow\infty}\eval{c^{(N)}_j}_{\substack{b = 1\\c=0}}  = \Gamma_{21}\Gamma^j_{12}\sum_{i=0}^{\infty} \binom{i+k}{k}\Gamma_{21}^{i} = \frac{\Gamma_{21}}{\Gamma_{12}},
\end{align}
and therefore obtain that
\begin{equation}
\lim_{N\to \infty} \,\frac{\sum_{i=1}^{N} c^{(N)}_i}{N} = be^{-\beta (E_2-E_1)}.    
\end{equation}
Substituting this result to Eq.~\eqref{Eq:first_limit}, we prove that
\begin{equation}
 \lim_{N\to \infty} \frac{\sum_{i=1}^{N} b^{(N)}_i}{N} = c+b[1-e^{-\beta (E_2 - E_1)}].
\end{equation}
As a result, in the limit of $N\to\infty$, the protocol $\mathcal{P}^{(12)}$ achieves a $\beta$-swap. 


\begin{center}
    \emph{\textbf{Finite memory convergence rates}}\label{app:convergence_proofs}
\end{center}

To complete the proofs of Lemma~\ref{Lem:transposition} and Theorem~\ref{Thm:beta-swap}, we will now analyse what the approximation error is for a finite size of the memory $N$. This will tell us how quickly the initial state convergences to a $\beta$-swap as a function of $N$. First, we define two functions governing the convergence
\begin{subequations}
\begin{align}\label{Eq:error_function1}
\!\!\! \mathbb{E}:=& \frac{1}{N} \sum_{i = 1}^N\eval{c^{(N)}_{i}}_{c=1,b=0} \!\! = \frac{\Gamma^N_{21}}{N}\sum_{i=1}^{N}\sum_{j=0}^{i-1} \binom{N\! +\! j\! -\! 1}{j}\Gamma^j_{12} , \\
\!\!\! \mathbb{F}:=& \frac{1}{N} \sum_{j = 1}^N\eval{b_j^{(N)}}_{c=0,b=1  } \!\! = \frac{\Gamma^N_{12}}{N}\sum_{j=1}^N \sum_{i=1 }^{j}\binom{j\! +\! N\! -\! 1\! -\! i}{N\! -\! 1}\Gamma^{j-i}_{21}\label{Eq:error_function2},\!
\end{align}
\end{subequations}
which allow us to write the final state of the system as

\begin{equation}
\label{eq:finalstate}
    \v{q} = b\mqty(\mathbb{F} \\ 1-\mathbb{F} ) + c\mqty(1-\mathbb{E}\\\mathbb{E}).
\end{equation}

Next, we will asymptotically expand Eq.~\eqref{Eq:error_function1}. The starting point is to reduce the double sum into one, then convert the binomial sums into regularised beta functions and use its properties given by Eqs.~\eqref{Eq:incomplete_beta_function_recurrence4},~\eqref{Eq:incomplete_beta_function_recurrence} and~\eqref{Eq:beta-function-expansion-binomial}:
\begin{align}
\mathbb{E}&= \frac{\Gamma^N_{21}}{N}\sum_{i=1}^N(N-i+1) \binom{N+i-2}{i-1}\Gamma_{12}^{i-1} \nonumber \\&= \Gamma^N_{21}\sum_{i=0}^{N-1}\binom{N+i-1}{i}\Gamma_{12}^{i} - \frac{\Gamma^{N}}{N}\sum_{i=1}^{N-1} i\: \binom{N+i-1}{i}\Gamma_{12}^{i} \nonumber \\
&= I_{\Gamma_{21}}(N,N)-\frac{\Gamma_{12}}{\Gamma_{21}}I_{\Gamma_{21}}(N+1,N-1)\nonumber \\ &=\qty(1-\frac{\Gamma_{12}}{\Gamma_{21}})I_{\Gamma_{21}}(N,N)+\frac{\Gamma_{12}}{\Gamma_{21}}\frac{\Gamma_{21}^N \Gamma_{12}^{N-1}}{NB(N,N)} \nonumber\\
&=\frac{1}{\Gamma_{21}}\qty[(\Gamma_{21}-\Gamma_{12})\frac{1}{2}I_{4\Gamma_{21}\Gamma_{12}}\qty(N,\frac{1}{2})+\frac{(\Gamma_{21} \Gamma_{12})^{N}}{N}\frac{\Gamma(2N)}{\Gamma(N)^2}] \label{Eq:E_function_before_exp}.
\end{align}

Now we expand Eq.~\eqref{Eq:E_function_before_exp} up to second order using Eq.~\eqref{Eq:regularised-beta-expansion-12}. Recall that such an expansion is an approximation up to terms of the order $O(x^N/N^2)$ (where $x$ will be given by \mbox{$4\Gamma_{21}\Gamma_{12}$}), which will be dropped since our final approximation will be up to the order $o(x^N/N^{3/2})$ or $o(1/N^{1/2})$ for finite and infinite temperatures, respectively. Then, the gamma functions appearing in the expansions are simplified using Eq.~\eqref{Eq:Gamma-function-legendre} and~Eq.~\eqref{Eq:beta-function-expansion-binomial}. Finally, using Stirling's approximation [Eq.~\eqref{Eq:Stirlings-approximation}], we arrive at 
\begin{align}\label{Eq:convergence_E1}
    \mathbbm{E} = & \frac{(4\Gamma_{21}\Gamma_{12})^N}{2\Gamma_{21}\sqrt{\pi N}}\Bigg[\frac{(\Gamma_{21}-\Gamma_{12})}{\sqrt{1-4\Gamma_{21}\Gamma_{12}}}-\frac{(\Gamma_{21}-\Gamma_{12})(4\Gamma_{21}\Gamma_{12})}{2(1-4\Gamma_{21}\Gamma_{12})^{3/2}(N+1)} \nonumber\\    &\hspace{3.5cm}+1+o\left(N^{-1}\right)\Bigg][1+O(N^{-1})].
\end{align}
As the last step, we simplify the above expression for $\Gamma_{12}>1/2$ (finite temperature case) by using the fact that \mbox{$\Gamma_{12}=1-\Gamma_{21}$}, to arrive at

\begin{equation}
    \!\!\!\! \!\!\eval{\mathbb{E}}_{\Gamma_{12}>\frac{1}{2}} \!\!\!\!\!\!\! =  (4\Gamma_{12}\Gamma_{21})^{N}\left[\frac{\Gamma_{12}}{(\Gamma_{12}\! -\!\Gamma_{21})^2\sqrt{\pi N}(N\! +\! 1)}+o\left( N^{-3/2}\right)\right].
\end{equation}
The infinite-temperature limit, $\beta =0$, is similarly analysed by using Eq.~\eqref{Eq:E_function_before_exp} and plugging $\Gamma_{12} = \Gamma_{21} = 1/2$. In this case, the scaling is slightly different and is given by
\begin{equation}
    \eval{\mathbb{E}}_{\Gamma_{12} = \frac{1}{2}} = \frac{1}{\sqrt{\pi N}}+o\left(N^{-1/2}\right).
\end{equation}

The other half of estimating the convergence rate stems from considering Eq.~\eqref{Eq:error_function2}. Proceeding in the same manner as before, we re-write it as
\begin{align}
\mathbb{F} &= \frac{\Gamma^N_{12}}{N}\sum_{k=1}^N(N-k+1)\binom{N+k-2}{N-1}\Gamma_{21}^{k-1} \nonumber\\
&=\Gamma^N_{12}\sum_{k=0}^{N-1}\binom{N+k-1}{k}\Gamma_{21}^{k} -\frac{\Gamma^N_{12}}{N}\sum_{k=1}^{N-1} k \: \binom{N+k-1}{k}\Gamma_{21}^{k} \nonumber \\
&= [1-I_{\Gamma_{21}}(N,N)]-\frac{\Gamma_{21}}{\Gamma_{12}}[1-I_{\Gamma_{21}}(N-1,N+1)] \nonumber\\
&=[1-I_{\Gamma_{21}}(N,N)]-\frac{\Gamma_{21}}{\Gamma_{12}}\qty[1-I_{\Gamma_{21}}(N,N)-\frac{(\Gamma_{21}\Gamma_{12})^N}{\Gamma_{21}NB(N,N)}] \nonumber \\
&=\frac{(\Gamma_{12}-\Gamma_{21})}{\Gamma_{12}}\qty[1-I_{4\Gamma_{12}\Gamma_{12}}\qty(N,\frac{1}{2})]+\frac{(4\Gamma_{21}\Gamma_{12})^N}{2\Gamma_{12}}\frac{1}{\sqrt{\pi N}}.\label{Eq:error_function4}
\end{align}
Again, we expand Eq.~\eqref{Eq:error_function4} up to second order using Eq.~\eqref{Eq:regularised-beta-expansion-12}. Then, the gamma functions appearing in the expansions are simplified using Eq.~\eqref{Eq:Gamma-function-legendre}, the remaining ones are simplified by using Eq.~\eqref{Eq:beta-function-expansion-binomial}. Finally, using Stirling's approximation [Eq.~\eqref{Eq:Stirlings-approximation}], we arrive at 
\begin{align}\label{Eq:convergence_F}
    \mathbbm{F} = & \qty(\frac{\Gamma_{12}-\Gamma_{21}}{2\Gamma_{12}}) \Biggl[ 2-\frac{(4\Gamma_{21}\Gamma_{12})^N}{\sqrt{1-4\Gamma_{21}\Gamma_{12}}}\frac{1}{\sqrt{\pi N}} \qty(1-\frac{4\Gamma_{21}\Gamma_{12}}{2(N+1)(1-4\Gamma_{21}\Gamma_{12})})\Biggr] \nonumber\\
    &\hspace{4.5cm}+\frac{(4\Gamma_{21}\Gamma_{12})^N}{2\Gamma_{12}}\left[\frac{1}{\sqrt{\pi N}}+o(N^{-3/2})\right].
\end{align}
For $\Gamma_{12}>1/2$ (finite temperature case), we use the fact that \mbox{$\Gamma_{12}=1-\Gamma_{21}$}, to simplify the above as
\begin{equation}
    \mathbb{F} \simeq \underbrace{\frac{\Gamma_{12} - \Gamma_{21}}{\Gamma_{12}}}_{=1 - e^{-\beta (E_2-E_1)}} + \underbrace{\Gamma_{21}\frac{(4\Gamma_{12}\Gamma_{21})^N}{(N+1)\sqrt{\pi N}\qty(\Gamma_{12} - \Gamma_{21})^2}}_{\equiv \mathbb{G}},
\end{equation}
where $\simeq$ hides the $o$-terms. The infinite-temperature limit, \mbox{$\beta =0$}, is obtained by using Eq.~\eqref{Eq:error_function4} and plugging \mbox{$\Gamma_{12} = \Gamma_{21} = 1/2$}. This yields the following convergence
\begin{equation}
    \eval{\mathbb{F}}_{\Gamma_{12} = \frac{1}{2}} = \frac{1}{\sqrt{\pi N}}+o\left(N^{-1/2}\right)
\end{equation}
with the expression for $\mathbb{G}$ modified to \mbox{$(\pi N)^{-1/2}$}.

As a final step to prove 
Lemma~\ref{Lem:transposition} and Theorem~\ref{Thm:beta-swap}, we calculate explicitly the state of the primary system after the protocol~$\mathcal{P}^{(12)}$. This is done by substituting the results for $\mathbb{E}$ and $\mathbb{F}$ to Eq.~\eqref{eq:finalstate}, yielding
\begin{align}
    \v q = \mqty(1 - e^{-\beta E}& 1 \\ e^{-\beta E} & 0 )\mqty(b\\c) + \qty(b\mathbb{G} - c \mathbb{E})\mqty(1 \\ -1) = \Pi^\beta_{12}\v p + \qty(b\mathbb{G} - c \mathbb{E})\mqty(1 \\ -1).
\end{align}
Thus, the distance between $\v q$ and the target $\Pi^\beta_{12} \v p$ is given by
\begin{align}
    \delta\qty(\Pi^\beta_{12} \v p,\,\v q) & = \abs{b\mathbb{G} - c \mathbb{E}}.
\end{align}
For $\beta=0$ case, the above gives
\begin{align}
    \delta\qty(\Pi^\beta_{12} \v p,\,\v q) & = \frac{\abs{b - c }}{\sqrt{\pi N}}+o\left(N^{-1/2}\right),
\end{align}
whereas $\beta\neq 0$ case, it gives
\begin{align}
    \delta\qty(\Pi^\beta_{12} \v p,\,\v q) 
     &= (4\Gamma_{12}\Gamma_{21})^{N}\Bigg[\frac{\abs{b \Gamma_{21} - c\Gamma_{12}}}{(\Gamma_{12}-\Gamma_{21})^2\sqrt{\pi N}(N+1)}+o(N^{-3/2})\Bigg].
\end{align}
These prove Lemma~\ref{Lem:transposition} and Theorem~\ref{Thm:beta-swap} after going to the notation used therein, i.e., $b \rightarrow p_1,\,c\rightarrow p_2,\,\Gamma_{12}\rightarrow\Gamma_1$ and $\Gamma_{21}\rightarrow\Gamma_2$. \qed


\begin{center}
    \emph{\textbf{Strengthening Corollary \ref{corr:general_bound}}}\label{app:bound}
\end{center}

Corollary~\ref{corr:general_bound} deals with a very general approach to bounding the distance between the target state $\v p^{\pi}$ and its approximation obtained from $\v{p}$ via the MeMTP protocol $\mathcal{P}^{\Pi}$. However, it can be improved by taking into account the set of indices on which the permutation acts.

\begin{corollary} \label{corr:detailed_bound}
    Consider states $\v{p}$ and $\v q\in C_+^{TO}(\v{p})$. Then, in the infinite temperature limit, $\beta=0$, and for an $N$-dimensional memory, there exists a MeMTP protocol $\mathcal{P}$ such that
    \begin{equation}
        \P(\v{p}\otimes \v{\eta}_M) = \v{q}'\otimes \v{\eta}_M,
    \end{equation}
    with
    \begin{equation}
        \delta(\v{q}',\v{q})\leq \frac{1}{2\sqrt{\pi N}} \sum_{\substack{k,l\\ k\neq l}} \abs{p_{i_k} - p_{i_l}} + o\qty(N^{-1/2})
        =: \epsilon,
    \end{equation}
    where $\qty{i_1,\hdots,i_{d'}}\subset\qty{1,\hdots,d}$ is a subset of indices neighbouring in the $\beta$-order, $\pi_{\v p}(i_{j}) + 1 = \pi_{\v p}(i_{j+1})$, such that \mbox{$\Pi_{\v q} = \prod_{i=1} \Pi_{j_i k_i}$} with $j_i,k_i\in\qty{i_1,\hdots,i_{d'}}$.
\end{corollary}

\begin{proof}
First, we consider a target state to be an extreme point~$\v p^{\pi^*}$ such that the $\beta$-order $\Pi_{\v p^*}\equiv\Pi^*$ can be decomposed into the maximal number of $d'(d'-1)/2$ neighbour swaps on the levels $i_1$ through $i_{d'}$. Taking explicitly Eq.~\eqref{eq:Lambda_correction_op} from Theorem~\ref{Thm:permutation}, one finds that
\begin{equation}
    \delta\left(\v{p}^{\pi^*},\v r^*\right) = \frac{1}{2\sqrt{\pi N}} \sum_{\substack{k,l\\ k\neq l}} \abs{p_{i_k} - p_{i_l}} + o\qty(N^{-1/2}),
\end{equation}
where for convenience we used $\mathcal{P}^{\Pi^*}(\v p\otimes\eta_M) = \v r^* \otimes \eta_M$.
We note that the above expression in fact provides a general upper bound for any permutation $\Pi$ on the aforementioned subset of $d'$ levels -- defining $\mathcal{P}^{\Pi}(\v p\otimes\eta_M) = \v r \otimes\eta_M$ we find that
\begin{equation}
    \delta\left(\Pi\v{p},\v r\right) \leq \epsilon.
\end{equation}

Now, there are two cases to be considered. First, take a state $\v{q}$ which is in the future of the approximation point $\v r$, $\v q \in C_+^{MTP}\qty(\v r)$, from which it follows that

\begin{equation} \label{eq:exact_achiev}
     \exists\mathcal{O}\in\text{MTP}:\delta\left(\v{q},\mathcal{O}(\v r)\right) = 0.
\end{equation}
Otherwise, $\v q$ is not in the future cone of $\v r$. In this case, we first note that there exists a ball $B(\v r,\epsilon')\ni \Pi\v p$ with radius $\epsilon' \leq \epsilon$ with respect to $\delta(\cdot,\cdot)$. Thanks to the planarity of the boundaries $\partial C_+^{MTP}\qty(\v r)$ and $\partial C_+^{MTP}\qty(\Pi \v p)$ when restricted to a fixed $\beta$-order, we can consider the extreme case
\begin{equation}
    \v q \in \partial C_+^{MTP}\qty(\Pi \v p) \!\Rightarrow\! \exists \v r'\!\in\! \partial C_+^{MTP}\qty(\v r):\delta\qty(\v{q},\v{r}')\!\leq\! \epsilon'\! \leq\! \epsilon
\end{equation}
and the same argument applies for any \mbox{$\v q \in C_+^{MTP}\qty(\Pi \v p)\backslash C_+^{MTP}\qty(\v r)$}, thus concluding the proof.
\end{proof}

The bound presented in Corollary~\ref{corr:detailed_bound} can be further improved by taking into account the possibility of dividing the set $\qty{i_1,\hdots,i_{d'}}$ into subsets that are not mixed at any step when considering the decomposition of $\Pi_{\v q}$ into neighbour transpositions. Finally, we point out that, in agreement with Eq.~\eqref{eq:exact_achiev}, there will exist such states $\v{q}$ that are attainable exactly, and moreover, their volume will increase together with the size of memory $N$.


\begin{center}
    \emph{\textbf{Proof of Theorem~\ref{Thm:beta-3-cycle}}}\label{app:beta-3-cycle}
\end{center}

To prove Theorem~\ref{Thm:beta-3-cycle}, we will consider a composite system consisting of the main $d$-dimensional system and an $N$-dimensional memory system. Without loss of generality, we can simply assume that the main system has three levels with $i_1 = 1, i_2 = 2$ and $i_3 = 3$, and its state is described by an unnormalised probability vector \mbox{$\v p = (a, b, c)$}. The Hamiltonian is then given by \mbox{$H =\sum_{i=1}^3 E_i \ketbra{E_i}{E_i}$}, while the memory system is described by a trivial Hamiltonian $H_M = 0$ and prepared in a maximally mixed state \mbox{$\v \eta_M = (1/N, ..., 1/N)$}. The joint state of the composite system, \mbox{$\v r:= \v p \otimes \v \eta_M$}, is then given by
\begin{equation}
\label{Eq:entries_three_level_system}
 \v r^{(0)} \equiv \v r = \frac{1}{N}\Big(\underbrace{a, ..., a}_{N \text{  times}} | \underbrace{b, ..., b}_{N \text{  times}}| \underbrace{c, ..., c}_{N \text{  times}}\Big),
\end{equation}
and the joint thermal state is given by
\begin{equation}
\v \Gamma = \frac{1}{ZN}[e^{-\beta E_1}, ...,e^{-\beta E_1}| e^{-\beta E_2}, ..., e^{-\beta E_2}|e^{-\beta E_3}, ..., e^{-\beta E_3}],
\end{equation}
where $Z=\sum_{i=1}^3 e^{-\beta E_i}$.

As before, the starting point consists of understanding how the joint state of the composite system changes under the action of the composite protocol $\widetilde{\mathcal{P}}^{(13)}_N \circ  \widetilde{\mathcal{P}}^{(23)}_N$, whose action is summarised in two steps:
\begin{enumerate}
\item Two-level thermalisation between second and third energy levels. 
\item Two-level thermalisation between first and third energy levels.
\end{enumerate}
The final state $\v r^{(N)}$ is then given by
\begin{align}
   \!\!\! \widetilde{\P}^{\Pi}(\v r) &\equiv \v r^{(N)} = \v q \otimes \v \eta_M  \nonumber \\&=\frac{1}{N}{\small\qty[a^{(N)}_1,\hdots,a^{(N)}_N \, \bigg|\, b^{(N)}_1,\hdots,b^{(N)}_N \, \bigg| \, c^{(N)}_{1},\hdots, c^{(N)}_{N}]},\label{eq:beta-3-cycle-final}
\end{align}
where $\Pi=\Pi_{13}\Pi_{23}$. Note that due to probability conservation, characterising the second and third entries of Eq.~\eqref{eq:beta-3-cycle-final} is sufficient.

After the first protocol $\widetilde{\mathcal{P}}^{23}$, the second energy level remains ``untouched'' and, as a result, its entries are given by Eq.~\eqref{Eq:bj_entries} (with $\Gamma_{12}$ and $\Gamma_{21}$ replaced by $\Gamma_{23}$ and $\Gamma_{32}$, respectively). The other two entries are obtained in a similar way as Eqs.~\eqref{Eq:recurrence-relation-entries_bj}-\eqref{Eq:recurrence-relation-entries_cj}. The action of the protocol generates a recurrence formula that allows us to write the last entry $c^{(N)}_k$ as
\begin{align}
 c^{(N)}_k = a\Gamma^{k-1}_{13} \sum_{i=1}^N &\binom{N+k-1-i}{k-1}\Gamma_{31}^{N+1-i}+\Gamma^N_{31}\sum_{l=0}^{N-1}\Gamma_{13}^{l}\binom{N+l-1}{l}c_{k-l},
\end{align}
where $c_k$ is given by
\begin{equation}\label{ck}
    c_k = c\Gamma^N_{32} \sum_{i=0}^{k-1} \binom{N+i-1}{i}\Gamma^{i}_{23} +b\Gamma_{32}\Gamma^{k}_{23}\sum_{i=0}^{N-1} \binom{i+k}{k}\Gamma_{32}^{i}.
\end{equation}

Since, without loss of generality, we assumed that $\v p$ has $\beta$-ordering $(123)$, the proof boils down to demonstrating that $\widetilde{\P}^{\Pi}(\v r)$ sends $\v p$ to the following extreme point
\begin{equation}
    \v p^{(321)} = \qty[a + \frac{\Gamma_{32}}{\Gamma_{23}}b - a \frac{\Gamma_{31}}{\Gamma_{13}}, c+b\left(1-\frac{\Gamma_{32}}{\Gamma_{23}}\right), \frac{\Gamma_{31}}{\Gamma_{13}}a].
\end{equation}
Therefore, we need to prove the following limits
\begin{subequations}
\begin{align}\label{Eq:first-limit-beta-cyclic-perm}
    \lim_{N \to \infty}\frac{1}{N}\sum_{i=1}^N b^{(N)}_i &= c+b \, \qty(1-\frac{\Gamma_{32}}{\Gamma_{23}}), \\ \lim_{N \to \infty}\frac{1}{N}\sum_{i=1}^N c^{(N)}_i &= a \frac{\Gamma_{31}}{\Gamma_{13}}. \label{Eq:second-limit-beta-cyclyc-perm}
\end{align}
\end{subequations}


\begin{center}
    \emph{\textbf{Proof of limit (\ref{Eq:first-limit-beta-cyclic-perm}) }}
\end{center}

We start by recalling that $b^{(N)}_j$ is given by:
\begin{equation}\label{Eq:bjn}
 b_j^{(N)}\! =c\! \frac{\Gamma_{23}}{\Gamma_{32}}\Gamma_{32}^{j}\sum_{i=0}^{N-1}\binom{j\! +\! i\! -\! 1}{i}\Gamma_{23}^{i}+b\Gamma_{23}^{N}\sum_{i=1}^{j}\binom{j\! +\! N\! -\! 1\! -\! i}{N\! -\! 1}\Gamma_{32}^{j-i}.
\end{equation}
Comparing Eqs. \eqref{Eq:bjn} and the right-hand side of \eqref{Eq:first-limit-beta-cyclic-perm}, we see that in order to prove Eq.~\eqref{Eq:first-limit-beta-cyclic-perm}, we need to prove the following two limits: 
\begin{align}\label{eq:b_part_1}
 \lim_{N \to \infty}\sum_{j=1}^N\eval{\frac{b^{(N)}_j}{N}}_{\substack{b = 0\\c=1}} &=\lim_{N\rightarrow\infty}\frac{\Gamma_{23}}{\Gamma_{32}}\frac{1}{N}\sum_{j=1}^N\Gamma_{32}^j\sum_{i=0}^{N-1}\binom{j+i-1}{i}\Gamma_{23}^{i} \nonumber \\ &=1 
\end{align}
and
\begin{align}\label{eq:b_part_2}
 \lim_{N \to \infty}\sum_{j=1}^N\eval{\frac{b^{(N)}_j}{N}}_{\substack{b = 1\\c=0}} &=\lim_{N\rightarrow\infty} \Gamma_{23}^{N}\frac{1}{N}\sum_{j=1}^N\Gamma_{32}^j\sum_{i=1}^j\binom{j+N-1-i}{N-1}\Gamma_{32}^{-i} \nonumber \\ &=\qty(1-\frac{\Gamma_{32}}{\Gamma_{23}}).  
\end{align}

We begin by proving Eq.~\eqref{eq:b_part_1}. First, we rewrite this expression as:
\begin{align}
\sum_{j=1}^N \eval{\frac{b^{(N)}_j}{N}}_{\substack{b = 1\\c=0}} &=\frac{1}{N} \frac{\Gamma_{23}}{\Gamma_{32}}\sum_{i=0}^{N-1}\Gamma_{23}^{i}\sum_{j=1}^N\Gamma_{32}^j\binom{j+i-1}{i} \nonumber\\&=  \frac{1}{N} \Gamma_{23}\sum_{i=0}^{N-1}\Gamma_{23}^{i}\sum_{j=0}^{N-1}\Gamma_{32}^j\binom{j+i}{i}.\label{eq:rewrite_b}
\end{align} 
We can evaluate the second sum in Eq.~\eqref{eq:rewrite_b} as follows:
\begin{align}
\sum_{j=0}^{N-1}\Gamma_{32}^j\binom{j+i}{i} &=  (\Gamma_{23})^{-i-1}\Bigg(1-\frac{B_{\Gamma_{32}}(N,i+1)}{B(N,i+1)}\Bigg)\nonumber\label{eq:second_sum} \nonumber \\ &=(\Gamma_{23})^{-i-1}[1-I_{\Gamma_{23}}(N,i+1)].
\end{align}
Thus, substituting Eq.~\eqref{eq:second_sum} into Eq.~\eqref{eq:rewrite_b}, we obtain
\begin{align}
\sum_{j=1}^N \eval{\frac{b^{(N)}_j}{N}}_{\substack{b = 1\\c=0}} =1-\frac{1}{N}\sum_{i=0}^{N-1}I_{\Gamma_{32}}(N,i+1).\label{eq:expr_b}
\end{align}
Using Eq.~\eqref{Eq:beta-regularised-limit-sum}, we conclude that the second term in Eq.~\eqref{eq:expr_b} vanishes in the limit of $N\rightarrow \infty$, and therefore
\begin{equation}
     \lim_{n \to \infty}\sum_{j=1}^N\eval{b^{(N)}_j}_{\substack{b = 0\\c=1}} = 1,\label{eq:beta-3-cycle-first-lim}
\end{equation}
so that we have proved Eq.~\eqref{eq:b_part_1}. 

To prove Eq.~\eqref{eq:b_part_2}, we begin by manipulating it so that we can express it in a simpler form:\! 
\begin{align}\label{eq:simplified_b_semistep} 
\sum_{j=1}^N \eval{\frac{b^{(N)}_j}{N}}_{\substack{b = 1\\c=0}}\!\!\!\! &= \frac{\Gamma_{23}^{N}}{N}\sum_{j=1}^N\Gamma_{32}^j\sum_{i=1}^j\binom{N+j-1-i}{N-1}\Gamma_{32}^{-i}\nonumber \\ 
&= \frac{\Gamma_{23}^{N}}{N}\sum_{j=0}^{N-1}(N-j)\binom{N+j-1}{j}\Gamma_{32}^{j}\nonumber \\
&=
\Gamma_{23}^{N}\sum_{j=0}^{N-1}\binom{N\! +\! j\! -\! 1}{j}\Gamma_{32}^{j}-\frac{\Gamma_{23}^{N}}{N}\sum_{j=0}^{N-1}j\,\binom{N\! +\! j\! -\! 1}{j}\Gamma_{32}^{j}.
\end{align} 
Applying Eq.~\eqref{Eq:incompleta-beta-function-difference} to transform the first term of Eq.~\eqref{eq:simplified_b_semistep} into a difference of regularised beta functions, and then using its asymptotic expansion, we obtain
\begin{equation}
    \Gamma_{23}^{N}\sum_{j=0}^{N-1}\binom{N+j-1}{j}\Gamma_{32}^{j} = I_{\Gamma_{32}}(0,N)-I_{\Gamma_{32}}(N,N)\simeq 1.
\end{equation}
Next, we consider the second term in Eq.~\eqref{eq:simplified_b_semistep}, which can be directly evaluated as
\begin{align}
-\frac{\Gamma_{23}^{N}}{N}&\sum_{j=0}^{N-1}j\binom{N+j-1}{j}\Gamma_{32}^{j} \nonumber\\
&=-\frac{\Gamma_{32}}{\Gamma_{23}}\Gamma_{23}^{N+1}\sum_{j=-1}^{N-2}\binom{(N+1)+j-1}{j}\Gamma_{32}^j\nonumber\\
&=-\frac{\Gamma_{32}}{\Gamma_{23}}\qty[(I_{\Gamma_{32}}(0,N+1)-I_{\Gamma_{32}}(N-1,N+1)] \nonumber\\ &\simeq -\frac{\Gamma_{32}}{\Gamma_{23}},
\end{align}
where in the last line we used the asymptotic expansion of $I_x(a,b)$ to approximate the difference between regularised beta functions. Collecting all the terms, we conclude that the limit is given by
\begin{equation}
     \lim_{N \to \infty}\sum_{j=1}^N\eval{\frac{b^{(N)}_j}{N}}_{\substack{b = 1\\c=0}}  = 1-\frac{\Gamma_{32}}{\Gamma_{23}}.
\end{equation}
Therefore, combining the above with Eq.~\eqref{eq:beta-3-cycle-first-lim}, we get the desired limit: 
\begin{equation}
    \lim_{N \to \infty}\frac{1}{N}\sum_{j=1}^N \frac{b^{(N)}_j}{N} = c+b \, \qty(1-\frac{\Gamma_{32}}{\Gamma_{23}}) .
\end{equation}


\newpage

\begin{center}
    \emph{\textbf{Proof of limit (\ref{Eq:second-limit-beta-cyclyc-perm})}}
\end{center}

As before, in order to prove Eq.~\eqref{Eq:second-limit-beta-cyclyc-perm}, we will also need to prove two other limits. Recall that $c^{(N)}_j$ is given by
\begin{align}
   c^{(N)}_j = a\Gamma^{j-1}_{13} \sum_{i=1}^N \binom{N+j-1-i}{j-1}\Gamma_{31}^{N+1-i}+\Gamma^N_{31}\sum_{l=0}^{N-1}\Gamma_{13}^{l}\binom{N+l-1}{l}c_{j-l},
\end{align}
with
\begin{equation}\label{ck2}
    c_j = c\Gamma^N_{32} \sum_{i=0}^{j-1} \binom{N+i-1}{i}\Gamma^{i}_{23} +b\Gamma_{32}\Gamma^{j}_{23}\sum_{i=0}^{N-1} \binom{i+j}{j}\Gamma_{32}^{i}.
\end{equation}
Since $c_{j-l}$ does not depend on $a$, the problem reduces to showing that
\begin{align}
 \lim_{N\to \infty}\frac{1}{N}\sum_{j=1}^N\eval{c^{(N)}_j}_{\substack{a = 1\\b,c=0}} &=\lim_{N\rightarrow\infty}\frac{1}{N}\sum_{j=1}^N\Gamma_{13}^{j}\sum_{i=1}^{N}\binom{N+j-i}{j}\Gamma_{31}^{N+1-i} \nonumber \\ &=\frac{\Gamma_{31}}{\Gamma_{13}},\label{eq:c_part_1}
\end{align}
and 
\begin{align}
\lim_{N\to \infty}\frac{1}{N}\sum_{j=1}^N \eval{c^{(N)}_j}_{\substack{a = 0}} &= \lim_{N\rightarrow\infty}\frac{1}{N}\sum_{j=1}^N\Gamma_{31}^{N}\sum_{l=0}^{N-1}\Gamma_{13}^{l}\binom{N+l-1}{l}c_{j-l}\nonumber\\&=0.\label{eq:ckl}
\end{align}

Let us start by proving Eq.~\eqref{eq:c_part_1}. First, we manipulate $c^{(N)}_j$ and rewrite it in terms of the incomplete beta function as follows:
\begin{align}
\eval{c^{(N)}_j}_{\substack{a = 1\\b,c=0}}&=\Gamma_{13}^{j-1}\sum_{i=1}^{N}\binom{N+j-1-i}{j-1}\Gamma_{31}^{N+1-i}\nonumber\\
&=\sum_{i=1}^N\qty[(I_{\Gamma_{13}}(0,N+1-i)-I_{\Gamma_{13}}(N,N+1-i)]\nonumber\\ &= \sum_{i=1}^N I_{\Gamma_{31}}(N+1-i,N).
\end{align}
Using Eq.~\eqref{Eq:regularised-beta-function-sum-2}, we obtain 
\begin{eqnarray}
\sum_{i=1}^N I_{\Gamma_{31}}(N+1-i,N) = N\frac{\Gamma_{31}}{\Gamma_{13}}.
\end{eqnarray}
Thus, collecting all the terms, we get the desired limit
\begin{equation}
\lim_{N\to \infty}\frac{1}{N}\sum_{j=1}^N\eval{c^{(N)}_j}_{\substack{a = 1\\b,c=0}}=\frac{\Gamma_{31}}{\Gamma_{13}}.
\end{equation}

The final step is to show that the remaining limit from Eq.~\eqref{eq:ckl} is zero, namely
\begin{equation}
    \lim_{N\rightarrow\infty}\frac{1}{N}\Gamma_{31}^N\sum_{k=1}^N\sum_{l=0}^{N-1}\Gamma_{13}^lc_{k-l}=0.
\end{equation}
Since $c_{k-l}$ has two contributions, one needs to show that both limits go to zero. Treating each separately, we first write the first term of Eq.~\eqref{ck2} in terms of the incomplete beta function,
\begin{equation}
\eval{c_k}_{\substack{c = 1, b=0}} = \Gamma^N_{32} \sum_{i=0}^{k-1} \binom{N+i-1}{i}\Gamma^{i}_{23} = I_{32}(N,k) \leq 1,
\end{equation}
where we could bound by one because the incomplete beta function is a CDF. Thus,
\begin{align}
\sum_{k=1}^N \frac{\Gamma^N_{31}}{N}\sum_{l=0}^{N-1}\Gamma_{13}^{l}&\binom{N+l-1}{l}\eval{c_{k-l}}_{\substack{c = 1, b=0}} \leq \sum_{k=1}^N \frac{\Gamma^N_{31}}{N}\sum_{l=0}^{N-1}\Gamma_{13}^{l}\binom{N+l-1}{l} \nonumber \\ &= I_{\Gamma_{31}}(N,N),
\end{align}
and this term goes to zero for $\Gamma_{31} \leq 1/2$. This can be seen from the asymptotic expansion of $I_{\Gamma_{31}}(N,N)$. 

Finally, we need to show that the second term of Eq.~\eqref{ck2} is zero. To do so, we first re-write the second term as
\begin{align}
\!\!\!\! \sum_{k=1}^{N}\eval{c^{(N)}_k}_{\substack{b = 1\\a,c=0}}\!\! & = \frac{1}{N}\Gamma_{31}^N\sum_{l=0}^{N-1}\Gamma_{13}^l \binom{N+l-1}{l}\sum_{k=1}^N\eval{c_{k-l}}_{\substack{c = 0\\b=1}}  \nonumber \\
    & =
    \frac{1}{N}\Gamma_{31}^N\sum_{l=0}^{N-1}\Gamma_{13}^l \binom{N+l-1}{l}\sum_{k=l+1}^N\eval{c_{k-l}}_{\substack{c = 0\\b=1}} \nonumber \\
    & =
    \frac{1}{N}\Gamma_{31}^N\sum_{l=0}^{N-1}\Gamma_{13}^l \binom{N+l-1}{l}\sum_{k=1}^{N-l}\eval{c_{k}}_{\substack{c = 0\\b=1}}. 
\end{align}  
Notice that the above expression can be bounded by
\begin{align}\label{Eq:bound_ck}
  \sum_{k=1}^{N}\eval{c^{(N)}_k}_{\substack{b = 1\\a,c=0}} &\leq
    \frac{1}{N}\qty(\Gamma_{31}^N\sum_{l=0}^{N-1}\Gamma_{13}^l \binom{N+l-1}{l})\qty(\sum_{k=1}^{N}\eval{c_{k}}_{\substack{c = 0\\b=1}} ), \nonumber 
\end{align}
and the right-hand side of equation Eq.~\eqref{Eq:bound_ck} can be expressed in terms of the regularised beta function as follows:
\begin{align}
\frac{1}{N}\Gamma_{31}^N&\sum_{l=0}^{N-1}\Gamma_{13}^l \binom{N+l-1}{l}\qty(\sum_{k=1}^{N}\eval{c_{k}}_{\substack{c = 0\\b=1}} ) \nonumber\\
     & =
    \Gamma_{32}\frac{1}{N}I_{31}(N,N)\sum_{k=1}^N\Gamma^{k}_{23}\sum_{i=0}^{N-1} \binom{i+k}{k}\Gamma_{32}^{i} \nonumber\\
    & =
    \Gamma_{32}\frac{1}{N}I_{31}(N,N)\sum_{k=1}^N\Gamma^{k}_{23}\Gamma_{23}^{-1-k}\qty[1 - I_{32}(N,k+1)]\nonumber \\
    & =
    \frac{\Gamma_{32}}{\Gamma_{23}}\frac{1}{N}I_{31}(N,N)\qty(\frac{1}{N}\sum_{k=1}^N I_{23}(k+1,N)).
\end{align}
Now, note that the first term goes to zero when $\Gamma_{31} < 1/2$, whereas the second term is also bounded by one as
\begin{equation}
\frac{1}{N}\sum_{k=1}^NI_{23}(k+1,N) \leq \frac{1}{N}\sum_{k=1}^N 1 = 1.
\end{equation}
Since 
\begin{equation}
   0 \leq \lim_{N\to \infty}\sum_{k=1}^N\frac{1}{N}\eval{c^{(N)}_k}_{\substack{b = 1\\a,c=0}} \leq 0,
\end{equation}
we conclude that the resulting limit is zero. Therefore, 
\begin{equation}
    \lim_{N\rightarrow\infty}\frac{1}{N}\Gamma_{31}^N\sum_{k=1}^N\sum_{l=0}^{N-1}\Gamma_{31}^lc_{k-l}=0.
\end{equation}

\section{Concluding remarks}~\label{sec:outlook}

In this chapter, we proposed a novel approach to investigate memory effects in thermodynamics by introducing the concept of memory-assisted Markovian thermal processes. These were defined by extending the framework of Markovian thermal processes with ancillary memory systems brought in thermal equilibrium states. Our construction allowed us to interpolate between the regime of memoryless dynamics and the one with full control over all degrees of freedom of the system and the bath. Using a family of protocols composed of Markovian thermal processes, we demonstrated that energy-incoherent states achievable from a given initial state via thermal operations could be approximated arbitrarily well via our algorithmic procedure employing memory. Furthermore, we analysed the convergence of our protocols in the infinite memory limit, finding polynomial and exponential convergence rates for infinite and finite temperatures, respectively. In the infinite temperature limit, we provided analytic convergence to the entire set of states accessible via thermal operations. For finite temperatures, we proved the convergence to a subset of accessible states and, based on extensive numerical evidence, we conjectured that a modified version of our protocol can realise arbitrary transitions achievable via thermal operations with an exponential convergence rate that grows with memory size. Our model-independent approach can be seen as a significant step forward in understanding ultimate limits of the Markovian evolution in general, which should be contrasted with the model-specific approaches to the so-called Markovian embedding \cite{siegle2010markovian,budini2013embedding,campbell2018system}. On the other hand, it may be seen as far less general than the approach taken in Ref.~\cite{ende2023finitedimensional}, where our work would correspond to a step towards simulating arbitrary evolution with Markov-Stinespring curves.

We also explained how our results can be employed to quantitatively assess the role of memory for the performance of thermodynamic protocols. In this context, we discussed the dependence on the memory size of the amount and quality of work extracted from a given non-equilibrium state. However, the method can be used as well to investigate other thermodynamic protocols, such as information erasure or thermodynamically free encoding of information~\cite{korzekwa2019encoding}. Furthermore, we revealed that all transitions accessible via thermal operations can be accomplished using a restricted set of thermal operations that exclusively affect only two energy levels (of the system extended by a memory) at any given time. These findings carry important implications, not only for the development of efficient thermodynamic protocols, such as optimal cooling and Landauer erasure, but also for the exploration of novel avenues of research focused on characterising memory effects in thermodynamics. Finally, we also commented on the role played by the memory system as a free energy storage that enables non-Markovian effects.

Our results offer many possibilities for generalisation and further research. First, one can try proving that the future thermal cone for memory-assisted Markovian thermal processes agrees with that of thermal operations in the limit of infinite memory, \mbox{$\lim_{N\rightarrow\infty} C^+_{\text{MeMTP}} = C^+_{\text{TO}}$}, as suggested by Conjecture~\ref{Conj:beta_permutations}. This can be built upon the proofs for $\beta$-swaps (Theorem \ref{Thm:beta-swap}) and $\beta$-$3$-cycles (Theorem \ref{Thm:beta-3-cycle}) presented in this work. Second, one may also attempt to show that the convergence of the proposed protocols $\mathcal{P}^{\Pi}$ and $\widetilde{\mathcal{P}}^{\Pi}$ is optimal with respect to the memory size. In other words, one could investigate the upper-bound on the power of memory-assisted Markovian thermal processes with a given size of memory $N$. Third, from a more practical point of view, it may be worthwhile to explore MeMTPs involving finite and infinite memory with non-trivial energy level structure. The practical relevance of this direction can be understood by considering the introduction of non-degenerate splitting of the levels for the full system, which would allow the level pairs to be addressed independently.

In addition to the above, there are also less clear-cut goals for future efforts, such as expanding the studies beyond energy-incoherent states into the full range of quantum states. Furthermore, while our work focused on a single main system, an interesting avenue for future work could be to investigate many non-interacting subsystems. This extension could shed light on the combined consequences of finite-size and memory effects, providing valuable insights into the behaviour of larger, more complex systems. Specifically, characterising such effects could help to identify strategies for improving the efficiency of thermodynamic protocols in practical applications. Finally, one can also consider memory composed of many equivalent systems (such as a multi-qubit memory), and analyse the potential challenges arising from energy-level degeneration in such a setting.

Finally, the feasibility of the introduced algorithm can be studied from a control perspective, following the approaches outlined in~\cite{wolpert2019space, PRXQuantum.4.010332}. The first approach introduces the notion of a space-time trade-off, which refers to the minimal amount of memory and time steps required to classically implement a given process. The second approach deals with control complexity, defining it as the number of levels a given operation non-trivially acts versus the time steps needed to implement that process. Our algorithm has specific time and memory requirements, namely $N$-dimensional memory and $N^2$ time steps. Furthermore, it is limited to the simplest two-level processes at any given time, meaning its control complexity is as low as possible. Nonetheless, future work might explore variations of our protocol (or any of the variants presented in Section~\ref{App:protocols}). Such explorations could focus on enabling parallelisation of certain steps by expanding available memory, thereby illustrating the space-time trade-off.
\chapter{Fluctuation-dissipation relations for thermodynamic distillation processes}\label{C:finite-size}

Almost two centuries ago, Robert Brown observed that pollen seeds immersed in water move randomly in erratic motion~\cite{brown1828particles}. It was not until the 1905 papers by Einstein and Smoluchowski~\cite{Einstein1906, Smoluchowski1906} that scientists understood that this ``Brownian'' motion is induced by the bombardment of pollen particles by water molecules. Crucially, by noting that these collisions would also create friction for the particle being pulled through the fluid, Einstein realised that the two processes, fluctuations of particle's position and dissipation of its energy, have the same origin and thus must be related. Over the years, physicists generalised and formalised this observation into fundamental fluctuation-dissipation relations describing the behavior of systems driven out of equilibrium~\cite{kubo1966fluctuation,marconi2008fluctuation}. 

Now, it is well known that near-equilibrium, linear response theory provides a general proof of the fluctuation-dissipation theorem, which states that the response of a given system when subject to an external perturbation is expressed in terms of the fluctuation properties of the system in thermal equilibrium~\cite{kubo1966fluctuation}. The theoretical description underlying the fluctuation-dissipation relations is usually expressed in terms of the stochastic character of thermodynamic variables. This approach is strongly motivated since it is experimentally viable~\cite{batalhao2014, An2015}.

As introduced in the Chapter~\ref{C:resource_theory_of_thermodynamics}, the resource theory of thermodynamics aims to go beyond the thermodynamic limit and the assumption of equilibrium. It is often presented as an extension of statistical mechanics to scenarios with large fluctuations, referred to as single-shot statistical mechanics.~\cite{Dahlsten_2011,Yunger2018}. A natural question is then whether fluctuation-dissipation relations are present in such a resource-theoretic description. Although important insights have been obtained in trying to connect the information-theoretic and fluctuation theorem approaches~\cite{Alhambra2016, Halpern2015}, they have, so far, not been explicitly related to dissipation. The tools required for the analysis of free energy dissipation in a resource-theoretic framework were developed in Refs.~\cite{Hayashithermo2017, Chubb2018beyondthermodynamic,Chubb2019_2,korzekwa2019avoiding}, where the authors investigated irreversibility of thermodynamic processes due to finite-size effects. However, the relation between fluctuations and actual dissipation was not derived and, moreover, these results were obtained for quasi-classical case of energy-incoherent states. Thus, they were not able to account for quantum effects that come into play when dealing with even smaller systems, when fluctuations around thermodynamic averages are no longer just thermal in their origin.

This chapter pushes towards a genuinely quantum framework characterising optimal thermodynamic state transformations and links fluctuations with free energy dissipation. We investigate a special case of state interconversion processes known as \emph{thermodynamic distillations}. These are thermodynamic processes in which a given initial quantum system is transformed, with some transformation error, to a pure energy eigenstate of the final system. In particular, we focus on the initial system consisting of asymptotically many non-interacting subsystems that are either energy-incoherent and non-identical (in different states and with different Hamiltonians), or pure and identical. Within this setting, our main results are given by two theorems. The first one yields the optimal transformation error as a function of the free energy difference between the initial and target states, and the free energy fluctuations in the initial state. This can be seen as an extension of previously derived results on optimal thermodynamic state transformations~\cite{Hayashithermo2017, Chubb2018beyondthermodynamic,Chubb2019_2,korzekwa2019avoiding}. The second theorem provides a precise relation between the free energy fluctuations of the initial state and the minimal amount of free energy dissipated in the optimal thermodynamic distillation process. It is conceptually novel and does not form an extension of previously known results, and as such it constitutes our main contribution. Note that the second theorem employs the first one as one of the building blocks.

Our results allow us for a rigorous study of important thermodynamic protocols. First of all, we extend the analysis of work extraction to the regime of not necessarily identical incoherent states, as well as to pure states. By directly applying our main results, we obtain a second-order asymptotic expression for the optimal transformation error while extracting a given amount of work from the initial system. Moreover, we also verify the accuracy of the obtained expression by comparing it with the numerically optimised work extraction process. As a second application, we analyse the optimal energetic cost of erasing $N$ independent bits prepared in arbitrary states. In this case, we obtained the optimal transformation error for the erasure process as a function of invested work. The last application we consider is the optimal thermodynamically-free communication scheme, i.e., the optimal encoding of information into a quantum system without using any extra thermodynamic resources. Applying our theorems gives us the optimal number of messages that can be encoded into a quantum system in a thermodynamically free-way, which we show to be directly related to the non-equilibrium free energy of the system. This result can be interpreted as the inverse of the Szilard engine, as in this process we use the ability to perform work to encode information. Furthermore, our results connect the fluctuations of free energy and the optimal average decoding error. Finally, our findings also provide new tools to study approximate transformations and corresponding asymptotic interconversion rates. Here, we not only extend previous distillation results~\cite{Chubb2018beyondthermodynamic} to non-identical systems, but also to genuinely quantum states in superposition of different energy eigenstates. 

The chapter is organised as follows. We start in Sec.~\ref{sec:high} with a high-level description that can give a flavour of our investigations and explains the physical intuition behind them to a broad audience without the necessity to get into the technicalities of the framework we work in. We then introduce the central notion of this chapter namely thermodynamic distillation processes and adapt the information-theoretic quantities from Chapter~\ref{C:resource_theory_of_thermodynamics} to the scenario discussed here. In Section~\ref{sec:results}, we state our main results concerning the optimal transformation error and fluctuation-dissipation relation for incoherent and pure states, discuss their thermodynamic interpretation and apply them to three thermodynamic protocols of work extraction, information erasure and thermodynamically-free communication. The technical derivation of the main results can be found in Section~\ref{sec:math}. Finally, we conclude with an outlook in Section~\ref{sec:out}. This way we prove a general fluctuation-dissipation relation for thermodynamic distillation processes.
\newpage

\section{High-level description}
\label{sec:high}

Before formally stating the setting studied in this chapter, let us present a high-level description of our investigations. Our aim is to identify the fluctuation-dissipation phenomenon in the realm of resource-theoretic approach to thermodynamics. To do so, we will extract the main feature captured by the original works of Einstein and Smoluchowski: in order to obtain any ordered motion of a state of the system that is subjected to random forces, we necessarily need to dissipate energy that is proportional to the fluctuations of energy induced by these random forces. The main point of the resource theory of thermodynamics is to determine whether one state can be thermodynamically transformed into another. In our framework, we will examine the effect of the fluctuations present in the initial state of the system on the minimal amount of dissipation during a state transformation process. As we shall see at the end of this section, the proper fluctuating and dissipated quantity in the thermodynamic context will be given by the free energy of the system.

As a warm up, let us start with a simple example, where the goal is to draw work from a given system (this is indeed an example of a state transformation if one includes explicitly an ancillary weight system). More precisely, consider a model system with a continuous, non-degenerate energy spectrum with the ground state of energy $E_0=0$ that is prepared in a probabilistic mixture $\rho$ of different energy eigenstates $\ket{E}$ corresponding to energy $E$, i.e.,
\begin{equation}
    \rho=\int\limits_{0}^{\infty} p(E) \ketbra{E}{E}\, dE, 
\end{equation}
where $p(E)$ is a probability density function describing the system's distribution over energy levels. Our aim is now to use this model system with probabilistic (``fluctuating'') amounts of energy to make an almost deterministic (``ordered'') change of energy of another system. More formally, we are interested in performing $\epsilon$-deterministic work $W$, i.e., in changing the state of the ancillary weight system from one energy eigenstate $\ket{W_0}$ to another energy eigenstate $\ket{W_0+W}$ with probability $1-\epsilon$. 

How can we achieve this? If the distribution $p(E)$ is vanishing for \mbox{$E\in[0,W]$}, then we can couple the two systems and transfer the amount of energy $W$ between them by simply shifting the entire distribution $p(E)$ down by $W$, while moving the weight system up by $W$ (see the top panel of Fig.~\ref{fig:highleveldescript}). Similarly, if the bulk of $p(E)$ is localized far away from the ground energy 0, we can try to perform an analogous protocol, but this time we will fail with probability
\begin{equation}
    \epsilon=\int\limits_{0}^{W} p(E)\, dE,
\end{equation}
since the states with $E\in[0,W]$ cannot be lowered by $W$, as they would need to to be lowered below the ground state.

\begin{figure}[t]
    \centering
    \includegraphics{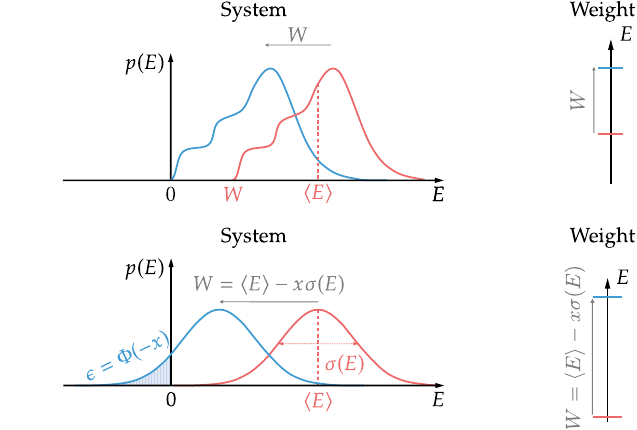}
    \caption{\label{fig:highleveldescript} \emph{Transforming ``fluctuating'' to ``deterministic'' energy.} \textit{Top:} Despite fluctuations of energy, the lowest occupied state of the system is far away from the ground state, and so deterministic amount of work $W$ can drawn from it. However, $W$ is much smaller than the average energy content $\langle E\rangle$ of the system.\textit{ Bottom:} accepting probability of failure $\epsilon$, one can extract deterministic amount of work approaching the average energy $\langle E\rangle$, with the loss proportional to energy fluctuations $\sigma(E)$, where the proportionality constant is determined by $\epsilon$. Here, the initial distribution is assumed to be Gaussian with average $\langle E\rangle$ and standard deviation $\sigma(E)$.}
\end{figure}
To illustrate this more clearly, consider an important example of $p(E)$ given by a Gaussian distribution with a mean $\langle E\rangle$ and standard deviation $\sigma$:
\begin{equation}
    p(E) = \frac{1}{\sqrt{2\pi} \sigma} \exp\left(-\frac{1}{2}\frac{(E-\langle E\rangle)^2}{\sigma^2}\right).
\end{equation}
The importance of this example stems from the fact that in thermodynamics we are interested in total energy distributions of a large number $N$ of particles, and results like the central limit theorem tell us that the distributions of total quantities in such a case are approximated very well by Gaussian distributions. Of course, a Gaussian distribution is non-vanishing below the ground state energy 0, but as long the average energy $\langle E \rangle$ is far away from zero this can be neglected for the sake of our example. In the bottom panel of Fig.~\ref{fig:highleveldescript}, we present how shifting down a Gaussian distribution by its mean $\langle E\rangle$ decreased by a number of standard deviations $x\sigma$ results in an error $\epsilon=\Phi(-x)$, where $\Phi$ is the cumulative normal distribution function. Thus, for a fixed success probability of extracting work $W$ we can extract the average energy content of the system, $\langle E\rangle$, decreased by the quantity proportional to energy fluctuations $\sigma$. In other words, in order to transform the fluctuating type of energy into (almost) deterministic one, we need to lose (dissipate) some of it due to fluctuations. This simple scenario gives an intuition for why fluctuations may be related to dissipation.

Of course, the toy example we analysed above does not account for many features of realistic scenarios. First of all, it deals merely with mechanical work, whereas in thermodynamics one also has access to a thermal bath and can use it to draw even more thermodynamical work. Second, when considering systems of many particles we do not deal with non-degenerate spectrum, but rather at each energy we have a corresponding density of states. As a result, one may not be able to simply shift the distribution down, as there may be less low energy states then high energy states. Next, in the analysed example we only considered the protocol of work extraction, which is a very particular type of a general thermodynamic state transformation that physicists are interested in. Finally, since we deal with quantum mechanics, within each degenerate energy subspace we may deal with coherent superposition of states that can constructively or destructively interfere. Hence, the picture gets even more complex and requires a formalism that can account both for coherent and incoherent contributions to fluctuations. 

Despite these complications, our work extends the original intuition from the simple toy example to general quantum thermodynamic scenarios, including all the features described above. The crucial modification required is replacing the concept of average energy and its fluctuations (relevant in the case of mechanical systems) with the average free energy and its fluctuations (relevant for thermodynamic scenarios).To account for quantum systems prepared in arbitrary non-equilibrium states, one needs to use the non-equilibrium quantum generalisations of the classical expression for free energy, which also allows for the rigorous definition of free energy fluctuations. With these modifications in place, one can employ the above intuition to investigate general thermodynamic distillation processes that transform generic states with fluctuations of free energy into states with no free energy fluctuations (the equivalent of "ordered energy" states). Specifically, we show that, for a fixed success probability of transformation $(1-\epsilon)$, during such a process, the amount of free energy dissipated must be proportional to the initial free energy fluctuations.

\section{Setting the scene}
\label{sec:setting}

\subsection{Thermodynamic distillation processes}
\label{sec:distillation}

A \emph{thermodynamic distillation} process is a thermodynamically free transformation from a general \emph{initial system} described by a Hamiltonian $H$ and prepared in a state $\rho$, to a \emph{target system} described by a Hamiltonian $\tilde{H}$ and in a state $\tilde{\rho}$ that is an eigenstate of $\tilde{H}$.\footnote{In fact, all of our results apply to a slightly more general setting with target states being proportional to the Gibbs state on their support, e.g. for \mbox{$\tilde{\rho}=\frac{\tilde{\gamma}_k}{\tilde{\gamma}_k+\tilde{\gamma}_l}\ketbra{\tilde{E}_k}{\tilde{E}_k}+\frac{\tilde{\gamma}_l}{\tilde{\gamma}_k+\tilde{\gamma}_l}\ketbra{\tilde{E}_l}{\tilde{E}_l}$}, where $\ket{\tilde{E}_i}$ denotes the eigenstate of $\tilde{H}$ and $\tilde{\gamma}_i$ is its thermal occupation.} An \emph{$\epsilon$-approximate thermodynamic distillation} process from $(\rho,H)$ to $(\tilde{\rho},\tilde{H})$ is a thermal operation that transforms the initial system $(\rho,H)$ to the \emph{final system} $(\rho_{\text{fin}},\tilde{H})$ with $\rho_{\text{fin}}$ being $\epsilon$ away from $\tilde{\rho}$ in the infidelity distance $\delta$,
\begin{equation}
        \delta(\rho_1,\rho_2):=1-\left(\mathrm{Tr}{\sqrt{\sqrt{\rho_1}{\rho_2}\sqrt{\rho_1}}}\right)^2.
\end{equation}

Here, we will study the distillation process from $N$ independent initial systems to arbitrary target systems, e.g., to $\tilde{N}$ independent target systems as illustrated in Fig.~\ref{fig:distillation}. In particular, we will be interested in the asymptotic behaviour for large $N$. Thus, our distillation setting is specified by a family of initial and target systems indexed by a natural number $N$. For each fixed $N$, the initial system $(\rho^N,H^N)$ consists of $N$ non-interacting subsystems with the total Hamiltonian $H^N$ and a state $\rho^N$ given by
\begin{marginfigure}[-1.34cm]
	\includegraphics[width=4.651cm]{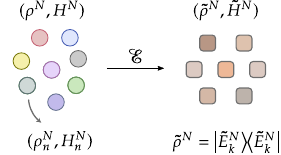}
	\caption[Fig:thermal-cones]{\emph{Thermodynamic distillation process.} The arrow depicts the existence of a thermal operation transforming  $N$ independent initial systems to $\tilde{N}$ independent target systems. The circles and squares represent the initial and target systems with each subsystem described by a different Hamiltonian and prepared in a different state.}
	\label{fig:distillation}
\end{marginfigure}

\begin{equation}\label{eq:initial}
    H^N=\sum_{n=1}^N H^N_{n},\qquad \rho^N=\bigotimes_{n=1}^N \rho^N_{n}, 
\end{equation}
while the target system is described by an arbitrary Hamiltonian $\tilde{H}^N$ and a state $\tilde{\rho}^N=\ketbra{\tilde{E}_k^{N}}{\tilde{E}_k^{N}}$, with $\ket{\tilde{E}_k^{N}}$ being an eigenstate of $\tilde{H}^N$ corresponding to some energy $\tilde{E}_k^{N}$. Note that since $\tilde{H}^N$ is arbitrary, it does not need to describe $N$ particles; in fact, it can even be a Hamiltonian of a single qubit. 

A typical example of this setting is when initial and target systems are given by copies of independent and identical subsystems. More precisely, in this case, the family of initial systems is given by $H^N$ with $H^N_n=H$ and $\rho^N=\rho^{\otimes N}$, while the family of target systems is given by $\tilde{N}$ subsystems, each with a Hamiltonian $\tilde{H}$ and in a state $\ketbra{\tilde{E}_k}{\tilde{E}_k}$. One is then interested in the optimal distillation rate $\tilde{N}/N$ as $N$ tends to infinity. However, we will investigate a more general setting, allowing the subsystems to differ in both state and Hamiltonian, as long as the initial state is uncorrelated.


\subsection{Information-theoretic quantities}

Before we proceed to present our results, let us introduce the necessary information-theoretic quantities. For an initial system $(\rho^N,H^N)$, we introduce the following notation for free energy and free energy fluctuations:
    \begin{subequations}
    \begin{align}
        F^N := \frac{1}{\beta }\sum_{n=1}^N D(\rho^{N}_n\|\gamma^{N}_n),\label{eq:dn}\\
        \sigma^2(F^N) := \frac{1}{\beta^2 }\sum_{n=1}^N V(\rho^{N}_n\|\gamma^{N}_n),\label{eq:vn}\\ 
        \kappa^3(F^N) := \frac{1}{\beta^3 }\sum_{n=1}^N Y(\rho^{N}_n\|\gamma^{N}_n).\label{eq:yn}
    \end{align}
\end{subequations}
We also introduce
\begin{equation}
    \label{eq:deltaF}
    \Delta F^N := \frac{1}{\beta }\left(\sum_{n=1}^N D(\rho^N_n\|\gamma^N_n)-D(\tilde{\rho}^N\|\tilde{\gamma}^N)\right),
\end{equation}
which describes the free energy difference between the initial and target states, as well as
\begin{equation}
    \label{eq:F_diss}
    F^N_{\text{diss}} := \frac{1}{\beta }\left(\sum_{n=1}^N D(\rho^N_n\|\gamma^N_n)-D(\rho^N_{\text{fin}}\|\tilde{\gamma}^N)\right),
\end{equation}
which quantifies the amount of free energy that is dissipated in the distillation process, i.e., the free energy difference between the initial and final states. 

Let us also make two final technical comments. First, we only consider families of initial systems for which the limits of $\sigma^2(F^N)/N$ and $\kappa^3(F^N)/N$ as $N\rightarrow \infty$ are well-defined and non-zero. Second, in what follows, we will use a shorthand notation with $\simeq$, $\lesssim$ and $\gtrsim$ denoting equalities and inequalities up to terms of order $o(\sqrt{N})$.


\section{Optimal distillation error and fluctuation-dissipation relations}
\label{sec:results}
The first pair of our main results concerns thermodynamic distillation processes from incoherent systems. The first theorem connect the optimal distillation error to the free energy fluctuations of the initial state of the system:

\begin{theorem}[Optimal distillation error for incoherent states]
    \label{thm:incoherent}
    For a distillation setting with energy-incoherent initial states, the transformation error $\epsilon_N$ of the optimal $\epsilon$-approximate distillation process in the asymptotic limit is given by
    \begin{equation}
    \label{eq:error_incoherent}
        \lim_{N\to\infty}\epsilon_N = \lim_{N\to\infty}\Phi\left(-\frac{\Delta F^N}{\sigma(F^N)} \right),
    \end{equation}
    where $\Phi$ denotes the cumulative normal distribution function. Moreover, for any $N$ there exists an $\epsilon$-approximate distillation process with the transformation error $\epsilon_N$ bounded by
    \begin{align}
        \label{eq:error_bound}
        \epsilon_N&\leq \Phi\left[-\frac{\Delta F^N}{\sigma(F^N)} \right]+\frac{C\kappa^3(F^N)}{\sigma^3(F^N)},
    \end{align}
    where $C$ is a constant from the Berry-Esseen theorem that is bounded by
    \begin{equation}
        0.4097\leq C \leq 0.4748.
    \end{equation}
\end{theorem}

In such a distillation process, the minimal amount of free energy dissipated is related to free energy fluctuations via:

\begin{theorem}[Fluctuation-dissipation relation for incoherent states]
    \label{thm:incoherent2}
    The minimal amount of free energy dissipated in the optimal (minimising the transformation error $\epsilon$) distillation process from identical incoherent states asymptotically satisfies
    \begin{equation}
        \label{eq:diss}
        F^N_{\text{diss}} \simeq a(\epsilon_N)\sigma(F^N),
    \end{equation}
    where 
    \begin{equation}
        \label{eq:a}
       a(\epsilon)=-\Phi^{-1}(\epsilon)(1-\epsilon)+\frac{\exp\Big\{\frac{-[\Phi^{-1}(\epsilon)]^2}{2}\Big\}}{\sqrt{2\pi}}
    \end{equation}
    and $\Phi^{-1}$ is the inverse function of the cumulative normal distribution function $\Phi$.
\end{theorem}

We prove the above theorems in Secs.~\ref{sec:proof1a}~and~\ref{sec:proof2a}, and here we will briefly discuss their scope and consequences. We start by noting that combining Eqs.~\eqref{eq:error_incoherent}~and~\eqref{eq:diss} yields the optimal amount of dissipated free energy as a function of $\Delta F^N$ and $\sigma(F^N)$:
\begin{align}
    F^N_{\text{diss}} \simeq &  \left[1-\Phi\left(-\frac{\Delta F^N}{\sigma(F^N)}\right)\right]\Delta F^N +\frac{\exp\left[-\frac{(\Delta F^N)^2}{2\sigma^2(F^N)}\right]}{\sqrt{2\pi}}\sigma(F^N).\label{eq:diss_delta_sigma}
\end{align}
Now, for the analysed case of independent initial subsystems, free energy fluctuations $\sigma(F^N)$ scale as $\sqrt{N}$. Thus, we can distinguish three regimes, depending on how the free energy difference between the initial and target states, $\Delta F^N$, behaves with growing $N$:
\begin{equation}
    \lim_{N\to\infty} \frac{\Delta F^N}{\sqrt{N}}=\left\{
    \begin{array}{l}
         \infty,  \\
         -\infty, \\
         \alpha\in \mathbb{R}.
    \end{array}
    \right.
\end{equation}

The first case corresponds to the target state having much smaller free energy than the initial state (as compared to the size of free energy fluctuations). According to Eq.~\eqref{eq:error_incoherent}, the transformation error then approaches zero in the asymptotic limit; while according to Eq.~\eqref{eq:diss_delta_sigma}, the amount of dissipated free energy $F^N_{\text{diss}}\simeq \Delta F^N$, i.e., up to second order asymptotic terms the target and final states have the same free energy. This means that one can get arbitrarily close to the target state with much lower free energy than the initial state. The second case corresponds to the target state having much larger free energy than the initial state. The transformation error then approaches one in the asymptotic limit, while the amount of dissipated free energy $F^N_{\text{diss}}$ approaches zero. This means that it is impossible to even get slightly closer to the target state with much higher free energy than the initial state, and so the optimal process corresponds to doing nothing (that is why there is no free energy dissipated). 

Finally, the third case that forms the essence of Theorems~\ref{thm:incoherent}~and~\ref{thm:incoherent2} corresponds to the target state having free energy very close to that of the initial state (again, the scale is set by the magnitude of free energy fluctuations). Our theorems then directly link the optimal transformation error and the minimal amount of dissipated free energy in the process to the free energy fluctuations of the initial state of the system. For two processes with the same free energy difference $\Delta F^N$, the process involving the initial state with smaller fluctuations will yield a smaller transformation error according to Eq.~\eqref{eq:error_incoherent}. Similarly, since the derivative of Eq.~\eqref{eq:diss_delta_sigma} over $\sigma(F^N)$ for a fixed $\Delta F^N$ is always positive, states with smaller free energy fluctuations will lead to smaller free energy dissipation. As a particular example consider a battery-assisted distillation process, i.e. a thermodynamic transformation from \mbox{$(\rho^N\otimes \ketbra{1}{1}_B,H^N+H_B)$} to \mbox{$(\tilde{\rho}^N\otimes \ketbra{0}{0}_B,H^N+H_B)$}, where the energy gap of the battery system $B$ is $W^N_{\mathrm{cost}}$. Now, the quality of transformation from $\rho^N$ to $\tilde{\rho}^N$ (measured by transformation error $\epsilon_N$) depends on the amount of work $W^N_{\mathrm{cost}}$ that we invest into the process. As expected, to achieve $\epsilon\leq 1/2$, we need to invest at least the difference of free energies $[D(\tilde{\rho}^N\|\tilde{\gamma}^N)-D(\rho^N\|\gamma^N)]/\beta$. However, Theorem~\ref{thm:incoherent} tells us how much more work is needed to decrease the transformation error to a desired level: the more free energy fluctuations there were in $\rho^N$, the more work we need to invest.

Let us also compare our two theorems to the results presented in Ref.~\cite{Chubb2018beyondthermodynamic}. There, the authors studied the incoherent thermodynamic interconversion problem between identical copies of the initial system, $\rho^{\otimes N}$, and identical copies of the target system, $\tilde{\rho}^{\otimes \tilde{N}}$. Here, for the price of the reduced generality of the target state (it has to be an eigenstate of the target Hamiltonian), we obtained a four-fold improvement. First, our result applies to general independent systems, not only to identical copies. Second, the Hamiltonians of the initial and target systems can vary, which is particularly important for applications like work extraction or thermodynamically-free communication. Third, we went beyond the second-order asymptotic result and found a single-shot upper bound on the optimal transformation error $\epsilon_N$, Eq.~\eqref{eq:error_bound}, that holds for any finite $N$. Thus, even in the finite $N$ regime, one can get a guarantee on the transformation error that is approaching the asymptotically optimal value as $N\to\infty$. Finally, we derived the expression for the actual amount of dissipated free energy in the optimal process and related it to the fluctuations of the free energy content of the initial state.

Our second pair of main results is analogous to the first pair, but concerns thermodynamic distillation process from $N$ identical copies of a pure quantum system. Thus, the following two theorems connect the optimal distillation error to the free energy fluctuations of the initial state of the system, and the minimal amount of free energy dissipated in such a distillation process to these fluctuations. To formally state these theorems, we need to introduce a technical notion of a Hamiltonian with incommensurable spectrum. Given any two energy levels, $E_i$ and $E_j$, of such a Hamiltonian, there does not exist natural numbers $m$ and $n$ such that $m E_i=n E_j$. We then have that the optimal distillation error for identical pure state is given by the Theorem:

\begin{theorem}[Optimal distillation error for identical pure states]
    \label{thm:pure}
    For a distillation setting with $N$ identical initial systems, each in a pure state $\ketbra{\psi}{\psi}$ and described by the same Hamiltonian $H$ with incommensurable spectrum, the transformation error $\epsilon_N$ of the optimal $\epsilon$-approximate distillation process in the asymptotic limit is given by
    \begin{equation}
        \label{eq:error_pure}
            \lim_{N\to\infty}\epsilon_N = \lim_{N\to\infty}\Phi\left(-\frac{\Delta F^N}{\sigma(F^N)} \right),
    \end{equation}
    where $\Phi$ denotes the cumulative normal distribution function. Moreover, the result still holds if both the initial and target systems get extended by an ancillary system with an arbitrary Hamiltonian $H_A$, with the initial and target states being some eigenstates of $H_A$.
\end{theorem}
As before, such processes are accompanied by some free energy dissipation, which specifically satisfies:

\begin{theorem}[Fluctuation-dissipation relation for identical pure states]
    \label{thm:pure2}
   The minimal amount of free energy dissipated in the optimal (minimising the transformation error $\epsilon$) distillation process from identical pure states asymptotically satisfies
    \begin{equation}
        \label{eq:diss_pure}
        F^N_{\text{diss}} \gtrsim a(\epsilon_N)\sigma(F^N),
    \end{equation}
    where $a(\epsilon)$ is given by Eq.~\eqref{eq:a}.
\end{theorem}

We prove the above theorems in Secs.~\ref{sec:proof1b}~and~\ref{sec:proof2b}, while here we will only add one comment to the previous discussion. Namely, since for a pure state the free energy fluctuations are just the energy fluctuations
\begin{align}
    \label{eq:energy_fluct0}
    \frac{1}{\beta^2}V(\ketbra{\psi}{\psi}\|\gamma)&=\bra{\psi} H^2\ket{\psi}-\bra{\psi} H\ket{\psi}^2,
\end{align}
and because in the considered scenario all pure states are identical, we have
\begin{equation}
    \label{eq:energy_fluc}
    \frac{1}{N}\sigma^2(F^N)=\langle H^2 \rangle_{\psi}-\langle{H}\rangle_{\psi}^2,
\end{equation}
where we use a shorthand notation \mbox{$\langle \cdot\rangle_\psi=\matrixel{\psi}{\cdot}{\psi}$}. Analogously to the incoherent case, the only non-trivial behaviour of the optimal transformation error happens when $\Delta F^N\simeq \alpha \sqrt{N}$, and its value is then specified by the ratio $\alpha/(\langle H^2 \rangle_{\psi}-\langle{H}\rangle_{\psi}^2)^{1/2}$.


\section{Applications}

In the last section, we characterised optimal thermodynamic distillation processes and subsequently proved a relation between the amount of free energy dissipated in such processes and the free energy fluctuations of the system's initial state. We now apply these results to determine the optimal performance of thermodynamic protocols, including work extraction, information erasure, and thermodynamically-free communication, up to second-order asymptotics in the number $N$ of processed systems.

\subsection{Optimal work extraction}
\label{sec:work-res}

As the first application of our results, we focus on work extraction process from a collection of $N$ non-interacting subsystems with Hamiltonians $H^N_n$ and in incoherent states $\rho^N_n$. As already described in Sec.~\ref{Subsec:work-extraction}, this is just a particular case of a thermodynamic distillation process. We only need to note that the pure battery state does not contribute to fluctuations $\sigma$ and $\kappa$, and that the difference between non-equilibrium free energies of the ground and excited battery states is just the energy difference $W^N_{\text{ext}}$. Then, Theorem~\ref{thm:incoherent} tells us that, in the asymptotic limit, the optimal transformation error for extracting the amount of work $W^N_{\text{ext}}$ is
\begin{equation}
\label{eq:workanderror}
    \lim_{N\to \infty} \epsilon_N = \lim_{N\to\infty}\Phi\left[\frac{W_{\text{ext}}^N-F^N}{\sigma(F^N)}\right].
\end{equation}
We thus clearly see that again we have three cases dependent on the difference $(W_{\text{ext}}^N-F^N)$. To get the asymptotic error different from zero and one, the extracted work $W^N_{\text{ext}}$ has to be of the form
\begin{equation}
    W^N_{\text{ext}}\simeq F^N - \alpha\sqrt{N},
\end{equation}
for some constant $\alpha$. Combining the above two equations yields the following second-order asymptotic expression for the extracted work:
\begin{equation}
    \label{eq:work}
    W_{\text{ext}}^N\simeq F^N + \sigma(F^N)\Phi^{-1}(\epsilon).
\end{equation}
Thus, for a fixed quality of extracted work measured by~$\epsilon$, more work can be extracted from states with smaller free energy fluctuations (assuming that the initial free energy $F^N$ is fixed). This is a direct generalisation of the result obtained in Ref.~\cite{Chubb2018beyondthermodynamic} to a scenario with non-identical initial systems and with a cleaner interpretation of the error in the battery system. We present the comparison between our bounds and the numerically optimised work extraction processes in Fig.~\hyperref[fig:work-numerics]{\ref{fig:work-numerics}a} 

\begin{figure*}
	\includegraphics{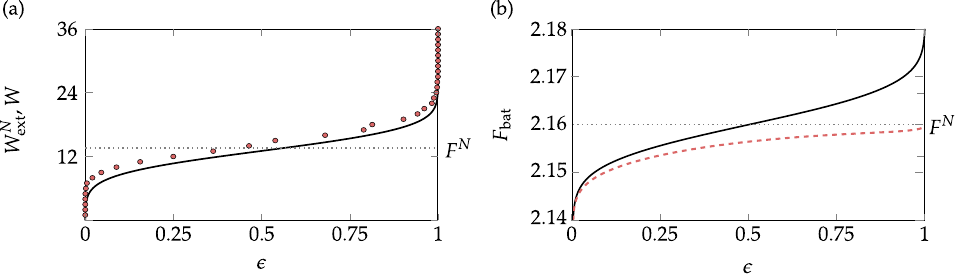}    
	\caption{\emph{Optimal work extraction.} 
		(a) Comparison between the asymptotic approximation, Eq.~\eqref{eq:work}, for the optimal amount of extracted work $W^N_{\mathrm{ext}}$ (solid black line) as a function of transformation error $\epsilon$, and the actual optimal value $W$ (red circles) obtained by explicitly solving the thermomajorisation conditions (see Chapter~\ref{C:mathematical_preliminaries} for details). The inverse temperature of the thermal bath is chosen to be $\beta=1$, while the initial system is composed of 100 two-level subsystems. The first 59 subsystems are described by the Hamiltonian corresponding to a thermal state $0.6\ketbra{0}{0}+0.4\ketbra{1}{1}$, and the remaining 41 subsystems have the Hamiltonian leading to a thermal state $0.75\ketbra{0}{0}+0.25\ketbra{1}{1}$. The initial state of the system is given by 59 copies of a state $0.9\ketbra{0}{0}+0.1\ketbra{1}{1}$ and 41 copies of a state $0.7\ketbra{0}{0}+0.3\ketbra{1}{1}$. The non-equilibrium free energy of the total initial system, $F^N$, is indicated by a grey dotted line. (b) Non-equilibrium free energy $F_{\textrm{bat}}$ of the two-level battery system calculated for the final (dashed red line) and target (solid black line) state of the optimal work extraction process. The inverse temperature of the thermal bath is chosen to be $\beta = 1$, and the initial state that the work is extracted from is composed of 100 copies of a state $0.7\ketbra{0}{0}+0.3\ketbra{1}{1}$. Each subsystem is described by the Hamiltonian corresponding to a thermal state $0.6\ketbra{0}{0}+0.4\ketbra{1}{1}$ and the non-equilibrium free energy of the total initial system, $F^N$, is indicated by a grey dotted line
		\label{fig:work-numerics}}
\end{figure*}

Similarly, by employing Theorem~\ref{thm:pure}, we can investigate the optimal work extraction process from a collection of $N$ non-interacting subsystems with identical Hamiltonians $H$ and each in the same  pure state $\ketbra{\psi}{\psi}$. We simply need to choose the ancillary system $A$ to be the battery $B$ with energy splitting $W^N_{\mathrm{ext}}$ and the initial and target states to be given by $\ket{0}_B$ and $\ket{1}_B$. Also, since all systems are in identical pure states and have the same Hamiltonian, we have $\sigma(F^N)$ specified by Eq.~\eqref{eq:energy_fluc} and
\begin{align}
    \frac{1}{N}{F}^N&=\langle H\rangle_\psi+\frac{\log Z}{\beta}.
\end{align}
As a result, the optimal amount of work extracted from $N$ pure quantum systems up to second-order asymptotic expansion is given by:
\begin{equation}
    \label{eq:work_pure}
    \!\!\!W_{\text{ext}}\!\simeq\! N\left(\!\langle H\rangle_\psi+\frac{\log Z}{\beta} + \frac{\langle H^2\rangle_\psi-\langle H\rangle_\psi^2}{\sqrt{N}}\Phi^{-1}(\epsilon)\!\right)\!.\!
\end{equation}

Finally, let us note that we can employ Theorem~\ref{thm:incoherent2} (and to some extent also Theorem~\ref{thm:pure2}) to investigate the meaning of work quality measured by the transformation error~$\epsilon$. So far we have measured the extracted work as the difference between the free energy of the initial battery's state and its target state that was obtained with success probability $1-\epsilon$. However, due to the aforementioned theorems, we know precisely the free energy of the actual final state of the battery, which can be used to quantify the actual amount of extracted work (with no error). In Fig.~\hyperref[fig:work-numerics]{\ref{fig:work-numerics}b} we present the behaviour of both measures as a function of $\epsilon$, where it is clear that the two notions coincide for small error~$\epsilon$.


\subsection{Optimal cost of erasure}
\label{sec:erasure-res}

In order to obtain the optimal work cost of erasing $N$ two-level systems prepared in incoherent states $\rho^N_n$, we apply Theorem~\ref{thm:incoherent} analogously as in the previous section, but this time to the scenario described in Sec.~\ref{subsec:information-erasure}. We then get the optimal transformation error in the erasure process given by
\begin{equation}
    \lim_{N\to \infty} \epsilon_N = \lim_{N\to\infty}\Phi\left(\frac{\frac{1}{\beta}S(\rho^N)-W^N_{\mathrm{cost}}}{\sigma(F^N)}\right),
\end{equation}
where $S(\rho^N)$ is the entropy of the initial state, and \mbox{$W^N_{\mathrm{cost}}$} is the invested work cost. Using analogous reasoning as in the case of work extraction, we can now obtain the second-order asymptotics for the cost of erasure:
\begin{equation}
    W^N_{\mathrm{cost}}\simeq \frac{S(\rho^N)}{\beta}  - \sigma(F^N)\Phi^{-1}(\epsilon).
\end{equation}

Let us make three brief comments on the above result. First, we only considered the application of the incoherent result, Theorem~\ref{thm:incoherent}, as in the case of trivial Hamiltonians, the erasure of a pure state $\ketbra{\psi}{\psi}^{\otimes N}$ is free (because all unitary transformations are then thermodynamically-free). Of course, our results straightforwardly extend to non-trivial Hamiltonians, but we believe that the simple case we described above is most illustrative and recovers the spirit of the original Landauer's erasure scenario. Second, since the maximally mixed initial state has vanishing free energy fluctuations, $\sigma(F^N)=0$, we cannot directly apply our result (that relates fluctuations of the initial state to dissipation) to get the erasure cost of $N$ completely unknown bits of information. However, using the tools described in Sec.~\ref{sec:math}, it is straightforward to show that in this case, the exact expression (working for all $N$) for the erasure cost is given by
\begin{equation}
    W^N_{\mathrm{cost}}= \frac{N}{\beta}\left( \log 2 - \frac{\log(1-\epsilon)}{N}\right).
\end{equation}
Thus, for the case of zero error one recovers the Landauer's cost of erasure~\cite{Alicki2004}. Third, analogously to the case of work extraction, here also Theorem~\ref{thm:incoherent2} can be employed to investigate the meaning of erasure quality quantified by $\epsilon$.


\subsection{Optimal thermodynamically-free communication rate}
\label{sec:encoding-res}

Finally, we now explain how Theorems~\ref{thm:incoherent}~and~\ref{thm:pure} allow one to obtain the optimal thermodynamically-free encoding rate into a collection of $N$ identical subsystems in either incoherent or pure states. We simply choose the target system to be a single $M$-dimensional quantum system with a trivial Hamiltonian $\tilde{H}=0$ that is prepared in any of the degenerate eigenstates of $\tilde{H}$. Note that the non-equilibrium free energy of such a target system is given by
\begin{equation}
    \frac{1}{\beta}D(\tilde{\rho}^N\|\tilde{\gamma}^N)=\frac{1}{\beta}\log M.
\end{equation}
Our theorems then tell us that in the asymptotic limit, the optimal transformation error $\epsilon$ in the considered distillation process is given by
\begin{equation}
    \lim_{N\to \infty} \epsilon_N = \lim_{N\to\infty}\Phi\left(\frac{\frac{1}{\beta}\log M-F^N}{\sigma(F^N)}\right).
\end{equation}
Rewriting the above, we get the following second-order asymptotic behaviour:
\begin{equation}
    \log M\simeq \beta F^N+\beta\sigma(F^N)\Phi^{-1}(\epsilon).
\end{equation}

Now, the distillation process above can be followed by unitaries that map between $M$ degenerate eigenstates of $\tilde{H}$ that we will simply denote $\ket{1},\dots,\ket{M}$. Crucially, note that such unitaries are thermodynamically-free because they act in a fixed energy subspace. Such a protocol then allows one to encode $M$ messages into $M$ states $\sigma_i$, each one being $\epsilon$-close in infidelity to $\ket{i}$ for $i\in\{1,\dots,M\}$. Decoding the message using a measurement in the eigenbasis of $\tilde{H}$ then leads to the average decoding error $\epsilon_{\mathrm{d}}$ satisfying:
\begin{equation}
    1-\epsilon_{\mathrm{d}}:=\frac{1}{M}\sum_{i=1}^M \matrixel{i}{\sigma_i}{i}=1-\epsilon,
\end{equation}
so that $\epsilon_{\mathrm{d}}=\epsilon$.

Using the communication protocol described above, we then get the following asymptotic lower bound on the optimal thermodynamically-free encoding rate into a state $\rho^N$ [recall Eq.~\eqref{eq:optimal-encoding-rate}]\sidenote{The optimal encoding rate is given by:
\begin{equation*}
    R(\rho^N, \epsilon_{\mathrm{d}}) := \frac{\log [M(\rho^N, \epsilon_{\mathrm{d}})]}{N} \,.   
\end{equation*}}:
\begin{equation}
    R(\rho^N,\epsilon_{\mathrm{d}}) \geq \frac{\beta}{N}(F^N+\sigma(F^N)\Phi^{-1}(\epsilon_{\mathrm{d}}))+o\left(\frac{1}{\sqrt{N}}\right).
\end{equation}
The above lower bound is exactly matching the upper bound for $R(\rho^N,\epsilon_{\mathrm{d}})$ recently derived in Ref.~\cite{korzekwa2019encoding} for a slightly different scenario with $\rho^N_n=\rho$ and $H^N_n=H$ for all $n$, with $\tilde{H}^N=H^N$, and with Gibbs-preserving operation instead of thermal operations. However, the proof presented there can be easily adapted to work in the current case if we keep the first restriction, i.e., when the initial state is $\rho^N=\rho^{\otimes N}$ and all initial subsystems have equal Hamiltonians. We explain in detail how to adapt that proof in Section~\ref{app:optimality}, where we also explain what technical result concerning hypothesis testing relative entropy needs to be proven in order to make the proof also work when subsystems are not identical. Here we conclude that
\begin{equation}
  \!\!\!R(\rho^{\otimes N},\epsilon_{\mathrm{d}}) \!=\! D(\rho\|\gamma)+\frac{\sqrt{V(\rho\|\gamma)}}{\sqrt{N}}\Phi^{-1}(\epsilon_{\mathrm{d}})+o\!\left(\!\frac{1}{\sqrt{N}}\!\right)\!,\!
\end{equation}
where $\rho$ is either a pure or incoherent state.

The above result can be thermodynamically interpreted as the inverse of the Szilard engine. While the Szilard engine converts bits of information into work, the protocol studied here employs the free energy of the system (i.e., the ability to perform work) to encode bits of information. While the asymptotic result was recently proven in Ref.~\cite{narasimhachar2019quantifying}, here we proved that this relation is deeper as it also connects fluctuations of free energy to the optimal average decoding error.

\section{Derivation of the results}
\label{sec:math}

In what follows, we first introduce the mathematical formalism used to study the incoherent distillation process. We then use it to prove Theorems~\ref{thm:incoherent}~and~\ref{thm:incoherent2}. Finally, we also prove Theorems~\ref{thm:pure}~and~\ref{thm:pure2} by first mapping the problem of distillation from pure states to an equivalent incoherent problem, and then using the formalism of incoherent distillations.

\subsection{Incoherent distillation process}
\label{sec:incoherent}

\begin{center}
    \emph{\textbf{Distillation conditions via approximate majorisation}}
\end{center}

We begin by restating the crucial theorem based on Ref.~\cite{brandao2015second} regarding thermodynamic interconversion for incoherent states.
\begin{theorem}[Approximate thermomajorisation]
	\label{thm:thermo_int}
	For the initial and target system with the same thermal distribution $\v{\gamma}$, there exists a thermal operation mapping between an energy-incoherent state $\v{p}$ and a state $\epsilon$-close to $\v{q}$ in infidelity distance, if and only if $\hat{\v{p}}\succ_\epsilon \hat{\v{q}}$.
\end{theorem}

Despite the fact that in our case we want to study the general case of initial and final systems with different Hamiltonians, with a little bit of ingenuity we can still use the above theorem. Namely, we consider a family of total systems composed of the first $N$ subsystems with initial Hamiltonians $H^N_n$, and the remaining part described by the target Hamiltonian $\tilde{H}^N$. We choose initial states of the total system on the first $N$ subsystems to be a general product of incoherent states $\v{p}^N_n$, while the remaining part to be prepared in a thermal equilibrium state $\tilde{\v{\gamma}}^N$ corresponding to $\tilde{H}^N$. Since Gibbs states are free, this setting is thermodynamically equivalent to having just the first $N$ systems with Hamiltonians $H^N_n$ and in states $\v{p}_n^N$. Moreover, for target states of the total system, we choose thermal equilibrium states $\v{\gamma}^N_n$ for the first $N$ subsystems, and sharp states $\tilde{\v{s}}_{k}^N$ of the Hamiltonian $\tilde{H}^N$ for the remaining part. Again, this is thermodynamically equivalent to having just the system with Hamiltonian $\tilde{H}^N$ and in a state $\tilde{\v{s}}_{k}^N$. Thus, employing Theorem~\ref{thm:thermo_int}, an $\epsilon$-approximate distillation process for incoherent states exists if and only if:
\begin{equation}
    \left(\bigotimes_{n=1}^N \hat{\v{p}}^N_n \otimes \hat{\tilde{\v{\gamma}}}^N\right) \succ_\epsilon \left(\bigotimes_{n=1}^{N}\hat{\v{\gamma}}^N_n\otimes\hat{\tilde{\v{s}}}_{k}^N \right) .
\end{equation}
This way, using a single fixed Hamiltonian, we can encode transformations between different Hamiltonians.

Let us introduce the following shorthand notation:
\begin{equation}
    \hat{\v{P}}^N:=\bigotimes_{n=1}^N \hat{\v{p}}^{N}_n,\qquad \hat{\v{G}}^N:=\bigotimes_{n=1}^N {\hat{\v{\gamma}}}^N_n=\bigotimes_{n=1}^N {\v{\eta}}^N_n. \label{eq:PN_GN}
\end{equation}
Then, we can use the previous facts on the embedding map to conclude with the following statement: there exists an $\epsilon$-approximate thermodynamic distillation process from $N$ systems with Hamiltonians $H^N_n$ and in energy-incoherent states $\v{p}^N_n$ to a system with a Hamiltonian $\tilde{H}^N$ and in a sharp energy eigenstate $\tilde{\v{s}}_{k}^N$ if and only if
\begin{equation}
\label{eq:e-thermomajorisation}
    \hat{\v{P}}^N \otimes \tilde{\v{\eta}}^N \succ_{\epsilon} \hat{\v{G}}^N \otimes \tilde{\v{f}}^N_k.
\end{equation}

\begin{center}
    \emph{\textbf{Information-theoretic intermission}}\label{sec:intermission}
\end{center}

Before we proceed, we need to make a short intermission for a few important comments concerning information-theoretic quantities introduced in Section~\ref{subsec:information-theoretic-notions}. For incoherent states $\rho$ and $\gamma$ represented by probability vectors $\v{p}$ and $\v{\gamma}$, these simplify and take the following classical form:
\begin{subequations}
	\begin{align}
	D(\v{p}\|\v{\gamma}):=&\sum_i p_i\left(\log \frac{p_i}{\gamma_i}\right),\\
	V(\v{p}\|\v{\gamma}):=&\sum_i p_i \left(\log \frac{p_i}{\gamma_i} - D(\v{p}\|\v{\gamma})\right)^2,\\
	Y(\v{p}\|\v{\gamma}):=&\sum_i p_i \left|\log \frac{p_i}{\gamma_i} - D(\v{p}\|\v{\gamma})\right|^3.
	\end{align}
\end{subequations}
Moreover, by direct calculation, one can easily show that the above quantities are invariant under embedding, i.e., \mbox{$D(\v{p}\|\v{\gamma})=D(\hat{\v{p}}\|\v{\eta})$}, and the same holds for $V$ and $Y$. Therefore
\begin{subequations}
	\begin{align}
	D(\v{p}\|\v{\gamma})=&D(\hat{\v{p}}\|\v{\eta})=\log D-H(\hat{\v{p}}),\label{eq:invariance1}\\
	V(\v{p}\|\v{\gamma})=&V(\hat{\v{p}}\|\v{\eta})=V(\hat{\v{p}}),\label{eq:invariance2}\\
	Y(\v{p}\|\v{\gamma})=&Y(\hat{\v{p}}\|\v{\eta})=Y(\hat{\v{p}}),
\end{align}
\end{subequations}
where
\begin{subequations}
	\begin{align}
	H(\v{p}):=&\sum_i p_i(-\log p_i),\\
	V(\v{p}):=&\sum_i p_i (\log p_i + H(\v{p}))^2,\\
	Y(\v{p}):=&\sum_i p_i \left|\log p_i + H(\v{p})\right|^3,\label{eq:invariance3}
\end{align}
\end{subequations}
and note that $V(\v{p})=0$ if and only if $\v{p}$ is a flat state.


\begin{center}
    \emph{\textbf{Optimal error for a distillation process}}
\end{center}

In order to transform the approximate majorisation condition from Eq.~\eqref{eq:e-thermomajorisation} into an explicit expression for the optimal transformation error, we start from the following result proven by the authors of Ref.~\cite{Chubb2018beyondthermodynamic}.
\begin{lemma}[Optimal transformation error]
	\label{lem:1shot-distill}
	Let $\v p$ and $\v q$ be distributions with $V(\v q)=0$. Then
	\begin{align}
	\min\left\lbrace \epsilon \middle| \v p\succ_\epsilon \v q \right\rbrace = 1-\sum_{i=1}^{\exp H(\v q)}p_i^{\downarrow}.
	\end{align}
\end{lemma}
Applying the above lemma to Eq.~\eqref{eq:e-thermomajorisation} yields the following expression for the optimal error $\epsilon_N$:
\begin{equation}
    \label{eq:optimalesum}
    \epsilon_N = 1-\sum_{i=1}^{\exp [H(\hat{\v{G}}^{N})+H(\tilde{\v{f}}_k^N)]} \left(\hat{\v{P}}^N\otimes \tilde{\v{\eta}}^N\right)_i^\downarrow.
\end{equation}
Now, for an arbitrary distribution $\v{p}$ and any flat state~$\v{f}$, we make two observations: the size of the support of~$\v{f}$ is simply $\exp (H(\v{f}))$, and the entries of $\v{p}\otimes\v{f}$ are just the copied and scaled entries of $\v{p}$. As a result, the sum of the $l$ largest elements of $\v{p}$ can be expressed as
\begin{equation}
\label{eq:simpleobservation}
    \sum_{i=1}^l p_i^\downarrow = \sum_{i=1}^{l \exp (H(\v{f}))} (\v{p}\otimes \v{f})_i^\downarrow.
\end{equation}
Inverting the above expression we can write
\begin{equation}
\label{eq:simpleobservation2}
    \sum_{i=1}^{l} (\v{p}\otimes \v{f})_i^\downarrow = \sum_{i=1}^{l\exp (-H(\v{f}))} p_i^\downarrow,
\end{equation}
where the summation with non-integer upper limit $x$ should be interpreted as:
\begin{equation}
    \sum_{i=1}^x p_i :=\sum_{i=1}^{\lfloor x \rfloor} p_i +(x-\lfloor x \rfloor) p_{\lceil x \rceil}.
\end{equation}
Since $\tilde{\v{\eta}}$ is a flat state, we conclude that
\begin{equation}
\label{eq:optimal_e}
    \epsilon_N = 1- \sum_{i=1}^{\exp [H(\hat{\v{G}}^{N})+H(\tilde{\v{f}}_k^N)-H(\tilde{\v{\eta}}^N)]} (\hat{\v{P}}^N)_i^\downarrow.
\end{equation}

We see that the error depends crucially on partial ordered sums as above. To deal with these kind of sums, we introduce the function $\chi_{\v{p}}$ defined implicitly by the following equation
\begin{equation}
    \sum_{i=1}^{\chi_{\v{p}}(l)} p_i^\downarrow=\sum_i \{p_i | p_i\geq 1/l\}.
\end{equation}
In words: $\chi_{\v{p}}(l)$ counts the number of entries of $\v{p}$ that are larger than $1/l$. Now, we have the following lemma that will be crucial in proving our theorems.
\newpage
\begin{lemma}[Ordered summation bounds]
    Every $d$-dimensional probability distribution $\v{p}$ satisfies the following for all $l\in\{1,\dots,d\}$ and for all $\alpha\geq 1$:
    \begin{subequations}
        \begin{align}
            \sum_{i=1}^{l} p_i^\downarrow&\geq \sum_{i=1}^{\chi_{\v{p}}(l)} p_i^\downarrow,\label{eq:lower}\\
            \sum_{i=1}^{l} p_i^\downarrow&\leq \sum_{i=1}^{\chi_{\v{p}}(\alpha l)/c} p_i^\downarrow,\label{eq:upper}
        \end{align}
    \end{subequations}
    where 
    \begin{equation}\label{eq:Expofc}
        c=\sqrt{\alpha}\sum_{i=\chi_{\v{p}}(\sqrt{\alpha}l)}^{\chi_{\v{p}}(\alpha l)} p_i^\downarrow.
    \end{equation}
    Moreover, as the probabilities are ordered in the sums given in Eq.~\eqref{eq:lower} and Eq.~\eqref{eq:upper}, it simply follows that
    \begin{equation}
        \label{chi_order}
        \chi_{\v{p}}(l) \leq l\leq \chi_{\v{p}}(\alpha l)/c.
    \end{equation}
\end{lemma}
\begin{proof}
     The first inequality is very easily proven by observing that the number of entries larger than $1/l$, i.e., $\chi_{\v{p}}(l)$, is bounded from above by $l$ due to normalisation. Now, to prove the second inequality, we start from the following observation:
    \begin{equation}
        \sum_{i=1}^{\chi_{\v{p}}(\sqrt{\alpha}l)} \left(p_i^\downarrow-\frac{1}{\sqrt{\alpha}l}\right)\geq
        \sum_{i=1}^{\chi_{\v{p}}(\alpha l)}
        \left(p_i^\downarrow-\frac{1}{\sqrt{\alpha}l}\right),
    \end{equation}
    which comes from the fact that all the extra terms on the right hand side of the above are negative by definition. By rearranging terms we arrive at
    \begin{equation}
        \chi_{\v{p}}(\alpha l)-\chi_{\v{p}}(\sqrt{\alpha}l)\geq c l,
    \end{equation}
    which obviously implies $l \leq \frac{\chi_{\v{p}}({\alpha}l)}{c}$
\end{proof}


\subsection{Proof of Theorem~\ref{thm:incoherent}}
\label{sec:proof1a}

The proof of Theorem~\ref{thm:incoherent} will be divided into two parts. First, we will derive the upper bound for the optimal transformation error $\epsilon_N$, Eq.~\eqref{eq:error_bound}. Then, we will provide a lower bound for $\epsilon_N$ and show that it is approaching the derived upper bound in the asymptotic limit, and so we will prove Eq.~\eqref{eq:error_incoherent}. 


\begin{center}
    \emph{\textbf{Upper bound for the transformation error}}
\end{center}

We start by introducing the following averaged entropic quantities for the total initial distribution $\hat{\v{P}}^N$:
\begin{subequations}\label{Moment_incoherent}
\begin{align}
    \!\! h_N:=&\frac{1}{N} H(\hat{\v{P}}^N)=\frac{1}{N}\sum_{n=1}^N H(\hat{\v{p}}^N_n)=:\frac{1}{N}\sum_{n=1}^N h^N_n,\label{eq:hn}\\
    \!\! v_N:=&\frac{1}{N} V(\hat{\v{P}}^N)=\frac{1}{N}\sum_{n=1}^N V(\hat{\v{p}}^N_n)=:\frac{1}{N}\sum_{n=1}^N v^N_n,\label{eq:vn2}\\
    \!\! y_N:=&\frac{1}{N} Y(\hat{\v{P}}^N)=\frac{1}{N}\sum_{n=1}^N Y(\hat{\v{p}}^N_n)=:\frac{1}{N}\sum_{n=1}^N y^N_n. \label{eq:yn2}
\end{align}
\end{subequations}
Note that the above $v_N$ and $y_N$ are, up to rescaling by $N/\beta$, incoherent versions of $\sigma^2(F^N)$ and $\kappa^3(F^N)$ defined in Eqs.~\eqref{eq:vn}-\eqref{eq:yn}. We also define the function~$l$:
\begin{align}\label{eq:Defn_f}
    l(z):=\exp\left(N h_N+z\sqrt{N v_N}\right).
\end{align}
We now rewrite the upper summation limit appearing in Eq.~\eqref{eq:optimal_e} employing the above function:
\begin{equation}
    \exp [H(\hat{\v{G}}^{N})+H(\tilde{\v{f}}^N_k)-H(\tilde{\v{\eta}}^N)]=l(x),
\end{equation}
so that
\begin{align}\label{eq:Def_x}
    x &= \frac{D(\hat{\v{P}}^N\|\hat{\v{G}}^{N})- D(\tilde{\v{f}}^N_k\|\tilde{\v{\eta}}^N)}{\sqrt{V(\hat{\v{P}}^N)}}.
\end{align}
This can be further transformed by employing the invariance of relative entropic quantities under embedding, Eqs.~\eqref{eq:invariance1}-\eqref{eq:invariance2}, leading to
\begin{align}
\label{eq:dist_cond}
    x&=\frac{\sum\limits_{n=1}^N D(\v{p}_{n}^N\|\v{\gamma}_n^N)-D(\tilde{\v{s}}_{k}^N\|\tilde{\v{\gamma}}^N)}{\left(\sum\limits_{n=1}^N V(\v{p}_{n}^N\|\v{\gamma}_{n}^N)\right)^{\frac{1}{2}}},
\end{align}
which is precisely the argument of $\Phi$ appearing in the statement of Theorem~\ref{thm:incoherent} in Eq.~\eqref{eq:error_incoherent}: $
x = \Delta F^N/\sigma(F^N)$. We conclude that with this $x$, we can then rewrite the expression for the optimal transformation error, Eq.~\eqref{eq:optimal_e}, as
\begin{equation}
    \label{eq:error_exact}
    \epsilon_N = 1- \sum_{i=1}^{l(x)} (\hat{\v{P}}^N)_i^\downarrow.
\end{equation}

Next, we will find an upper bound for the error employing Eq.~\eqref{eq:lower}:
\begin{align}
\label{eq:first_bound}
    \epsilon_N&\leq 1-\sum_{i=1}^{\chi_{\hat{\v{P}}^N}(l(x))}(\hat{\v{P}}^N)_i^\downarrow=1- \sum_{i}\left\lbrace 
    \hat{P}^N_i \middle|
    \hat{P}^N_i \geq \frac{1}{l(x)}	\right\rbrace \!.
    \end{align}
\newpage
In order to evaluate the above sum, consider $N$ discrete random variables $X_n$ taking values $-\log(\hat{\v{p}}^N_n)_i$ with probability $(\hat{\v{p}}^N_n)_i$, so that
\begin{subequations}
\begin{align}\label{eq:Moment_Inc}
    \langle X_n \rangle & = h^N_{n},\\
    \langle (X_n-\langle X_n \rangle)^2 \rangle & = v^N_{n},\\
    \langle\left| X_n-\langle X_n \rangle\right|^3\rangle & = y^N_{n},	\end{align}
\end{subequations}
where the average $\langle \cdot\rangle$ is taken with respect to the distribution $\hat{\v{p}}_{n}^N$. We then have the following
\begin{align}\label{eq: prob_on_set}
    \sum_{i}\left\lbrace 
    \hat{P}^N_i \middle|
    \hat{P}^N_i \geq \frac{1}{l(x)}		
    \right\rbrace &=\sum_{i_1,\dots,i_N}\left\lbrace \prod_{n=1}^{N}
    (\hat{\v{p}}^N_{n})_{i_n} \middle|\prod_{n=1}^{N}(\hat{\v{p}}^N_{n})_{i_n}\geq \frac{1}{l(x)}	\right\rbrace \nonumber\\
    \quad&=\sum_{i_1,\dots,i_N}\left\lbrace \prod_{n=1}^{N}(\hat{\v{p}}^N_{n})_{i_n} \middle|-\sum_{n=1}^{N}\log (\hat{\v{p}}^N_{n})_{i_n}\leq \log l(x)	\right\rbrace \nonumber\\
    &=\Pr\left[\sum_{n=1}^N X_n\leq Nh_N+x\sqrt{Nv_N}\right] \nonumber\\
    &=\Pr\left[\frac{\sum_{n=1}^N (X_n-\langle X_n\rangle)}{\sqrt{\sum_{n=1}^N \langle (X_n-\langle X_n \rangle)^2 \rangle}}\leq x\right].
\end{align}
Now, the Berry-Esseen theorem~\cite{berry1941accuracy,Esseen1942} tells us that
\begin{align}
\left| \Pr\!\left[\!\frac{\sum_{n=1}^N (X_n-\langle X_n\rangle)}{\sqrt{\sum_{n=1}^N \langle (X_n\!-\!\langle X_n \rangle)^2 \rangle }}\leq x\!\right]\! -\! \Phi(x)\right|\!\leq \!\frac{C y_N}{\sqrt{Nv_N^3}},\!\!
\end{align}
where $C$ is a constant that was bounded in Refs.~\cite{Esseen1956, Shevtsova2011} by
\begin{equation}
\label{eq:c_bounds}
    0.4097\leq C \leq 0.4748.
\end{equation}
We thus have
\begin{equation}
    \label{eq:bound_Pn}
    \left|\sum_{i}\left\lbrace 
    \hat{P}^N_i \middle|
    \hat{P}^N_i \geq \frac{1}{l(x)}		
    \right\rbrace -  \Phi(x)\right| \leq \frac{Cy_N}{\sqrt{Nv_N^3}},
\end{equation}
and so we conclude that the error $\epsilon_N$ is bounded from above by
\begin{equation}
\label{eq:error_bound2}
    \epsilon_N\leq \Phi\left(-\frac{\Delta F^N}{\sigma(F^N)}\right)+\frac{C\kappa^3(F^N)}{\sigma^3(F^N)},
\end{equation}
which proves the single-shot upper bound on transformation error, Eq.~\eqref{eq:error_bound}, presented in Theorem~\ref{thm:incoherent}. Also, note that from Eq.~\eqref{eq:error_bound2}, it is clear that if $\lim_{N\to\infty} v_n$ and $\lim_{N\to\infty} y_n$ are well-defined and non-zero (as we assume), then
\begin{equation}
    \label{eq:asymptotic_upper}
    \lim_{N\rightarrow\infty} \epsilon_N \leq \lim_{N\rightarrow\infty} \Phi\left(-\frac{\Delta F^N}{\sigma(F^N)}\right).
\end{equation}

\begin{center}
    \emph{\textbf{Lower bound for the transformation error}}\label{sec:incoherent_lower}
\end{center}

In order to lower bound the expression for the optimal error in the asymptotic limit, we choose \mbox{$\alpha=\exp(\delta\sqrt{N})$} with $\delta>0$ in Eq.~\eqref{eq:upper}. Thus, from Eq.~\eqref{chi_order} and Eq.~\eqref{eq:Defn_f} we have 
\begin{equation}\label{eq:Boundonlx}
    l(x)\leq \frac{\chi_{\hat{\v{P}}^N}[e^{\delta\sqrt{N}}l(x)]}{c}=\frac{\chi_{\hat{\v{P}}^N}[l(x+\delta)]}{c} ,
\end{equation}
where we can evaluate $c$ from Eq.~\eqref{eq:Expofc}:
\begin{align}\label{eq:c}
 c=e^{\frac{\delta \sqrt{N}}{2}} \sum_{i=\chi_{\hat{\v{P}}^N}[l(x+\delta/2)]}^{\chi_{\hat{\v{P}}^N}[l(x+\delta)]} (\hat{\v{P}}^N)_i^\downarrow=e^{\frac{\delta \sqrt{N}}{2}} \Biggl( \sum_{i}
    &\biggl\{\hat{P}^N_i \bigg|
    \hat{P}^N_i \geq \frac{1}{l(x+\delta)}	
    \biggl\}  \nonumber\\
     & -\sum_{i} 
    \hat{P}^N_i \bigg|\hat{P}^N_i\geq\frac{1}{l(x+\delta/2)} \Biggl).
\end{align}
Using Eq.~\eqref{eq:bound_Pn} we can bound the above expression from below as
\begin{equation}
   c \geq e^{\frac{\delta \sqrt{N}}{2}}\Bigg[\Phi(x+\delta)-\Phi(x+\delta/2)-\frac{2Cy_N}{\sqrt{Nv_N^3}}\Bigg].
\end{equation}
Now, for any finite $\delta>0$ it is clear that there exists $N_0$ such that for all $N\geq N_0$ we have $c>1$. Combining with Eq.~\eqref{eq:Boundonlx}, for large enough $N$ we finally have 
\begin{equation}
    \label{eq:Boundonlx2}
    l(x)\leq \frac{\chi_{\hat{\v{P}}^N}[l(x+\delta)]}{c} \leq \chi_{\hat{\v{P}}^N}[l(x+\delta)].
\end{equation}
Hence, using Eq.~\eqref{eq:Boundonlx2}, we have the following lower bound on transformation error
\begin{align}
     \label{eq:first_upper}
    \epsilon_N = 1-\sum_{i=1}^{l(x)}(\hat{\v{P}}^N)_i^\downarrow &\geq 1- \sum_{i=1}^{\chi_{\hat{\v{P}}^N}[l(x+\delta)]} (\hat{\v{P}}^N)_i^\downarrow \nonumber \\ &= 1- \sum_{i}\left\lbrace 
    \hat{P}^N_i \middle|
    \hat{P}^N_i \geq \frac{1}{l(x+\delta)}
    \right\rbrace \nonumber\\
    &\geq 1-\Phi(x+\delta)-\frac{Cy_N}{\sqrt{Nv_N^3}},
\end{align}
where in the last line we used Eq.~\eqref{eq:bound_Pn} again. It is thus clear that
\begin{equation}
    \lim_{N\to\infty} \epsilon_N\geq 1-\lim_{N\to\infty} \Phi(x+\delta)=\lim_{N\to\infty} \Phi(-x-\delta)
\end{equation}
and, since it works for any $\delta>0$, we conclude that
\begin{equation}
    \lim_{N\to\infty} \epsilon_N\geq \lim_{N\to\infty} \Phi\left[-\frac{\Delta F^N}{\sigma(F^N)}\right].
\end{equation}
Combining the above with the bound obtained in Eq.~\eqref{eq:asymptotic_upper}, we arrive at
\begin{equation}\label{eq:err_opt}
    \lim_{N\to\infty}\epsilon_N=\lim_{N\to\infty} \Phi\left[-\frac{\Delta F^N}{\sigma(F^N)}\right],
\end{equation}
which proves the asymptotic expression for the transformation error, Eq.~\eqref{eq:error_incoherent}, presented in Theorem~\ref{thm:incoherent}.


\subsection{Proof of Theorem~\ref{thm:incoherent2}}
\label{sec:proof2a}

The proof of Theorem~\ref{thm:incoherent2} will be divided into two parts. First, we will find the embedded version of the optimal final state minimising the dissipation of free energy $F^N_{\mathrm{diss}}$, and derive the expression for $F^N_{\mathrm{diss}}$ as a function of the initial state. Then, we will calculate $F^N_{\mathrm{diss}}$ up to second order asymptotic terms by upper and lower bounding the expression, and showing that the bounds coincide.

\begin{center}
    \emph{\textbf{Deriving optimal dissipation}}
\end{center}

We start by presenting an extension of Lemma~\ref{lem:1shot-distill} that not only specifies the optimal transformation error for approximate majorisation but also yields the optimal final state.

\begin{lemma}[Optimal final state]\label{Opt_Diss_free}
    Let $\v{p}$ and $\v{q}$ be distributions of size $d$ with \mbox{$V(\v{q})=0$} and \mbox{$H(\v{q})=\log L$}. Then, states $\v{r}$ saturating $\v{p}\succ_\epsilon \v{q}$, i.e., such that \mbox{$\v{p}\succ\v{r}$} and $F(\v{q},\v{r})=1-\epsilon$ with $\epsilon$ being the minimal value specified by Lemma~\ref{lem:1shot-distill}, are given by $\Pi \v{q}^*$, where $\Pi$ is an arbitrary permutation,
    \begin{eqnarray}\label{eq:qstar}
     \v{q}^*&:=&
     \begin{cases}
       \frac{1-\epsilon}{L} &\quad \mathrm{for~} i\leq L,\\
       B \v{p}'_{i-L} &\quad\mathrm{for~}i>L,\\
     \end{cases}\, 
\end{eqnarray}
$B$ is an arbitrary $(d-L)\times(d-L)$ bistochastic matrix, and $\v{p}'$ is a vector of size $(d-L)$ with $p'_i=p^\downarrow_{i+L}$. Moreover, among all such distributions $\Pi\v{q}^*$, the entropy is minimised for the one with $B$ being the identity matrix.
\end{lemma}

The proof of the above lemma can be found in Section~\ref{app:min_out}. Here, we apply it to the central Eq.~\eqref{eq:e-thermomajorisation} that specifies the conditions for the investigated $\epsilon$-approximate thermodynamic distillation process. As a result, the actual total final state in the embedded picture, $\hat{\v{F}}^N$, which is $\epsilon$ away from the target state $\hat{\v{G}}^N \otimes \tilde{\v{f}}^N_k$, is given, up to permutations, by
\begin{eqnarray}
\label{eq:distribution-total}
     \hat{\v{F}}^N&=&
     \begin{cases}
       \frac{1-\epsilon}{K} &\quad \mathrm{for~} i\leq K,\\
       (\hat{\v{P}}^N\otimes\tilde{\v{\eta}}^N)_{i}^\downarrow &\quad\mathrm{for~}i>K.\\
     \end{cases}\, 
\end{eqnarray}
where
\begin{equation}
    K=\exp(H(\hat{\v{G}}^N)+H(\tilde{\v{f}}^N_k)).
\end{equation}
Note that we have chosen $\hat{\v{F}}^N$ to minimise entropy since, due to Eq.~\eqref{eq:invariance1}, it translates into the real (unembedded) final state $\v{F}^N$ with maximal free energy, i.e., it leads to minimal free energy dissipation.

The next step consists of going from the embedded to the unembedded picture. Importantly, if a state in the embedded picture is not uniform within each embedding box, the unembedding will effectively lead to the loss of free energy. However, we note that the final state is given by Eq.~\eqref{eq:distribution-total} only up to permutations. Thus, one may freely rearrange its elements to minimise such a loss. In particular, in Section~\ref{app:DDF} we show that for the case of identical initial states (i.e., $\v{P}^n=\v{p}^{\otimes N}$), there exists a permutation that transforms $\hat{\v{F}}^N$ so that it is uniform in almost all embedding boxes, which leads to exponentially small dissipation of free energy (i.e., no dissipation up to second-order asymptotics).  Therefore, employing the definition of dissipated free energy from Eq.~\eqref{eq:F_diss} and the above discussion, we have
\begin{align}
    F^N_{\mathrm{diss}}=&\frac{1}{\beta}\left( D(\v{P}^N\|\v{G}^N)-D(\v{F}^N\|\v{G}^N\otimes \tilde{\v{G}}^N)\right)\nonumber\\
    \simeq&\frac{1}{\beta}\left( H(\hat{\v{F}}^N)-H(\hat{\v{P}}^N)-\log \tilde{D}\right).
\end{align}
Here, $\tilde{D}$ is the embedding constant defined by $\tilde{\v{G}}^N$ according to Eq.~\eqref{eq:definitiongibbsvec}, and similarly $D$ will denote this embedding constant for $\v{G}^N$.
Note, however, that these can always be chosen to be equal, since the only thing that matters is that $\tilde{G}^N_k=\tilde{D}_k/\tilde{D}$ and ${{G}}^N_k={D}_k/D$, i.e., the change in $\tilde{D}$ or $D$ can be compensated by the appropriate change in $\tilde{D}_k$ or $D_k$.

\vspace{-0.1cm}
Next, noting that $K=D\tilde{D}_k=\tilde{D}\tilde{D}_k$, we calculate the entropy of $\hat{\v{F}}^N$:
\begin{align}
    H(\hat{\v{F}}^N) =& -\sum_{i=1}^K \frac{1-\epsilon}{K}\log\left(\frac{1-\epsilon}{K}\right) - \sum_{i > \tilde{D}_k} (\hat{\v{P}}^N)_i{^{\downarrow}} \log \frac{(\hat{\v{P}}^N)_i{^{\downarrow}}}{\tilde{D}}\nonumber\\
    =&-(1-\epsilon)\log(1-\epsilon)+(1-\epsilon)\log K\nonumber\\
    & \hspace{2cm}- \sum_{i > \tilde{D}_k} (\hat{\v{P}}^N)_i{^{\downarrow}} \log (\hat{\v{P}}^N)_i{^{\downarrow}} +\epsilon \log \tilde{D}\nonumber\\
    \simeq & \log \tilde{D} +(1-\epsilon)\log\tilde{D}_k- \sum_{i > \tilde{D}_k} (\hat{\v{P}}^N)_i{^{\downarrow}} \log (\hat{\v{P}}^N)_i{^{\downarrow}},
\end{align}
where we have dropped the term $(1-\epsilon) \log (1-\epsilon)$ as it is constant (i.e. it does not scale with $N$). 
Therefore, the dissipated free energy in the optimal distillation process is simply given by
\begin{equation}
    \label{eq:dissipatedfreeenergyapprox}
   \!\! F^N_{\text{diss}} \simeq\frac{1}{\beta}\!\left(\! (1-\epsilon)\log \tilde{D}_k +\sum_{i = 1}^{\tilde{D}_k} (\hat{\v{P}}^N)_i{^{\downarrow}} \log (\hat{\v{P}}^N)_i{^{\downarrow}} \!\right)\!.\!
\end{equation}


\begin{center}
    \emph{\textbf{Calculating optimal dissipation}}
\end{center}

We now proceed to provide bounds for $F^N_{\text{diss}}$. To do so, we first make use of Eq.\eqref{chi_order} and the fact that $\log x$ is negative for all $x\in(0,1)$ to write 
\begin{align}
\label{eq:ineq}
    \!\!\sum_{i = 1}^{\tilde{D}_k} (\hat{\v{P}}^N)_i{^{\downarrow}} \log (\hat{\v{P}}^N)_i{^{\downarrow}} \leq \sum_{i = 1}^{\chi_{\hat{\v{P}}^N}(\tilde{D}_k)} (\hat{\v{P}}^N)_i{^{\downarrow}} \log (\hat{\v{P}}^N)_i{^{\downarrow}} \!.\!
\end{align}
The right hand side of the above can then be recast as follows,
\begin{align}
&\!\!\sum_{i = 1}^{\chi_{\hat{\v{P}}^N}(\tilde{D}_k)} (\hat{\v{P}}^N)_i{^{\downarrow}} \log (\hat{\v{P}}^N)_i{^{\downarrow}}\nonumber =\sum_{i}\left\lbrace 
    \hat{P}^N_i\log (\hat{P}^N_i) \middle|
    \hat{P}^N_i \geq \frac{1}{\tilde{D}_k}		
    \right\rbrace \nonumber \\
    &~=\!\!\sum_{i_1,\dots,i_N}\!\left\lbrace \prod_{n=1}^{N}
    (\hat{\v{p}}^N_{n})_{i_n}\sum_{m=1}^N\log(\hat{\v{p}}^N_{m})_{i_m} \middle|\prod_{n=1}^{N}(\hat{\v{p}}^N_{n})_{i_n}\geq \frac{1}{\tilde{D}_k}\!\right\rbrace  \nonumber\\
    &~=\!\!\sum_{i_1,\dots,i_N}\!\left\lbrace \prod_{n=1}^{N}
    (\hat{\v{p}}^N_{n})_{i_n}\sum_{m=1}^N\log(\hat{\v{p}}^N_{m})_{i_m}\right. \times \nonumber\\
    &\hspace{4cm} \times \left| -\sum_{m=1}^N\log(\hat{\v{p}}^N_{m})_{i_m}\leq \log \tilde{D}_k \right\rbrace.
    \label{eq:PNlogPN}
\end{align}
Let us now note that from the definition in Eq.~\eqref{eq:deltaF} we have (here we employ the notation introduced in Eqs.~\eqref{eq:hn}-\eqref{eq:yn2}):
\begin{align}
    \beta\Delta F^N &= \sum\limits_{n=1}^N D(\v{p}_{n}^N\|\v{\gamma}_n^N)-D(\tilde{\v{s}}_{k}^N\|\tilde{\v{\gamma}}^N)\nonumber\\
    &= H(\tilde{\v{f}}_{k}^N) - \sum\limits_{n=1}^N H(\hat{\v{p}}_{n}^N)=\log \tilde{D}_k - N h_N .
\end{align}
Next, since for non-trivial error we need
\begin{equation}
\label{eq:xsigmF}
    \beta\Delta F^N = x\beta\sigma(F^N) = x \sqrt{N v_N},
\end{equation}
with some constant $x$, we can rewrite $\log \tilde{D}_k$ as
\begin{equation}
    \log \tilde{D}_k = N h_N +x \sqrt{N v_N}. 
    \label{eq:Dktilde}
\end{equation}
\vspace{-0.1cm}
Coming back to Eq.~\eqref{eq:PNlogPN}, we can rewrite it by introducing $N$ discrete random variables $\{X_n\}_{n=1}^N$ assuming values $-\log (\hat{\v{p}}_n^N)_{i_n}$ with probability $(\hat{\v{p}}^N_n)_{i_n}$. Crucially, since the sum of their averages is $Nh_N$ and their total variance is $Nv_N$, the condition in the summation in Eq.~\eqref{eq:PNlogPN} simply becomes:
\begin{equation}
    \label{eq:constraint}
    Y_N:=\frac{\sum_{n=1}^N(X_n-\langle X_n\rangle)}{\sqrt{\sum_{n=1}^N\langle(X_n-\langle X_n\rangle)^2\rangle}} \leq x.
\end{equation}
Thus, we can write Eq.~\eqref{eq:PNlogPN} in a compact way as follows
\begin{align}
    &\sum_{i = 1}^{\chi_{\hat{\v{P}}^N}(\tilde{D}_k)} (\hat{\v{P}}^N)_i{^{\downarrow}} \log (\hat{\v{P}}^N)_i{^{\downarrow}}=-\int_T \sum_{n=1}^N
    X_N  d P, 
\end{align}
where $P$ is discrete probability measure given by $P(i_1,\ldots,i_N)=\prod_{n=1}^{N} (\hat{\v{p}}^N_{n})_{i_n}$ and $T$ is a region satisfying the constraint $Y_N<x$. Noting that 
\begin{align}
    \sum_{n=1}^N X_n= \sqrt{N v_n} Y_N + N h_N
\end{align}
we can further rewrite it as 
\begin{align}
    \!\!\!\int\limits_T \sum_{n=1}^N
    X_N  d P=
    \sqrt{Nv_N}\!\!\int\limits_{Y_N\leq x}\!\!  
    Y_N dP   + N h_N \!\!\int\limits_{Y_N\leq x}\!\! dP.\!
\end{align}
The second integral on the right hand side of the above was already calculated and is equal to $1-\epsilon$, where $\epsilon$ is the optimal transformation error from Theorem~\ref{thm:incoherent}. Further, since $Y_N$ is a standarized sum of independent random variables, its distribution tends to a normal Gaussian distribution with density denoted by  $\phi(x)$, i.e.,
\begin{equation}
    \int\limits_{Y_N\leq x}\!\!  
    Y_N dP\simeq\int_{-\infty}^x  y \phi(y) d y = -\frac{e^{-x^2/2}}{\sqrt{2\pi}}.
\end{equation}
We thus obtain: 
\begin{align}
    \sum_{i = 1}^{\chi_{\hat{\v{P}}^N}(\tilde{D}_k)} (\hat{\v{P}}^N)_i{^{\downarrow}} \log (\hat{\v{P}}^N)_i{^{\downarrow}}&\simeq \sqrt{N v_N}\frac{e^{-x^2/2}}{\sqrt{2\pi}} - N h_N(1-\epsilon)\nonumber\\
    &= \sqrt{Nv_N}\frac{e^{-x^2/2}}{\sqrt{2\pi}} -(1-\epsilon) (\log \tilde{D}_k-x\sqrt{Nv_N})\nonumber\\
    &= \beta\sigma(F^N)\!\left[\!\frac{e^{-x^2/2}}{\sqrt{2\pi}}+x(1-\epsilon)\!\right]\! -\!(1-\epsilon) \log \tilde{D}_k.
\end{align}

We can now use the inequality from Eq.~\eqref{eq:ineq} and substitute the above to Eq.~\eqref{eq:dissipatedfreeenergyapprox} to arrive at:
\begin{equation}
   F^N_{\text{diss}} \lesssim \sigma(F^N)\left(\frac{e^{-x^2/2}}{\sqrt{2\pi}}+x(1-\epsilon)\right).
\end{equation}
Next, employing Eq.~\eqref{eq:xsigmF} and the expression for optimal transformation error from Theorem~\ref{thm:incoherent}, we can re-express $x$ as
\begin{equation}
    x=-\Phi^{-1}(\epsilon)
\end{equation}
to finally obtain
\begin{equation}    
\label{eq:diss_inc_upper}
   F^N_{\text{diss}} \lesssim \sigma(F^N)\left[\frac{e^{\frac{-(\Phi^{-1}(\epsilon))^2}{2}}}{\sqrt{2\pi}}-\Phi^{-1}(\epsilon)(1-\epsilon)\right].
\end{equation}

To provide a lower bound of $F^N_{\text{diss}}$ given in Eq.~\eqref{eq:dissipatedfreeenergyapprox}, we simply follow the argument we have given in Sec.~\ref{sec:incoherent_lower}. From Eq.~\eqref{eq:upper} it straightforwardly follows that 
\begin{align}
    \sum_{i=1}^{\tilde{D}_k} (\hat{\v{P}}^N)_i^\downarrow\log(\hat{\v{P}}^N)_i^\downarrow &\geq \sum_{i=1}^{\frac{\chi_{\hat{\v{P}}^N}(\tilde{D}_k e^{\delta\sqrt{N}})}{c}} (\hat{\v{P}}^N)_i^\downarrow\log(\hat{\v{P}}^N)_i^\downarrow,
\end{align}
where $c$, as before, can be lower bounded by 1 for large enough $N$. We thus get
\begin{equation}
 F^N_{\mathrm{diss}} \gtrsim\frac{1}{\beta}\left[ (1-\epsilon)\log \tilde{D}_k +\sum_{i=1}^{\chi_{\hat{\v{P}}^N}(\tilde{D}_k e^{\delta\sqrt{N}})} (\hat{\v{P}}^N)_i^\downarrow\log(\hat{\v{P}}^N)_i^\downarrow \right].
\end{equation}
Since the above inequality holds for any $\delta>0$, therefore the limit $\delta\rightarrow 0$ we have
\begin{equation}
   F^N_{\mathrm{diss}} \gtrsim \sigma(F^N)\left\{\frac{e^{\frac{-[\Phi^{-1}(\epsilon)]^2}{2}}}{\sqrt{2\pi}}-\Phi^{-1}(\epsilon)(1-\epsilon)\right\}.
\end{equation}
Combining the above with the upper bound from Eq.~\eqref{eq:diss_inc_upper} we finally arrive at
\begin{equation}
      F^N_{\text{diss}}\simeq a(\epsilon) \sigma(F^N), 
\end{equation}
where $a(\epsilon)$ is given by Eq.~\eqref{eq:a}.


\subsection{Proof of Theorem~\ref{thm:pure}}\label{sec:proof1b}

The proof of Theorem~\ref{thm:pure} will be divided into three parts. First, we will show that a thermodynamic distillation process from a general state $\rho$ can be reduced to a distillation process from an incoherent state that is a dephased version of $\rho$. Employing this observation, we will recast the problem under consideration in terms of approximate majorisation and thermomajorisation as described in Sec.~\ref{sec:incoherent}. Then, in the second part of the proof, we will derive the upper bound for the optimal transformation error $\epsilon_N$. Finally, in the third part, we will provide a lower bound for $\epsilon_N$ and show that it is approaching the derived upper bound in the asymptotic limit, and so we will prove Eq.~\eqref{eq:error_pure}. 


\begin{center}
    \emph{\textbf{Reducing the problem to the incoherent case}}
\end{center}

The thermodynamic distillation problem under investigation is specified as follows. The family of initial systems consists of a collection of $N$ identical subsystems, each with the same Hamiltonian
\begin{equation}
    H=\sum_{i=1}^d E_i \ketbra{E_i}{E_i},
\end{equation}
and an ancillary system with an arbitrary Hamiltonian~$H_A$ (note that the ancillary system can always be ignored by simply choosing its dimension to be 1). The family of initial states is given by
\begin{equation}\label{eq:Initial_pure_state_tensor}
    \rho^N= \psi^{\otimes N} \otimes \ketbra{E^A_0}{E^A_0},
\end{equation}
where
\begin{equation}
    \psi=\ketbra{\psi}{\psi},\quad \ket{\psi}=\sum_{i=1}^d \sqrt{p_i} e^{i\phi_i} \ket{E_i},
\end{equation}
is an arbitrary pure state and $\ket{E^A_0}$ is an eigenstate of $H_A$ with energy $E^A_0$. The family of target systems is composed of subsystems described by arbitrary Hamiltonians~$\tilde{H}^N$ and a subsystem described by the Hamiltonian~$H_A$. The family of target states is given by 
\begin{equation}
    \tilde{\rho}^N= \ketbra{\tilde{E}^{N}_k}{\tilde{E}^{N}_k} \otimes \ketbra{E^A_1}{E^A_1},
\end{equation}
where $\ket{\tilde{E}^{N}_k}$ is some eigenstate of $\tilde{H}^N$ and $\ket{E^A_1}$ is an eigenstate of $H_A$ with energy $E^A_1$. We are thus interested in the existence of a thermal operation $\E^{\beta}$ approximately performing the following transformation:
\begin{equation}
    \psi^{\otimes N} \otimes \ketbra{E^A_0}{E^A_0} ~\xrightarrow{~\E^{\beta}~}~ \ketbra{\tilde{E}^{N}_k}{\tilde{E}^{N}_k} \otimes \ketbra{E^A_1}{E^A_1}.
\end{equation}

We now have the following simple, but very useful, lemma.
\begin{lemma}[Dephasing invariance of TOs]
	\label{lem:pure-to-incoherent}
   Every incoherent state $\sigma$ achievable from a state $\rho$ through a thermal operation is also achievable from $\D(\rho)$, where $\D$ is the dephasing operation destroying coherence between different energy subspaces:
	\begin{align}
	\exists \E^{\beta}: \E^{\beta}(\rho) = \sigma \quad\Leftrightarrow \quad \E^{\beta}[\D(\rho)] = \sigma.
	\end{align}
\end{lemma}
\begin{proof}
First, for a given $\rho$ and incoherent $\sigma$, assume that there exists a thermal operation $\E^{\beta}$ such that $\E^{\beta}(\rho) = \sigma$. Now, employing the fact that every thermal operation is covariant with respect to time-translations~\cite{lostaglio2015description}, and using the fact that incoherent $\sigma$ by definition satisfies $\mathcal{D}(\sigma) = \sigma$, we get 
\begin{equation}
\label{dephasingrho}
    \E^{\beta}[\D(\rho)]=\D[\E^{\beta}(\rho)]=\D(\sigma)=\sigma.  
\end{equation}
Likewise, the reverse implication holds by noting that the dephasing operation is a thermal operation.
\end{proof}

Because the target state in our case is incoherent, we can use the above result to restate our problem as the existence of a thermal operation $\E^{\beta}$ approximately performing the following transformation 
\begin{equation}
    \D(\psi^{\otimes N} \otimes \ketbra{E^A_0}{E^A_0}) ~\xrightarrow{~\E^{\beta}~}~ \ketbra{\tilde{E}^{N}_k}{\tilde{E}^{N}_k} \otimes \ketbra{E^A_1}{E^A_1}.
\end{equation}
Since
\begin{equation}\label{eq: Dephased initial state}
    \D(\psi^{\otimes N} \otimes \ketbra{E^A_0}{E^A_0})=\D(\psi^{\otimes N}) \otimes \ketbra{E^A_0}{E^A_0},
\end{equation}
our problem further reduces to understanding the structure of the incoherent state $\D(\psi^{\otimes N})$. It is block-diagonal in the energy eigenbasis and can be diagonalised using thermal operations (since unitaries in a fixed energy subspace are free operations). After such a procedure, we end up with an incoherent state that is described by the probability distribution $\v{P}^N$ over the multi-index set $\v{k}$ 
\begin{equation}
    \label{eq:Multinomial_prob_vector}
    P^N_{\v{k}} =\binom{N}{k_1,...,k_d}\prod_{i=1}^d p^{k_i}_i.
\end{equation}
Note that $P^N_{\v{k}}$ specifies the probability of $k_1$ systems being in energy state $E_1$, $k_2$ systems being in energy state $E_2$, and so on; and that we made a technical assumption that energy levels are incommensurable, so that each vector $\v{k}$ corresponds to a different value of total energy.

We have thus reduced the problem of thermodynamic distillation from pure states to thermodynamic distillation from incoherent states. More precisely, let us denote the sharp distributions corresponding to $\ket{E^A_i}$ by $\v{s}^A_i$ and the corresponding flat states after embedding by $\v{f}^A_i$. As before, we also use $\tilde{\v{s}}_k^N$ and $\tilde{\v{f}}_k^N$ to denote distributions related to the sharp state $\ket{\tilde{E}^{N}_k}$ and its corresponding flat state. The embedded Gibbs state corresponding to $H^N$  will be again denoted by $\hat{\v{G}}^N$, however now it has an even simpler form than in Eq.~\eqref{eq:PN_GN}, as the initial systems have identical Hamiltonians:
\begin{equation}
    \hat{\v{G}}^N=\hat{\v{\gamma}}^{\otimes N} = {\v{\eta}}^{\otimes N}.
\end{equation}
Similarly, $\hat{\v{P}}^N$ will be used to denote the embedded initial state (even though it now has a different form than in Eq.~\eqref{eq:PN_GN}):
\begin{equation}
     \hat{P}^N_{\v{k}, g_{\v{k}}} = \binom{N}{k_1,...,k_d}\prod^{d}_{i=1}\frac{p^{k_i}_i}{D^{k_i}_i}, 
\end{equation}
where $\gamma_i=D_i/D$ and 
\begin{equation}
    g_{\v{k}} \in \bigg\{1, ..., \prod_{i=1}^d D^{k_i}_i \bigg\}
\end{equation}
is an index for the degeneracy coming from embedding. With the notation set, our distillation problem can now be written as
\begin{equation}
\label{thermomajorisation__embedded_pure_state}
    {\hat{\v{P}}}^N \otimes \v{f}^A_0 \otimes \tilde{\v{\eta}} \succ_{\epsilon}  \hat{\v{G}}^N \otimes \v{f}^A_1 \otimes \tilde{\v{f}_k}.
\end{equation}

\begin{center}
    \emph{\textbf{{Upper bound for the transformation error}}}
\end{center}

We begin by observing that our target distribution in Eq.~\eqref{thermomajorisation__embedded_pure_state} is flat, and so \mbox{$V(\hat{\v{G}}^N \otimes \v{f}^A_1 \otimes \tilde{\v{f}_k})=0$}. Thus, we can employ Lemma~\ref{lem:1shot-distill} and Eq.~\eqref{eq:simpleobservation2} to get the following expression for the optimal transformation error:
\begin{equation}
    \label{eq:Exactexpoferror}
    \epsilon_N= 1-\sum_{j=1}^{L}(\hat{\v{P}}^N)^{\downarrow}_j
\end{equation}
where $L$ is given by
\begin{align}
    L & = \exp[H(\hat{\v{G}}^{N})+H(\v{f}^A_1)+H(\tilde{\v{f}}_k)-H(\v{f}^A_0)-H(\tilde{\v{\eta}})]\nonumber\\
    & = \exp[H(\hat{\v{G}}^{N})-D(\tilde{\v{f}}_k\|\tilde{\v{\eta}})-\beta(E_1^A-E_0^A)]. \label{eq:L}
\end{align}
Notice that in the current case $\Delta F^N$, defined in Eq.~\eqref{eq:deltaF}, is given by
\begin{align}
   \!\!\! \Delta F^N&\!=\!\frac{1}{\beta}\Big(D(\psi^{\otimes N}\|\gamma^{\otimes N})+ D(\ketbra{E_0^A}{E_0^A}\|\gamma_A)\nonumber\\
    \!\!&\qquad\quad-D(\ketbra{\tilde{E}^N_{k}}{\tilde{E}^N_{k}}\|\tilde{{\gamma}}) -D(\ketbra{E_1^A}{E_1^A}\|\gamma_A)\Big)\nonumber\\
   \!\!\!  &\!=\! \frac{1}{\beta}\Big(ND(\psi\|\gamma)\!-\!D(\tilde{\v{f}}_k\|\tilde{\v{\eta}})\!-\!\beta(E_1^A\!-\!E_0^A)\Big).
\end{align}
Using the above we can then rewrite $L$ as
\begin{align}
    \label{eq:lnL}
    \log L =&  \beta \Delta F^N+H(\hat{\v{G}}^N)-ND(\psi\|\gamma).
\end{align}
Now, employing Eq.~\eqref{eq:lower} and the above, we provide the upper bound for $\epsilon_N$:
\begin{align}
\epsilon_N&\leq 1- \sum_{\v{k},g_{\v{k}}}\left\lbrace 
    \hat{P}^N_{\v{k},g_{\v{k}}} \middle|
    \hat{P}^N_{\v{k},g_{\v{k}}}\geq \frac{1}{L}	\right\rbrace= 1- \sum_{\v{k}}\left\lbrace 
    {P}^N_{\v{k}} \middle|
    {P}^N_{\v{k}}\geq \frac{\prod_{i=1}^d D_i^{k_i}}{L}	\right\rbrace\nonumber\\
    \!\!\!&= 1- \sum_{\v{k}}\left\lbrace 
    {P}^N_{\v{k}} \middle|
    \log{P}^N_{\v{k}}\geq \sum_{i=1}^d k_i\log D_i-\log L	\right\rbrace \nonumber\\
    \!\!\!&= 1- \sum_{\v{k}}\bigg\{
    {P}^N_{\v{k}} \Bigg|
    \frac{\log{P}^N_{\v{k}}}{N}\geq \sum_{i=1}^d \frac{k_i}{N}\log \gamma_i+D(\psi\|\gamma)-\frac{\beta}{N} \Delta F^N \Bigg\}.\label{eq:ErrorCalc}
 \end{align}

To simplify the calculation of the upper bound of $\epsilon_N$, we rewrite each $\v{k}$ as a function of a vector $\v{s}$ such that 
\begin{equation}
    \label{eq:k_of_s}
    \v{k}=\v{k}(\v{s})=N\v{p}+\sqrt{N}\v{s},
\end{equation}
with $\sum_{i=1}^d s_i=0$. We then note that
\begin{equation}\label{eq:KLD_of_psi}
    D(\psi\|\gamma)=-\sum_{i=1}^d p_i \log \gamma_i
\end{equation}
and so the condition in Eq.~\eqref{eq:ErrorCalc} can be rewritten as
\begin{align}\label{eq:logP^N}
    \frac{\log{P}^N_{\v{k}(\v{s})}}{N}&\geq \frac{1}{\sqrt{N}}\sum_{i=1}^d s_i\log \gamma_i-\frac{\beta}{N} \Delta F^N  = -\frac{\beta}{\sqrt{N}}\sum_{i=1}^d s_iE_i-\frac{\beta}{N} \Delta F^N.
 \end{align}
As we rigorously argue in Section~\ref{app:log}, the left-hand side of the above vanishes much quicker than the right-hand side when $N\to\infty$, leading to
\begin{align}
\lim_{N\to\infty}\epsilon_N\leq 1\!-\! \lim_{N\to\infty}\sum_{\v{s}}\bigg\{
    {P}^N_{\v{k}(\v{s})} \Bigg|
     \sum_{i=1}^d s_iE_i\geq- \frac{\Delta F^N}{\sqrt{N}}\Bigg\}.\label{Final_Bound_on_Epsilon}
\end{align}

\begin{figure*}
    \centering
    \includegraphics[width=15.940cm]{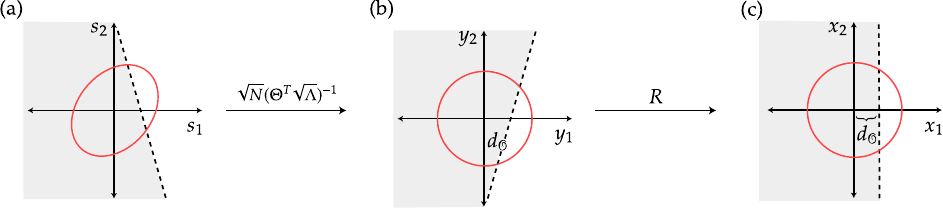}
    \caption{\label{fig:hyperplane} \emph{Standardising the bivariate normal distribution}. The points with equal probability density for the bivariate normal distribution are represented by a red ellipsis centred at the origin, and the black dashed line corresponds to the constraining hyperplane. The upper bound on $\epsilon_N$ is given by the probability mass within the shaded area. In order to calculate it, we first apply a rotation and scaling transformation, making the ellipsis symmetric with respect to the origin. Then, using the rotational symmetry of the standard bivariate normal distribution, one can rotate it such that the hyperplane becomes parallel to $x_2$.} 
\end{figure*}

Our goal is then to calculate the sum of $P^N_{\v{k}(\v{s})}$ in the limit $N\to\infty$ subject to the following hyperplane constraint 
\begin{equation}\label{eq:Defining_Hyperplane}
     \v{s}\cdot\v{E}\geq- \frac{\Delta F^N}{\sqrt{N}},
\end{equation} 
where $\v{E}$ is a vector of energies (eigenvalues of $H$). First, we approximate the multinomial distribution $\v{P}^N$ specified in Eq.~\eqref{eq:Multinomial_prob_vector} by a multivariate normal distribution $\mathcal{N}^{(\v{\mu},\v{\Sigma})}$ with mean vector $\v{\mu}=N\v{p}$ and covariance matrix $\v{\Sigma} = N(\text{diag }(\v{p})-\v{p}\v{p}^{T})$:
\begin{align}
\mathcal{N}^{(\v{\mu},\v{\Sigma})}_{\v{k}(\v{s})}&= \frac{1}{\sqrt{(2\pi)^d|\boldsymbol\Sigma|}}
\exp\left(-\frac{1}{2}({\v{k}}-{\v{\mu}})^T{\boldsymbol\Sigma}^{-1}({\v{k}}-{\v{\mu}})
\right)\nonumber\\
&=\frac{1}{\sqrt{(2\pi)^d|\boldsymbol\Sigma|}}
\exp\left(-\frac{1}{2}\v{s}^T{N\boldsymbol\Sigma}^{-1}\v{s}
\right).
\end{align}
As we explain in Section~\ref{app:clt}, such an approximation can always be made with an error approaching 0 as $N\rightarrow\infty$. Next, we standardise the multivariate normal distribution $\mathcal{N}^{(\v{\mu},\v{\Sigma})}$ using rotation and scaling transformations:
\begin{equation}
\label{ratemodifiedNew2}
    \v{\Sigma} = \Theta^T\sqrt{\v{\Lambda}}\sqrt{\v{\Lambda}}\Theta,
\end{equation}
where $\v{\Lambda}$ is a diagonal matrix with the eigenvalues of $\v{\Sigma}$ and $\Theta$ is an orthogonal matrix with columns given by the eigenvectors of $\v{\Sigma}$. We illustrate this process for a three-level system (so described by $s_1$ and $s_2$ since $\sum_i s_i=0$) in Fig.~\ref{fig:hyperplane}. This rotation and scaling of co-ordinates allows us to write $\mathcal{N}^{(\v{\mu},\v{\Sigma})}$ as a product of univariate standard normal distribution $\phi(y_i)$:
\begin{align}\label{MultivariateNormalmean0}
\!\!\!\!\mathcal{N}^{(\v{\mu},\v{\Sigma})}_{\v{k}(\v{s}(\v{y}))}&\!=\!\frac{1}{\sqrt{(2\pi)^d|\boldsymbol\Sigma|}}
\exp\left(\!-\frac{1}{2}\v{y}^T\v{y}
\!\right)=\prod_{i=1}^d \phi(y_i),\!
\end{align}
where
\begin{equation}\label{y_and_s}
   \v{y}= \sqrt{N}(\Theta^T\sqrt{\v{\Lambda}})^{-1} \v{s}.
\end{equation}

We then can equivalently write the equation specifying the hyperplane, Eq.~\eqref{eq:Defining_Hyperplane}, as
\begin{equation}
    \label{transhyperplane}
    (\Theta^T\sqrt{\v{\Lambda}}\v{y})\cdot\v{E}\geq -\Delta F^N.
\end{equation}
Observe that the standard normal distribution given in  Eq.~\eqref{MultivariateNormalmean0} is rotationally invariant about the origin. One can thus choose a coordinate system \mbox{$\v{x}=\{x_1,\ldots, x_d\}$} by applying a suitable rotation $R$ on \mbox{$\v{y}=\{y_1,\ldots, y_d\}$}, so that the hyperplane specified in Eq.~\eqref{eq:Defining_Hyperplane} becomes parallel to all coordinate axes but  the $x_1$ axis. Eq.~\eqref{MultivariateNormalmean0} can then be rewritten in the following form
\begin{align}\label{MultivariateNormalmean0onx}
\!\!\!\!\mathcal{N}^{(\v{\mu},\v{\Sigma})}_{\v{k}(\v{s}(\v{x}))}&\!=\!\frac{1}{\sqrt{(2\pi)^d|\boldsymbol\Sigma|}}
\exp\left(\!-\frac{1}{2}\v{x}^T\v{x}
\!\right)=\prod_{i=1}^d \phi(x_i).\!
\end{align}
As we have 
\begin{equation}\label{Rotation}
    \v{x}=R\v{y},
\end{equation}
we can use it together with Eq.~\eqref{y_and_s} to rewrite Eq.~\eqref{eq:Defining_Hyperplane} as
\begin{equation}\label{eq:Defining_Hyperplane_New}
     \Theta^T\sqrt{\v{\Lambda}}R^T\v{x}\cdot\v{E}\geq -\Delta F^N.
\end{equation}

To calculate the right hand side of the inequality given in Eq.~\eqref{Final_Bound_on_Epsilon} in the limit $N\rightarrow\infty$, we integrate Eq. $\eqref{MultivariateNormalmean0onx}$ from $-\infty$ to $d_{\mathcal{O}}$ along $x_1$, and  from $-\infty$ to $+\infty$ along any other $x_i\neq x_1$, where $d_\mathcal{O}$ is the signed distance of the hyperplane given in Eq.~\eqref{eq:Defining_Hyperplane_New} from the origin (see Fig.~\ref{fig:hyperplane}). This distance can be explicitly calculated as
\begin{eqnarray}
 d_{\mathcal{O}} &=& \frac{\Delta F^N}{\sqrt{\v{E}\cdot(R\sqrt{\v{\Lambda}}\Theta)^T(R\sqrt{\v{\Lambda}}\Theta)\v{E}}} = \frac{\Delta F^N}{\sqrt{\v{E}\cdot(\v{\Sigma} \v{E})}}
= \frac{\Delta F^N}{\sigma(F^N)},
\end{eqnarray}
where we have used the definition of $\sigma(F^N)$ from  Eq.~\eqref{eq:vn} in the last line. Thus, the upper bound on $\epsilon_N$ in the limit $N\rightarrow\infty$ given in Eq.~\eqref{Final_Bound_on_Epsilon} can be calculated as
\begin{align}
&1\!-\!\lim_{N\rightarrow\infty}\int_{-\infty}^{+\infty}\!dx_d\phi(x_d)\ldots\int_{-\infty}^{+\infty}\!dx_2\phi(x_2)\int_{-\infty}^{d_\mathcal{O}}\!dx_1\phi(x_1)\nonumber\\ 
&=  1\!-\!\lim_{N\rightarrow\infty}\Phi(d_\mathcal{O}) = \lim_{N\rightarrow\infty}\Phi\left(-\frac{\Delta F^N}{\sigma(F^N)}\right). \label{eq:Errorint}
\end{align}


\begin{center}
    \emph{\textbf{{Lower bound for the transformation error}}}\label{Lower_Bound_Error_transformation}
\end{center}

We start by writing $L$ from Eq.~\eqref{eq:lnL} as 
\begin{eqnarray}\label{Ldef_x}
    L= \exp\left(AN+x\sqrt{N v_N}\right)=:L(x),
\end{eqnarray}
where 
\begin{align}\label{A_and_x}
 A =&H(\v{\eta})-D(\psi\|{\gamma}),\qquad
 x=\frac{\Delta F^N}{\sigma(F^N)}.
\end{align}
In the previous section we have exactly calculated the right hand side of Eq.~\eqref{eq:ErrorCalc} in the limit $N\rightarrow\infty$ (see Eq.~\eqref{eq:Errorint}). Using Eqs.~\eqref{Ldef_x}-\eqref{A_and_x}, we can equivalently rewrite this as
\begin{equation}\label{eq:P_N_and_PHI}
     \lim_{N\rightarrow\infty}\sum_{\v{k},g_{\v{k}}}\left\lbrace 
    \hat{P}^N_{\v{k},g_{\v{k}}} \middle|
    \hat{P}^N_{\v{k},g_{\v{k}}}\geq \frac{1}{L(x)} \right\rbrace = \lim_{N\rightarrow\infty}\Phi(x),
\end{equation}
where $x$ depends on $N$ as per Eq.~\eqref{A_and_x}. Now, to prove the lower bound on transformation error $\epsilon_N$ we start with the exact expression, Eq.~\eqref{eq:Exactexpoferror}, and use the inequality from Eq.~\eqref{eq:upper}. Similarly as before, we choose \mbox{$\alpha=\exp(\delta\sqrt{N})$} such that $\delta>0$, in the Eq.~\eqref{eq:upper}. Thus from Eq.~\eqref{chi_order} and Eq.~\eqref{Ldef_x} we have,
\begin{equation}\label{eq:BoundonlxPure}
    L(x)\leq \frac{\chi_{\hat{\v{P}}^N}(e^{\delta\sqrt{N}}L(x))}{c}=\frac{\chi_{\hat{\v{P}}^N}(L(x+\delta))}{c}
\end{equation}
where $c$ can be evaluated similarly as before: 
\begin{align}\label{c_defn}
    c&=e^{\delta \sqrt{N}/2}  \sum_{i=\chi_{\hat{\v{P}}^N}(e^{\delta\sqrt{N}/2}L(x))}^{\chi_{\hat{\v{P}}^N}(e^{\delta\sqrt{N}}L(x))} (\hat{\v{P}}^N)_i^\downarrow=e^{\delta \sqrt{N}/2}  \sum_{i=\chi_{\hat{\v{P}}^N}(L(x+\delta/2))}^{\chi_{\hat{\v{P}}^N}(L(x+\delta))} (\hat{\v{P}}^N)_i^\downarrow \nonumber\\
    &=e^{\delta \sqrt{N}/2} \left( \sum_{i}\left\lbrace 
    \hat{P}^N_i \middle|
    \hat{P}_i^N \geq \frac{1}{L(x+\delta)}	
    \right\rbrace  \right.\left. -\sum_{i}\left\lbrace 
    \hat{P}^N_i \middle|
    \hat{P}_i^N \geq \frac{1}{L(x+\delta/2)}
    \right\rbrace \right).
\end{align}
Using Eq.~\eqref{eq:P_N_and_PHI}, we see that the limiting behaviour of $c$ from Eq.~\eqref{c_defn} is given by 
\begin{equation}
  \lim_{N\to\infty} c = \left(\Phi(x+\delta)-\Phi(x+\delta/2)\right) \lim_{N\to\infty} e^{\delta \sqrt{N}/2}.
\end{equation}
Thus, for any finite $\delta>0$, there always exists $N_0$ such that for all $N\geq N_0$ we have $c>1$. Combining with Eq.~\eqref{eq:BoundonlxPure}, for large enough $N$ we finally have, 
\begin{equation}\label{eq:Boundonlx2Pure}
    L(x)\leq \frac{\chi_{\hat{\v{P}}^N}(L(x+\delta))}{c} \leq \chi_{\hat{\v{P}}^N}(L(x+\delta)).
\end{equation}
Hence, using Eq.~\eqref{eq:Boundonlx2Pure}, we have the following lower bound on transformation error
\begin{align}
    \label{eq:second_upper3}
   \lim_{N\rightarrow\infty}\epsilon_N &\geq  1- \lim_{N\rightarrow\infty}\sum_{i=1}^{\chi_{\hat{\v{P}}^N}(L(x+\delta))} (\hat{\v{P}}^N)_i^\downarrow \, \nonumber\\
   \!\!\! &=1- \lim_{N\rightarrow\infty}\sum_{i}\left\lbrace 
    \hat{P}^N_i \middle|
    \hat{P}_i^N \geq \frac{1}{L(x+\delta)}
    \right\rbrace \nonumber\\
    \!\!\!&= 1-\lim_{N\rightarrow\infty}\Phi(x+\delta)=\lim_{N\rightarrow\infty}\Phi(-x-\delta),
\end{align}
where the first equality in the last line follows from Eq.~\eqref{eq:P_N_and_PHI}.
 Since the above inequality holds for any $\delta>0$, taking the limit $\delta\rightarrow0$ we conclude that
\begin{equation}
    \lim_{N\to\infty} \epsilon_N\geq \lim_{N\to\infty} \Phi\left(-\frac{\Delta F^N}{\sigma(F^N)}\right).
\end{equation}
Finally, combining the above with the bound obtained in Eq.~\eqref{eq:Errorint}, we have
\begin{equation}\label{Proof_end_error}
    \lim_{N\to\infty}\epsilon_N=\lim_{N\to\infty} \Phi\left(-\frac{\Delta F^N}{\sigma(F^N)}\right),
\end{equation}
which completes the proof.


\subsection{Proof of Theorem~\ref{thm:pure2}}\label{sec:proof2b}

The proof of Theorem~\ref{thm:pure2} will be divided into two parts. First, we will find the embedded version of the optimal final state minimising the dissipation of free energy $F^N_{\mathrm{diss}}$, and derive the lower bound for $F^N_{\mathrm{diss}}$ as a function of the initial state. And then, we will calculate this bound up to second order asymptotic terms.


\begin{center}
    \emph{\textbf{{Deriving bound for optimal dissipation}}}
\end{center}

We start by applying Lemma~\ref{Opt_Diss_free} to the central Eq.~\eqref{thermomajorisation__embedded_pure_state} that specifies the conditions for the investigated $\epsilon$-approximate thermodynamic distillation process. As a result, the actual total final state in the embedded picture, $\hat{\v{F}}^N$, which is $\epsilon$ away from the target state $\hat{\v{G}}^N  \otimes \v{f}_1^A\otimes \tilde{\v{f}}_k$, is given, up to permutations, by
\begin{eqnarray}
\label{eq:distribution-total2}
     \hat{\v{F}}^N&=&
     \begin{cases}
       \frac{1-\epsilon}{k} &\quad \mathrm{for~} i\leq k,\\
       ( {\hat{\v{P}}}^N \otimes \v{f}^A_0 \otimes \tilde{\v{\eta}})_{i}^\downarrow &\quad\mathrm{for~}i>k,
     \end{cases}\, 
\end{eqnarray}
where
\begin{equation}
    k=\exp\Big(H(\hat{\v{G}}^N)+H(\v{f}^A_1)+H(\tilde{\v{f}}_k)\Big).
\end{equation}
Similarly as before, we have chosen $\hat{\v{F}}^N$ with minimal entropy (i.e., the bistochastic map $B$ from Lemma~\ref{Opt_Diss_free} was chosen to be a permutation), as this minimises the dissipation of free energy. Due to the fact that the final state is specified only up to permutations, one may freely rearrange its elements to reduce the dissipation of free energy coming from unembedding (such dissipation happens when a state is not uniform within each embedding box). However, unlike in the previous case, this permutational freedom does not necessarily allow one to completely avoid additional dissipation. In the previous case, the number of non-uniform embedding boxes was exponentially small, which led to a negligible loss of free energy due to unembedding. But here, the number of non-uniform embedding boxes is not exponentially small and that can lead to a finite dissipation of free energy. Therefore, in this case, we can only provide a lower bound on dissipation of free energy as follows:
\begin{align}
\label{eq:Inequality_F_Diss}
     F^N_{\text{diss}} &=\frac{1}{\beta }\Big[D( \psi^{\otimes N}\|\v{G}^{ N})+ D(\ketbra{E_0^A}{E_0^A}\|\gamma_A)
     -D(\v{F}^N\|\v{G}^N\otimes\gamma_A\otimes\tilde{\v{G}}^N )\Big]\nonumber\\
     &\geq \frac{1}{\beta }\Big[D(\hat{\v{P}}^N\|\hat{\v{G}}^N)+ H(\v{P}^N) +D(\ketbra{E_0^A}{E_0^A}\|\gamma_A)\nonumber\\
     &\hspace{4cm}- D(\hat{\v{F}}^N\|\hat{\v{G}}^N\otimes\hat{\gamma}_A\otimes\hat{\tilde{\v{G}}}^N)\Big],
\end{align}
where we have used the fact that
\begin{align}
    D(\psi^{\otimes N}\|\v{G}^N)&=D(\D(\psi^{\otimes N})\|\v{G}^N) + D(\psi^{\otimes N}\|\D(\psi^{\otimes N}))    \nonumber\\
    &=D(\hat{\v{P}}^N\|\hat{\v{G}}^N)+ H(\v{P}^N).
\end{align}

Next, we recall that $H(\v{P}^N)=O(\log N)$ (see Section~\ref{app:log} for details), and note that \mbox{$D(\ketbra{E_0^A}{E_0^A}\|\gamma_A) = -\log\gamma_0^A$}. Thus, the inequality from Eq.~\eqref{eq:Inequality_F_Diss} can be simplified as follows:
\begin{align}
    F^N_{\text{diss}} \gtrsim&\frac{1}{\beta}\Big(D(\hat{\v{P}}^N\|\hat{\v{G}}^N)
     -D(\hat{\v{F}}^N\|\hat{\v{G}}^N\otimes\hat{\gamma}_A\otimes\hat{\tilde{\v{G}}}^N)-\log\gamma_0^A\Big)\nonumber\\
 =&\frac{1}{\beta}\Big( H(\hat{\v{F}}^N)-H(\hat{\v{P}}^N)\!-\!\log (\tilde{D}D^A)-\log\gamma_0^A\Big).\label{simplified_True_Diss_Pure}
\end{align}
Here, $\tilde{D}$ and $D^A$ are the embedding constants defined by Eq.~\eqref{eq:definitiongibbsvec} for $\tilde{\v{G}}^N$ and $\gamma_A$, respectively, and similarly $D$ will denote this embedding constant for $\v{G}^N$. As before, $\tilde{D}$ and $D$ can be chosen to be equal, since the only thing that matters is that $\tilde{{G}}^N_k=\tilde{D}_k/\tilde{D}$ and ${{G}}^N_k={D}_k/D$, i.e., the change in $\tilde{D}$ or $D$ can be compensated by the appropriate change in $\tilde{D}_k$ or $D_k$. Furthermore, noting that $k=D\tilde{D}_kD^A_1 =\tilde{D}\tilde{D}_kD^A_1 $, we can express the entropy of $\hat{\v{F}}^N$ as
\begin{align}
    H(\hat{\v{F}}^N)=&-\sum_{i=1}^k \frac{1-\epsilon}{k}\log\left(\frac{1-\epsilon}{k}\right)- \sum_{i > k} ( {\hat{\v{P}}}^N \otimes \v{f}^A_0 \otimes \tilde{\v{\eta}})_i^{\downarrow} \log ( {\hat{\v{P}}}^N \otimes \v{f}^A_0 \otimes \tilde{\v{\eta}})_i^{\downarrow}\nonumber\\
    =&-(1-\epsilon)\log(1-\epsilon)+(1-\epsilon)\log k\nonumber\\
    &\hspace{4cm}- \sum_{i > k} ( {\hat{\v{P}}}^N \otimes \v{f}^A_0 \otimes \tilde{\v{\eta}})_i^{\downarrow}\log \Big(\frac{{\hat{\v{P}}}^N}{\tilde{D}D_0^A}\Big)_i^{\downarrow}\nonumber\\
    \simeq&\log k-\epsilon\log L- \sum_{i > L} (\hat{\v{P}}^N)_i^{\downarrow}\log (\hat{\v{P}}^N)_i^{\downarrow} \label{eq:Entropy_of_J},
\end{align}
where in the last step of we used the fact that \mbox{$L=k/(\tilde{D}D_0^A)$}, which comes from Eq.~\eqref{eq:L}. Thus, from Eq.~\eqref{simplified_True_Diss_Pure}, the dissipated free energy in the optimal distillation process is simply bounded by
\begin{align}
    F^N_{\text{diss}}\gtrsim& \frac{1}{\beta}\Bigg[\sum_{i=1}^{L} (\hat{\v{P}}^N)_i^{\downarrow}\log (\hat{\v{P}}^N)_i^{\downarrow}+\log\Big(\frac{k}{\tilde{D}D^A}\Big)-\epsilon\log L- \log \gamma_0^A\Bigg]\nonumber\\
   =& \frac{1}{\beta}\Bigg[\sum_{i=1}^{L} (\hat{\v{P}}^N)_i^{\downarrow}\log (\hat{\v{P}}^N)_i^{\downarrow}+(1-\epsilon)\log L\!\Bigg].
\end{align}
Finally, using the expression for $L$ given in Eq.~\eqref{Ldef_x}, we can rewrite the above bound as 
\begin{align}
    F^N_{\text{diss}}\gtrsim& \frac{1}{\beta}\Bigg[\sum_{i=1}^{L} (\hat{\v{P}}^N)_i^{\downarrow}\log (\hat{\v{P}}^N)_i^{\downarrow}+(1-\epsilon)\left(AN+x\sqrt{Nv_N}\right)\Bigg],\label{eq: F_Diss_1}
\end{align}
where $A$ and $x$ are given by Eq.~\eqref{A_and_x}. 


\begin{center}
    \emph{\textbf{{Calculating bound for optimal dissipation}}}
\end{center}

We now proceed to bounding the first term on the right hand side of Eq.~\eqref{eq: F_Diss_1}. We start by noting that Eq.~\eqref{eq:Boundonlx2Pure} implies that the following holds for any $\delta>0$:   
\begin{eqnarray}\label{eq:ent_Bound2}
     \!\!\!\!\!\!\!\!\sum_{i=1}^{L(x)} (\hat{\v{P}}^N)_i^\downarrow\log(\hat{\v{P}}^N)_i^\downarrow\geq 
     \!\!\!\!\!\!\sum_{i=1}^{\chi_{\hat{\v{P}}^N}(L(x+\delta))}\!\!\!\!(\hat{\v{P}}^N)_i^\downarrow\log(\hat{\v{P}}^N)_i^\downarrow.
\end{eqnarray}
Next, from the definition of $\chi_{\hat{\v{P}}^N}[L(x+\delta)]$, it follows that
\begin{align}
    \label{eq:Entropy_Calc_2}
    \sum_{i=1}^{\chi_{\hat{\v{P}}^N}[(L(x+\delta)]} (\hat{\v{P}}^N)_i^\downarrow\log(\hat{\v{P}}^N)_i^\downarrow&=\sum_{i}\Big\{
    \hat{P}^N_i\log (\hat{P}^N_i) ~\Big|~
    \hat{P}^N_i \geq \frac{1}{L(x+\delta)}\Big\}\nonumber\\
    &= \sum_{\v{k},g_{\v{k}}}\left\lbrace 
    \hat{P}^N_{\v{k},g_{\v{k}}}\log(\hat{P}^N_{\v{k},g_{\v{k}}}) ~\middle|~
    \hat{P}^N_{\v{k},g_{\v{k}}}\geq \frac{1}{L(x+\delta)}	\right\rbrace\nonumber \\
    &=\sum_{\v{k}}\left\lbrace 
    {P}^N_{\v{k}}\log (\hat{P}^N_{\v{k},1}) ~\middle|~
    {P}^N_{\v{k}}\geq \frac{\prod_{i=1}^d D_i^{k_i}}{L(x+\delta)}	\right\rbrace\nonumber\\
    &=\sum_{\v{k}}\left\lbrace 
    {P}^N_{\v{k}}\log ({P}^N_{\v{k}}) ~\middle|~
    {P}^N_{\v{k}}\geq \frac{\prod_{i=1}^d D_i^{k_i}}{L(x+\delta)}	\right\rbrace\nonumber\\
    &- \sum_{\v{k}}\left\lbrace 
    {P}^N_{\v{k}}\Big(\sum_{j=1}^d k_j\log\gamma_j\Big) ~\middle|~
    {P}^N_{\v{k}}\geq \frac{\prod_{i=1}^d D_i^{k_i}}{L(x+\delta)}	\right\rbrace\nonumber\\
    &-(1-\epsilon)N\log D.
\end{align}

We now note that the constraint on the summand in Eq.~\eqref{eq:Entropy_Calc_2} can be modified using the parametrization of $\v{k}$ as a function of $\v{s}$ given in Eq.~\eqref{eq:k_of_s} as follows: 
\begin{align}
   \log (P^N_{\v{k}}) \geq& \sum_{i=1}^d k_i\log D_i - \log (L(x+\delta))\nonumber\\
    =& \sum_{i=1}^d (Np_i+\sqrt{N}s_i)\log \gamma_i+N\log D- \log L(x+\delta)\nonumber\\
    =& -\beta\sqrt{N}\sum_{i=1}^d s_iE_i+N\Big(\log D-D(\psi\|\gamma)\Big)\nonumber\\
    &\hspace{2cm}- N\Big(\log D-D(\psi\|\gamma)\Big)-(x+\delta)\sqrt{Nv_N}\nonumber\\
    =& -\beta\sqrt{N}\sum_{i=1}^d s_iE_i-(x+\delta)\sqrt{Nv_N}\nonumber\\
    =& -\beta\left(\sqrt{N}\v{s}\cdot\v{E}+\delta\sigma(F^N)+\Delta F^N\right),
    \label{eq:deltabound}
\end{align}
where in the second equality we used the form of $L$ given in Eq.~\eqref{Ldef_x} and employed Eq.~\eqref{eq:KLD_of_psi}, whereas in the final equality we used the definition $\sigma(F^N)=\sqrt{Nv_n}/\beta$ and Eq.~\eqref{A_and_x} saying that \mbox{$x=\Delta F^N/\sigma(F^N)$}. Moreover, since \mbox{$\log(P^N_{\v{k}(\v{s})})=O(\log N)$}, the condition from Eq.~\eqref{eq:deltabound} can be rewritten as
\begin{equation}
    \label{eq:constraint2}
    \sqrt{N}\v{s}\cdot\v{E}+\Delta F^N+\delta\sigma(F^N)\gtrsim 0.
\end{equation}

Coming back, we can now rewrite Eq.~\eqref{eq:Entropy_Calc_2} using the parametrization of $\v{k}$ as a function of $\v{s}$ from Eq.~\eqref{eq:k_of_s} and re-expressing the constraint using Eq.~\eqref{eq:constraint2}:

\begin{align}
    \label{eq:Entropy_Calc_3}
    &\sum_{i=1}^{\chi_{\hat{\v{P}}^N}(L(x+\delta))} (\hat{\v{P}}^N)_i^\downarrow\log(\hat{\v{P}}^N)_i^\downarrow=\sum_{\v{k}(\v{s})}\Bigg\{ 
    {P}^N_{\v{k}(\v{s})}\log ({P}^N_{\v{k}(\v{s})})-\sqrt{N}{P}^N_{\v{k}(\v{s})}\sum_{i=1}^d s_i\log\gamma_i \nonumber\\
    &\qquad\qquad\qquad\times\Bigg|    \sqrt{N}\v{s}\cdot\v{E}+\Delta F^N+\delta\sigma(F^N)\geq 0\Bigg\}\nonumber\\
    &\qquad\qquad\hspace{5cm}-(1-\epsilon)N\Big(\log D-D(\psi\|\gamma)\Big)\nonumber\\
    &\quad\simeq\beta\sqrt{N}\sum_{\v{s}}\Bigg\{ 
    {P}^N_{\v{k}(\v{s})} \v{s}\cdot\v{E}\Bigg| \sqrt{N}\v{s}\cdot\v{E}+\Delta F^N+\delta\sigma(F^N)\geq 0\Bigg\}-(1-\epsilon)AN,
\end{align}
where we used the definition of $A$ from Eq.~\eqref{A_and_x}. 

Our next goal goal is then to calculate the sum
\begin{equation}
    \label{eq:Defining_Hyperplane2}
    \sum_{\v{s}}\Big\{P^N_{\v{k}(\v{s})}\v{s}\cdot\v{E}  \bigg|\sqrt{N}\v{s}\cdot\v{E}+\Delta F^N+\delta\sigma(F^N)\geq 0\Big\}.
\end{equation}
Similarly as before, we will approximate the multinomial distribution $\v{P}^N$  by a multivariate normal distribution $\mathcal{N}^{(\v{\mu},\v{\Sigma})}$ with mean vector $\v{\mu}=N\v{p}$ and covariance matrix $\v{\Sigma} = N(\text{diag }(\v{p})-\v{p}\v{p}^{T})$:
\begin{align}
\mathcal{N}^{(\v{\mu},\v{\Sigma})}_{\v{k}(\v{s})}&= \frac{1}{\sqrt{(2\pi)^d|\boldsymbol\Sigma|}}
\exp\left(-\frac{1}{2}({\v{k}}-{\v{\mu}})^T{\boldsymbol\Sigma}^{-1}({\v{k}}-{\v{\mu}})
\right)\nonumber\\
&=\frac{1}{\sqrt{(2\pi)^d|\boldsymbol\Sigma|}}
\exp\left(-\frac{1}{2}\v{s}^T{N\boldsymbol\Sigma}^{-1}\v{s}
\right).
\end{align}
Next, we standardise the multivariate normal distribution $\mathcal{N}^{(\v{\mu},\v{\Sigma})}$ using rotation and scaling transformations:
\begin{equation}
\label{ratemodifiedNew}
    \v{\Sigma} = \Theta^T\sqrt{\v{\Lambda}}\sqrt{\v{\Lambda}}\Theta,
\end{equation}
where $\v{\Lambda}$ is a diagonal matrix with the eigenvalues of $\v{\Sigma}$ and $\Theta$ is an orthogonal matrix with columns given by the eigenvectors of $\v{\Sigma}$. This rotation and scaling of co-ordinates allows us to write $\mathcal{N}^{(\v{\mu},\v{\Sigma})}$ as a product of univariate standard normal distribution $\phi(y_i)$:
\begin{align}
\!\!\!\!\mathcal{N}^{(\v{\mu},\v{\Sigma})}_{\v{k}(\v{s}(\v{y}))}&\!=\!\frac{1}{\sqrt{(2\pi)^d|\boldsymbol\Sigma|}}
\exp\left(\!-\frac{1}{2}\v{y}^T\v{y}
\!\right)=\prod_{i=1}^d \phi(y_i),\!
\end{align}
where
\begin{equation}\label{eq:Def_y}
 \v{y} = \sqrt{N}(\Theta^{T}\sqrt{\Lambda})^{-1} \v{s},   
\end{equation}
such that
\begin{equation}\label{eq:Def_s}
\v{s} = \frac{1}{\sqrt{N}} \Theta^{T} \sqrt{\Lambda} \v{y}.
\end{equation}

From Eq.~\eqref{eq:Def_y} and Eq.~\eqref{eq:Def_s}, one can write equivalently the sum given in Eq.~\eqref{eq:Defining_Hyperplane2} as
\begin{equation}
    \label{eq:Defining_Hyperplane3}
    \sum_{\v{y}}\Big\{P^N_{\v{k}(\v{s}(\v{y}))}\frac{\tilde{\v{E}}\cdot \v{y}}{\sqrt{N}} \bigg|\tilde{\v{E}}\cdot\v{y}+\Delta F^N+\delta\sigma(F^N) \geq 0\Big\},
\end{equation}
where we defined $\tilde{\v{E}}:= \sqrt{\Lambda} \Theta \v{E}$ with the following normalisation
\begin{align}
    \| \tilde{\v{E}}\| &= \sqrt{\tilde{\v{E}}\cdot\tilde{\v{E}}} = \sqrt{(\sqrt{\Lambda} \Theta \v{E})\cdot(\sqrt{\Lambda} \Theta \v{E})} = \sqrt{\v{E}\cdot\Theta^{T}\Lambda \Theta \v{E}} = \sqrt{\v{E}\cdot \v{\Sigma} \v{E}} \nonumber\\ &=  \sigma( F^N).
\end{align}
One can then equivalently express Eq.~\eqref{eq:Defining_Hyperplane3} as 
\begin{align}
    \label{eq:sigma_times_P^N}
    \frac{\sigma(F^N)}{\sqrt{N}}\sum_{\v{y}}\Big\{P^N_{\v{k}(\v{s}(\v{y}))}\hat{\tilde{\v{E}}}\cdot \v{y} \;\bigg|\;&\sigma(F^N)\hat{\tilde{\v{E}}}\cdot\v{y}+\Delta F^N+\delta\sigma(F^N) \geq 0\Big\},
\end{align}
where $\hat{\tilde{\v{E}}}$ is unit vector along $\tilde{\v{E}}$. Now, we choose a rotation $R$ such that
\begin{equation}
    R\hat{\tilde{\v{E}}} = (1,\ldots, 0)^T
\end{equation}
and we call $R\v{y}=\v{x}$. Thus, we can rewrite Eq.~\eqref{eq:sigma_times_P^N} as
\begin{align}
    \label{eq:sigma_times_P^N2}
    &\frac{\sigma(F^N)}{\sqrt{N}}\sum_{\v{y}}\Big\{P^N_{\v{k}(\v{s}(\v{y}))}R\hat{\tilde{\v{E}}}\cdot R\v{y} \;\Bigg|\;\sigma(F^N)R\hat{\tilde{\v{E}}}\cdot R\v{y}+\Delta F^N+\delta\sigma(F^N) \geq 0\Big\}\nonumber\\
    &\hspace{0.4cm}=\frac{\sigma(F^N)}{\sqrt{N}}\sum_{\v{x}}\Big\{P^N_{\v{k}(\v{s}(\v{x}))}x_1 \;\Bigg|\;\sigma(F^N)x_1+\Delta F^N+\delta\sigma(F^N) \geq 0\Big\}.
\end{align}

Now, employing multivariate normal distribution approximation and the fact that it is isotropic, we can exactly calculate the above expression as an integral over $\v{x}$ of the following form 
\begin{align}
    \frac{\sigma(F^N)}{\sqrt{N}}\Bigg(\int_{-\infty}^{+\infty}dx_2 \phi(x_2)\ldots \int_{-\infty}^{+\infty}dx_d \nonumber&\phi(x_d) \int_{-\infty}^{d_\mathcal{O}} dx_1x_1 \phi(x_1)\Bigg)\\ &=\frac{\sigma( F^N)}{\sqrt{2\pi N}}\exp\left(-\frac{d_{\mathcal{O}}^2}{2}\right),
\end{align}
where $d_\mathcal{O}$ is the distance between the origin and the hyperplane specified by the constraint given in Eq.~\eqref{eq:sigma_times_P^N2}:
\begin{equation}
    d_{\mathcal{O}} = \frac{ \Delta F^N+\delta\sigma(F^N) }{\sqrt{\tilde{\v{E}}\cdot\tilde{\v{E}}}}=\frac{ \Delta F^N}{\sigma( F^N)}+\delta.
\end{equation}
Thus, finally we get the following expression for the desired sum from Eq.~\eqref{eq:Defining_Hyperplane2}:
\begin{align}
\sum_{\v{s}}\Big\{P^N_{\v{k}(\v{s})}\v{s}\cdot\v{E}  \bigg|\sqrt{N}\v{s}\cdot\v{E}&+\Delta F^N+\delta\sigma(F^N)\geq 0\Big\}\nonumber\\
    &=  \frac{\sigma( F^N)}{\sqrt{2\pi N}}\exp\left[-\frac{1}{2}\left(\frac{ \Delta F^N}{\sigma( F^N)}+\delta\right)^2\right].
\end{align}

Substituting the above to Eq.~\eqref{eq:Entropy_Calc_3}, and the resulting expression to Eq.~\eqref{eq:ent_Bound2}, we get the following bound:
\begin{align}
\sum_{i=1}^{L(x)} (\hat{\v{P}}^N)_i^\downarrow\log(\hat{\v{P}}^N)_i^\downarrow  \gtrsim -(1-\epsilon)A +\frac{\beta\sigma( F^N)}{\sqrt{2\pi }}\exp\left[-\frac{1}{2}\left(\frac{ \Delta F^N}{\sigma( F^N)}+\delta\right)^2\right],
\end{align}
which in turn can be used in Eq.~\eqref{eq: F_Diss_1} to give the following:
\begin{align}
    F^N_{\mathrm{diss}} \gtrsim &\frac{\sigma( F^N)}{\sqrt{2\pi }}\exp\left[-\frac{1}{2}\left(\frac{ \Delta F^N}{\sigma( F^N)}+\delta\right)^2\right] +(1-\epsilon)\Delta F^N.
\end{align}
Since the above holds for any $\delta>0$, by taking the limit $\delta\rightarrow 0$, we finally arrive at
\begin{equation}\label{eq:F_N_Diss_Simp}
    F^N_{\text{diss}} \gtrsim   (1-\epsilon)\Delta F^N +\frac{\sigma( F^N)}{\sqrt{2\pi}}\exp\Bigg[{-\frac{(\Delta F^N)^2}{2\sigma( F^N)^2}}\Bigg].
\end{equation}
Employing Eq.~\eqref{Proof_end_error} we have
\begin{equation}
    \Delta F^N = -\Phi^{-1}(\epsilon)\sigma(F^N),
\end{equation}
and so substituting this in Eq.~\eqref{eq:F_N_Diss_Simp} we obtain
\begin{eqnarray}
F^N_{\text{diss}}&\gtrsim& a(\epsilon)\sigma( F^N),
\end{eqnarray}
with $a$ given by Eq.~\eqref{eq:a}, which completes the proof.

\subsection{Optimality of the communication rate}
\label{app:optimality}

The following derivation will closely follow the proof of Lemma~1 of Ref.~\cite{korzekwa2019encoding}. Let us assume that for a system $(\rho^N,H^N)$ it is possible to encode $M$ messages in a thermodynamically-free way so that the average decoding error is $\epsilon$. It means that there exists a set of $M$ encoding thermal operations $\{\E^{\beta}_i\}_{i=1}^M$ and a decoding POVM $\{ \Pi_i \}_{i=1}^M$ such that
\begin{equation}
    \label{eq:assumption}
    1-\epsilon=\frac{1}{M}\sum_{i=1}^M \tr{\E^{\beta}_i(\rho^N) \Pi_i}.
\end{equation}
Note that every thermal operation $\E^{\beta}_i$ between the initial system $(\rho^N,H^N)$ and a target system $(\tilde{\rho}^N,\tilde{H}^N)$ preserves the thermal equilibrium state, i.e.,
\begin{equation}
    \E^{\beta}_i(\gamma^N)=\tilde{\gamma}^N.
\end{equation}
Now, let us introduce the following three states
\begin{subequations}
\begin{align}
    \tau:=&\frac{1}{M} \sum_{i=1}^M \ketbra{i}{i}\otimes \E^{\beta}_i(\rho^N),\\
    \zeta:=&\frac{1}{M} \sum_{i=1}^M \ketbra{i}{i}\otimes \gamma^N,\\
    \tilde{\zeta}:=&\frac{1}{M} \sum_{i=1}^M \ketbra{i}{i}\otimes \tilde{\gamma}^N.
\end{align}
\end{subequations}
The hypothesis testing relative entropy $D_H^{\epsilon}$ between $\tau$ and $\tilde{\zeta}$, defined by~\cite{wang10,buscemi2010quantum,brandao2011one} 
\begin{align}
D_H^{\epsilon}(\tau\|\tilde{\zeta}) \!:= - \log \inf \big\{ \tr{Q\tilde{\zeta}}\ \big|\ &0 \leq Q \leq \iden ,\tr{Q\tau} \geq 1-\epsilon \big\} ,
\label{eq:hypothesis_rel_ent}
\end{align}
satisfies the following
\begin{equation}
	D_H^{\epsilon}(\tau\|\tilde{\zeta})\geq -\log \tr{Q\tilde{\zeta}}
\end{equation}
for
\begin{equation}
	Q = \sum_{i=1}^M \ketbra{i}{i} \otimes \Pi_i.
\end{equation}
This is because the above (potentially suboptimal) choice of $Q$ clearly satisfies \mbox{$0\leq Q\leq \iden$} and also
\begin{equation}
	\tr{Q\tau}=	\frac{1}{M} \sum_{i=1}^M \tr{\E^{\beta}_i(\rho)\Pi_i}\geq 1-\epsilon,
\end{equation} 
due to our assumption in Eq.~\eqref{eq:assumption}. At the same time we have
\begin{equation}
	\tr{Q\tilde{\zeta}}=\frac{1}{M} \sum_{i=1}^M \tr{\tilde{\rho}^N \Pi_i}=\frac{1}{M},
\end{equation} 
so that 
\begin{equation}
	\label{eq:hypothesis_ineq_1}
	\log M\leq D_H^{\epsilon}(\tau\|\tilde{\zeta}).
\end{equation}

Next, we introduce the following encoding channel
\begin{equation}
	\mathcal{F}:=\sum_{i=1}^{M} \ketbra{i}{i} \otimes \E^{\beta}_i,
\end{equation}
which satisfies
\begin{equation}
    \mathcal{F}(\zeta)=\tilde{\zeta}.
\end{equation}
Employing the data-processing inequality twice, we get the following sequence of inequalities:
\begin{align}
	D_H^{\epsilon}(\tau\|\tilde{\zeta})&=D_H^{\epsilon}\left[\mathcal{F}\left(\frac{1}{M}\sum_{i=1}^M \ketbra{i}{i} \otimes \rho^N\right)\bigg\|\mathcal{F}({\zeta})\right]\nonumber\\
	&\leq D_H^{\epsilon}\left(\frac{1}{M}\sum_{i=1}^M \ketbra{i}{i} \otimes \rho^N\bigg\|\zeta\right)\nonumber\\
	&\leq D_H^{\epsilon}\big( \rho^N \big\| \gamma^N \big).
\end{align}
Combining this with Eq.~\eqref{eq:hypothesis_ineq_1}, we arrive at
\begin{equation}
	\label{eq:hypothesis_ineq_2}
	\log M\leq D_H^{\epsilon}\big( \rho^N \big\| \gamma^N \big).
\end{equation}

Finally, for the case of identical initial subsystems, $\rho^N=\rho^{\otimes N}$ and $\gamma^N=\gamma^{\otimes N}$, we can use the known second order asymptotic expansion of the hypothesis testing relative entropy~\cite{li12},
\begin{align}
    \frac{1}{N}D^\epsilon_H(\rho^{\otimes N}\|\gamma^{\otimes N})&\simeq D(\rho\|\gamma)+\sqrt{\frac{V(\rho\|\gamma)}{N}}\Phi^{-1}(\epsilon),
\end{align}
leading to
\begin{align}
    \frac{\log M}{N}\lesssim D(\rho\|\gamma)+\sqrt{\frac{V(\rho\|\gamma)}{N}}\Phi^{-1}(\epsilon).
\end{align}
For the above proof to work also in the case of non-identical subsystems, one would need to prove the following asymptotic behaviour of the hypothesis testing relative entropy:
\begin{align}
    \frac{1}{N}D^\epsilon_H\left(\bigotimes_{n=1}^N \rho_n\bigg\|\bigotimes_{n=1}^N \gamma_n\right) \simeq\frac{1}{N}\sum\limits_{n=1}^N D(\rho_n\|\gamma_n)+\sqrt{\frac{\frac{1}{N}\!\sum\limits_{n=1}^N \!V(\rho_n\|\gamma_n)}{N}}\Phi^{-1}(\epsilon).\!
\end{align}


\subsection{Proof of Lemma~\ref{Opt_Diss_free}}
\label{app:min_out}

Consider $\tilde{\v{q}}$ to be a probability vector that saturates $\v{p}\succ_{\epsilon}\v{q}$. By definition it means that $\v{p}\succ\tilde{\v{q}}$ and \mbox{$F(\v{q},\tilde{\v{q}})=1-\epsilon$}. Since $V(\v{q})=0$ and $H(\v{q})=\log L$, the probability vector $\v{q}$ contains exactly $L$ non-zero entries such that each of them is equal to $1/L$. Thus,
\begin{equation}
    \v{q}^{\downarrow}=\Big(\underbrace{\frac{1}{L},\ldots,\frac{1}{L}}_{L},0,\ldots,0\Big).
\end{equation}
From the definition of $F(\v{q},\tilde{\v{q}})$ given in Eq.~\eqref{eq:DefFidelity} we have the following
\begin{eqnarray}
\label{eq:Fidelity_majorize}
1-\epsilon=F(\v{q},\tilde{\v{q}}) &\leq& F(\v{q}^{\downarrow},\tilde{\v{q}}^{\downarrow})\nonumber\\
&=& \frac{1}{L}\left(\sum_{i=1}^L \sqrt{\tilde{q}_i^{\downarrow}}\right)^2= \frac{1}{L}\left(\sum_{i=1}^L \sqrt{\tilde{q}_i^{\downarrow}\cdot 1}\right)^2\nonumber\\&\leq& \frac{1}{L}\left(\sum_{i=1}^L\tilde{q}_i^{\downarrow}\right)  \left(\sum_{i=1}^L 1\right) = \sum_{i=1}^L\tilde{q}_i^{\downarrow},
\end{eqnarray}
where the first inequality follows from the definition of fidelity, while the second one comes from the Cauchy-Schwarz inequality. Now, from Lemma~\ref{Opt_Diss_free}, we have 
\begin{eqnarray}
\sum_{i=1}^L p_i^{\downarrow} = 1-\epsilon.
\end{eqnarray}
Combining this with Eq.~\eqref{eq:Fidelity_majorize} we obtain
\begin{eqnarray}
\sum_{i=1}^L p_i^{\downarrow}\leq \sum_{i=1}^L \tilde{q}_i^{\downarrow}.
\end{eqnarray}
On the other hand, $\v{p}\succ \tilde{\v{q}}$ gives 
\begin{equation}
    \sum_{i=1}^L p_i^{\downarrow}\geq \sum_{i=1}^L \tilde{q}_i^{\downarrow},    
\end{equation}
and so we conclude that
\begin{equation}\label{constrainoflagrangemultiplier}
    \sum_{i=1}^L p_i^{\downarrow}= \sum_{i=1}^L \tilde{q}_i^{\downarrow}=1-\epsilon.
\end{equation}

Next, note that as $\epsilon = 1-F(\v{q},\tilde{\v{q}})\geq 1-F(\v{q}^{\downarrow},\tilde{\v{q}}^{\downarrow})$, we see that to have the minimal value of error, $\tilde{\v{q}}^{\downarrow}$ needs to maximize the fidelity $F(\v{q}^{\downarrow},\tilde{\v{q}}^{\downarrow})$ subject to the constraint given in Eq.~\eqref{constrainoflagrangemultiplier}. Thus we can write
\begin{equation}\label{eq:Optimal_q}
\begin{split}
 &\max ~~ \qquad\frac{1}{L}\left(\sum_{i=1}^L \sqrt{\tilde{q}_i^{\downarrow}}\right)^2\\
 &\textrm{such that} \quad \sum_{i=1}^L \tilde{q}_i^{\downarrow}=1-\epsilon.
\end{split}
\end{equation}
The Lagrangian of the aforementioned optimization problem in Eq.~\eqref{eq:Optimal_q} is given by 
\begin{equation}
    \mathcal{L}= \frac{1}{L}\left(\sum_{i=1}^L \sqrt{\tilde{q}_i^{\downarrow}}\right)^2-\lambda\left(\sum_{i=1}^L \tilde{q}_i^{\downarrow}-1+\epsilon\right),
\end{equation}
where $\lambda$ is a Lagrange multiplier. To find the solution of the problem we calculate 
\begin{eqnarray}
\frac{\partial \mathcal{L}}{\partial \tilde{q}_j^{\downarrow}}= \frac{1}{\Big(L\sqrt{\tilde{q}_j^{\downarrow}}\Big)}\left(\sum_{i=1}^L\sqrt{\tilde{q}_i^{\downarrow}}\right)-\lambda
\end{eqnarray}
for all $j \in \{1,\ldots,L\}$. Solving $\partial \mathcal{L}/\partial \tilde{q}_j^{\downarrow}=0$ we get 
\begin{eqnarray}\label{eq:qjtilde}
\tilde{q}^{\downarrow}_j=\Bigg[\frac{1}{\lambda L}\left(\sum_{i=1}^L\sqrt{\tilde{q}_i^{\downarrow}}\right)\Bigg]^2
\end{eqnarray}
for all $j \in \{1,\ldots,L\}$. Substituting $\tilde{q}^{\downarrow}_j$ from Eq.~\eqref{eq:qjtilde} to the constraint specified in Eq.~\eqref{constrainoflagrangemultiplier} gives 
\begin{equation}\label{eq:lambda_soln}
    \lambda=\frac{1}{\sqrt{(1-\epsilon)L}}\left(\sum_{i=1}^L\sqrt{\tilde{q}_i^{\downarrow}}\right).
\end{equation}
Therefore, putting $\lambda$ from Eq.~\eqref{eq:lambda_soln} into Eq.~\eqref{eq:qjtilde} finally gives
\begin{eqnarray}\label{eq:ordertilde}
\quad\forall i\in\{1,\ldots, L\}:\qquad \tilde{q}_i^{\downarrow} = \frac{1-\epsilon}{L}.
\end{eqnarray}
Since fidelity is a concave function, and the constraint defined in Eq.~\eqref{eq:Optimal_q} is linear, this solution is optimal.

Next, we note that $\v{p}\succ \tilde{\v{q}}$ implies $Q\v{p}^{\downarrow}=\tilde{\v{q}}^{\downarrow}$ where $Q$ is a bistochastic matrix. Since any $Q$ can be decomposed as a convex sum of permutations, we can write
\begin{equation}\label{Qmatrix}
    Q=\sum_{i}\alpha_i\Pi_i,
\end{equation}
where $\Pi_i$ is a permutation matrix and \mbox{$\sum_i \alpha_i=1$} with $\alpha_i\in[0,1]$ for all~$i$. Let us now define a vector $\v{v}$ as 
\begin{equation}
    \v{v}:=\Big(\underbrace{1,\ldots,1}_{L},0\ldots,0\Big).
\end{equation}
Using this $\v{v}$, we can write the following
\begin{eqnarray}\label{eq:QTv}
    \v{v}\cdot Q\v{p}^{\downarrow}= \v{v}\cdot\tilde{\v{q}}^{\downarrow}=\v{v}\cdot\v{p}^{\downarrow}=(1-\epsilon)\Rightarrow 
     \sum_{i}\alpha_i\v{v}\cdot\Pi_i\v{p}^{\downarrow}= \v{v}\cdot\v{p}^{\downarrow},
\end{eqnarray}
where we use Eq.~\eqref{constrainoflagrangemultiplier} in the first line, and the convex decomposition of $Q$ from Eq.~\eqref{Qmatrix} in the second line. Note that Eq.~\eqref{eq:QTv} can be equivalently written as 
\begin{eqnarray}\label{eq:Trans_equality}
    \sum_{i}\alpha_i(\Pi_i^T\v{v})\cdot\v{p}^{\downarrow}= \v{v}\cdot\v{p}^{\downarrow},
\end{eqnarray}
where $\Pi_i^T$ is the transpose of the permutation matrix~$\Pi_i$. Because the elements of $\v{v}$ and $\v{p}^{\downarrow}$ are arranged in a decreasing order, the maximum value of $(\Pi_i^T\v{v})\cdot\v{p}^{\downarrow}$ is $\v{v}\cdot\v{p}^{\downarrow}$. Thus, we see that equality in Eq.~\eqref{eq:Trans_equality} holds if and only if $\v{v}$ is invariant under $(\Pi_i^T)$ for all $i$. Therefore, $\Pi_i$ can be only of the form 
\begin{equation}
    \label{eq:PiDecomposition}
    \Pi_i = \Pi_i^L\oplus\Pi_i^{L^c},
\end{equation}
where $\Pi_i^L$ and $\Pi_i^{L^c}$ are permutations that act trivially for indices $i>L$ and $i\in\{1,\ldots,L\}$, respectively. Thus, we infer using Eq.~\eqref{Qmatrix} that $Q$ is a block-diagonal matrix of the form
\begin{equation}
    Q=\tilde{B}\oplus B,
\end{equation}
where $\tilde{B}$ and $B$ are $L\times L$ and $(d-L)\times(d-L)$ bistochastic matrices. Thus, 
\begin{eqnarray}
    \label{eq:ordertilde2}
    \forall i > L ,\quad \tilde{q}_i^{\downarrow} = \sum_{j>L}B_{ij}p_j^{\downarrow}.
\end{eqnarray}
Combining Eq.~\eqref{eq:ordertilde} and Eq.~\eqref{eq:ordertilde2} we conclude that $\tilde{\v{q}}^{\downarrow}={\v{q}}^{*\downarrow}$ with $\v{q}^{*}$ defined in Eq.~\eqref{eq:qstar}. This finally implies $\tilde{\v{q}}=\Pi\v{q}^{*}$ with $\Pi$ being some permutation. Moreover, it is straightforward to show that the minimal entropy of $\v{q}^{*}$ is achieved for $B$ being the identity matrix, since the entropy is increasing under application of any non-trivial bistochastic matrix.


\newpage
\subsection{Eliminating the logarithmic term}
\label{app:log}

We start with the following lemma that will be needed to prove our claim.

\begin{lemma}[Logarithmic growth of probability]\label{Orderoflogpro}
    For a fixed value of $b>0$ and any $\v{s}$, such that $\|\v{s}\|=\sqrt{\sum_{i=1}^ds_i^2}\leq b$, we have
    \begin{equation}
        \log P^N_{\v{k}(\v{s})}= O(\log N),
    \end{equation}
    where $P^N_{\v{k}}$ and $\v{k}(\v{s})$ are defined by Eqs.~\eqref{eq:Multinomial_prob_vector}~and~\eqref{eq:k_of_s}, respectively.
\end{lemma}
\begin{proof}
We start by using the definition,
\begin{equation}\label{Multinomialwiths}
    P^N_{\v{k}(\v{s})} = \binom{N}{k_1(s_1),...,k_d(s_d)}p^{k_1(s_1)}_1 ... p^{k_d(s_d)}_d,
\end{equation}
to write $\log P^N_{\v{k}(\v{s})}$ as
\begin{equation}
   \!\!\!\log P^N_{\v{k}(\v{s})}\!=\!\log N! -\sum_{i=1}^d \log k_i(s_i)!+\sum_{i=1}^d k_i(s_i)\log p_i.\!
\end{equation}
Employing Stirling inequality, 
\begin{align}
\log\sqrt{N}+N\log N-N\leq\log N!\leq 1+\log\sqrt{N}+N\log N-N,
    \label{Stirlingineq}
\end{align}
we first provide a lower bound for $\log P^N_{\v{k}(\v{s})}$,
\begin{align}
    \log P^N_{\v{k}(\v{s})}\geq& (\log\sqrt{N}+N\log N-N)+\sum_{i=1}^d k_i(s_i)\log p_i\nonumber\\
    &-\sum_{i=1}^d(1\!+\!\log\sqrt{k_i(s_i)}\!+\!k_i(s_i)\log k_i(s_i)\!-\!k_i(s_i)) \nonumber\\
    =&-\sum_{i=1}^d k_i(s_i)\log\Big(\frac{k_i(s_i)}{Np_i}\Big)-d+\frac{1}{2}\Big[\log N-\sum_{i=1}^d\log k_i(s_i)\Big].
    \label{eq:UpperBound}
\end{align}
Recall that $k_i(s_i)=Np_i+\sqrt{N}s_i$, which implies $\sum_{i=1}^d s_i=0$. To simplify the above further, we lower bound the first term by employing the inequality \mbox{$\log(1+g)<g$} for $g>-1$ in the following way:
\begin{align}
     \sum_{i=1}^d k_i(s_i)\log\Big(\frac{k_i(s_i)}{Np_i}\Big)&=\sum_{i=1}^d k_i(s_i)\log\left(1+\frac{s_i}{\sqrt{N}p_i}\right)\nonumber\\
     &\leq\sum_{i=1}^d k_i(s_i)\frac{s_i}{\sqrt{N}p_i}=\sqrt{N}\sum_{i=1}^d(1+\frac{s_i}{\sqrt{N}p_i})s_i\nonumber\\
     &=\sum_{i=1}^d \frac{s_i^2}{p_i}\leq\sum_{i=1}^d \frac{s_i^2}{p_{\min}}\leq \frac{b^2}{p_{\min}},
\end{align}
where $p_{\min}=\text{min }\{p_1,\ldots, p_d\}$. Moreover, observing that
\begin{equation}
     \log N \geq\left(\log N-\sum_{i=1}^d\log k_i(s_i)\right) \geq -(d-1)\log N ,
\end{equation}
we can conclude that 
\begin{equation}\label{eq:O(LogN)}
    \log N-\sum_{i=1}^d\log k_i(s_i)=O(\log N).
\end{equation}
Putting it together, we further simplify the bound given in Eq.~\eqref{eq:UpperBound} as
\begin{equation}
     \!\!\!\log P^N_{\v{k}(\v{s})}\geq -d-\frac{b^2}{p_{\min}}-\frac{(d-1)}{2}\log{N}=O(\log N).
\end{equation}

Similarly, by employing the Stirling inequality from Eq.~\eqref{Stirlingineq}, we also prove an upper bound for $\log P^N_{\v{k}(\v{s})}$ as follows
\begin{align}
    \log P^N_{\v{k}(\v{s})}
    \leq& (1+\log{\sqrt{N}}+N\log N-N)+\sum_{i=1}^dk_i(s_i)\log p_i\nonumber\\
    &-\sum_{i=1}^d\Big(\log\sqrt{k_i(s_i)}+k_i(s
    _i)\log k_i(s_i)-k_i(s_i)\Big)\nonumber\\
    =&-\sum_{i=1}^dk_i(s_i)\log\Big(\frac{k_i(s_i)}{Np_i}\Big)+1+\frac{1}{2}(\log N-\sum_{i=1}^d\log k_i(s_i)).
\end{align}
Using the inequality $\log(1+g)\geq\frac{g}{1+g}$ for $g>-1$, we have
\begin{align}
    &\sum_{i=1}^dk_i(s_i)\log\Big(\frac{k_i(s_i)}{Np_i}\Big)= \sum_{i=1}^dk_i(s_i)\log\Big(1+\frac{s_i}{\sqrt{N}p_i}\Big)\nonumber\\
    &\qquad\geq\sum_{i=1}^dk_i(s_i)\frac{s_i}{s_i+\sqrt{N}p_i}=\sqrt{N}\sum_{i=1}^d s_i = 0.
    \label{log_kisi_Npi}
\end{align}
The above inequality together with Eq.~\eqref{eq:O(LogN)} imply that $\log P^N_{\v{k}(\v{s})}\leq O(\log N)$ which completes the proof.
\end{proof}

Using Lemma~\ref{Orderoflogpro}, we will now be able to prove our claim that is captured by the following result.
\begin{lemma}[Equality between limits]
The following limits are equal:
\label{Lemma_to_remove_lnN}
\begin{align}
&\lim_{N\rightarrow\infty}\sum_{\v{s}}\Big\{P^N_{\v{k}(\v{s})}\Big| \frac{1}{N}\log(P^N_{\v{k}(\v{s})})+\frac{\beta}{\sqrt{N}}\sum_{i}s_iE_i +F_{\text{diss}}^N\geq 0\Big\} \nonumber \\
&\hspace{2cm}=\lim_{N\rightarrow\infty}\sum_{\v{s}}\Big\{P^N_{\v{k}(\v{s})}\Big| \frac{\beta}{\sqrt{N}}\sum_{i}s_iE_i+F_{\text{diss}}^N\geq 0\Big\}.
\end{align}
\end{lemma}
\newpage
\begin{proof}
We start by introducing the following notation
\begin{align}
 &A(N) := \sum_{\v{s}}\Big\{P^N_{\v{k}(\v{s})}\Big| \frac{1}{N}\log(P^N_{\v{k}(\v{s})})+\frac{\beta}{\sqrt{N}}\sum_{i}s_iE_i+F_{\text{diss}}^N\geq 0\Big\},\nonumber\\
 &A(b,N) := \sum_{\v{s}}\Big\{P^N_{\v{k}(\v{s})}\Big| \frac{1}{N}\log(P^N_{\v{k}(\v{s})})+\frac{\beta}{\sqrt{N}}\sum_{i}s_iE_i+F_{\text{diss}}^N\geq 0,\nonumber\\
 &\hspace{6.5cm}\text{  such that  } \|\v{s}\|\leq b\Big\},\nonumber\\
 &B(N) := \sum_{\v{s}}\Big\{P^N_{\v{k}(\v{s})}\Big| \frac{\beta}{\sqrt{N}}\sum_{i}s_iE_i+F_{\text{diss}}^N\geq 0\Big\},\nonumber\\
 &B(b,N) := \sum_{\v{s}}\Big\{P^N_{\v{k}(\v{s})}\Big| \frac{\beta}{\sqrt{N}}\sum_{i}s_iE_i+F_{\text{diss}}^N\geq 0, \text{  such that  } \|\v{s}\|\leq b\Big\},\nonumber\\
 &\Omega(b,N):= \sum_{\v{s}}\Big\{P^N_{\v{k}(\v{s})}\Big| \text{  such that  } \|\v{s}\|\geq b\Big\}.
 \end{align}
 Our goal is to show that
 \begin{equation}
     \lim_{N\rightarrow\infty} A(N)=\lim_{N\rightarrow\infty} B(N).
 \end{equation}
From the definition it follows that 
\begin{subequations}
\begin{eqnarray}
     A(N)-\Omega(b,N)&\leq& A(b,N)\leq A(N)\label{LimA},\\
      B(N)-\Omega(b,N)&\leq& B(b,N)\leq B(N)\label{LimB}.
\end{eqnarray}
\end{subequations}
Taking the limit $N\rightarrow\infty$ of Eqs.~\eqref{LimA}~and~\eqref{LimB}, we have 
\begin{subequations}
\begin{align}
    \lim_{N\rightarrow\infty}\Big(A(N)-\Omega(b,N)\Big)&\leq\lim_{N\rightarrow\infty} A(b,N)\leq \lim_{N\rightarrow\infty} A(N),\label{LimA1}\\
    \lim_{N\rightarrow\infty}\Big(B(N)-\Omega(b,N)\Big)&\leq \lim_{N\rightarrow\infty}B(b,N)\leq \lim_{N\rightarrow\infty}B(N)\label{LimB1}.
\end{align}
\end{subequations}
Now, let us define
\begin{eqnarray}
     \lim_{N\rightarrow\infty}\Omega(b,N)=:\epsilon(b).
\end{eqnarray}
As the multinomial distribution concentrates around mean for $N\rightarrow\infty$, it follows that $\lim_{b\rightarrow\infty}\epsilon(b)=0$. Therefore, taking the limit $b\rightarrow\infty$ in Eq.~\eqref{LimA1}  we have
\begin{align}
     & \lim_{N\rightarrow\infty}A(N)\leq\lim_{b\rightarrow\infty} \lim_{N\rightarrow\infty}A(b,N)\leq \lim_{N\rightarrow\infty}A(N)\nonumber\\
     &\qquad\Rightarrow\quad\lim_{b\rightarrow\infty} \lim_{N\rightarrow\infty}A(b,N)= \lim_{N\rightarrow\infty}A(N).\label{eq:ANABN}
\end{align}
Analogously, taking the limit $b\rightarrow\infty$ in Eq.~\eqref{LimB1} we can show that
\begin{equation}\label{eq:BNBBN}
    \lim_{b\rightarrow\infty} \lim_{N\rightarrow\infty}B(b,N)=\lim_{N\rightarrow\infty}B(N).
\end{equation}

Moreover, for any fixed $b$, by employing Lemma~\ref{Orderoflogpro}, we see that $\frac{1}{N}\log(P^N_{\v{k}(\v{s})})=O(\frac{\log N}{N})$, which vanishes as \mbox{$N\rightarrow\infty$}. Therefore, we have 
 \begin{equation}
     \lim_{N\rightarrow\infty} A(b,N) =\lim_{N\rightarrow\infty} B(b,N),
 \end{equation}
and so by taking the limit $b\rightarrow\infty$, we arrive at
\begin{equation}\label{eq:ANABNBNBBN}
     \lim_{b\rightarrow\infty}\lim_{N\rightarrow\infty} A(b,N) =\lim_{b\rightarrow\infty}\lim_{N\rightarrow\infty} B(b,N).
\end{equation}
Combining the above with Eqs.~\eqref{eq:ANABN}-\eqref{eq:BNBBN} we have
\begin{equation}
    \lim_{N\rightarrow\infty} A(N) =\lim_{N\rightarrow\infty} B(N),
\end{equation}
which completes the proof.
\end{proof}


\subsection{Central limit theorem for multinomial distribution}
\label{app:clt}

Let us now prove the central limit theorem for multinomial distribution. We start by stating the following lemma:

\begin{lemma}[Cetral limit theorem for multinomial distribution]\label{CLT_Multinomial}
    The multinomial distribution with mean \mbox{$\v{\mu}=N\v{p}$} and a covariance matrix $\v{\Sigma}$ can be approximated in the asymptotic limit by a multivariate normal distribution $\mathcal{N}^{(\v{\mu},\v{\Sigma})}$ with mean $\v{\mu}$ and a covariance matrix~$\v{\Sigma}$.
\end{lemma}
\begin{proof}
Assume $\v{X_1}\ldots \v{X_N}$ are independent and identically distributed random vectors each
of them with the following distribution
\begin{equation}
\!\!\!\text{Prob} (\v{X}=\v{x}) = 
     \begin{cases}
       \prod_{i=1}^d p_i^{x_i} &\quad\!\text{if $\v{x}$ is unit vector},\\
       0 &\quad\!\text{otherwise}.
     \end{cases}
\end{equation}
Then, the mean vector of $\v{X}$ is $\v{p}$ and the covariance matrix $\frac{1}{N}\v{\Sigma}=\text{ diag }(\v{p})-\v{p}\v{p}^T$. Define  $\v{S}_N:=\v{X}_1+\ldots+\v{X}_N$. Then 
\begin{eqnarray}
     \text{Prob}(\v{S}_N=\v{k}) &=&\binom{N}{k_1\ldots k_d}p_1^{k_1}\ldots p_d^{k_d}.
\end{eqnarray}
We thus see that a multinomial distribution arises from a sum of independent and identically distributed random variables. Therefore, using the central limit theorem, we obtain that the distribution of $\v{k}$ approaches the distribution $\mathcal{N}^{(\v{\mu},\v{\Sigma})}$ arbitrarily well for \mbox{$N\to\infty$}, which completes the proof.
\end{proof}


\subsection{Entropy difference between uniformised and non-uniformised embedding boxes}
\label{app:DDF}

In this section, we show that there exists a permutation which transforms the total final state $\hat{\v{F}}^N$ from Eq.~\eqref{eq:distribution-total} into a state that is uniform in almost all embedding boxes, leading to no dissipation up to higher-order asymptotics. We start with the following observations about the total initial and final states. First, since we consider the initial state composed of $N$ independent systems in identical incoherent states $\v{P}^N=\v{p}^{\otimes N}$, we note that $\hat {\v P}^{N} \otimes \hat{\tilde{\v{G}}}^N$ has a polynomial number of distinct entries with exponential degeneracy and exponentially many embedding boxes. Second, note that according to Eq.~\eqref{eq:distribution-total} the entries of the embedded total final state, $\hat{\v{F}}^N$, are essentially given by the entries of $\hat {\v P}^{N} \otimes \hat{\tilde{\v{G}}}^N$ (plus many equal entries $(1-\epsilon)/K$). Thus, $\hat{\v{F}}^N$ has polynomially many distinct entries with exponential degeneracy and exponentially many embedding boxes. Employing permutational freedom, we can then rearrange the entries of $\hat{\v{F}}^N$ so that they are uniform in almost every embedding box, except poly($N$) of them. We will denote the probability distribution over the embedding boxes by $\v{q}$ and note that it is essentially equal to $\v P^{N} \otimes \, \tilde{\v G}^N$. Moreover, we will denote by $\v{r}^{(i)}$ the normalised distribution within the $i$-th embedding box.

The next step is to write the entropy of the embedded total final state as the entropy of probability distribution over different embedding boxes, $H(\v q)$, plus the average entropy of normalised probabilities within each box, $H(\v{r}^{(i)})$,
\begin{equation}
\label{eq:Hnormalisedbox}    
    H(\hat{\v{F}}^N) = H(\v{q}) + \sum_i q_i H(\v{r}^{(i)}).
\end{equation}
Note that whether we uniformise or not a given box, the entropy $H(\v q)$ does not change. Consequently, the entropy of the final state uniformised within each embedding box takes the form of
\begin{align}
\label{eq:entr-totalfinuni}
    H(\hat{\v{F}}_{\text{uni}}^N) &= H(\v{q}) + \sum_i q_i H(\v{r}^{(i)}_{\text{uni}}) ,
\end{align}
with $\v{r}^{(i)}_{\text{uni}}$ representing the normalised and uniformised probability within the $i$-th embedding box. Thus, the entropy difference between the uniformised and non-uniformised distributions reads 
\begin{equation}
    H(\hat{\v{F}}_{\text{uni}}^N) - H(\hat{\v{F}}^N) = \sum_i q_i \left [H(\v{r}^{(i)}_{\text{uni}})-H(\v{r}^{(i)}) \right] .
\end{equation}

Now, note that due to the previous argument, the above sum is performed only over a polynomial number of boxes. Denoting the set with size poly($N$) as $\Omega$, one can write
\begin{align}
  H(\hat{\v{F}}_{\text{uni}}^N) - H(\hat{\v{F}}^N) &= \sum_{i\in\Omega} q_i \left [H(\v{r}^{(i)}_{\text{uni}})-H(\v{r}^{(i)}) \right] \leq \sum_{i\in\Omega} q_iH(\v{r}^{(i)}_{\text{uni}}).
\end{align}
Let us analyse the right-hand side of the above equation. First, note that $q_i$ is exponentially small, $q_i \propto \exp(-N)$; while $H(\v{r}^{(i)}_{\text{uni}})$ is the entropy of a uniform state over the dimension of the $i$-th embedding box (equal to $\prod_i D^{k_i}_i$), so it will scale linearly in $N$:
\begin{equation}
    H(\v{r}^{(i)}_{\text{uni}})=\log \left( \prod_i D^{k_i}_i\right) \propto O(N) \, .
\end{equation}
Therefore, we conclude that the entropy difference vanishes exponentially in $N$
\begin{equation}
H(\hat{\v{F}}_{\text{uni}}^N) - H(\hat{\v{F}}^N) \propto \exp(-N).
\end{equation}


\section{Concluding remarks}
\label{sec:out}

In this chapter, we have derived a version of the fluctuation-dissipation theorem for state interconversion under thermal operations. We achieved this by establishing a relation between optimal transformation error and the amount of free energy dissipated in the process on the one hand, and fluctuations of free energy in the initial state of the system on the other hand. We addressed and solved the problem in two different regimes: for initial states being either energy-incoherent or pure, with the target state in both cases being an energy eigenstate, and with the possibility to change the Hamiltonian in the process. For the case of finitely many independent but not necessarily identical energy-incoherent systems, we have provided the single-shot upper bound on the optimal transformation error as a function of average dissipated free energy and free energy fluctuations. Moreover, in the asymptotic regime we obtained the optimal transformation error up to second order asymptotic corrections, which extends previous results of Ref.~\cite{Chubb2018beyondthermodynamic} to the regime of non-identical initial systems and varying Hamiltonians. For the first time we have also performed the asymptotic analysis of the thermodynamic distillation process from quantum states that have coherence in the energy eigenbasis. As a result, we expressed the optimal transformation error from identical pure states and free energy dissipated during this transformation up to second order asymptotic corrections as a function of free energy fluctuations.

The obtained results can be naturally extended in the following directions. Firstly, one could generalise our analysis to arbitrary initial states. We indeed believe that an analogous result to ours will hold for such general mixed states with coherence. That is because dephasing into fixed energy subspaces leads to free energy change of the order $O(\log N)$, which is negligible compared to the second order asymptotic corrections of the order $O(\sqrt{N})$ that we focus on. In other words, the contribution of coherence to free energy per copy of the system vanishes faster with growing $N$ than what we are interested while studying second order corrections. 

Secondly, it would be extremely interesting to generalise the thermodynamic state interconversion problem to arbitrary final states, and see how the interplay between the fluctuations of the initial and target states affects dissipation. For energy-incoherent initial and final states one can infer from Ref.~\cite{korzekwa2019avoiding} that appropriately tuned fluctuations can significantly reduce dissipation, however nothing is known for states with coherence. Unfortunately, since thermal operations are time-translation covariant, such that coherence and athermality form independent resources~\cite{lostaglio2015description, lostaglio2015description, PRLHorodeckicoherence}, it seems unlikely that the current approach can be easily generalised. Thirdly, one could try to extend our results on pure states to allow for non-identical systems and to derive a bound working for all $N$, not only for $N\to\infty$ (i.e., replace the proving technique based on central limit theorem by the one based on a version of Berry-Esseen theorem).

In this chapter we have also provided a number of physical applications of our fluctuation-dissipation theorems by considering several scenarios and explaining how our results can be useful to describe fundamental and well-known thermodynamic and information-theoretic processes. We derived the optimal value of extractable work in a thermodynamic distillation process as a function of the transformation error associated to the work quality. This, together with the knowledge of the actual final free energy of the battery system provided by Theorem~\ref{thm:incoherent2}, could potentially be used to clarify the notion of imperfect work~\cite{Aberg2013, Woods2019maximumefficiencyof, Ying_Ng_NJP}, and to construct a comparison platform allowing one to continuously distinguish between work-like and heat-like forms of energy. We have also shown how our results yield the optimal trade-off between the work invested in erasing $N$ independent bits prepared in arbitrary states, and the erasure quality measured by the infidelity distance between the final state and the fully erased state. This can of course be straightforwardly extended to higher-dimensional systems and arbitrary final erased state (not necessarily the ground state). Finally, we have investigated the optimal encoding rate into a collection of non-interacting subsystems consisting of energy-incoherent or pure states using thermal operations. We derived the optimal rate (up to second-order asymptotics) of encoding information with a given average decoding error and without spending thermodynamic resources. This provides an operational interpretation of the resourcefulness of athermal quantum states for communication scenarios under the restriction of using thermal operations.

We would also like to point out to some possible technical extensions of our results. Firstly, we used infidelity as our quantifier of transformation error, but we expect that similar results could be derived using other quantifiers, e.g., the trace distance. Secondly, our investigations were performed in the spirit of small-deviation analysis (where we look for constant transformation error and total free energy dissipation of the order $O(\sqrt{N})$), but possibly other interesting fluctuation-dissipation relations could be derived within the the moderate and large deviation regimes. Thirdly, our result for pure states is limited to Hamiltonians with incommensurable spectrum, but we believe this is just a technical nuisance that one should be able to get rid of. Lastly, within the framework of general resource theories, it might be possible to derive analogous fluctuation-dissipation relations, but with free energy replaced by a resource quantifier relevant for a given resource theory.
\chapter{Quantum catalysis in cavity quantum electrodynamics}\label{C:CQED}

The effect of catalysis involves using an auxiliary system (a catalyst) to enable a process that would either not occur spontaneously or would occur very slowly. Catalysis manifests across a variety of fields (see e.g.,~\cite{chorkendorff2017concepts}), including biological processes activated by enzymes, the speed-up of chemical reactions, and the synthesis of nanomaterials. 

More recently, the phenomenon of catalysis has also become relevant in the context of quantum information; see recent reviews \cite{datta2022catalysis,lipkabartosik2023catalysis}. First examples focused on entanglement manipulation~\cite{jonathan1999entanglement,vanDam2003,turgut2007catalytic,daftuar2001mathematical,sun2005existence,feng2005catalyst,PhysRevLett.127.080502,Kondra_2021,Datta2022}, and then spread to quantum thermodynamics~\cite{brandao2015second,Ng_2015,Wilming2017,Mueller2018,Lipka-Bartosik2021,shiraishi2021quantum,Gallego2016,boes2019bypassing,Henao_2021,Henao2022,Jeongrak2022,son2023hierarchy,Lipka_Bartosik_2023,czartowski2023thermal}, coherence theory~\cite{Aberg2014,Vaccaro2018,lostaglio2019coherence,takagi2022correlation,char2023catalytic,van2023covariant} and other areas~\cite{Marvian2019,wilming2021entropy,Wilming2022correlationsin,rubboli2022fundamental,lie2021catalytic,boes2019neumann}. These results are typically formulated within the framework of quantum resource theories~\cite{chitambar2019quantum}. This abstract approach is particularly useful for characterising the fundamental limits of manipulating quantum resources, including scenarios involving catalytic systems of arbitrary complexity. 

An interesting question is whether quantum catalysis is also relevant and useful in a more practical context, potentially even in experiments. Here we investigate quantum catalysis in a paradigmatic setup of quantum optics, namely the Jaynes-Cummings model~\cite{Jaynes1963,Larson_2021,Greentree_2013}, where a two-level atom interacts with a single-mode optical cavity. In this chapter, we uncover a catalytic process enabling the generation of a non-classical state of light in the cavity, using the atom as a catalyst. Specifically, we consider the cavity to be initially prepared in a ``classical'' coherent state, and uncorrelated to the atom. By carefully setting the initial state of the atom and the interaction time, we obtain a final state such that \emph{(i)} the atom is back in its initial state exactly, and \emph{(ii)} the state of the cavity is now non-classical, i.e., featuring Wigner negativity or sub-Poisonian statistics. Hence, non-classicalilty of the cavity has been generated without perturbing the state of the atom (see Fig.~\ref{scheme}). The process is catalytic and the atom could be re-used, for example by coupling it to another cavity. 

We investigate the mechanism of this catalytic process, and identify two crucial ingredients. First, the final state of the atom and cavity must feature correlations. Second, the evolution of the state of the catalyst must involve quantum coherence (i.e., superpositions of the energy basis states). The latter is an interesting aspect, as typical instances of quantum catalysis in resource theories involve only diagonal states (i.e., without coherence), so that they can be understood as a stochastic process involving the probability distributions of the system and catalysis. Instead here, the system and catalyst experience a genuinely quantum evolution. This is a novel instance of the effect of \emph{coherent quantum catalysis}, recently investigated in quantum thermodynamics~\cite{Lipka_Bartosik_2023}.

Before proceeding, it is worth discussing previous works in quantum optics that relate to the concept of catalysis. Notably, the pioneering proposal for quantum computing in ion traps~\cite{Cirac1995} (see also~\cite{Phoenix1993,Hagley1997}) considers two spin qubits that become entangled via an interaction with a cavity that can be considered catalytic~\cite{lipkabartosik2023catalysis}. Another relevant direction is that of ``multi-photon catalysis'' (see e.g.,~\cite{Lvovsky2002,Bartley2012,hu2016multiphoton,Hu2017}) which is a heralded catalytic process, where the catalyst is returned only with some probability. In contrast, our catalytic protocol is deterministic. 

\begin{figure}[t]
\includegraphics[width=10.7cm]{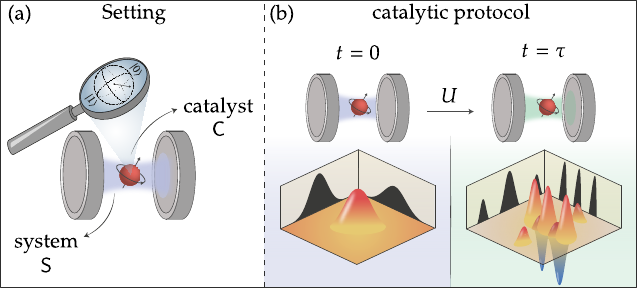}
	\caption[catalysis]{\emph{Quantum catalysis in the Jaynes-Cummings model.} (a) An atom (the catalyst $\ms C$) interacts with a single-mode optical cavity (the system $\ms S$), initially prepared in a ``classical'' coherent state. (b) We consider the evolution $U$ over a well-chosen time interval (from $t=0$ to $t=\tau$) such that (i) the final state of the cavity is non-classical, and (ii) the atom is returned to its intial state exactly. Hence, non-classicality has been generated via catalysis.}
	\label{scheme}
\end{figure}

\section{Setting the scene}

Consider a setup comprising a system $(\ms{S})$ initially prepared in a state $\rho_{\ms S}$ and a catalyst $(\ms{C})$ in an initial state $\chi_{\ms C}$. The total system $\ms{SC}$ is assumed to be closed and evolves via an energy-conserving process for some time~$\tau$. This evolution is represented by the unitary $U=\exp({-i H_{\ms{SC}} \tau})$, where $H_{\ms{SC}}$ denotes the joint Hamiltonian. Consequently, the final state of the total system is given by $\sigma_{\ms{SC}} := U(\rho_{\ms{S}} \ot \chi_{\ms{C}})U^{\dagger}$.

The evolution is said to be \emph{catalytic} when the catalyst is returned in exactly the same state as it was initially prepared. Formally, we demand that
\begin{equation}\label{Eq:catalytic-constrain}
    \sigma_{\ms C} := \Tr_{\ms{S}}[U(\rho_{\ms{S}} \ot \chi_{\ms{C}})U^{\dagger}] = \chi_{\ms C},
\end{equation}
which we refer to as the \emph{catalytic constraint}. Satisfying this constraint typically requires to carefully choose the initial states of the system $\rho_{\ms{S}}$ and the catalyst $\chi_{\ms{C}}$, as well as the interaction time $\tau$.

The main goal of a catalytic evolution is to induce an interesting local dynamics on the system $\ms{S}$, i.e.,
\begin{align} \label{eq:local_evol_S}
    \rho_{\ms{S}} \rightarrow \sigma_{\ms{S}} := \Tr_{\ms{C}}[U(\rho_{\ms{S}} \ot \chi_{\ms{C}})U^{\dagger}],
\end{align}
while leaving the state of the catalyst unchanged. Notably, it is possible to induce an evolution on $\ms{S}$ [as in Eq. (\ref{eq:local_evol_S})] that would not be possible without the presence of the catalyst.

\subsection{Jaynes-Cummings model}

In this Section, we discuss the phenomenon of catalysis in the Jaynes-Cummings (JC) model, describing the interaction between a single-mode optical cavity and a two-level atom~\cite{Jaynes1963} (see Fig.~\hyperref[scheme]{\ref{scheme}a}). 

We choose the cavity to represent the system $\ms S$, while the atom will play the role of the catalyst $\ms C$. The cavity is characterised by the bosonic annihilation operator $a$ with the photon number operator $n_{\ms{S}} := a^{\dagger}a$. The atom has energy levels $\ket{g}$ and $\ket{e}$ and its energy is captured by $\sigma_z = \ketbra{e}{e}-\ketbra{g}{g} $. We work in the resonant regime, where the atom and cavity have the same frequency $\omega$. The evolution is governed by the JC Hamiltonian, which in the rotating-wave approximation reads
\begin{align}\label{Eq:Total_Hamiltonian}
    H_{\ms{SC}} = \omega a^{\dagger} a + \frac{\omega}{2} \sigma_z +  g \left(\sigma_+ a + \sigma_- a^{\dagger} \right),
\end{align}
where $g$ is the coupling constant and \mbox{$\sigma_{+} = \ketbra{e}{g}$} and \mbox{$\sigma_- = \ketbra{g}{e}$} are the raising and lowering operators. Note that, as we focus on the resonant regime, the evolution specified by Eq.~(\ref{Eq:Total_Hamiltonian}) is energy-preserving. 

The Jaynes-Cummings interaction Hamiltonian $H_{\text{int}} = g \left(a\sigma_+ + a^{\dagger}\sigma_- \right)$ couples pairs of atom-field states $\{\ket{n+1,g},\ket{n,e} \}$. Consequently, the Hamiltonian $H_{\ms{SC}}$ decouples into a direct product of $2\times2$-matrix Hamiltonians, i.e., $H_{\ms{SC}}=\bigoplus_{n=0}^{\infty}H_{\ms{SC}}^{(n)}$, where

\begin{equation}\label{Eq:JC_eigeinproblem}
    H_{\ms{SC}}^{(n)} \begin{bmatrix}
\ket{n+1,g}\\ \ket{n,e} 
\end{bmatrix} = \begin{bmatrix}
(n+1/2)\omega & g\sqrt{n+1}\, \\ \\
g\sqrt{n+1} &  (n+1/2) \omega 
\end{bmatrix}\begin{bmatrix}
\ket{n+1,g}\\ \ket{n,e} 
\end{bmatrix}.
\end{equation}
The eigenvalue problem for this Hamiltonian yields the eigenfrequencies
\begin{equation}\label{Eq:eigenfrequencies}
    \omega^{(n)}_{\pm} = \qty(n+\frac{1}{2})\omega \pm \frac{1}{2}\mu_n,
\end{equation}
where $\mu_n = 2g \sqrt{n+1} $ is the $n$-photon Rabi frequency. In the resonance regime, the corresponding eigenstates are:
\begin{align}
    \ket{n,+} =   \frac{1}{\sqrt{2}}(\ket{n+1,g} + \ket{n,e}) \: , \:
    \ket{n,-} =  \frac{1}{\sqrt{2}} (\ket{n+1,g} - \ket{n,e}).
\end{align}

The eigenstates $\ket{n,\pm}$ are referred to as dressed atom states. The unperturbed atomic eigenstates, denoted as $\ket{g}$ and $\ket{e}$, are modified (or 'dressed') due to their interaction with the cavity field, resulting in a shift of their eigenfrequencies by an amount that is determined by the strength of the coupling, as illustrated in Fig.~\ref{F-degeneracy-levels}.
\begin{marginfigure}[-2.6cm]
	\includegraphics[width=4.753cm]{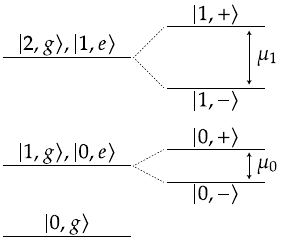}
	\caption{\emph{Jaynes-Cummings ladder}. Resonant atom-light interaction lifts the degeneracy of the unperturbed states $\{\ket{n,e},\ket{n+1,g} \}$. The level splitting $\mu_n = 2g\sqrt{n+1}$ depends on the number of photons.}
	\label{F-degeneracy-levels}
\end{marginfigure}
Given that we know the eigenvalues and eigenstates of the Jaynes–Cummings Hamiltonian, the unitary time evolution operator $U(t) = e^{-iHt}$ takes on the form of
\begin{align}\label{Eq:time-evolution-operator}
   &\hspace{-0.2cm} U(t) = e^{i\omega t} \ketbra{0,g}{0,g} \nonumber\\&\hspace{0.6cm}+\sum_{n=0}^{\infty}e^{-i\qty(n+\frac{1}{2})\omega t}  \Bigl\{\cos \frac{\mu_n t}{2}\Bigl(\ketbra{n\!+\!1,g}{n\!+\!1,g}+\ketbra{n,e}{n,e}\Bigl)\nonumber\\ &\hspace{2.7cm}-i\sin\frac{\mu_n t}{2}\Bigl(\ketbra{n\!+\!1,g}{n,e}+\ketbra{n,e}{n\!+\!1,g}\Bigl)\Bigl\}.
\end{align}

The unitary operator~\eqref{Eq:time-evolution-operator} describes the full dynamics of the Jaynes-Cummings model. In particular, we derive the final reduced states of the cavity $\ms{S}$ and the atom $\ms{C}$, and explicitly determine the set of atomic states that satisfy the catalytic constraint [Eq.~\eqref{Eq:catalytic-constrain}]. Here, we will omit the technical details of the derivation and instead focus on the discussion of our results~(see Sec. \ref{Appsub:generalsolution} for the derivation of the reduced states of the system and cavity, and Sec. \ref{Appsub:set-catalytic-states} for the solution of the catalytic constraint.).

Our objective is to catalytically generate non-classical light in the cavity. To investigate the possibility of doing so, we will use two complementary figures of merit as non-classicality witnesses. 
\begin{marginfigure}[1.2cm]
	\includegraphics[width=4.7618 cm]{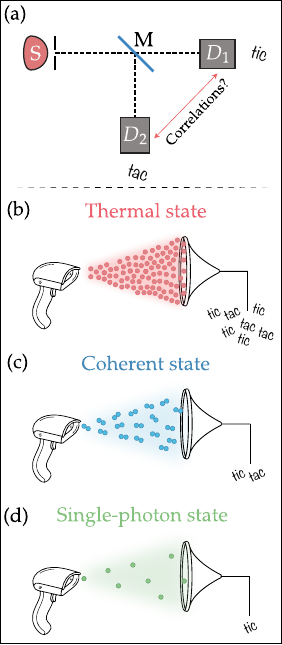}
	\caption{\emph{Hanbury Brown and Twiss experiment}. (a) A light source $S$ emits photons that are transmitted through or reflected by a mirror M. Two detectors, namely $D_1$ and $D_2$, count the number of photons, and the correlation between them is determined. Photon detections are recorded as a function of time for (b)~bunching (chaotic light), (c) randomness (e.g., a coherent state, laser beam), and (d) antibunching (e.g., light emitted from a single atom).
}
	\label{Fig:second-order-exp}
\end{marginfigure}

\subsection{Figures of merit}

\begin{center}
    \emph{Second-order coherence}
\end{center}

As a first figure of merit, we quantify non-classicality via the second-order auto correlation function (or simply second-order coherence) of the final state of the cavity~\cite{PhysRev.130.2529}, i.e.,
\begin{align}\label{eq:def_g2}
   g^{(2)}(\sigma_{\ms{S}}) = \frac{\langle a^{\dagger 2}a^2 \rangle_{\sigma}}{\langle a^{\dagger}a \rangle^2_{\sigma}} = \frac{\langle n_{\ms{S}}^2 \rangle_{\sigma}  - \langle n_{\ms{S}} \rangle_{\sigma}}{\langle n_{\ms{S}} \rangle_{\sigma}^2},
\end{align}
where $\langle n_{\ms{S}}\rangle_{\rho} := \Tr[\rho_{\ms{S}} n_{\ms{S}}]$. 

Physically, the concept of second-order coherence can be interpreted as a response to the following question~\cite{loudon1976photon}: ``\emph{If I have detected a photon or multiple photons at time $t$, what is the probability or correlation of detecting another photon or a similar number of photons at a later time $t + \tau$?}''. Depending on the type of light source, this question yields different answers. 

To delve deeper into the physics behind the answer to such a question, let us discuss the Hanbury Brown and Twiss experiment~\cite{brown1956correlation}. Firstly, let us remember that quantum theory views a beam of light as a stream of photons, with the measurement of beam intensity interpreted as the counting of photon arrivals at a phototube. In the Hanbury Brown and Twiss experiment [see Fig.~\hyperref[Fig:second-order-exp]{\ref{Fig:second-order-exp}a} for a schematic representation], two detectors, $D_1$ and $D_2$, count the number of photons transmitted through or reflected from the mirror $M$. Consider a series of measurements where the numbers of photons, $n_1$ and $n_2$, counted by the two detectors during a constant time interval, are recorded. If the series of identical experiments is sufficiently long, each detector registers the same average number of counts, $\langle n_1\rangle = \langle n_2\rangle$. However, the number of counts in each detector in a single run is not necessarily the same. In the quantum realm, each incident photon either passes through the half-silvered mirror or is reflected from it. Consequently, the two split beams are not exact replicas of each other; it is only their average properties over an extended series of experiments that are identical.

When a beam of light from an ordinary source is directed onto a photomultiplier tube, which records the arrival of photons, an intriguing phenomenon known as \emph{photon bunching} can be observed. In this phenomenon, the photons detected by the phototube do not scatter randomly in time; rather, they are detected in the form of clusters or bunches~[see Fig.~\hyperref[Fig:second-order-exp]{\ref{Fig:second-order-exp}b}]. Computing the degree of second-order coherence of thermal (chaotic) light yields a value of 2. This number indicates that a thermal light field has \emph{super-Poissonian statistics}. However, if the light source is now prepared as coherent light, it exhibits no classical intensity fluctuations or photon bunching. As a result, the counts in the two detectors are uncorrelated, and only random coincidences occur, leading to a correlation function equal to $g^{(2)} = 1$~[see Fig.\hyperref[Fig:second-order-exp]{\ref{Fig:second-order-exp}d}]. Coherent light exhibits Poissonian statistics, yielding random photon spacing. Interestingly, within such an experimental setup, we can observe a phenomenon known as \emph{antibunching}, a counterpart to photon bunching. This unique type of light is generated through nonlinear interactions of laser light with matter and is solely compatible with the quantum theory. Unlike coherent light, where the degree of second-order coherence equals one, antibunched light exhibits a degree of second-order coherence below unity, i.e., $g^{(2)}<1$~[see Fig.~\hyperref[Fig:second-order-exp]{\ref{Fig:second-order-exp}b}]. Consequently, it is considered to be anticorrelated or antibunched, highlighting its quantum nature. In this case, the light field follows a sub-Poissonian photon statistics, which means that the photon number distribution has a variance that is less than the mean. Such a discussion can be summarised as follows:
\begin{equation}
    g^{(2)}=\begin{cases}
2 \quad\:\: \Longrightarrow \: \text{Super-Poissonian statistics.} \\
1 \quad\:\: \Longrightarrow \: \text{Poissonian statistics}\\
< 1 \:\: \Longrightarrow \: \text{Sub-Poissonian photon statistics}
\end{cases}
\end{equation}

\begin{center}
    \emph{Wigner logarithmic negativity}
\end{center}

Our second figure of merit is the Wigner logarithmic negativity (WLN)~\cite{veitch2014resource,Albarelli2018,PhysRevA.97.062337}, defined as
\begin{align} \label{Wigner Neg}
\mathsf{W} \left ( \rho \right) := \log \left( \int \! \mathrm{d} x \,\mathrm{d} p \, \left| W_\rho \left(x,p \right) \right| \right),
\end{align}
where $W_\rho \left(x,p \right)$ is the Wigner function
\begin{align}\label{eq:w_fun}
    W_{\rho}(x,p) = \frac{1}{\pi} \int e^{2i p x'} \langle x-x'|\rho|x+x' \rangle \, \text{d}x'.
\end{align}
Negativity of the Wigner function has long been recognised as an important quantum feature, especially the volume of the negative part, which has been introduced as a nonclassicality quantifier~\cite{kenfack2004negativity}. The crucial aspect of the Wigner function's logarithmic negativity (WLN) is its property as an additive monotone\sidenote{A resource theory for continuous-variable systems has been proposed~\cite{Albarelli2018,PhysRevA.97.062337}, which quantifies both quantum non-Gaussianity (defined as the convex hull of Gaussian states being the set of free states) and Wigner negativity (with states having positive Wigner functions considered as free states). In both cases, the Wigner logarithmic negativity is a resource monotone.}, as the Wigner function of separable states can be factorised. Additionally, it is computable through numerical integration.

Since the Wigner function is normalised, it is greater than zero ($\mathsf{W} > 0$) whenever $W_{\rho}(x, p)$ is negative for some region of the phase space.

\section{Generating non-classical states of light}

The task we would like to achieve catalytically is the generation of non-classical light in the cavity. We consider an initial state of the cavity that is classical, namely a coherent state
\begin{equation}
    \ket{\alpha} = e^{-|\alpha|^2/2}\sum_{n = 0}^{\infty} \frac{\alpha^n}{\sqrt{n !}} \ket{n}.
\end{equation}
This state represents the closest approximation to an ideal, perfectly stable, classical field with a well-defined amplitude and phase. Coherent states exhibit Poissonian statistics with the mean number of photons $\langle n_{\mathrm{S}}\rangle_{\ket{\alpha}} = |\alpha|^2$. The second-order coherence for a coherent state is equal to unity, i.e., $g^{(2)}(\ketbra{\alpha}{\alpha}) = 1$, and, since the Wigner function of a coherent state is Gaussian, it implies that it is everywhere non-negative, resulting in $\mathsf{W} = 0$.

\begin{figure}[t]
    \centering
    \includegraphics[width=10.407cm]{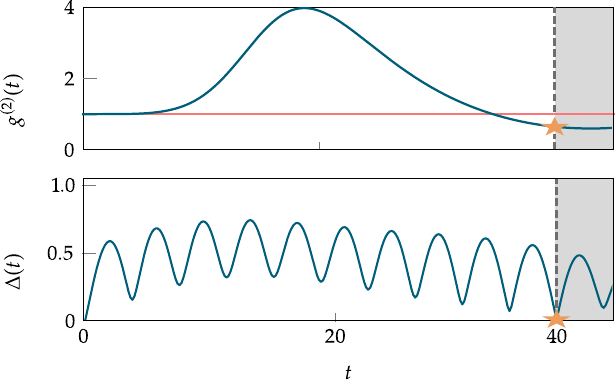}
    \caption{\emph{First illustrative example.} Catalytic process for generating non-classicality in the cavity, as captured by the second-order coherence. The time evolution, $g^{(2)}(t) := g^{(2)}(\sigma_{\mathrm{S}}(t))$, is shown in the top panel, while the bottom panel displays the change of the atomic state, measured via the distance $\Delta(t) := |\chi_{\mathrm{C}}-\sigma_{\mathrm{C}}(t)|_1$ concerning the initial state. The orange stars indicate the final time ($\tau \approx 40$) for which the evolution is catalytic. The atom returns to its initial state ($\Delta(\tau)=0$), while non-classicality has been activated, $g^{(2)}(\tau) \approx 0.5 < g^{(2)}(0)=1$. The values of the parameters are: $\alpha = 1/\sqrt{2}$, $\omega = 2\pi$, $g = \pi$.}
    \label{fig:2}
\end{figure}

Our goal is to find an initial state of the atom $\chi_{\ms{C}}$ and an interaction time $\tau$ such that the evolution is catalytic [i.e., satisfying Eq. (\ref{Eq:catalytic-constrain})], while also leading to a final state of the cavity $\sigma_{\ms{S}}$ that exhibits non-classicality. Below, we present two illustrative examples of such processes, using two complementary figures of merit to witness non-classicality. A detailed analysis of these examples will be discussed in the next section.

In Fig.~\ref{fig:2}, we present the catalytic activation of non-classicality. The upper plot depicts the time evolution of $g^{(2)}$. Catalysis occurs at time $\tau \approx 40$, indicated by the orange stars. Starting from the initial state of the cavity $\rho_{\mathrm{S}} = \dyad{\alpha}_{\mathrm{S}}$ (with $\alpha=1/\sqrt{2}$), which has $g^{(2)}(\rho_{\mathrm{S}})=1$ as any coherent state, we obtain a final state $\sigma_{\mathrm{S}}$ for which $g^{(2)}(\sigma_{\mathrm{S}}) \approx 0.5$, indicating non-classicality. Simultaneously, we monitor the time evolution of the atomic state by calculating the trace distance to its initial state, i.e.,
\begin{align}
    \Delta(t) := \norm{{\chi_{\ms{C}}-\sigma_{\ms{C}}(t)}}_1.
\end{align}
We observe that at the final time $\tau \approx 40$, we find $\Delta(\tau)=0$, indicating that the atom has returned to its initial state
\begin{equation}
    \chi_{\ms C} \approx \begin{pmatrix}
0.16 &0.36i  \\
0.36i & 0.84 \\
\end{pmatrix} 
\end{equation}
as required for catalysis.

As a second example, we demonstrate how Wigner negativity can be catalytically generated. Starting from a coherent state $\rho_{\ms{S}} = \dyad{\alpha}_{\ms{S}}$ which has a positive Wigner function, hence $\mathsf{W}(\rho_{\ms{S}})=0$, we aim to obtain a final state of the cavity $\sigma_{\ms{S}}$ with $\mathsf{W}(\sigma_{\ms{S}})>0$, therefore certifying its non-classicality. In Fig.~\ref{F-Wigner-WLN-nav}, we present an example of such an evolution. First, we plot WLN as a function of time $t$. At the final time $ \tau \approx 5$, we obtain a state $\sigma_{\ms{S}}$ that has $\mathsf{W}(\sigma_{\ms{S}}) \approx 0.1$, and we plot its Wigner function. To verify the catalytic nature of the evolution, we display the evolution of the atomic state via its trajectory in the Bloch sphere. Crucially, the trajectory is closed, as the initial and final state of the atom exactly coincide (red dot). In this example, the atom was prepared in the following state:
\begin{equation}
    \chi_{\ms C} \approx \begin{pmatrix}
0.154 & 6.84.10^{-4}i  \\
-6.84.10^{-4}i & 0.846 \\
\end{pmatrix} .
\end{equation}

\begin{figure*}
    \centering
    \includegraphics{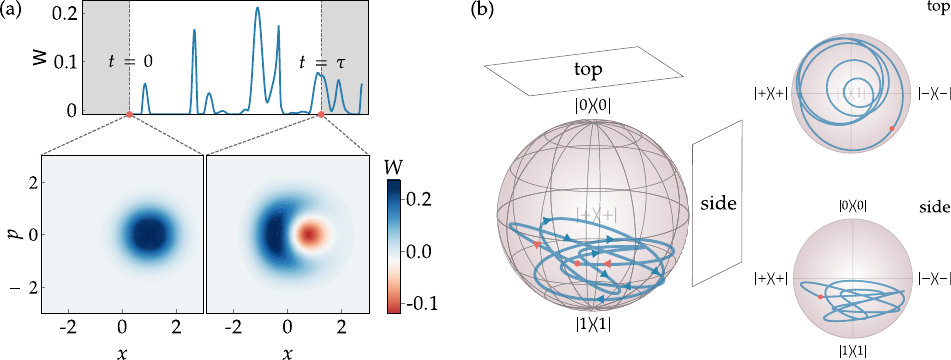}
    \caption{\emph{Second illustrative example.} Catalytic generation of Wigner negativity. Panel $(a)$ shows the time evolution of the Wigner logarithmic negativity $\mathsf W$, as well as the Wigner function at initial and final time $\tau$. Non-classicality is clearly generated, while the process is catalytic. Panel $(b)$ shows the trajectory of the atomic state, i.e., the catalyst $\ms{C}$, on the Bloch sphere. The initial and final state coincide (red dot). The values of the parameters are: $\alpha = 1/\sqrt{2}$, $\omega = 2\pi$ and $g = \pi$.}
    \label{F-Wigner-WLN-nav}
\end{figure*}

Note that the two examples we analysed are complementary. The first one demonstrates the generation of non-classicality as witnessed by the $g^{(2)}$ function, while the final state $\sigma_{\mathrm{S}}$ has a positive Wigner function. In the second example, the final state $\sigma_{\mathrm{S}}$ has a negative Wigner function, despite $g^{(2)}(\sigma) > 1$.

\section{Mechanism of catalysis}\label{sec:meachanism}

Here, we present a more intuitive understanding of quantum catalysis by identifying a mechanism that allows for the activation of non-classicality. We start by stating the following lemma:

\begin{lemma}[Higher order moments]\label{Lemma:higher-moments}
Let $O_{\ms{SC}} = O_{\ms{S}} \otimes \mathbbm{1}_{\ms{C}} + \mathbbm{1}_{\ms{S}} \otimes O_{\ms{C}}$ be an additive and conserved observable on the joint system $\ms{SC}$. In a catalytic protocol governed by a unitary operator $U$, such that $\sigma_{\ms{SC}} = U(\rho_{\ms{S}} \otimes \chi_{\ms{C}})U^{\dagger}$, where $\sigma_{\ms{C}} = \chi_{\ms{C}}$, and $[U, O_{\ms{SC}}] = 0$, the $k$-th moment of the observable $O$ on the system $\ms{S}$ satisfies:
\begin{equation}\label{Eq:momentOSC}
\langle O^k_{\ms{S}} \rangle_{\sigma}  = \langle O_{\ms{S}}^k\rangle_{\rho}  + \Tr[\Delta_k (\rho_{\ms{SC}}-\sigma_{\ms{SC}})],    
\end{equation}
where 
\begin{equation}\label{Eq:delta-k}
    \Delta_k = \sum_{i=1}^{k-1}\binom{k}{i}O^{k-i}_{\ms S}\otimes O^i_{\ms C}.
\end{equation}
\end{lemma}
\begin{proof}
Let us use the short  notation $O_{\ms S} = O_{\ms{S}}\otimes \mathbbm{1}_{\ms{C}}$ and $O_{\ms C} = \mathbbm{1}_{\ms{S}}\otimes O_{\ms C}$. Then, to calculate $\langle O^k_{\ms S} \rangle_{\sigma}$, we write $O_{\ms{SC}}$ as $O^k_{\ms{S}\ms{C}} = O^k_{\ms{S}} +  O^k_{\ms{C}} + \Delta_k$, where $\Delta_k = O^{k}_{\ms{S}\ms{C}} -O^k_{\ms{S}} - O^k_{\ms{C}}$ as given by Eq.~\eqref{Eq:delta-k}. Thus, the $k$th moment of the observable $O^k_{\ms{S}\ms{C}}$ in the state $\mathcal A$ is given by
\begin{equation}
    \langle O^k_{\ms{S}\ms{C}} \rangle_{A} = \tr(O^k_{\ms{S}} \mathcal A) + \tr(O^k_{\ms{C}} \mathcal A) + \tr(\Delta_k \mathcal A).
\end{equation}
Next, re-writing $O^k_{\ms{S}}$ as a function of $O_{\ms{S}\ms{C}}$, $O_{\ms{C}}$, and $\Delta_k$ allows one to express $\langle O^k_{\ms{S}} \rangle_{\sigma}$ as
\begin{align}\label{Eq:moment-calculation}
    \langle O^k_{\ms{S}} \rangle_{\sigma} &= \tr(O^k_{\ms{S}\ms{C}}\sigma) \!-\! \tr(O^k_{\ms{C}}\sigma) \!-\! \tr(\Delta_k \sigma) = \tr(O^k_{\ms{S}\ms{C}}\rho) - \tr(O^k_{\ms{C}} \rho) - \tr(\Delta_k \sigma) \nonumber \\
    &= \tr(O^k_{\ms{S}}\rho) + \tr(O^k_{\ms{C}}\rho) + \tr(\Delta_k \rho) -  \tr(\Delta_k \sigma) - \tr(O^{k}_{\ms{C}} \rho).
\end{align}
In the first line, we first used the fact that $[U, O_{\ms{S}\ms{C}}] = 0$ and that all moments of $\ms{C}$ are preserved. This allows one to replace $\sigma$'s to $\rho$'s. Second, we use Eq.~\eqref{Eq:momentOSC} to re-write $\tr(O^k_{\ms{S}\ms{C}}\rho)$. Finally, simplifying Eq.~\eqref{Eq:moment-calculation} gives
\begin{equation}
    \langle O^k_{\ms{S}} \rangle_{\sigma} = \langle O_{\ms{S}}^k\rangle_{\rho} + \Tr[\Delta_k (\rho_{\ms{SC}}-\sigma_{\ms{SC}})].       
\end{equation}
\end{proof}

In the Jaynes-Cummings model, the energies of the cavity $\ms{S}$ and the atom $\ms{C}$ are specified by local number operators $n_{\ms{S}}$ and $n_{\ms{C}}$, respectively. Hence, the total energy of both systems is proportional to the number of excitations, and described by a joint operator \mbox{$n_{\ms{SC}} := n_{\ms{S}} + n_{\ms{C}}$}. Since the JC evolution $U(t)$ that takes $\rho_{\ms{S}} \ot \chi_{\ms{C}}$ into $\sigma_{\ms{SC}}$ conserves the total energy, we have that $[U(t), n_{\ms{S}} + n_{\ms{C}}] = 0$ for all $t$. When the evolution is catalytic, all moments of $n_{\ms{C}}$ must remain unchanged, in particular $\langle n_{\ms{C}} \rangle_{\chi}$ = $\langle n_{\ms{C}} \rangle_{\sigma}$. Consequently, the first moment of $n_{\ms{S}}$ is also preserved, that is $\langle n_{\ms{S}} \rangle_{\rho} = \langle n_{\ms{S}} \rangle_{\sigma}$. Importantly,  this is not the case for higher moments of $n_{\ms{S}}$. Particularly, from Lemma.~\ref{Lemma:higher-moments}, the second moment satisfies
\begin{align}\label{eq:cat_variance_formula}
    \langle n^2_{\ms{S}} \rangle_{\sigma} = \langle n^2_{\ms{S}} \rangle_{\rho} + 2\Bigg(\langle n_{\ms{S}}\rangle_{\sigma} \langle n_{\ms{C}}\rangle_{\sigma} - \langle n_{\ms{S}}  \ot n_{\ms{C}}\rangle_{\sigma}\Bigg).
\end{align}
Hence, the second moment in the final state of the cavity, $\langle n^2_{\ms{S}} \rangle_{\sigma}$, can become smaller (or larger) than the second moment in the initial state $\langle n^2_{\ms{S}} \rangle_{\rho}$. This means that using a catalyst allows for modifying the distribution of the local observable $n_{\ms{S}}$ of the system $\ms{S}$: while its average must remain the same, the higher moments can change. Importantly, this can only happen if the system becomes correlated with the catalyst, i.e., $\langle n^2_{\ms{S}} \rangle_{\sigma} \neq \langle n^2_{\ms{S}} \rangle_{\rho}$ only if $\sigma \neq \sigma_{\ms{S}} \ot \sigma_{\ms{C}}$ [as seen from Eq. \eqref{eq:cat_variance_formula}]. Thus, these correlations are essential for observing quantum catalysis. Note that the above analysis also applies beyond the JC model to arbitrary observables and moments. The only requirement is the conservation of local observables and the catalytic constraint.

The above analysis will also serve as a basis for the characterisation of the parameter regime leading to catalysis. In particular, we derive a necessary and sufficient condition for satisfying the catalytic constraint of Eq.~\eqref{Eq:catalytic-constrain}. This, in turn, allows one to obtain an analytic expression for the second-order coherence $g^2$, which is based on the second moment $ \langle n^2_{\ms{S}} \rangle_{\sigma}$. To do so, let us consider an arbitrary initial state of the atom:
\begin{align}
    \chi_{\ms C}= q \ketbra{g}{g}+ r\ketbra{g}{e}+r^*\ketbra{e}{g}+[1-q] \ketbra{e}{e},
\end{align}
as well as a general initial state of the cavity $\rho_{\ms{S}} = \sum_{n,m}^{\infty} p_{n,m}\ketbra{n}{m}$ with $ p_n :=  p_{n,n}$. Combining Eq. \eqref{eq:cat_variance_formula} with the fact that mean energy of the cavity is conserved, i.e., $\mbox{$\langle n_{\ms{S}} \rangle_{\sigma} = \langle n_{\ms{S}} \rangle_{\rho}$}$, we arrive at
\begin{align}\label{eq:g2_conserved}
    \!g^{(2)}(\sigma_{\ms{S}}) = g^{(2)}(\rho_{\ms{S}}) - \frac{2}{\langle n_{\ms{S}}\rangle_{\rho}^2} \!\left[\langle n_{\ms{S}} \ot n_{\ms{C}}\rangle_{\sigma} \!- (1-q) \langle n_{\ms{S}} \rangle_{\rho}\right]\!,\!
\end{align}
where we have
\begin{align} \label{eq:corr_term}
    \!\!\!\langle n_{\ms{S}} \ot n_{\ms{C}}\rangle_{\sigma} \!=\! \sum_{n = 0}^{\infty} n \left[(1 - q) p_{n} c_n^2 + y_n + q p_{n+1} s_n^2 \right],
\end{align}
with $s_n \!:=\! \sin(gt\sqrt{n+1}), c_n \!:=\! \cos(g t \sqrt{n+1})$ and $y_{n} \!:=\! 2\operatorname{Im}(r p_{n+1,n})s_n c_n$. 

In order to satisfy the catalytic constraint, we obtain a set of equations for the components of the atomic state. Decomposing the diagonal term as \mbox{$q = q_{\text{inc}} +  q_{\text{coh}}$}, we get
 \begin{align}\label{Eq-population}
 q_{\text{inc}} = \frac{1}{Q} \sum_{n=0}^{\infty}p_{n} s_n^2, \qquad  
q_{\text{coh}} = \frac{1}{Q} \sum_{n=0}^{\infty}y_n,  
\end{align}
with $Q := \sum_{n=0}^{\infty}(p_{n}+p_{n+1})s_n^2$. Interestingly, $q_{\text{inc}}$ is specified by the occupations of $\rho_{\ms{S}}$, while $q_{\text{coh}}$ depends on its coherence in the Fock basis. Moreover, the off-diagonal term $r$ satisfies 
\begin{equation}\label{Eq-coherence}
    r = \frac{i (a_3 a^{*}_4+a^{*}_1a_4)}{|a_1|^2 - |a_3|^2}-\frac{i(a_3 a^{*}_2+a^{*}_1a_2)}{|a_1|^2 - |a_3|^2}q,
\end{equation}
with $a_i$ being auxiliary functions defined as
\begin{align}\label{eq:aterms}
\begin{split}
    a_1 &= \sum_{n=0}^{\infty} p_{n,n}c_{n-1}c_{n}- e^{-i\omega \tau}, \qquad
    a_3 = \sum_{n=0}^{\infty} p_{n,n+2} s_ns_{n+1},  \\
    a_2 &= \sum_{n=0}^{\infty}p_{n,n+1}s_n\qty[c_{n-1}+c_{n+1}], \hspace{10pt}
    a_4 = \sum_{n=0}^{\infty}p_{n,n+1}s_nc_{n+1}. 
\end{split}
\end{align}
For a detailed derivation of Eqs.~\eqref{eq:aterms} see Sections~\ref{Appsub:generalsolution} and~\ref{Appsub:set-catalytic-states}.

Importantly, Eqs~\eqref{Eq-population} and~\eqref{Eq-coherence} are necessary and sufficient for ensuring the catalytic constraint of Eq.~\eqref{Eq:catalytic-constrain}. In combination with Eq.~\eqref{eq:g2_conserved}, we can now characterise analytically the effect of non-classicality activation for the $g^2$ function in the catalytic regime. This is done for finding the parameters (initial states and interaction time) for the first illustrative example of a catalytic process (see Fig.~\ref{fig:2}). For the second  example, we use again such analytic results to ensure the validity of the catalytic constraint, while the Wigner functions are computed numerically.

\section{How general is catalysis?}

An interesting problem is to understand how typical the effect of catalysis is. Here we discuss different aspects of this question. 

To begin with, note that it is not obvious a priori whether the catalytic constraint of Eq. \eqref{Eq:catalytic-constrain} can be satisfied. However, due to the quantum version of the Perron-Frobenius theorem, every quantum channel has at least one positive semi-definite fixed point \cite{fannes1992finitely}. Now, for a fixed input state $\rho_{\ms{S}}$ and a fixed interaction time $\tau$, the state of the atom $\ms{C}$ evolves according to an effective quantum channel $\chi_{\ms{C}} \rightarrow \Tr_{\ms{S}}[U(\tau)(\rho_{\ms{S}}\ot \chi_{\ms{C}})U^{\dagger}(\tau)]$. Consequently, there always exists an initial state $\chi_{\ms{C}}$ which is left unchanged by this channel, hence providing at least one solution to Eq. \eqref{Eq:catalytic-constrain}.

Next, one might wonder how often a catalytic evolution leads to a non-classical state of the cavity. In particular, when preparing the cavity in a coherent state $\ket{\alpha}$, does there always exist a state of the catalyst $\chi_{\ms{C}}$ which allows to generate non-classicality? To address this question, we investigate the minimum value of $g^{(2)}$ of the final state $\sigma_{\ms{S}}$, as a function of $|\alpha|$ [see Fig.~~\hyperref[fig:1]{\ref{fig:1}a}]. This is done by combining Eqs.~(\ref{eq:g2_conserved}-\ref{eq:corr_term}) and (\ref{Eq-population}-\ref{Eq-coherence}), and imposing a bound on the final time, i.e $g\tau \leq 100$. For $\alpha \in (0, 2]$, we observe that $g^{(2)}(\sigma_{\ms S}) < 1$, indicating that non-classicality generation via catalysis is generic here. 
Note that when the initial coherent state has low energy, the final state of the cavity is close to the intial one, but with a slightly reduced variance, leading to a value of $g^{(2)}$ approaching zero.

\begin{figure}[t]
    \centering
    \includegraphics[width=10.358cm]{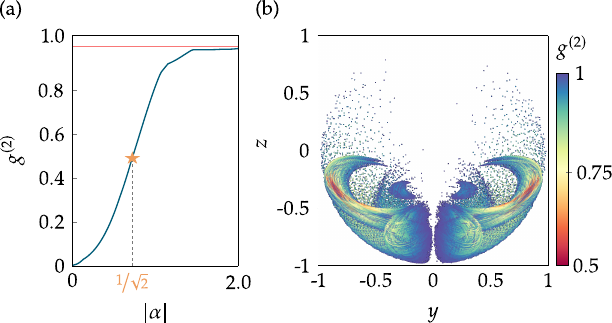}
    \caption{\emph{Which states lead to catalysis?} Panel (a) shows the minimal value of $g^2$ obtained as a function of the amplitude $|\alpha|$ of the initial coherent state of the cavity $\rho_{\ms{S}} = \dyad{\alpha}_{\ms{S}}$. The orange star corresponds to our first illustrative example. Panel (b) displays the atomic states (in the $y-z$ plane of the Bloch sphere) that satisfy the catalytic constraint and generate non-classical states, for an intial coherent state $\alpha=1/\sqrt{2}$. The colour represents different values of $g^{(2)}<1$. We impose a limit on the interaction time $\tau\leq 100$ and take $10^6$ samples. In both panels, parameters are $\omega = 2\pi$, $g = \pi$.}
    \label{fig:1}
\end{figure}

Although we have focused on the generation of states with Wigner negativity or sub-Poissonian distribution, catalysis also allows the generation of specific non-classical states. In this regard, we show in Sec.~\hyperref[App:generating-squeezed]{C-2} that one can also use our catalytic protocol to generate squeezed states of light. We further determine the possible final states of the cavity that can be obtained under a catalytic protocol.

Let us now ask the converse question, i.e., whether every atomic state can lead to a catalytic evolution. A first observation is that pure states cannot act as useful catalysts in general, and in particular cannot catalytically generate non-classicality. Indeed, a key ingredient for catalysis is the fact that the system and catalyst become correlated [see Eq. (\ref{eq:cat_variance_formula})], which is impossible when the state of the catalyst is pure. 

To further explore this question, we also investigated which states of the atom can catalytically generate sub-Poissonian statistics (i.e., $g^{(2)}<1$), given an initial coherent state of the cavity and a limited interaction time $g\tau \leq 100$. In Fig.~\hyperref[fig:1]{\ref{fig:1}b}, we display an example of such a set of catalytic states. Interestingly, this set appears to contain states that are almost pure. 
The structure of the catalytic set is highly non-trivial, and, in particular, we observe a strong dependence on the initial state of the cavity. This leads to the following question:

\begin{center}
    \emph{Which states of the atom allow to generate non-classicality?}
\end{center}

We refer to the set of all atomic states that satisfy Eq.~\eqref{Eq:catalytic-constrain} as the \emph{set of catalytic states}. Formally, one can define it as follows:
\begin{equation}
    \mathcal{C} := \big\{\omega_{\ms{C}} \,\big|\, \omega_{\ms{C}} = \Tr_{\ms{S}}[U(\tau)(\rho_{\ms{S}} \ot \omega_{\ms{C}})U^{\dagger}(\tau)],\, \tau \geq 0 \big\}.
\end{equation}

To explore the geometry of this set, we characterise the catalytic set for three initial coherent states of the cavity corresponding to $\alpha \in \{0.2, 1/\sqrt{2}, 25\}$. The results are presented in Fig.~\ref{Fig:set-of-catalytic-states}, where the catalytic set is shown in grey. Additionally, we determine which states in the catalytic set can generate non-classicality, focusing on the $g^{(2)}$ function. Similar to Fig.~\hyperref[fig:1]{\ref{fig:1}b}, we represent these states in color, with the latter indicating the level of non-classicality being generated. Interestingly, these sets vary significantly with $|\alpha|$. Moreover, when $|\alpha|$ is large, catalytic states are distributed close to the equatorial plane of the Bloch sphere, and we could find no instance where non-classicality is generated. This can be understood by noticing that the atom is ``too small'' to significantly perturb the field. Furthermore, when the initial state of the cavity is prepared as an incoherent mixture of Fock states, the set of catalytic states is diagonal and lies along the $z$ axis. 

\begin{figure*}
    \centering
    \includegraphics{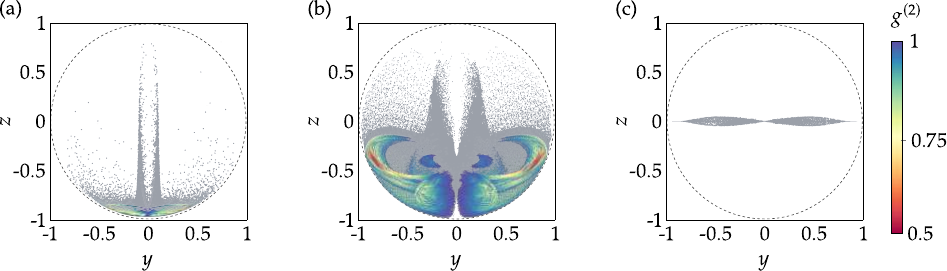}
    \caption{{\label{Fig:set-of-catalytic-states}\emph{Set of catalytic state for different initial state preparation}. Atomic states (projection of the $y-z$ plane of the Bloch ball) that satisfy the catalytic constraint (gray) for (a) $\alpha = 0.2$, (b) $\alpha = 1/\sqrt{2}$, and (c) $\alpha = 25$ are highlighted by a uniform gradient if they produce $g^2(\sigma_\ms{S}) < 1$ All panels were generated by imposing a limit on interaction time with $\tau \leq 100$ and taking $10^{6}$ samples. The parameters are $\omega = 2\pi$ and $g = \pi$.
    }}
\end{figure*}

Finally, one can ask what happens when the initial state of the cavity is not a coherent one. Let us start by examining Eqs.~\eqref{Eq:r_solutionr}~and~\eqref{Eq:q_solution} under the assumption that the initial state of the cavity is an incoherent mixture of Fock states, i.e., $\rho_{\ms{S}} = \sum_{n,m} p_{n,m}\dyad{n}{m}$ with $p_{n,m} = 0$ if $n \neq m$.  In this case, from Eqs. (\ref{Eq-population}) and (\ref{Eq-coherence}), we can infer that the only feasible states of the catalyst are those with $q_{\text{coh}} = 0$ and $r = 0$. Consequently, the atomic state is incoherent in the energy basis and its ground state occupation takes the form of
\begin{equation}\label{Eq:incoherent-catalyst}
    q = \frac{\sum_{n=0}^{\infty}p_n\sin^2(gt\sqrt{n+1})}{\sum_{n=0}^{\infty}(p_n+p_{n+1})\sin^2(gt\sqrt{n+1})} .
\end{equation}
This leads to the following theorem:
\begin{theorem}[Fock states under catalytic transformations] The non-classicality of a pure Fock state is non-decreasing under catalytic evolution 
\end{theorem}
\begin{proof}
    Proving that the second-order second-order coherence of a Fock state $\rho_{\ms{S}} = \dyad{k}_{\ms{S}}$, cannot decrease under a catalytic evolution is equivalent to show the quantity
\begin{align}
    g^{(2)}(\sigma_{\ms S}) - g^{(2)}(\rho_{\ms S}) &= \frac{\Tr[(a^{\dagger 2}a^2)\sigma_{\ms S}]}{\Tr[(a^{\dagger} a) \sigma_{\ms S}]} - \frac{\Tr[(a^{\dagger 2}a^2)\rho_{\ms S}]}{\Tr[(a^{\dagger} a) \rho_{\ms S}]} \nonumber \\ &= \frac{\Tr[(a^{\dagger 2}a^2)(\sigma_{\ms S}-\rho_{\ms S})]}{\Tr[(a^{\dagger} a) \rho_{\ms S}]}
\end{align}
is nonnegative. This will be accomplished by showing that the following inequality holds: 
\begin{equation}\label{Eq:inequality-g}
 \Tr[(a^{\dagger 2}a^2)(\sigma_{\ms S}-\rho_{\ms S})] \geq 0 .   
\end{equation}
In what follows, we omit the index $\ms{S}$ as well as any explicit reference to the variables' dependence on time. When the initial state of the cavity is a Fock state, then Eq.~\eqref{Eq:cavity-state} implies that the state of the cavity after the catalytic evolution is given by
\begin{align}\label{Eq:cavity-state-fock2}
\sigma = &\qty[(1-q)\cos^2\qty(gt\sqrt{k+1})+q\cos^2\qty(gt\sqrt{k})]\ketbra{k}{k} \nonumber \\ &+q\sin^2\qty(gt\sqrt{k})\ketbra{k\!-\!1}{k\!-\!1}+(1-q)\sin^2\qty(gt\sqrt{k+1})\ketbra{k\!+\!1}{k\!+\!1},
\end{align}
where $q$ is determined by Eq.~\eqref{Eq:incoherent-catalyst}, which for this particular case takes the form:
\begin{align}\label{Eq:q_Fock-states}
    q = \frac{\sin^2\qty(gt\sqrt{k+1})}{\sin^2\qty(gt\sqrt{k+1})+\sin^2\qty(gt\sqrt{k})} 
\end{align}
By substituting Eq.\eqref{Eq:q_Fock-states} into Eq.\eqref{Eq:cavity-state} and introducing the notation \mbox{$\psi = \frac{1}{2}(\ketbra{n-1}{n-1}+\ketbra{n+1}{n+1})$}, we obtain:
\begin{align}
    \sigma =& \frac{\sin^2\qty(gt\sqrt{k})\cos^2\qty(gt\sqrt{k+1})+\sin^2\qty(gt\sqrt{k+1})\cos^2\qty(gt\sqrt{k})}{\sin^2\qty(gt\sqrt{k+1})+\sin^2\qty(gt\sqrt{k})}\rho\nonumber\\&\hspace{3.5cm}+ \frac{2\sin^2\qty(gt\sqrt{k})\sin^2\qty(gt\sqrt{k+1})}{\sin^2\qty(gt\sqrt{k+1})+\sin^2\qty(gt\sqrt{k})}\psi.
\end{align}
Alternatively, we can express the above equation as \mbox{$\sigma = t\rho + (1-t)\psi$}, where
\begin{equation}
t = \frac{\sin^2\qty(gt\sqrt{k})\cos^2\qty(gt\sqrt{k+1})+\sin^2\qty(gt\sqrt{k+1})\cos^2\qty(gt\sqrt{k})}{\sin^2\qty(gt\sqrt{k+1})+\sin^2\qty(gt\sqrt{k})}.
\end{equation}
With these results at hand, we can manipulate Eq.~\eqref{Eq:inequality-g} to obtain
\begin{align}
    \Tr[(a^{\dagger 2}a^2)(\sigma-\rho)] &= (1-t)\Tr[(a^{\dagger 2}a^2)(\psi-\rho)] \nonumber\\&= \frac{(1-t)}{2}\Tr[(a^{\dagger 2}a^2)(\ketbra{k+1}{k+1}+\ketbra{k-1}{k-1})]\nonumber\\&\hspace{4cm}-(1-t)\Tr[(a^{\dagger 2}a^2)\ketbra{k}{k}] \nonumber \\
    &= \frac{(1-t)}{2}\qty[(k+1)k + (k-1)(k-2) - 2k(k-1)] \nonumber\\&= (1-t) \geq 0.    
\end{align}
Thus, we conclude that second-order coherence $g^{(2)}$ can only increase during a catalytic processes involving a pure Fock state, i.e.,
\begin{equation}
    g^{(2)}(\sigma) - g^{(2)}(\rho) \geq 0.
\end{equation}
\end{proof}

However, note that the above result is only valid for pure Fock states. For an incoherent mixture of Fock states, we observed that catalysis can further boost non-classicality (see Section~\hyperref[Appsub:incoherentmixtureg2]{C-2} for details).

\subsection{Generating squeezed states}\label{App:generating-squeezed}

The state of the cavity after the catalytic process is described by Eq.~\eqref{Eq:cavity-state}, with $q$ and $p$ given by Eqs.~\eqref{Eq:q_solution} and~\eqref{Eq:r_solutionr}. In the main text, we have discussed that both states with Wigner negativity and sub-Poissonian states can be produced catalytically, but the specific form of the states of light was not directly addressed. Here, we focus on this question and show that it is possible to prepare squeezed states in a catalytic way.

Defining the field quadratures by $X_1 = (a^{\dagger}+a)/\sqrt{2}$ and $X_2 = i(a^{\dagger}-a)/\sqrt{2}$, which satisfy the commutation relation $[X_1,X_2]=i$. The Heisenberg uncertainty principle is given by $\Delta X_1 \Delta X_2 \geq 1/2$, where $(\Delta X_1)^2 = \langle X_1^2\rangle - \langle X_1\rangle^2$ and similarly for $(\Delta X_2)^2$. For the coherent field state, the uncertainties are equal to $\Delta X^{(\alpha)}_1 = \Delta X^{(\alpha)}_2 = 1/\sqrt{2}$. A field state is called squeezed if the uncertainty of one of the quadratures is below the vacuum level, i.e., $\Delta X_1 < 1/\sqrt{2}$ or $\Delta X_2 < 1/\sqrt{2}$. Here, we consider the squeezing parameter
\begin{equation}
    \xi = \frac{\Delta X_1}{\Delta X^{(\alpha)}_1} = \sqrt{2} \Delta X_1
\end{equation}
and search for conditions where $\xi < 1$, which manifests field squeezing.

The catalytic generation of squeezed states can be observed in Fig~\ref{F:squeezing-fig}. The left plot depicts the minimal value of $\xi$ obtained as a function of the amplitude $|\alpha|$ of the initial coherent state of the cavity. We impose a limit on the interaction time $\tau \leq 100$ and observe that $\xi$ goes below the shot noise level, implying the presence of squeezing. In the right plot, we depict the evolution of $\xi$. Catalysis occurs at time $\tau = 18$ (represented by an orange star). Starting from the initial state of the cavity $\rho_{\ms S} = \ketbra{\alpha}_{\ms S}$, which has $\xi = 1$, we obtain a final state $\sigma_{\ms S}$ for which $\xi \approx 0.79$.

\begin{figure*}
    \centering
    \includegraphics{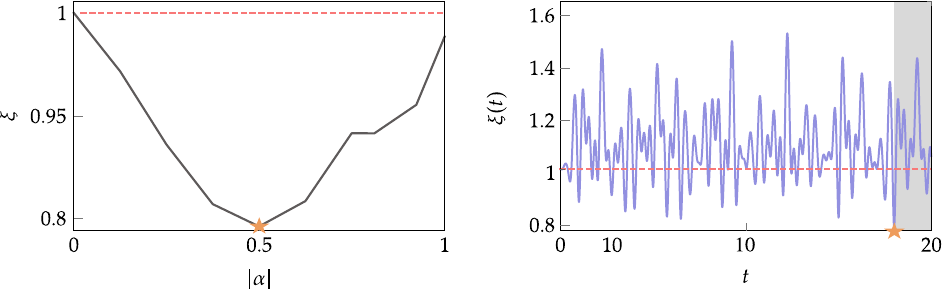}  
    \caption{\emph{Generating squeezed states via quantum catalysis.} The left panel shows the minimal value of $\xi$ obtained as a function of the amplitude $|\alpha|$ of the initial coherent state of the cavity within the interval of interaction time $\tau \leq [0,100]$. Taking the minimal value of $|\alpha|$, the right panel displays its time evolution, defined as $\xi(t) := \xi (\sigma_{\ms{S}}(t))$. The orange star indicates the final time ($\tau \approx 18$) at which the evolution becomes catalytic. At this point, the atom returns to its initial state, while a squeezed state with $\xi(\tau) \approx 0.79 < \xi(0) = 1$ has been produced. The values of the parameters are: $\alpha = 1/2$, $\omega = \pi$, and $g = 2\pi$.  \label{F:squeezing-fig} } 
\end{figure*}

\section{Catalysis in realistic scenarios}\label{sec:catalysis-realistic-scenarios}

\subsection{Catalysis in the presence of dissipation}

So far we have discussed an idealized (noise-free) scenario in quantum optics to understand the physics behind catalysis in quantum systems. However, any experiment implementing the Jaynes-Cummings model will necessarily feature cavity loss and atomic decay. We will now illustrate that the catalytic effect is robust even after incorporating these two effects in the model. To understand why our results remain valid in the presence of dissipation, it is worth noting that the Perron-Frobenius theorem, which assures that the catalytic constraint can be satisied, is valid for all quantum channels (hence also dissipative). Consequently, even in the presence of losses and noise, the catalytic constraint can still be satisfied. The interesting problem is whether non-classicality can still be generated in such a dissipative protocol. 

In order to understand the role of dissipation in our model let us assume that the cavity $\ms{S}$ is coupled to an environment characterized by temperature $T$ which accounts for the cavity losses. We also consider the possibility of atom decay due to photon emission. Consequently, the dynamics of the atom-cavity system are described by a Lindblad master equation of the form~\cite{Agarwal1986,Agarwal19862,Briegel1993,puri2001mathematical}:
\begin{equation}\label{Eq:rho-diss}
 \dot{\rho}_{\ms{SC}}\!=\!-i\left[H_{\ms{JC}},\rho_{\ms{SC}}\right]+ \kappa(n_{\ms{th}}+1)\mathcal{L}[a] +  \kappa n_{\ms{th}}\mathcal{L}[a^{\dagger}] + \Gamma \mathcal{L}[\sigma_-], 
\end{equation}
where as before we take $\rho_{\ms{SC}} := \dyad{\alpha}_{\ms{S}} \ot \chi_{\ms{C}}$, and $\kappa$ and $\Gamma$ are the cavity and atom dissipation rate, respectively. Moreover, $n_{\ms{th}}=(e^{1/T}-1)^{-1}$ is the average excitation number and \mbox{$\mathcal{L}[L] = L\rho_{\ms{SC}} L^{\dagger} -\frac{1}{2}\{L^{\dagger}L,\rho_{\ms{SC}} \}$} is the Lindblad dissipator with $\{.\,,\,.\}$ denoting the anticommutator. In the case of dissipative evolution \eqref{Eq:rho-diss} the (generalized) catalytic constraint becomes
\begin{equation} \label{eq:gen_catalytic_constraint}
    \sigma_{\ms{C}} := \Tr_{\ms{S}}[ \mathcal{E}_{\tau}(\rho_{\ms{SC}})] = \chi_{\ms{C}},
\end{equation}
where $\mathcal{E}_{\tau}(\cdot)$ is a quantum channel generated by the evolution from Eq.~\eqref{Eq:rho-diss} acting for some time $\tau$. 

For a given stopping time $\tau$ the solution to Eq.\eqref{Eq:rho-diss} that simultaneously satisfies the catalytic constraint can be determined numerically. 

In order to study the robustness of our results to dissipation we will now analyse two figures of merit, WLN and the auto-correlation function $g^{(2)}$, as a function of the stopping time $\tau$ for closed (reversible) and open (irreversible) quantum dynamics. At all stopping times $\tau$ we ensure that the generalized catalytic constraint from Eq. (\ref{eq:gen_catalytic_constraint}) is met. Our results are summarized in in Fig.~\ref{Fig:catalytic-power} which shows the amount of non-classicality that can be generated in the mode $\ms{S}$ for different choice of interaction times $\tau$. Our first observation is that a catalytic generation of non-classicality is possible even in the presence of dissipation. Secondly, the amount of non-classicality that can be generated in this way generally diminishes with longer interaction times. This is intuitive, as the longer we wait, the higher is the probability of losing photons from the cavity, especially since we are considering a constant dissipation rate for both the cavity and atom. Very interestingly, even for short evolution times, we can achieve relatively high levels of non-classicality as compared to the unitary case. 
\begin{marginfigure}[-8cm]
	\includegraphics[width=4.7618 cm]{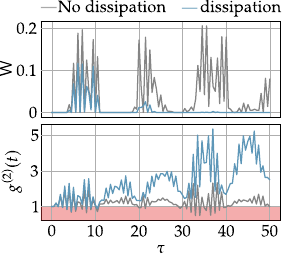}
	\caption{\label{Fig:catalytic-power}\emph{Catalysis in realistic scenarios.} Wigner logarithmic negativity and auto-correlation function are plotted as functions of the interaction time $ \tau $. The blue solid line represents the scenario with dissipation, while the gray solid line represents the scenario without dissipation. For each time $ \tau $, the catalytic constraint is satisfied, and the degree of non-classicality activated in the mode is computed. While $\mathsf W \geq 0$ indicates non-classical behaviour, $g^{(2)} < 1$ also witnesses non-classicality (reddish area in the right panel). Both plots use constant dissipation rates: an atom dissipation rate of $ \Gamma = 0.05$ and a cavity dissipation rate of $ \kappa =0.005$. The average number of thermal bath excitations is $ n_{\text{th}} = 0.1 $ and the other parameters are $\omega = 2\pi$, $g = 0.1\pi$ and $\alpha = 1/\sqrt{2}$.}
\end{marginfigure}

\subsection{Catalytic protocol in the presence of multiple cavities}

A relevant question one might have when investigating the catalytic protocol is whether the atom can be reused in a \emph{truly} catalytic way. That is, whether the catalytic protocol can be performed multiple times, while the catalyst is prepared only once. To see whether our protocol fulfills this requirement we consider a single atom $\ms{C}$ interacting sequentially with $N$ cavities $\ms{S} := \ms{S}_1 \ms{S}_2 \ldots \ms{S}_N$ with the same matching mode according to a catalytic process satisfying Eq. (\ref{Eq:catalytic-constrain}). More specifically, we imagine the following protocol:

\begin{enumerate}
\item For a given choice of parameters $\omega$, $g$, and $\tau$, determine the state of the atom $\chi_{\ms C}$ as discussed in Sec. \ref{sec:meachanism}.
\item Let the atom interact sequentially with $N$ cavities with the same frequency $\omega$ and interaction strenght $g$ and for the same amount of time $\tau$. This leads to the following unitary transformation of the composite system
\begin{equation}
    \sigma_{\ms{SC}} = \bigotimes_{i=1}^{N}[U_i\,(\rho_{\ms S_{i}} \otimes \chi_{\ms C})\,U^{\dagger}_i],
\end{equation}
where for all $i \in \{1, ..., N\}$, we have $\rho_{{\ms{S}_i}} = \ketbra{\alpha}{\alpha}_{\ms{S}_i}$ and $U_i$ is the unitary from Eq.~\eqref{Eq:time-evolution-operator} applied to $\ms{S}_i$ and $\ms{C}$.
\item After interacting with all the cavities the atom is removed from the last cavity ($\ms{S}_N$) and allowed to dissipate its energy to the environment. This removes the correlations built between the atom and the cavities. 
\end{enumerate}

After discarding the atom $\ms{C}$ the joint state of the cavities is given by $\sigma_{\ms{S}} = \Tr_{\ms{C}}[\sigma_{\ms{SC}}] = \sigma_{\ms{S}_1 \ms{S}_2 \ldots \ms{S}_n}$. Importantly, each of the states $\sigma_{\ms{S}_i}$ produced in this way is the same and non-classical, and can be further used in some information processing task. 

In certain applications one might be also interested in the amount of correlations developed between the cavities themselves. Such correlations could, in principle, limit the protocol's applicability to certain tasks which require producing multiple uncorrelated non-classical states. Luckily, the amount correlations developed between the cavities after discarding the atom is relatively low. To see this we will use the quantum fidelity $F(\sigma_{\ms{SC}}, \sigma_{\text{target}})$ where $\sigma_{\text{target}}$ is a product state $\sigma_{\text{target}} = \sigma_{\ms{S}_1} \ot \sigma_{\ms{S}_2} \otimes \ldots \otimes \sigma_{\ms{S}_N}$ and $\sigma_{\ms{S}_i} = \sigma_{\ms{S}_j}$ for $1\leq i, j \leq N$.  In Fig.~\ref{F-fidelity-multiple-cavities}, we show the resulting fidelity $F(\sigma_{\ms{SC}}, \sigma_{\text{target}})$ for up to $N = 5$ cavities. Notably, we observe that even for $N = 5$ the fidelity is relatively high, indicating tha the cavities are approximately uncorrelated after the catalytic protocol, and hence can be treated as independent. 
\begin{marginfigure}[-5.2cm]
    \centering    \includegraphics{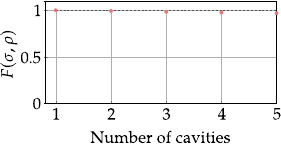}  
    \caption{\label{F-fidelity-multiple-cavities}\emph{Fidelity as a function of the number of cavities}. In the regime $g\tau \approx \pi $ (red dot), a high fidelity is observed, indicating the feasibility of employing the catalytic protocol to activate the non-classicality across five modes. 
} 
\end{marginfigure}

\section{Derivation of the results} \label{App:Jaynes-Cummings}

The first part of this Section provides the eigenproblem solution for the Jaynes-Cummings model (for further information, see references~\cite{Knight2005,Haroche2006}). In the following, we derive the final reduced states of the cavity $\ms{S}$ and the atom $\ms{C}$, and explicitly determine the set of atomic states that satisfy the catalytic constraint [Eq.~\eqref{Eq:catalytic-constrain}]. 

\subsection{Reduced states of subsystems}\label{Appsub:generalsolution}
The following unitary operator describes the dynamics of the Jaynes-Cummings model in resonance and under the rotating-wave approximation:
\begin{align}
   \hspace{-0.43cm} U(t) = &e^{i\omega t} \ketbra{0,g}{0,g} \nonumber\\&\!\!+\sum_{n=0}^{\infty}e^{-i\qty(n+\frac{1}{2})\omega t}  \Bigl\{\cos \frac{\mu_n t}{2}\Bigl(\ketbra{n+1,g}{n+1,g}+\ketbra{n,e}{n,e}\Bigl)\nonumber\\ &\!\!\hspace{2.1cm}-i\sin\frac{\mu_n t}{2}\Bigl(\ketbra{n+1,g}{n,e}+\ketbra{n,e}{n+1,g}\Bigl)\Bigl\}.
\end{align}
After the interaction, the joint system \mbox{$\sigma_{\ms SC}(t) = U(t)(\rho_{\ms S} \otimes \chi_{\ms C})U^{\dagger}$} becomes correlated. Let the cavity and the atom be prepared in general mixed states, i.e., $\rho_{\ms S} = \sum^{\infty}_{n,m=0}p_{n,m}\ketbra{n}{m}$ and $\chi_{\ms C} = q\ketbra{g}{g}+ r\ketbra{g}{e}+r^{*}\ketbra{e}{g}+(1-q) \ketbra{e}{e}$. Then, the reduced state of the cavity at time $t$ is obtained by taking the partial trace over the atom's degrees of freedom, i.e.,
$\sigma_\ms{S}(t) := \text{tr}_{\ms C}[{\sigma_{\ms SC}}(t)]$. More specifically,
\begin{align}\label{Eq:cavity-state}
\hspace{-0.25cm}\sigma_{\ms S}(t) = &qp_{0,0} \ketbra{0}{0}+ \Bigl\{\sum_{n=0}^{\infty} e^{i(n+1)\omega t}\qty[q p_{0,n+1}c_n+irp_{0,n}s_n]\ketbra{0}{n+1}+\text{h.c}\Bigl\}\nonumber \\&
\hspace{-0.15cm}+\sum_{n,m=0}^{\infty}e^{-i(n-m)\omega t}qp_{n+1,m+1}\qty[s_ns_m\ketbra{n}{m}+c_nc_m\ketbra{n\!+\!1}{m\!+\!1}] \nonumber\\
&\hspace{-0.15cm}+ \sum_{n,m=0}^{\infty}e^{-i(n-m)\omega t}(1-q)p_{n,m}\qty[c_nc_m\ketbra{n}{m}+s_ns_m\ketbra{n\!+\!1}{m\!+\!1}] \nonumber\\
 &\hspace{-0.15cm}+\Bigl\{\sum_{n,m=0}^{\infty} i e^{-i(n-m)\omega t}c_n s_m\qty[p_{n+1,m} r \ketbra{n\!+\!1}{m\!+\!1}+ p_{n,m+1}r^*\ketbra{n}{m}] \nonumber\\&\hspace{7.3cm}+\text{h.c}\Bigl\}.
\end{align}
The atomic state is obtained by marginalising over the photonic degrees of freedom, i.e., $\chi(t) := \text{tr}_{\ms S}[{\sigma_{\ms SC}}(t)]$, which leads to
\begin{align}\label{Eq:catalyst-state-t}
\chi_{\ms C}(t)= q(t)\ketbra{g}{g}+ r(t)\ketbra{g}{e}+r^*(t)\ketbra{e}{g}+[1-q(t)] \ketbra{e}{e},
\end{align}
Note that we identify $r:=r(0)$ and $q:=q(0)$ in what follows. The coefficients $q(t)$ and $r(t)$ are given by
\begin{align}
 q(t) &= q\sum_{n=0}^{\infty}p_{n}c^2_{n-1}+(1-q)p_ns^2_n+2\textrm{Re}[irp_{n+1,n}]s_n c_n ,\\
r(t) &=-ie^{i\omega t}\sum_{n=0}^{\infty}p_{n,n+1}s_nc_n+
    re^{i\omega t}\sum_{n=0}^{\infty}p_{n}c_{n-1}c_n+r^{*}\sum_{n=0}^{\infty}p_{n,n+2}s_ns_{n+1})\nonumber\\
    &\hspace{4cm}+ ie^{i\omega t}q\sum_{n=0}^{\infty}p_{n,n+1}s_n\qty[c_{n-1}+c_n].
\end{align}

\subsection{Catalytic constraint}\label{Appsub:set-catalytic-states}
Here we determine the set of atomic states that evolve catalytically by explicitly solving the catalytic constraint from Eq.~(\ref{Eq:catalytic-constrain}). More specifically, we are looking for the solution to the following operator equation:
\begin{align}\label{Eq:operator-equation}
    \chi_{\ms C}(\tau) = \Tr_{\ms{S}}\{U(\tau)[\rho_{\ms S} \otimes \chi_{\ms C}(\tau)] U(\tau)^{\dagger}\},
\end{align} 
for a fixed time $\tau$. From this point forward, we will abbreviate the diagonal elements of the state $\rho_{\ms{S}}$ as $ p_n := p_{n,n}$. To obtain the set of states that satisfy Eq.\eqref{Eq:operator-equation}, we first define the auxiliary functions:
\begin{align}\label{Eq:auxiliary-equation}
    \tilde{a}_1(t) &= e^{i\omega t}\sum_{n=0}^{\infty} p_n\cos(gt\sqrt{n})\cos(gt\sqrt{n+1}) -1,\nonumber \\
    \tilde{a}_2(t) &= ie^{i \omega t}\sum_{n=0}^{\infty}p_{n,n+1}\sin(gt\sqrt{n+1})\qty[\cos(gt\sqrt{n})+\cos(gt\sqrt{n+2})],\nonumber  \\
    \tilde{a}_3(t) &= e^{i\omega t}\sum_{n=0}^{\infty} p_{n,n+2} \sin(gt\sqrt{n+1})\sin(gt\sqrt{n+2}), \nonumber\\
    \tilde{a}_4(t) &= -ie^{i\omega t}\sum_{n=0}^{\infty}p_{n,n+1}\sin(gt\sqrt{n+1})\cos(gt\sqrt{n+2}). 
\end{align}
Note that $\tilde{a}_i = e^{i\omega t}a_i$. Next, we observe that Eq.~\eqref{Eq:operator-equation} gives rise to a set of two equations with two variables. By considering the ground state occupation $q(t)$, we find that the states satisfying Eq.~\eqref{Eq:operator-equation} are given by:
\begin{equation}\label{Eq:population_rawq}
    q(t) = \frac{\sum_{n=0}^{\infty}p_n\sin^2(gt\sqrt{n+1})+\operatorname{Re}\qty[ir(t)p_{n+1,n}]\sin(2gt\sqrt{n+1})}{\sum_{n=0}^{\infty}(p_n+p_{n+1})\sin^2(gt\sqrt{n+1})},
\end{equation}
whereas the coherence $r(t)$ obeys the equation 
\begin{equation}\label{Eq:p_and_q_relation}
    r(t) \tilde{a}_1(t) + q(t) \tilde{a}_2(t) + r^*(t) \tilde{a}_3(t) + \tilde{a}_4(t) = 0,
\end{equation}
whose solution is given by
\begin{equation}\label{Eq:r_solutionr}
    r(t) = \frac{\tilde{a}_3(t) \tilde{a}^{*}_4(t)-\tilde{a}^{*}_1(t)\tilde{a}_4(t)}{|\tilde{a}_1(t)|^2 - |\tilde{a}_3(t)|^2}+\frac{\tilde{a}_3(t) \tilde{a}^{*}_2(t)-\tilde{a}^{*}_1(t)\tilde{a}_2(t)}{|\tilde{a}_1(t)|^2 - |\tilde{a}_3(t)|^2}q(t).
\end{equation}
Substituting Eq.~\eqref{Eq:r_solutionr} into Eq.~\eqref{Eq:population_rawq}, we find that
\begin{equation}\label{Eq:q_solution}
    q(t) = \frac{\sum_{n=0}^{\infty}p_ns^2_n+2\operatorname{Re}\qty[i\qty(\frac{\tilde{a}_3(t) \tilde{a}^{*}_4(t)-\tilde{a}^{*}_1(t)\tilde{a}_4(t)}{|\tilde{a}_1(t)|^2 - |\tilde{a}_3(t)|^2})p_{n+1,n}]s_n c_n}{\sum_{n=0}^{\infty}(p_{n+1}+p_n)s^2_n-2\operatorname{Re}\qty[i\qty(\frac{\tilde{a}_3(t) \tilde{a}^{*}_2(t)-\tilde{a}^{*}_1(t)\tilde{a}_2(t)}{|\tilde{a}_1(t)|^2 - |\tilde{a}_3(t)|^2})p_{n+1,n}]s_n c_n}.
\end{equation}
Therefore, for a given value of $g$ and time $\tau$, Eqs.~\eqref{Eq:r_solutionr}~and~\eqref{Eq:q_solution} uniquely determine a state of the catalyst.



\subsection{Second moment of photon statistics in the catalytic Jaynes-Cummings evolution}\label{App:higher-moments-of-observables}

Here, we use Lemma.~\eqref{Lemma:higher-moments} to obtain an explicit expression for the second moment in a catalytic evolution as specified by the Jaynes-Cummings Hamiltonian. We start by considering the second moment of photon statistics
\begin{align}
    \langle n^2_{\ms{S}} \rangle_{\sigma} = \langle n^2_{\ms{S}} \rangle_{\rho} + 2\left[ (1-q)\,\langle n_{\ms{S}} \rangle_{\rho}  - \langle n_{\ms{S}} \ot \dyad{e}_{\ms{C}}\rangle_{\sigma}\right],
\end{align}
where $q := \langle e |\chi_{\ms{C}}|e \rangle$ is the excited-state occupation of the catalyst,  $\sigma = U(\rho_{\ms{S}} \ot \chi_{\ms{C}})U^{\dagger}$ and $\Tr_{\ms{S}}[\sigma] = \chi_{\ms{C}}$. Let us focus on the following term:
\begin{align}\label{Eq:term-second-moment}
     \langle n_{\ms{S}} \ot \dyad{e}_{\ms{C}}\rangle_{\sigma} &= \Tr\left[U^{\dagger}(n_{\ms{S}} \ot \dyad{e}_{\ms{C}})U (\rho_{\ms{S}} \ot \chi_{\ms{C}})\right] \nonumber \\ &= \sum_{k=0}^{\infty} k \Tr[U^{\dagger}\dyad{k, e} U (\rho_{\ms{S}} \ot \chi_{})],  
\end{align}
Using Eq.~(\ref{Eq:time-evolution-operator}), we can write
\begin{align}
    U^{\dagger} \ket{k, e} &= e^{i(k+\frac{1}{2})\chi t} \left(c_k \ket{k,e} - i s_k \ket{k+1, g}\right),
\end{align}
Substituting the above result into Eq.~\eqref{Eq:term-second-moment} leads to
\begin{align}
    \langle n_{\ms{S}} \ot \dyad{e}_{\ms{C}}\rangle_{\sigma} = \sum_{k = 0}^{\infty} k \bigl[(1-q)c_k^2 p_{k,k} &+ 2 s_k c_k \text{Im}(p_{k+1,k} r) \nonumber \\ &+ (1-q) s_{k}^2 p_{k+1,k+1} \bigl].
\end{align}

\subsection{Boosting non-classicality via a catalytic process for incoherent mixtures of Fock states}\label{Appsub:incoherentmixtureg2}
For an incoherent mixture of Fock states \mbox{$\rho  = \sum_{n} p_n \ketbra{n}{n}$}, the second-order coherence is given by
\begin{equation}\label{Eq-app-second-order}
    g^{(2)}(\rho) = \frac{\sum_{n=0}^{\infty}n(n-1) p_n}{(\sum_{n=0}^{\infty}n p_n)^2}.
\end{equation}
Assuming that the initial state of the cavity is prepared in a state $\rho_{\ms S} = \frac{1}{4}\ketbra{0}{0} +\frac{3}{4}\ketbra{2}{2})$, then its second-order coherence is $g^{(2)}(\rho_{\ms S}) = 2/3$. According to Eq.~\eqref{Eq:cavity-state}, the state of the cavity after the catalytic protocol (at time $t = \tau$) takes the form of 
\begin{align}
    \sigma_{\ms S} =& \frac{1}{4}\qty[q\!+\!(1\!-\!q)\cos^2 g t]\! \ketbra{0}{0}\!+\!\frac{3q}{4}\qty[\sin^2 \qty(g\tau\sqrt{2})\!+\!\frac{(1-q)}{4}\sin^2 g t]\!\ketbra{1}{1}\nonumber\\ &+\frac{3}{4}\qty[q\cos^2 \qty(g\tau\sqrt{2})+(1-q)\cos^2 \qty(g\tau\sqrt{3})]\ketbra{2}{2}\nonumber\\&+\frac{3(1-q)}{4}\sin^2 \qty(g\tau\sqrt{3})\ketbra{3}{3}
\end{align}
where $q$ is determined by Eq.~\eqref{Eq:incoherent-catalyst}:
\begin{equation}
    q = \frac{\sin^2\qty(g\tau)+3\sin^2\qty(g\tau\sqrt{3})}{\sin^2\qty(g\tau)+3\qty[\sin^2\qty(g\tau\sqrt{3})+\sin^2\qty(g\tau\sqrt{2})]}.
\end{equation}
Using Eq.~\eqref{Eq-app-second-order}, the second-order coherence for the final state is given by
\begin{equation}
 g^{(2)}(\sigma_{\ms S}) = \frac{2}{3}\qty{q\cos^2 \qty(g\tau\sqrt{2})+(1-q)\qty[1+2\sin^2 \qty(g\tau\sqrt{3})]}.
\end{equation}
Therefore, for $g\tau = 7.5\pi$, we obtain $g^{(2)} \approx 0.505$, indicating that non-classicality in the mode was catalytically increased. 

\section{Concluding remarks}

We presented a catalytic process for generating non-classical states of light in an optical cavity via interaction with an atom. Our results are valid for any coupling regime in which the Jaynes-Cummings model is applicable and also for realistic scenarios where the presence of noise and loss cannot be neglected. This fundamentally stems from the fact that a catalytic state of the atom exists for any coupling strength and even under the presence of dissipation. This work shows that quantum catalysis, a concept so far explored in the abstract framework of resource theories, is relevant in a practical context (see also Ref.~\cite{YungerHalpern2017}, which discuss the need to turn resource-theory notions into real-world applications). Furthermore, our protocol could potentially be implemented in state-of-the-art experimental setups \cite{frisk2019ultrastrong,periwal2021programmable,wenniger2022coherencepowered}, e.g. in cavity QED \cite{Ballester2012,Blais2021} or trapped ions \cite{Leibfried2003,Lv2018}. Beyond proof-of-principle experiments, it would also be interesting to investigate whether such a catalytic protocol offers a practical advantage. Indeed, the key point of catalysis is that the catalyst is returned exactly in the same state as it was initially prepared. Hence the same atom could in principle be used repeatedly for activating non-classicality in different cavities, or in the same cavity, but at different times. It would be interesting to uncover further instances of quantum catalysis in realistic setups, e.g. exploring other platforms and models in quantum optics, as well as for applications in quantum information and metrology.

Another relevant aspect of our work is that the catalyst (atom) is, in general, a coherent quantum system. Specifically, to generate non-classicality, the atom's state must exhibit coherence in the energy basis over the course of its evolution. This constitutes an example of coherent quantum catalysis~\cite{Lipka_Bartosik_2023}, which contrasts with most previous examples of quantum catalysis. A deeper understanding of the role of coherence in catalysis is an exciting future direction.  

Finally, our analysis also shed light on the key role of correlations in quantum catalysis. This represents a distinctive feature of quantum catalysis. It would be interesting to see if this effect can occur in different physical models.

\chapter{Conclusions}\label{C:conclusions}

Quantum thermodynamics arose from a desire to formalise and generalise the concepts of macroscopic thermodynamics to the quantum realm. Its emergence evolved naturally with our increasing ability to manipulate and control systems at finer scales. This led theoreticians to probe the intricate dynamics of small-scale systems, asking how quantum phenomena, such as entanglement and coherence, might affect classical thermodynamic formulations. However, theoretical inquiries were not the sole catalyst for the advent of this field. As we have witnessed the race towards so-called "quantum technologies", understanding the path to optimality brings us closer to the idea that one day we might develop devices that outperform current ones by using purely quantum effects. Since then, joint theoretical efforts from quantum optics, statistical mechanics, and quantum information theory are bundled to pave the way for a new era of quantum engineering ahead.

This thesis asks simple questions that naturally arise when we attempt to probe the quantum realm thermodynamically. By relying on a model-independent approach and using minimal assumptions, our main goal consisted of characterising thermodynamic transformations across different regimes and identifying optimal protocols. To ask and answer these questions, we lay out a powerful theoretical toolkit that provides a robust approach to studying the thermodynamics of small systems: the resource-theoretic approach. The main mathematical tools that consistently appear within this framework were presented in Chapter~\ref{C:mathematical_preliminaries}, while the framework itself was detailed in Chapter~\ref{C:resource_theory_of_thermodynamics}.

Chapter~\ref{C:thermal_cones} analysed the structure of thermal cones. Specifically, for a $d$-dimensional classical state of a system interacting with a heat bath, we found explicit construction of the past thermal cone and the incomparable region. Moreover, we provided a detailed analysis of their behaviour based on thermodynamic monotones given by the volumes of thermal cones. Then, we discussed the applicability of these results to other majorisation-based resource theories (such as that of entanglement and coherence), since the partial ordering describing allowed state transformations is then the opposite of the thermodynamic order in the infinite temperature limit. Finally, we also generalised the construction of thermal cones to account for probabilistic transformations and for a coherent qubit. 

In Chapter~\ref{C:memory-MTP}, we developed a resource-theoretic framework that allowed one to bridge the gap between two approaches to quantum thermodynamics based on Markovian thermal processes. Our approach was built on the notion of memory-assisted Markovian thermal processes, where memoryless thermodynamic processes were promoted to non-Markovianity by explicitly modelling ancillary memory systems initialised in thermal equilibrium states. Within this setting, we proposed a family of protocols composed of sequences of elementary two-level thermalisations that approximated all transitions between energy-incoherent states accessible via thermal operations. We proved that, as the size of the memory increased, these approximations became arbitrarily good for all transitions in the infinite temperature limit, and for a subset of transitions in the finite temperature regime. Furthermore, we presented solid numerical evidence for the convergence of our protocol to any transition at finite temperatures. We also explained how our framework could be used to quantify the role played by memory effects in thermodynamic protocols such as work extraction. Finally, our results showed that elementary control over two energy levels at a given time was sufficient to generate all energy-incoherent transitions accessible via thermal operations if one allowed for ancillary thermal systems.

In Chapter~\ref{C:finite-size}, we derived a fluctuation-dissipation theorem version within a resource-theoretic framework, where one investigates optimal quantum state transitions under thermodynamic constraints. More precisely, we first characterised optimal thermodynamic distillation processes, and then proved a relation between the amount of free energy dissipated in such processes and the free energy fluctuations of the initial state of the system. Our results applied to initial states given by either asymptotically many identical pure systems or an arbitrary number of independent energy-incoherent systems, and allowed not only for a state transformation, but also for the change of Hamiltonian. The fluctuation-dissipation relations we derived enabled us to find the optimal performance of thermodynamic protocols such as work extraction, information erasure, and thermodynamically-free communication, up to second-order asymptotics in the number $N$ of processed systems. We thus provided a first rigorous analysis of these thermodynamic protocols for quantum states with coherence between different energy eigenstates in the intermediate regime of large but finite $N$.

Finally, in the last chapter of this thesis, we went beyond thermodynamics and presented a catalytic process in a paradigmatic quantum optics setup, namely the Jaynes-Cummings model, where an atom interacted with an optical cavity. The atom played the role of the catalyst, and allowed for the deterministic generation of non-classical light in the cavity. Considering a cavity prepared in a "classical'' coherent state, and choosing appropriately the atomic state and the interaction time, we obtained an evolution with the following properties. First, the state of the cavity had been modified, and now featured non-classicality, as witnessed by sub-Poissonian statistics or Wigner negativity. Second, the process was catalytic, in the sense that the atom was deterministically returned to its initial state exactly, and could then in principle be re-used multiple times. We investigated the mechanism of this catalytic process, in particular highlighting the key role of correlations and quantum coherence.

}

\backmatter 
\setchapterstyle{plain} 


\printbibliography[heading=bibintoc, title=Bibliography] 

\end{document}